\documentclass[a4paper,11pt, twoside]{book}

\usepackage[T1]{fontenc}
\usepackage[utf8]{inputenc}
\usepackage{amsfonts,amsmath}
\usepackage{amsthm}
\usepackage{lmodern}
\usepackage[hidelinks]{hyperref}
\usepackage{bookmark}
\usepackage{color}
\usepackage{blindtext}  
\usepackage{fancyhdr}
\usepackage{titlesec}
\usepackage{etoolbox}

\patchcmd{\chapter}{\thispagestyle{plain}}{\thispagestyle{empty}}{}{}   

\usepackage[figuresright]{rotating}
\usepackage[font=small,labelfont=bf]{caption}
\usepackage{subcaption}
\graphicspath{{figs/}}
\usepackage{appendix}

\usepackage[table, dvipsnames]{xcolor}
\usepackage{multirow}

\usepackage{calc}

\usepackage{enumitem}       % para modificar la indentación de itemize

\usepackage{verbatim}
\usepackage{fancyvrb}

\usepackage{setspace}

\usepackage{graphicx, epsfig}

\usepackage{indentfirst}        % indent first paragraph

\usepackage[nottoc,numbib]{tocbibind}
\usepackage{emptypage}
\usepackage{mathtools,amssymb}
\usepackage{cite}
\usepackage[capitalise]{cleveref}
\usepackage[ruled,linesnumbered]{algorithm2e}
\let\forallsymb\forall
\usepackage{booktabs}
\usepackage{tikz, pgfplots}
\pgfplotsset{compat=1.17}
\pgfplotstableset{col sep=semicolon}
\usepackage{moresize} % provide more sizes for pgfplot
\usepgfplotslibrary{groupplots}
\usepgfplotslibrary{external}
\tikzset{every mark/.append style={scale=1.75, solid}, font=\small}  %font=\small
\pgfplotsset{
    width=1\textwidth,
    legend style={
        font=\small ,  %\scriptsize,  %\ssmall,
        inner xsep=1pt,
        inner ysep=1pt,
        nodes={inner sep=1pt}},
    legend cell align=left,
    every axis/.append style={line width=0.5pt},
 	every axis plot/.append style={line width=1.5pt},
 	every axis y label/.append style={yshift=-4pt}
}

\input{mysymbols.sty}

\titleformat{\chapter}[display] {\bfseries\Large}
                                {\filleft \Huge Chapter \thechapter}
                                {4ex}
                                {\titlerule \vspace{2ex} \filright}
                                [\vspace{2ex} \titlerule]

\titleformat{\section}[block] {\normalfont\bfseries}{\thesection}{1em}{}
\titleformat{\subsection}[block] {\normalfont\bfseries}{\thesubsection}{1em}{}
\titleformat{\subsubsection}[block] {\normalfont\bfseries}{\thesubsubsection}{1em}{}

\usepackage[acronym,nonumberlist,nopostdot]{glossaries}
\makeglossaries
\newacronym{gsp}{GSP}{graph signal processing}
\newacronym{gso}{GSO}{graph-shift operator}
\newacronym{gft}{GFT}{graph Fourier Transform}
\newacronym{gfs}{GFs}{graph filters}
\newacronym{nns}{NNs}{neural networks}
\newacronym{gnns}{GNNs}{graph neural networks}
\newacronym{iid}{i.i.d.}{independent identically distributed}
\newacronym{sem}{SEM}{structural equation model}
\newacronym{tls}{TLS}{total least squares}
\newacronym{agss}{AGSS}{aggregation sampling scheme}
\newacronym{dsgs}{DSGS}{diffused sparse graph signal}
\newacronym{bgs}{BGS}{bandlimited graph signals}
\newacronym{gcnn}{GCNN}{graph convolutional NN}

\newacronym{ls}{LS}{least-squares}
\newacronym{gmrf}{GMRF}{Gaussian Markov random field}
\newacronym{ml}{ML}{maximum likelihood}
\newacronym{er}{ER}{Erd\H{o}s-Rényi}
\newacronym{sbm}{SBM}{stochastic block model}
\newacronym{lv}{$\mathrm{LV}$}{local variation}

\newacronym{sgd}{SGD}{stochastic gradient descent}
\newacronym{gcg}{GCG}{graph-convolutional generator}
\newacronym{gdec}{GDec}{graph decoder}
\newacronym{nmse}{NMSE}{normalized mean square error}
\newacronym{sw}{SW}{small-world}

\newacronym{mse}{MSE}{mean square error}
\newacronym{zmuvg}{ZMUVG}{zero-mean unit variance Gaussian}
\newacronym{ss}{SS}{selection sampling}

\newacronym{ar}{AR}{autoregressive}
\newacronym{scp}{SCP}{sequential convex programming}
\newacronym{mm}{MM}{majorization-minimization}
\newacronym{ma}{MA}{moving-average}

\newacronym{gl}{GL}{graphical Lasso}
\newacronym{lvgl}{LVGL}{latent variable graphical Lasso}
\newacronym{rbf}{RBF}{radial basis function}

\def \sqJacob {\boldsymbol{\mathcal{X}}}
\def \barsqJacob {\bar{\boldsymbol{\mathcal{X}}}}
\def \ccalbA {\boldsymbol{\mathcal{A}}}
\def \ccaltbA {\boldsymbol{\mathcal{\tilde{A}}}}
\def \ccalbH {\boldsymbol{\mathcal{H}}}
\def \ccaltbH {\boldsymbol{\mathcal{\tilde{H}}}}
\def \ccalbU {\boldsymbol{\mathcal{U}}}
\def \ccalbC {\boldsymbol{\mathcal{C}}}
\def \ccaltbC {\boldsymbol{\mathcal{\tilde{C}}}}
\def \relu {\mathrm{ReLU}}
\def \vvec {\mathrm{vec}}
\def \EE{{\mathbb E}}
\def \bchkH{{\ensuremath{\mathbf{\check H}}}}

\def \bchks{{\ensuremath{\mathbf{\check s}}}}

\def \bbCo {\bbC_{{\scriptscriptstyle \ccalO}}}

\def \bbCoh {\bbC_{{\scriptscriptstyle \ccalO\ccalH}}}
\def \bbCho {\bbC_{{\scriptscriptstyle \ccalH\ccalO}}}

\def \hbCo {\hbC_{{\scriptscriptstyle \ccalO}}}

\def \bbXo {\bbX_{{\scriptscriptstyle \ccalO}}}
\def \bbXh {\bbX_{{\scriptscriptstyle \ccalH}}}

\def \bbSo {\bbS_{{\scriptscriptstyle \ccalO}}}
\def \bbSh {\bbS_{{\scriptscriptstyle \ccalH}}}
\def \bbSoh {\bbS_{{\scriptscriptstyle \ccalO\ccalH}}}
\def \bbSho {\bbS_{{\scriptscriptstyle \ccalH\ccalO}}}

\def \hbSo {\hbS_{{\scriptscriptstyle \ccalO}}}

\newcommand{\norm}[1]{\left\lVert#1\right\rVert}

\def \mytitle{Robust Network Topology Inference and Processing of Graph Signals}
\def \author{Samuel Rey Escudero}
\def \director{Antonio García Marqués}
\def \yearthesis{2022}

\begin{document}
%%%%%%%%%%%%%%% Portada %%%%%%%%%%%%%%%%%%%%
\frontmatter  
\hypersetup{pageanchor=false}
\begin{titlepage}
\begin{center}
    \includegraphics[width=0.5\columnwidth]{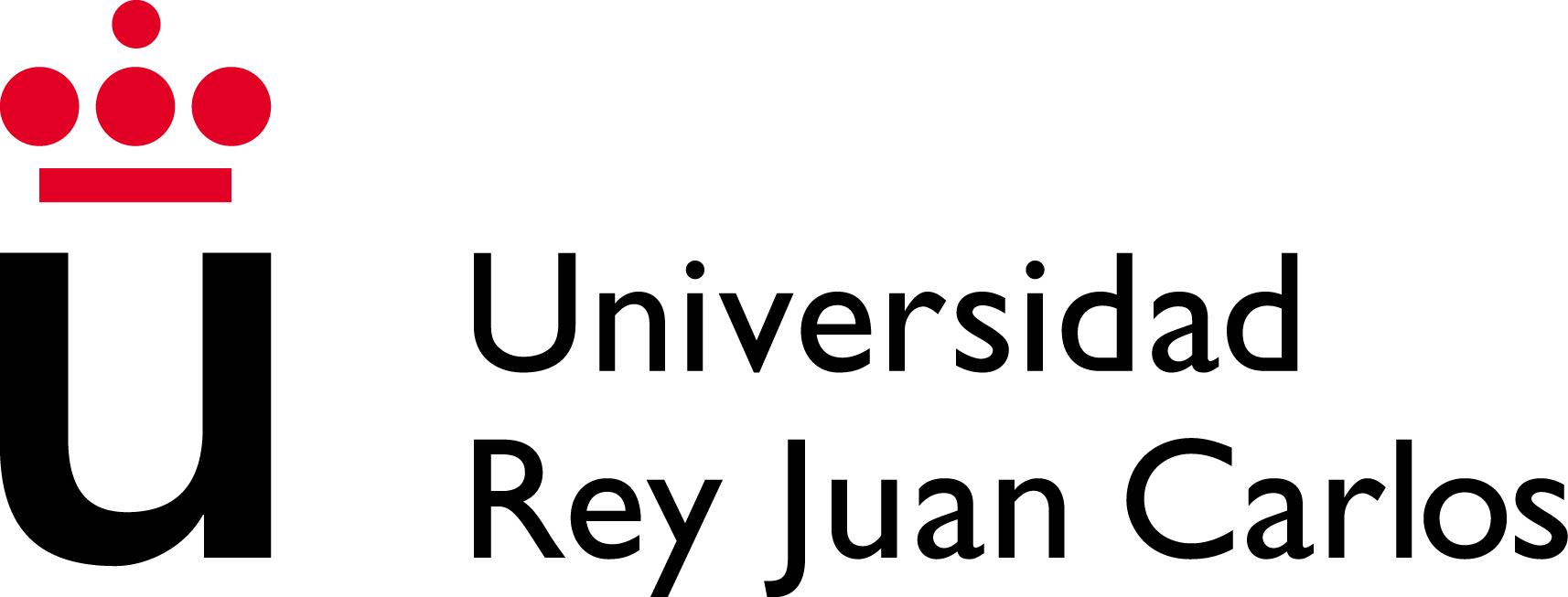}
    
    \vspace{0.08\textheight}
    {\Huge \bf TESIS DOCTORAL}
    \vspace{0.08\textheight}
    
    \begin{minipage}[c]{\textwidth}
        \centering  
        \setstretch{1.2}
        \Huge{\textit{\mytitle}} 
    \end{minipage}
    
    \vspace{0.08\textheight}
     \begin{minipage}[c]{\textwidth}
     \bf
        \centering  
        \setstretch{1.2}
        \Large{Autor: \\ {\sc \author} \\ 
        \vspace{0.05\textheight}
        Director: \\ {\sc \director} \vspace{0.01\textheight}}
    \end{minipage}

    \vfill
    \bf
    \Large{Programa de Doctorado Interuniversitario \\en Multimedia y Comunicaciones} \\
    \vspace{0.03\textheight}
    \Large{Escuela Internacional de Doctorado}\\
    \vspace{0.03\textheight}
    \LARGE{\yearthesis}

    \vspace{0.04\textheight}

\end{center}
\end{titlepage}

%\pagenumbering{Roman} % para comenzar la numeración de paginas en números romanos
\hypersetup{pageanchor=true}

%%%%%%%%%%%%%%% Índices %%%%%%%%%%%%%%%%%%%%
\tableofcontents
\clearpage
% \listoffigures % índice de figuras
% \clearpage

%%%%%%%%%%%%%%% Preamble %%%%%%%%%%%%%%%%%%%%
\chapter{Resumen}
% Motivación grafos
En los últimos años, la creciente presencia de sistemas vastos y heterogéneos está causando que los datos sean cada vez más abundantes y con estructuras más complejas, motivando así el rápido desarrollo de nuevos modelos y herramientas capaces de lidiar con el dominio irregular (no Euclídeo) donde están definidas las señales.
Un método particularmente interesante consiste en representar la estructura subyacente de las señales mediante un \emph{grafo} e interpretar las señales como \emph{señales definidas en el grafo}.
Precisamente, este es el mecanismo empleado en el procesado de señales definidas en grafos o \emph{graph signal processing (GSP)}, un área relativamente nueva que encapsula la topología del grafo en una matriz conocida como el \emph{graph-shift operator (GSO)}, y aprovecha la relación entre las señales y el grafo para desarrollar herramientas que lidien eficazmente con el dominio irregular de las señales.

% Motivación perturbaciones
Además de la estructura irregular de las señales, otra limitación fundamental es que los datos observados son propensos a presentar imperfecciones o perturbaciones, muchas veces inherentes al propio proceso de recolección de los mismos.
La naturaleza de estas perturbaciones es muy variada y, si se ignoran, pueden perjudicar significativamente el rendimiento de los algoritmos que utilicen los datos perturbados.
En GSP, las perturbaciones pueden clasificarse atendiendo a si afectan a las señales observadas o a la topología del grafo, y la atención recibida en trabajos previos depende, entre otros factores, del tipo de perturbación y de la complejidad resultante de su modelado.

\section*{Antecedentes}
% Trabajos con perturbaciones en la señal
El desarrollo de métodos robustos a perturbaciones en las señales suele desembocar en problemas tratables que, además, han sido estudiados en otras áreas relacionadas con GSP, propulsando la aparición de una considerable cantidad de trabajos.
En este ámbito, la \emph{presencia de ruido} en las señales constituye un problema omnipresente y ampliamente tratado en la literatura que puede combatirse eficazmente mediante métodos que separen la señal del ruido, una tarea conocida como \emph{signal denoising}.
Existe una gran variedad de métodos tradicionales basados en problemas de optimización convexa~\cite{chen2014signal,wang2015trend,pang2017graph,onuki2016graph}, y también alternativas más recientes enfocadas en modelos no lineales~\cite{tay2020time,do2020graph,chen2021graph}. 
Otro tipo de perturbación relevante en las señales es la presencia de \emph{valores perdidos} o \emph{missing values}, lo que implica que solo un subconjunto de los valores de la señal son observados.
Esto hace que sea necesario reconstruir la señal original, un problema ampliamente tratado en trabajos de muestreo e interpolación bajo el supuesto de que las observaciones pertenecen a nodos  distintos~\cite{anis2014towards,chen2015discrete,anis2016efficient,puy2018random}, o se corresponden a agregaciones sucesivas de los valores de la señal en nodos vecinos~\cite{marques2015sampling,chen2016signal}.

% Trabajos con perturbaciones en el grafo
Por otro lado, las perturbaciones en la topología del grafo originan problemas particularmente desafiantes y cuentan con menos trabajos previos en otras áreas, siendo estos algunos de los motivos por los que han sido menos estudiadas en la literatura de GSP.
En esta clase de perturbaciones destacan los errores en los enlaces del grafo, lo que se refleja en errores en el GSO y perjudica gravemente a la mayoría de las herramientas dentro de GSP basadas en el espectro o en polinomios del GSO.
Inicialmente,~\cite{ceci2020graph,miettinen2019modelling} modelaron la influencia de estas perturbaciones bajo algunas suposiciones simplificadoras y, posteriormente,~\cite{ceci2020_semtls,natali2020topology} introdujeron métodos que tenían en cuenta las imperfecciones en los enlaces para la estimación de filtros definidos en grafos o \emph{graph filters (GFs)}.
Por último, en el contexto de inferencia de topología, la presencia de \emph{nodos ocultos} supone una perturbación crítica.
En este escenario solo se dispone de observaciones de un subconjunto de los nodos y, aunque típicamente se busca inferir el subgrafo correspondiente a los nodos observados, es imperativo tener en cuenta la influencia ejercida por los nodos ocultos.
Este relevante problema se ha comenzado a abordar en el contexto de selección de modelos gráficos Gaussianos~\cite{chandrasekaran2012latent,yang2020network}, inferencia lineal de redes Bayesianas~\cite{anandkumar2013learning}, o regresión no lineal~\cite{anandkumar2013learning}.

\section*{Objetivos}
% Objetivo final
El objetivo final de esta tesis es avanzar hacia un modelado robusto dentro de GSP en el que los algoritmos sean cuidadosamente diseñados para incorporar la influencia de las perturbaciones en los datos. 
Para conseguirlo, nos proponemos analizar y comprender el impacto de diferentes tipos de perturbaciones en problemas fundamentales dentro del marco de GSP y, después, diseñar una formulación robusta capaz de aliviar los efectos perjudiciales producidos por estas perturbaciones.
Como este objetivo es bastante amplio, para ayudarnos a delimitar la extensión de la tesis planteamos los siguientes objetivos específicos. 

\vspace{3mm}\noindent
\textbf{(O1) Eliminación de ruido en señales definidas en grafos.}
Cuando las señales observadas presenten ruido, nuestro objetivo será diseñar arquitecturas no lineales capaces de separar el ruido de la señal.
Investigaremos distintas alternativas para incorporar la información codificada en la topología del grafo en la arquitectura y desarrollaremos garantías teóricas que evidencien la capacidad de eliminar ruido de las arquitecturas propuestas.  %, mostrando que esta capacidad está directamente influenciada por las propiedades del grafo subyacente.

\vspace{3mm}\noindent
\textbf{(O2) Interpolación de señales.}
Cuando las señales observadas contengan valores perdidos, nos centraremos en reconstruir la señal a partir de las muestras observadas.
Para esto, asumiremos que las observaciones pueden representarse como muestras tomadas mediante un esquema de muestreo de agregación local sucesiva y que la señal original es una señal \emph{sparse} difundida en el grafo.

\vspace{3mm}\noindent
\textbf{(O3) Identificación robusta de GFs.}
Nuestro objetivo será proponer un algoritmo capaz de estimar un filtro a partir de pares de señales de entrada y salida que tenga en cuenta la presencia de perturbaciones en los enlaces del grafo.
Además, trataremos con un escenario relacionado donde el objetivo será estimar simultáneamente varios filtros definidos sobre el mismo grafo y propondremos una implementación eficiente del algoritmo resultante.

\vspace{3mm}\noindent
\textbf{(O4) Inferencia de topología robusta y conjunta.}
Enfocaremos el problema clásico de inferencia de topología de una red mediante un enfoque robusto a la presencia de nodos ocultos y considerando que hay varias redes relacionadas que deben ser estimadas.
En primer lugar, asumiremos que las señales observadas son estacionarias en el grafo y, después, platearemos un modelo de inferencia conjunta que aproveche la similitud de distintos grafos para inferirlos simultáneamente, mejorando así la calidad de la estimación.  

\vspace{2mm}
Aunque el carácter de esta tesis es principalmente teórico, además de estos objetivos dedicados a un tipo de perturbación específica, también consideramos como objetivo transversal la evaluación de los algoritmos desarrollados mediante datos reales para ilustrar su potencial aplicabilidad.

\section*{Metodología}
% Solución del problema
En el desarrollo de esta tesis hemos seguido un planteamiento sistemático que busca la obtención de soluciones óptimas.
Para cada problema específico, nos hemos centrado en la obtención de modelos matemáticos capaces de capturar toda la estructura inherente al problema para posteriormente aprovecharlos en el desarrollo de problemas de optimización rigurosamente formulados.
Debido a la complejidad de las tareas abordadas, los problemas de optimización resultantes suelen ser no convexos, por lo que el siguiente paso se centra en proponer relajaciones convexas y/o algoritmos iterativos capaces de encontrar una solución (sub)óptima, demostrando matemáticamente la convergencia de los algoritmos iterativos a un punto estacionario cuando sea necesario. 

% Evaluación y código libre
Tras el desarrollo de los algoritmos pertinentes, es imprescindible evaluar su rendimiento numéricamente y compararlos con distintas alternativas existentes.
En este aspecto, además de la evaluación mediante datos sintéticos, es una parte fundamental la aplicación de los algoritmos a datos reales para estimar la potencial aplicación de las herramientas desarrolladas en escenarios prácticos. 
Finalmente, todo el código desarrollado se ha subido a repositorios online en GitHub para incrementar la visibilidad y difusión de los resultados obtenidos.

\section*{Resultados}
El trabajo realizado en esta tesis se ha visto reflejado en la escritura de \textbf{4 artículos de revista JCR} (3 de ellos actualmente en proceso de revisión), y \textbf{7 publicaciones en conferencias internacionales}.

% Resultados derivados de (O1)
La tarea de eliminación de ruido descrita en el objetivo \textbf{(O1)} se ha llevado a cabo en~\cite{rey2019underparametrized,rey2021untrained}.
Hemos desarrollado dos arquitecturas sobreparametrizadas y no entrenables que incorporan la topología del grafo de maneras distintas.
La primera se basa en GFs no entrenables que generalizan la operación de convolución, y la segunda se basa unos operadores de sobremuestreo construidos mediante esquemas de clustering jerárquico.
Por otro lado, se ha realizado un análisis matemático de ambas arquitecturas obteniendo garantías teóricas sobre su rendimiento, mejorando así nuestro entendimiento sobre arquitecturas no lineales y la influencia de incorporar la topología del grafo.
Finalmente, las arquitecturas propuestas han servido como elemento clave en otros problemas ajenos a la eliminación del ruido~\cite{rey2019deep,rey2021overparametrized}.

% Resultados derivados de (O2)
Los resultados derivados del objetivo \textbf{(O2)} se reflejan en~\cite{rey2019sampling}.
En este trabajo hemos generalizado los resultados del esquema de muestreo de agregación local sucesiva a escenarios donde la señal original es una señal sparse difundida a través de la red en lugar de ser una señal de banda limitada.
Después de definir el modelo de observaciones para las señales perturbadas, hemos propuesto un algoritmo de interpolación definido en el dominio espectral, y hemos generalizamos resultados existentes sobre deconvolución ciega a este esquema de muestreo de agregación local sucesiva con señales sparse difundidas en el grafo.

% Resultados derivados de (O3)
La solución propuesta al objetivo \textbf{(O3)} se ha traducido en las publicaciones~\cite{rey2021robust,rey2022robust}.
El método desarrollado está formulado en el dominio de los vértices, evitando problemas de inestabilidad numérica y las dificultades asociadas con la influencia de perturbaciones en el espectro del grafo.  
La identificación robusta del filtro se ha reformulado como un problema de optimización conjunto en el que el objetivo de identificación del filtro ha sido aumentado con un regularizador que remueve el ruido de la topología del grafo.
De esta forma, además de estimar el filtro deseado también se proporciona una estimación mejorada del grafo subyacente.
Por otro lado, hemos generalizado este problema a escenarios donde el objetivo es estimar simultáneamente múltiples GFs, todos definidos sobre el mismo grafo.
También en relación con este tipo de perturbaciones, en~\cite{tenorio2021robust} hemos desarrollamos una definición alternativa de GFs menos sensible a los errores en la topología.

% Resultados derivados de (O4)
Finalmente, la inferencia de topología en presencia de nodos ocultos planteada en \textbf{(O4)} se ha abordado en~\cite{buciulea2019network,rey2022joint}.
Inicialmente, hemos revisitado la definición clásica de estacionariedad para que refleje la influencia de los nodos ocultos y la hemos empleado en la formulación de un problema de optimización con restricciones adicionales que aprovechan la estructura resultante de la presencia de nodos ocultos.
Después, hemos presentado un método de inferencia de topología conjunta que estimaba la topología de varios grafos simultáneamente para explotar la similitud entre los distintos grafos.
Clave para este método robusto fue emplear la estructura por bloques resultante de la presencia de variables ocultas.
Adicionalmente, íntimamente relacionado con este objetivo está el trabajo desarrollado en~\cite{rey2022enhanced}, donde proponemos un algoritmo de inferencia de topología basado en información previa sobre la densidad de \emph{motivos} (o \emph{motifs}) del grafo objetivo.
Esta novedosa prior sobre la estructura del grafo tiene un carácter local y puede emplearse, por ejemplo, para medir la distancia entre grafos de distinto tamaño, un problema no trivial que es clave en los modelos de inferencia conjunta.

%\red{Resultado transversal}

\section*{Conclusiones}
% Conclusiones generales
En esta tesis se ha contribuido a construir los cimientos de un paradigma robusto en el que abordar problemas clásicos de GSP mientras se modela la influencia de perturbaciones en los datos observados.
Con esta finalidad, se han considerado varios tipos de perturbaciones clasificados en dos amplias clases: (i)~perturbaciones en las señales; y (ii)~perturbaciones en la topología del grafo.
La primera clase de perturbaciones está asociada con los objetivos \textbf{(O1)} y \textbf{(O2)} y suele originar problemas tratables que han sido estudiados en mayor profundidad.
Por otro lado, la segunda clase de perturbaciones aparece en los objetivos \textbf{(O3)} y \textbf{(O4)} y origina problemas más desafiantes, por lo que cuenta con un menor número de trabajos previos.

% Conclusiones sobre los distintos objetivos
En primer lugar, en el capítulo~\ref{chap:denoising} se han presentado distintas redes neuronales no lineales y no entrenables que incluyen la topología del grafo mediante dos estrategias distintas y, además, se ha caracterizado matemáticamente su capacidad para separar las señales del ruido asumiendo algunas simplificaciones, avanzando así en la comprensión sobre este tipo de arquitecturas.  
El capítulo~\ref{chap:interpolation} ha lidiado con la presencia de valores perdidos en las señales mediante la interpretación de las observaciones como muestras recogidas a través de un esquema de muestreo de agregaciones locales para, posteriormente, proponer un método espectral para su interpolación asumiendo se trataba de señales sparse difundidas en el grafo. 
Después, en el capítulo~\ref{chap:robust_filter_id} hemos presentado un algoritmo robusto a perturbaciones en los enlaces del grafo capaz de identificar uno o varios GFs a partir de un conjunto de observaciones de entrada y salida.
Como el problema de optimización era no convexo, hemos desarrollado un algoritmo iterativo basado en la resolución secuencial de varios problemas convexos, hemos demostrado su convergencia a un punto estacionario y, además de estimar los filtros de interés, hemos comprobado que elimina las perturbaciones existentes en la topología del grafo.  
Finalmente, en el capítulo~\ref{chap:nti_hidden} se ha presentado un método de inferencia de topología conjunta que tiene en cuenta la influencia de los nodos ocultos.
Este algoritmo lidia con la inferencia de varios grafos a partir de señales estacionarias y tiene en cuenta la particular estructura del problema para aprovechar la similitud entre nodos no observados.

% Evaluación numérica y líneas futuras
Adicionalmente, los algoritmos desarrollados en cada capítulo han sido evaluados mediante una extensiva batería de experimentos empleando datos sintéticos y reales.
En estos experimentos se han comparado los métodos propuestos en esta tesis con otras alternativas del estado del arte.
\clearpage
\chapter{Agradecimientos}
Antes de comenzar a presentar el contenido de esta tesis, me gustaría agradecer la financiación recibida\footnote{Este trabajo ha sido parcialmente financiado por las ayudas FPU17/04520, EST21/00420, y por el proyecto de investigación SPGRAPH (PID2019-105032GB-I00).} y, aún más importante, dedicar unas líneas en agradecimiento a una serie de personas que me han acompañado durante este viaje.

Sin lugar a dudas, la primera persona a la que debo darle las gracias es a mi director de tesis, Antonio.
Muchas gracias por tu esfuerzo y dedicación, por aquellas largas tutorías al inicio del doctorado dónde pasabas horas explicándome conceptos que, aunque ahora parecen triviales, al principio parecían magia, y por transmitirme la importancia de utilizar un vocabulario claro y preciso (aunque reconozco que esto último sigue siendo un trabajo en desarrollo).
En resumen, una gran parte de mi evolución como investigador en estos cuatro años te la debo a ti, y espero seguir trabajando y aprendiendo contigo en el futuro.
También quiero dar las gracias a Santiago y a su fantástico grupo en la Universidad de Rice.
Gracias por acogerme como a uno más del grupo y por ofrecerme un entorno tan motivador. Habéis conseguido que los 6 meses que pasé en Houston hayan sido una experiencia maravillosa.

Cambiando ahora al ámbito familiar, quiero darles las gracias a mis padres.
Ellos siempre han estado a mi lado para darme apoyo y han hecho de mi educación y la de mi hermano una de sus prioridades.
Si he sido capaz de llegar hasta aquí, ha sido gracias a vosotros.
Y ya que menciono a mi hermano, no puedo pasar la oportunidad de agradecerle esos viajes a Salamanca tan necesarios para desconectar y recargar baterías de cuando en cuando.

Por supuesto, no puedo olvidarme de Celia, una de las personas más importantes en mi vida.
Pero que voy a decirte a estas alturas que tu no sepas ya...
Simplemente, muchas gracias por haberme hecho feliz todos estos años y por embarcarte conmigo en una nueva aventura ahora que este viaje de 4 años llega a su fin.

Por último, pero no menos importante, quiero mencionar a mis compañeros del doctorado.
Gracias por contribuir a un entorno de trabajo maravilloso en el laboratorio, con esa combinación de discusiones científicas, dudas e inquietudes, pero también con desayunos y alguna que otra cerveza.
Sin duda alguna, habéis hecho que estos años se pasen volando.
Y en especial, gracias a Andrei y a Víctor por ayudarme con este último empujón para rematar la tesis.
\clearpage
\chapter{Abstract}
% Motivation, context, and goals
The abundance of large and heterogeneous systems is rendering contemporary data more pervasive, intricate, and with a non-regular structure.
With classical techniques facing troubles to deal with the irregular (non-Euclidean) domain where the signals are defined, a popular approach at the heart of graph signal processing (GSP) is to: (i) represent the underlying support via a graph and (ii) exploit the topology of this graph to process the signals at hand.
In addition to the irregular structure of the signals, another critical limitation is that the observed data is prone to the presence of perturbations, which, in the context of GSP, may affect not only the observed signals but also the topology of the supporting graph.
Ignoring the presence of perturbations, along with the couplings between the errors in the signal and the errors in their support, can  drastically hinder estimation performance. While many GSP works have looked at the presence of perturbations in the signals, much fewer have looked at the presence of perturbations in the graph, and almost none at their joint effect.  While this is not surprising (GSP is a relatively new field), we expect this to change in the upcoming years.  Motivated by the previous discussion, the goal of this thesis is to advance toward a robust GSP paradigm where the algorithms are carefully designed to incorporate the influence of perturbations in the graph signals, the graph support, and both.
To do so, we consider different types of perturbations, evaluate their disruptive impact on fundamental GSP tasks, and design robust algorithms to address them.

% Denoising and interpolation
The first part of the thesis addresses the presence of perturbations in the graph signals, which typically lead to more tractable problems.
When the observed signals are corrupted by additive noise, we introduce two untrained \textit{nonlinear} graph neural network architectures to remove the noise from the observations, develop theoretical guarantees for their denoising capabilities in a simple setup, and provide empirical evidence in more general scenarios.
Each of the architectures incorporates the information encoded by the graph in a different manner: one relying on graph convolutions, and the other employing graph upsampling operators based on hierarchical clustering. 
Intuitively, each architecture implements a different prior over the targeted signals.
Then, we move on to a setting where perturbations appear in the form of missing values.
In this case, we assume that the original signal is a \textit{diffused sparse graph signal}, interpret the missing values as samples gathered through a \textit{successive aggregation} sample scheme, and study the recovery (interpolation) of the original signal.
Depending on the particular application, the goal is to use the \textit{local} observations to recover the diffused signal or (the location and values of) the seeds.
Different sampling configurations are investigated, including those of known and unknown locations of the sources as well as that of the diffusing filter being unknown.

The second part of the thesis deals with perturbations in the topology of the graph, which give rise to more challenging formulations.
In this sense, we propose a novel approach for handling perturbations in the links of the graph and apply it to the problem of robust graph filter (GF) identification from input-output observations.
Different from existing works, we formulate a non-convex optimization problem that operates in the vertex domain and jointly performs GF identification and graph denoising, and hence, on top of learning the desired GF, an estimate of the graph is obtained as a byproduct.
To handle the resulting bi-convex problem, we design an algorithm that blends techniques from alternating optimization and majorization minimization, showing its convergence to a stationary point.
Then, moving on to the last type of perturbation, we investigate the problem of learning a graph from nodal observations for setups where only a subset of the nodes are observed, with the others remaining unobserved or hidden.
Our schemes assume the number of observed nodes is considerably larger than the number of hidden nodes, and build on recent GSP models to relate the signals and the underlying graph.
Specifically, we go beyond classical correlation and partial correlation approaches and assume that the signals are \emph{stationary} in the sought graph, and moreover, we propose a joint network topology inference framework where several related graphs are estimated together.
The underlying idea is to exploit the similarity of the different graphs to enhance the quality of the estimation.
Since the resulting problems are ill-conditioned and non-convex, the block matrix structure of the proposed formulations is leveraged and suitable convex-regularized relaxations are presented.

Although the methodology and focus of this thesis are more theoretical (defining an estimation problem, stating the considered assumptions, obtaining the estimates as solutions to rigorously formulated optimization problems, designing computationally efficient provably convergent algorithms and, whenever possible, characterizing the performance of those), the experimental results will also play an important role. 
To that end, we evaluate the performance of our algorithms over synthetic and real-world datasets and compare their results with state-of-the-art alternatives.
These experiments reflect the impact of ignoring the presence of perturbations, show the strengths and weaknesses of the proposed methods, demonstrate that in a number of settings our methods outperform current alternatives, and assess the applicability of our schemes to real-world problems. 
\clearpage

%%%%%%%%%%%%%%% Chapters %%%%%%%%%%%%%%%%%%
\mainmatter
\setlength{\parskip}{10pt}
\chapter{Introduction}\label{chap:introduccion}
% Description
We begin by providing a short overview of the research environment surrounding this thesis, motivating the relevance of graph-based methods and highlighting the impact of the presence of perturbations in the data.
After that, the chapter: (i) states the main objectives sought by this work; (ii) lists the resulting contributions; and (iii) presents a brief outline of the remaining chapters.

\section{Motivation and context}\label{sec:motivation}
% Motivate graphs/GSP - Examples of network data
In the last two decades, we have been experiencing a data deluge largely propelled by the pervasive deployment of networks of sensing devices, the massive use of online social media, and the unstoppable digitalization of our daily tasks.
At the same time, as contemporary interconnected systems grow in size and importance, the data generated by such systems becomes more complex and heterogeneous, motivating the fast development of new methods and techniques to process datasets defined over irregular (non-Euclidean) domains~\cite{easley2010networks,newman2003structure,castro2004network,leskovec2020mining}.
Among the novel approaches that emerged to handle contemporary data, one particularly tractable and fruitful consists in modeling the underlying irregular structure by means of a \emph{graph}, and then, interpreting the data as signals defined on the graph, which are commonly referred to as \emph{graph signals}.
This graph-based perspective has rapidly grown in popularity and it has been successfully applied to data obtained from power, communication, social, geographical, financial, or biological networks, to name a few~\cite{nodop1998field,kolaczyk2009book,sporns2012book}.
Moreover, it has attracted the attention of researchers from different areas, including statistics, machine learning, and signal processing.

%

% What is GSP
Precisely, interpreting signals with irregular support as graph signals and then exploiting the topology of the underlying graph to process the signals is at the core of \acrfull{gsp}, a relatively new field that is developing swiftly~\cite{shuman2013emerging,djuric2018cooperative,ortega2018graph,marques2020editorial}.
GSP is devoted to developing new models and algorithms for processing graph signals, oftentimes by generalizing classical tools originally conceived to process signals with regular support (time or space).
Based on the fundamental assumption that there exists a close relation between the properties of the signals and the topology of the graph where they are supported, the key to the success of GSP is to effectively exploit the relation between the graph and the signals.
To that end, a considerable proportion of the efforts in GSP are directed at analyzing how the algebraic and spectral characteristics of the graph impact the properties of the graph signals.
In this analysis, the so-called \acrfull{gso} plays a fundamental role.
The GSO is a sparse matrix whose sparsity pattern encodes the topology of the graph, rendering it a cornerstone element within the GSP framework~\cite{sandryhaila2013discrete,shuman2013emerging}.
For example, employing the GSO as a building block enables the definition of different spectral tools such as the \acrfull{gft}~\cite{sandryhaila2014discrete,isufi2016autoregressive,sardellitti2017graph}, or more general graph-signal operators such as \acrfull{gfs}, which may be expressed as polynomials of the GSO~\cite{sandryhaila2013discrete,segarra2016blind,segarra2017optimal,iglesias2018demixing}.

% Popular applications and examples
A wide range of graph-related problems have been addressed under the GSP umbrella, and even though a variety of goals and assumptions are considered for the different problems, the key idea of harnessing the relation between the graph and the signals remains a constant.
A popular problem consists in modeling an arbitrary linear transformation between some input and output graph signals through a GF.
This task is commonly referred to as GF design or {GF identification}, and the inferred GF may be interpreted as the dynamics driving a network-diffusion process of interest~\cite{segarra2016blind,segarra2017optimal,liu2018filter,tremblay2018design}.
The {sampling and reconstruction} of graph signals is also an interesting problem~\cite{chen2015discrete,anis2014towards,marques2015sampling,puy2018random}, with meaningful connections to semisupervised learning.
Note that while sampling signals defined over regular domains is relatively straightforward (regular sampling schemes are prudent and give rise to sample signals that are regular as well), this is not the case for graph signals, which are inherently irregular. Hence, the efforts when approaching this task focus on designing sampling schemes that exploit the graph structure allowing to effectively recover the whole signal from its sampled version.
The reconstruction of the signal is also known as {graph signal interpolation}, and it is related to solving an inverse problem that involves both the signal observations and the supporting graph~\cite{chen2014signal,chen2014inpainting,onuki2016graph,chen2015signal}.
Depending on the actual relation between the observations and the original signal, the problem has been addressed from a point of view of {signal denoising, inpainting, or signal super-resolution}, to name a few.
Another fundamental but considerably different problem is that of {network topology inference}, also known as {graph learning}~\cite{kalofolias2016learn,pavez2016generalized,dong2016learning,segarra2017network,segarra2018network,mateos2019connecting}.
In contrast with previous GSP problems, in network topology inference the focus is placed on the topology of the graph, which is unknown, and therefore, the goal is to infer the graph from a set of nodal observations.
Finally, in the context of deep learning, another line of research that has attracted attention is the development of non-linear architectures that exploit the relation between graphs and signals by incorporating the topology of the graph into their design.  
This popular family of \acrfull{nns} is known as \acrfull{gnns}, and it encompasses a gamut of different graph-based architectures that have been applied to a wide range of problems~\cite{kipf2016semi,gama2018convolutional,sakhavi2018learning,cui2019traffic,wang2018graphgan}.

% Presence of perturbation - two types
All the aforementioned GSP applications use as input the observed data (graph signals) and the observed/inferred support (the graph). Unfortunately, imperfect knowledge due to the presence of noise, missing values, or outliers is pervasive in contemporary data science applications. In this sense, we will use the generic term \textit{perturbation} to refer to any imperfection in the observed data, encompassing a variety of defects whose particularities will depend on the application at hand and the features of the data. To further illustrate the diverse nature of perturbations in a GSP context, consider the example of a network of sensors measuring some quantity of interest.
The process of acquiring the measurement will introduce a certain amount of noise, and furthermore, if any sensor is damaged then its measures will be completely lost. Equally important, the information about the connections between the sensors may not be fully accurate. This example clearly illustrates that, when dealing with GSP, one needs to account for (i)perturbations in the graph signals; (ii)~perturbations in the topology of the graph; and (iii)~the joint effects and interactions between these two.

% Perturbations in the signals
We start describing the presence of \emph{perturbations in the graph signals}.
Clearly, signal perturbations have been extensively investigated in signal processing, statistics, and data science, so that many of the classical results can also be leveraged in the GSP setup. From this point of view, the two key (distinctive) questions when dealing with perturbed graph signals are: (i)~how does the graph influence the perturbation? and (ii)~how can the graph be exploited to design the schemes that mitigate/eliminate the perturbations?
Due to the combination of practical relevance, tractability, and the existence of related works from classical signal processing, accounting for perturbations in the observed graph signals has attracted a considerable amount of attention.
Accounting for this, we focus on studying two specific types of perturbations that have been thoroughly analyzed (the presence of noise in the signals and the presence of missing values), considering graph-aware processing and acquisition architectures that had not been investigated before. 
\begin{itemize}
    \item \textbf{(P1) Noise in the graph signals.}
    This simple case assumes that the observed signals are corrupted by additive noise, typically modeled as an \acrfull{iid} random variable drawn from a particular distribution. Noisy graph signals arise in a gamut of graph-related applications such as measurements in electric, social and transportation networks, or monitoring biological signals\cite{kolaczyk2009book,chen2014signal,gimenez2017matrix,ortega2018graph}.
    While denoising schemes have been thoroughly investigated in classical signal processing, the noise perturbating graph signals is oftentimes related to the topology of the graph (e.g., the noise can be independent across nodes while its variance is proportional to the node degree or any other node centrality measure). Even more importantly, GSP denoising schemes must be particularized to exploit the graph when mitigating (eliminating) the noise. This process is known as \emph{graph-based signal denoising}, and traditional approaches include minimizing the graph total variation to push the signal values at neighboring nodes to be close~\cite{chen2014signal,wang2015trend}, promoting a notion of signal smoothness by adding a regularization parameter based on the quadratic form of the graph Laplacian~\cite{pang2017graph}, or encouraging the recovery of signals with a smooth gradient~\cite{onuki2016graph}.
    More recently, non-linear solutions for denoising graph signals have been proposed, with relevant examples based on median graph filters~\cite{tay2020time}, graph autoencoders~\cite{do2020graph}, or graph unrolling architectures~\cite{chen2021graph}.

    \item \textbf{(P2) Missing values.}
    We use this term to refer to setups where only a subset of the entries of the graph signal are available. This accounts for cases where the values are missing / totally corrupted, as well as for sampling setups where the remaining entries were purposely unobserved. Practical graph scenarios that can lead to missing values include damaged sensors in a sensor network, wrong or incomplete answers when the data is gathered through online forms, or just because sampling the signal values at every node is not feasible in large networks~\cite{kolaczyk2009book,acuna2004treatment,tsitsvero2016signals}.  A number of alternatives arise to deal with this type of perturbed signals, with naive alternatives including filling the missing values with zeros or using the mean value within the observed values in the one-hop neighborhood. To fill the missing values, a reasonable and rigorous approach is to look at the problem from the sampling perspective and design methods to perform \emph{graph signal interpolation}. Two critical aspects in this regard are the postulation of a parsimonious model for the graph signal (bandlimited, diffused, smooth...) and the impact of the scheme collecting the samples. Several works have investigated different instances of this problem, with a strong bias towards assuming that the original signal is graph bandlimited and that the observed values proceed from observations taken at a fixed subset of nodes~\cite{anis2014towards,chen2015discrete,anis2016efficient,puy2018random}. Alternatively, other works have postulated that the observed values correspond to successively aggregating the values of the signal from neighboring nodes~\cite{marques2015sampling,chen2016signal}, and designed the associated optimal interpolation schemes.
\end{itemize}

% Perturbations in the topology
We shift focus now to the second class of imperfections studied in this work: \emph{perturbations in the graph topology}.
In this case, recall that GSP builds upon exploiting the relation between the signals and the graph, and hence, it is not surprising that methods within the GSP framework are particularly sensitive to this type of imperfections.
More precisely, a fair amount of GSP methods rely on either the spectrum of the GSO or in polynomials of the GSO, and because the GSO captures the topology of the graph, the perturbations in the observed topology translate into perturbations in the GSO.
However, even when assuming additive models to represent the perturbations, which are the easiest type of models to deal with, measuring the impact of the perturbations in the topology on the spectrum of the GSO and the polynomials of the GSO is a challenging endeavor.  
Not only that, but when perturbations affect the number of observed nodes, even the size of the GSO will be uncertain.
In conclusion, the perturbations in the topology of the graph will not only hinder the performance of GSP algorithms but, furthermore, developing robust alternatives that model the presence of perturbations is a non-trivial and ill-posed problem, which has been barely studied in the GSP literature. In this work, we focus on two types of graph perturbations: uncertainty (imperfections) in the edges of the graph, and unobserved nodes (whose effects will be analyzed in the context of graph learning).
These two types of perturbations, which are both theoretically and practically relevant, are described in more detail next.
\begin{itemize}
    \item \textbf{(P3) Uncertainty in the edges.} 
    Here, we assume that the set of nodes is perfectly known and that imperfect information about the existence and strength of the links (graph topology) is available. The perturbations in the observed topology may encompass observing edges that do not exist in the true graph, missing/unobserved edges that exist in the true graph, noise present in the weights of the observed edges, or any combination of the previous options.
    These perturbations appear in a gamut of practical situations.
    On the one hand, when networks are given explicitly, perturbations may be due to observational noise and errors (e.g., link failures in power or wireless networks~\cite{isufi2017filtering}).
    On the other hand, when in lieu of physical entities, the graphs model (statistical) pairwise relationships among the observed variables, they need to be inferred from the data~\cite{friedman2008sparse,dong2016learning,segarra2017network}.
    While this type of perturbation is critical for many GSP methods, modeling the influence of the imperfect topological information and developing robust alternatives is a challenging task, and hence, there is a limited number of works approaching this problem.
    In the frequency domain, \cite{ceci2020graph} employs a small perturbation analysis to study the impact of perturbations in the spectrum of the graph Laplacian.
    In the vertex domain,~\cite{miettinen2018graph,miettinen2019modelling} postulates a graphon-based perturbation model applied to GFs of order one.
    Then, in more recent approaches, \cite{ceci2020_semtls} combines \acrfull{sem} with \acrfull{tls} to jointly infer the GF and the perturbations in the GSO, and~\cite{natali2020topology} proposes a robust GF identification alternative where the support of the graph is assumed to be known, so the perturbations are constrained to be noise in the observed edges. 
    
    \item \textbf{(P4) Hidden nodes.}
        Finally, we consider a perturbation where some elements of the nodal set are not known, which is a problem particularly acute in the context of network topology inference. Clearly, when hidden nodes are present, one has only access to signals (measurements) from the remaining (observed) nodes. However, this should not be confused with the (missing values) perturbations introduced in \textbf{(P2)}. There, the signals at some nodes were not observed, but the existence of the node and its connections to other nodes were known. In contrast, here not only the underlying graph is unknown, but even the number of hidden nodes is oftentimes not  known. A number of estimation goals arise in the context of {\textbf{(P4)}}: inferring the number of hidden nodes, the connections among them, or the values of the signals, to name a few. In the specific context of network topology inference under the presence of hidden nodes, the problem is more challenging because the links among observed nodes are unknown as well. As a result, the main goal is usually to estimate the links between the observed nodes (also known as learning the \emph{observed subgraph}). More ambitious approaches aim also at estimating the links between observed and unobserved nodes. Clearly, all these problems are related and challenging (highly correlated values from two observed nodes may be explained not only by an edge between the two nodes but by a third hidden node connected to them).
    Some network-inference methods have started to look at the problem of hidden nodes with examples in the context of Gaussian graphical model selection~\cite{chandrasekaran2012latent,yang2020network}, inference of linear Bayesian networks~\cite{anandkumar2013learning}, non-linear regression~\cite{mei2018silvar}, or brain connectivity~\cite{chang2019graphical}, to name a few. 
    Nonetheless, there are still many network-inference methods (including most in the context of GSP) that have not considered this type of perturbations.
\end{itemize}

% Closing notes
To summarize,  the presence of perturbations in GSP setups leads to having imperfect knowledge about the graph signals, the graph, or both.
Clearly, if these perturbations are ignored, the performance of naive algorithms that use as input the corrupted data will be drastically hindered. The solution is to design robust schemes that model and incorporate into their formulation the presence and effects of the perturbations. While relevant and timely, this is a challenging problem, especially when the perturbations affect the graph topology.
Bearing all this in mind, our goal, which is presented in detail in the next section, is to develop a set of GSP methods robust to different types of perturbations, understand their differences and similarities, and discuss how they can be combined in a meaningful and tractable way.

\section{Objectives}\label{sec:objectives}
% General goal: robust GSP
As discussed in the previous section, the presence of perturbations in the observed data constitutes a relevant and ubiquitous problem.
Motivated by this, the prevailing objective of this thesis is to advance towards a \emph{robust GSP} paradigm where the algorithms are carefully designed to deal with the presence of perturbations in the graphs and the signals.
To that end, we aim to analyze and understand the impact of the different types of perturbations in several fundamental GSP problems and then, design a robust formulation capable of: (i) recovering the original data from the perturbed observations; and/or (ii) approaching the desired task while taking into account the presence of perturbations in the data to minimize their disruptive influence.
To render these generic (and relatively ambitious) goals more reachable, we focus our research efforts on four specific objectives. Each of them considers a specific GSP problem and addresses one of the perturbations introduced in \cref{sec:motivation}. % in a different GSP task.

\vspace{3mm}\noindent
\textbf{(O1) Non-linear denoising of graph signals.}
Consider a setting where the observed graph signals are corrupted with noise as described in the perturbation type \textbf{(P1)}.
Our goal is to design non-linear architectures to denoise the observed graph signals.
Because dealing with noisy signals is a problem that has been studied to a considerable extent in the GSP literature, here we focus on the design of untrained \textit{non-linear} architectures and on their theoretical characterization. 
First, we will explore different ways of incorporating the information encoded in the graph and propose novel graph-aware NN architectures to denoise graph signals.
Second, we will provide theoretical guarantees for the denoising performance of the proposed architectures, and we will show that the denoising capability is directly influenced by the topology of the underlying graph.

\vspace{3mm}\noindent
\textbf{(O2) Signal interpolation of diffused sparse graph signals.}
Consider a setting where the observed signals have missing values as described in the perturbation type \textbf{(P2)}.
Our goal is to reconstruct the observed signal when the observations are taken at a particular node according to a successive local \textit{\acrfull{agss}}.
The signal to be reconstructed is assumed to be a \textit{\acrfull{dsgs}}, a class of signals that can be modeled as a signal with zeros everywhere except in a few seeding nodes, which is then diffused through the network via a GF.
Ultimately, we will develop a reconstruction algorithm to recover both the original (diffused) signal and the seeds from the perturbed observation.

\vspace{3mm}\noindent
\textbf{(O3) Robust GF identification.}
Consider a setting where the observed graph is perturbed as described in the perturbation type \textbf{(P3)}.
Our goal is to estimate a GF from some input-output signal pairs while accounting for the presence of errors in the supporting graph and in the observed signals.
The proposed approach needs to bypass the challenges associated with robust \emph{spectral} graph theory and avoid the numerical instability from computing high-order polynomials.
We will also address the scenario where several GFs need to be estimated, with all the GFs being defined as polynomials of the same GSO.
%The influence of perturbation in the topology of the graph will also be analyzed in non-linear architectures, and an alternative implementation of GFs robust to perturbations will be employed in GNNs.

\vspace{3mm}\noindent
\textbf{(O4) Robust network topology inference.}
Consider a network topology inference problem where a subset of the nodes remain hidden as described in the perturbation type \textbf{(P4)}.
Our goal is to estimate the joint topology of the observed subgraph while taking into account the presence of the hidden nodes.
%The observed signals will be assumed to be stationary in the unknown graph, so we will analyze how the presence of hidden nodes alters this prior information about the observed signals.
The observed signals will be assumed to be stationary in the unknown graph and our main focus will be considering the case where several related graphs, all defined over the same set of nodes, need to be learned.
Here, we will develop a \emph{joint network topology inference} algorithm that exploits the graph similarity while accounting for the presence of hidden nodes.
Note that leveraging the graph similarity between hidden nodes (which are not observed) is an interesting but non-trivial problem.

% Summary and transversal goals
Clearly, objectives \textbf{(O1)} and \textbf{(O2)} are concerned with perturbations involving the graph signals and they aim to design schemes to clean/infer the graph signal of interest.
In contrast, objectives \textbf{(O3)} and \textbf{(O4)} are concerned with perturbations involving the topology of the graph and they aim to solve higher order GSP tasks while mitigating (bypassing) the presence of the perturbations. On top of these four specific objectives, two transversal objectives are also considered. First, since many of the findings and challenges faced when addressing \textbf{(O1)}-\textbf{(O4)} are also present in related robust GSP setups, our aim is that the models and tools put forward in this thesis improve our understanding of the influence of perturbations in a general way, hoping that these robust approaches may be extended/generalized to other important GSP applications where the presence of perturbations is critical. Secondly, even though the scope of this thesis is markedly theoretical, showcasing the potential applicability of the algorithms developed herein is certainly important. Therefore, we define another transversal objective, which consists in assessing the practical value of our algorithms using real-world datasets.

\section{Summary of contributions}\label{sec:contributions}
This section provides the list of publications organized around the (results and contributions) of the four objectives presented in the previous section. 

% Denoising papers
The graph signal denoising task \textbf{(O1)} is addressed in \cite{rey2019underparametrized,rey2021untrained}.
The preliminary work in \cite{rey2019underparametrized} proposed an underparametrized deep decoder NN capable of learning non-linear representation for graph signals, which were used for compression and denoising. 
Later on,~\cite{rey2021untrained} presented two \emph{untrained} and overparametrized GNNs to address the graph signal denoising problem. 
To incorporate the topology of the graph, the first architecture employs a fixed (non-learnable) GF to generalize the convolutional layer in~\cite{kipf2016semi}.
The second architecture performs graph upsampling operations that, starting from a low-dimensional latent space, progressively increase the size of the input until it matches the size of the signal to denoise.
Furthermore, a mathematical analysis was conducted for each architecture offering bounds for their performance, improving our understanding of nonlinear architectures and the influence of incorporating the graph structure into NNs.
Interestingly, the decoder architecture introduced in~\cite{rey2021untrained} has proven useful for other problems than signal denoising.
The decoder was employed to design a graph deep decoder capable of learning the mapping between input-output signal pairs defined on different graphs~\cite{rey2019deep,rey2021overparametrized}.
The key idea is that the encoder uses the input graph to map the input signal onto a latent space, and then, the decoder uses the output graph to reconstruct the output signal from the latent representation.
The publications related to \textbf{(O1)} are listed below.
\begin{itemize}
    \item [\cite{rey2019underparametrized}] S. Rey, A. G. Marques, and S. Segarra, “An underparametrized deep decoder architecture for graph signals,” in IEEE Intl. Wrksp. Computat. Advances Multi-Sensor Adaptive Process. (CAMSAP). IEEE, 2019, pp. 231–235.
    
    \item [\cite{rey2021untrained}] S. Rey, S. Segarra, R. Heckel, and A. G. Marques, “Untrained graph neural networks for denoising,” arXiv preprint arXiv:2109.11700, 2021  (submitted to IEEE Trans. Signal Process.).
    
    \item [\cite{rey2019deep}] S. Rey, V. M. Tenorio, S. Rozada, L. Martino, and A. G. Marques, “Deep encoder-decoder neural network architectures for graph output signals,” in Conf. Signals, Syst., Computers (Asilomar). IEEE, 2019, pp. 225–229.
    
    \item [\cite{rey2021overparametrized}] S. Rey, V. M. Tenorio, S. Rozada, L. Martino, and A. G. Marques, “Overparametrized deep encoder-decoder schemes for inputs and outputs defined over graphs,” in European Signal Process. Conf. (EUSIPCO). IEEE, 2021, pp. 855–859.
\end{itemize}

% Sampling papers
The goal of graph signal interpolation \textbf{(O2)} is pursued in~\cite{rey2019sampling}, where the observed (non-missing) values of the perturbed signals are assumed to be taken at a particular node according to an AGSS.
In the AGSS, which was originally proposed for \acrfull{bgs}, the nodes successively aggregate the values of the signal in their neighborhood, and moreover, the recovery of the original signals can be guaranteed even if observations are gathered at a single node.
Firstly,~\cite{rey2019sampling} applied AGSS to the case where the observed signals are DSGS in lieu of BGS. Secondly, after defining the observational model for the perturbed signals, the paper proposed an interpolation algorithm defined in the spectral domain.
Finally, existing results for support identification and blind deconvolution were generalized to deal with AGSS and DSGS.
The publication related to \textbf{(O2)} is listed below.
\begin{enumerate}
    \item [\cite{rey2019sampling}] S. Rey, F. J. I. Garcia, C. Cabrera, and A. G. Marques, “Sampling and reconstruction of diffused sparse graph signals from successive local aggregations”, IEEE Signal Process. Lett., vol. 26, no. 8, pp. 1142–1146, 2019.
\end{enumerate}

% Robust GF ID papers
The robust robust GF identification problem \textbf{(O3)} is approached in \cite{rey2021robust,rey2022robust}.
In those works, the proposed solution was formulated in the vertex domain, avoiding the numerical instability of computing large polynomials and, at the same time, bypassing the challenges associated with robust spectral graph theory.  
The robust GF identification was recast as a joint optimization problem where the GF identification objective was augmented with a graph-denoising regularizer so that, on top of the desired GF, the proposed algorithm also provided an enhanced estimate of the supporting graph.
The joint formulation led  to a non-convex bi-convex optimization algorithm, for which a provably-convergent efficient algorithm able to find an approximate solution was developed.
Furthermore, to address scenarios where multiple GFs are present, the paper generalized the robust framework so that multiple GFs, all defined  over the same graph, were jointly identified.
Also related to \textbf{(O3)},~\cite{tenorio2021robust} introduced the neighborhood GF, a new type of GF that is numerically stable and robust to perturbations in the observed topology.
The definition of neighborhood GF, which replaced the powers of the GSO with $k$-hop adjacency matrices, was exploited to provide an alternative design of \acrfull{gcnn} that was employed in graph signal denoising and node classification problems.
The publications involved with \textbf{(O3)} are listed below.
\begin{enumerate}
    \item [\cite{rey2021robust}] S. Rey and A. G. Marques, “Robust graph-filter identification with graph denoising regularization”, in IEEE Int. Conf. Acoustics, Speech Signal Process. (ICASSP). IEEE, 2021, pp. 5300–5304.
    
    \item [\cite{rey2022robust}] S. Rey, V. M. Tenorio, and A. G. Marques, “Robust graph filter identification and graph denoising from signal observations,” arXiv preprint arXiv:2210.08488, 2022  (submitted to IEEE Trans. Signal Process.).
    
    \item [\cite{tenorio2021robust}] V. M. Tenorio, S. Rey, F. Gama, S. Segarra, and A. G. Marques, “A robust alternative for graph convolutional neural networks via graph neighborhood filters,” in Conf. Signals, Syst., Computers (Asilomar). IEEE, 2021, pp. 1573–1578.
\end{enumerate}

% NTI with hidden papers - NTI based on motif density
Lastly, the network topology inference in the presence of hidden nodes \textbf{(O4)} is addressed in \cite{buciulea2019network,rey2022joint}.
Initially,~\cite{buciulea2019network} investigated how the presence of hidden variables impacts the classical definition of graph stationarity.
%and formulated the network-recovery problem as a constrained optimization that explicitly accounted for the modified definitions.
Key to the proposed formulation was the consideration of a block matrix factorization approach and harnessing the low rankness and the sparsity pattern present in the blocks related to hidden variables.
Then, we exploited this block matrix factorization in~\cite{rey2022joint} to propose a topology inference method that, assuming that the observed signals are graph-stationary,  jointly learns multiple graphs while accounting for the presence of hidden variables.
To fully benefit from the joint inference formulation and successfully exploit the graph similarity among hidden nodes, the paper carefully exploited the structure inherent to the presence of latent variables with a regularization inspired by
group Lasso~\cite{simon2013sparse}.
An additional work that is closely related to the objective \textbf{(O4)} is presented in~\cite{rey2022enhanced}.
The paper presented a graph-learning algorithm that assumes that a reference graph with a density of motifs similar to that of the sought graph was known. Then, this similarity was harnessed to reveal a connection between the spectra of both graphs, which was exploited in the formulation of the inference problem and the associated algorithm.
The prior information about the density of motifs of the unknown graph is local and robust, in the sense that it enables the comparison of graphs of different sizes, an issue that was non-trivial.
Moreover, leveraging this prior to boost the performance of graph learning algorithms in the presence of hidden nodes arises as an interesting research problem, which is left as a future research direction.
The publications related to \textbf{(O4)} are listed below.
\begin{enumerate}
    \item [\cite{buciulea2019network}] A. Buciulea, S. Rey, C. Cabrera, and A. G. Marques, “Network reconstruction from graph-stationary signals with hidden variables”, in Conf. Signals, Syst., Computers (Asilomar). IEEE, 2019, pp. 56–60.
    
    \item [\cite{rey2022joint}] S. Rey, A. Buciulea, M. Navarro, S. Segarra, and A. G. Marques, “Joint inference of multiple graphs with hidden variables from stationary graph signals”, in IEEE Int. Conf. Acoustics, Speech Signal Process. (ICASSP). IEEE, 2022, pp. 5817–05821.
    
    \item [\cite{rey2022enhanced}] S. Rey, T. M. Roddenberry, S. Segarra, and A. G. Marques, “Enhanced graph-learning schemes driven by similar distributions of motifs”, arXiv preprintarXiv:2207.04747, 2022 (submitted to IEEE Trans. Signal Process.).
\end{enumerate}

\section{Outline of the dissertation}\label{sec:outline}
The remainder of this document is organized as follows.
First, \cref{chap:preliminaries} introduces fundamental definitions and concepts that will be employed during the thesis.
Then, \cref{chap:denoising} considers the presence of noise in the signals and proposes non-linear architectures to perform graph signal denoising. \cref{chap:interpolation} addresses the presence of missing values in the observed signals and introduces an interpolation method for DSGS.
Regarding the perturbations in the topology, \cref{chap:robust_filter_id} addresses the problem of GF identification assuming imperfect knowledge of the observed topology, and \cref{chap:nti_hidden} approaches the task of network topology inference while accounting for the presence of hidden nodes.
Finally, \cref{chap:conclusions} provides some concluding remarks and identifies some interesting future research directions.
\chapter{Fundamentals of graph signal processing}\label{chap:preliminaries}
% Chapter intro
This chapter introduces the main concepts and tools from GSP, which are the foundations to the research carried out in this thesis.
To that end, we begin by defining the basics of GSP, then introduce some fundamental tools and methods, and close the section describing a couple of more advanced GSP concepts directly related to this thesis.

\section{Graphs, graph signals, and the GSO}\label{sec:basic_elements}
% Graphs and adjacency matrix
A \emph{graph} is a mathematical structure formally defined as $\ccalG := (\ccalV, \ccalE)$, where $\ccalV$ and $\ccalE$ are, respectively,
the sets containing the \emph{nodes} and \emph{edges} conforming the graph, which are also commonly known as \emph{vertices} and \emph{links}.
The nodes collected in $\ccalV$ are typically labeled using integers so $\ccalV := \{1,2,...,N\}$, with $N$ denoting the number of nodes in the graph.
Then, the edges collected in $\ccalE$ are represented by pairs of nodes $(i,j)$ with $i,j \in \ccalV$ and $(i,j) \in \ccalE$ if and only if the node $i$ is connected to node $j$.
If none of the edges in the graph are directed, that is, if edges are agnostic to which node is the origin and which is the destiny, then the graph is called \emph{undirected}, and hence, $(i,j) \in \ccalE$ implies that $(j,i) \in \ccalE$.
In contrast, when the graph captures the direction of the edges it is called \emph{directed}, and we might encounter that $(i,j) \in \ccalE$ but $(j,i) \not\in \ccalE$.
\cref{fig:graph_and_signal} represents an undirected graph where the nodes in $\ccalV$ are represented in blue and the edges in $\ccalE$ are represented as gray lines.
Intuitively, a graph encodes pairwise relations between the nodes in $\ccalV$, with these relations being represented by the edges.
Then, an \emph{unweighted} graph captures whether an edge exists or not but it does not provide any information about the strength (closeness, similarity,...) of the connection.
On the other hand, this additional information about the distance or closeness between connected nodes is provided by \emph{weighted} graphs.
The distinction between weighted and unweighted graphs is apparent when looking at the \emph{adjacency matrix}, a widely used matrix representation of the topology of $\ccalG$. 
The adjacency matrix $\bbA$ is a sparse $N \times N$ matrix encoding the connectivity of the graph whose entry $A_{ij} \neq 0$ if and only if $(i,j) \in \ccalE$.
When the graph is unweighted the entries of $\bbA$ are binary, i.e., $\bbA \in \{0,1\}^{N \times N}$.
On the contrary, if the graph is weighted then $\bbA \in \reals^{N \times N}$ and the non-zero entries $A_{ij}$ capture the weight of the edge between the nodes $i$ and $j$.
Similarly, when the graph is undirected the matrix $\bbA$ is symmetric.
Another concept involving the connectivity of the graph is that of the neighborhood of a node.
For any node $i$, its \emph{neighborhood} $\ccalN_i := \{j \in \ccalV | (i,j) \in \ccalE\}$ is the set of nodes that are connected to $i$.
Furthermore, the \emph{degree} of a node $i$ is given by the number of neighbors, so it is formally defined as $d_i := |\ccalN_i| = [\bbA\mathbf{1}]_i$, where $\mathbf{1}$ denotes the vector of all ones.
In other words, $d_i$ can be computed by adding the entries of the $i$-th row of $\bbA$. If graphs are directed, one must account for defining incoming (outcoming) neighborhoods as well as incoming (outcoming) degrees.

% Graph signals
We move on to the definition of graph signals, which represent the other fundamental piece within the GSP framework and constitute the subject of study in most GSP problems.
Formally, a \emph{graph signal} can be modeled as a function from the node set to the real field\footnote{For simplicity, we focus our discussion on scalar, real-valued graph signals, but the values associated with each node could be discrete, complex, or even
vectors (e.g., when multiple features per node are observed).} $x: \ccalV \to \reals$ or, equivalently, as an $N$-dimensional real-valued vector $\bbx = [x_1,...,x_N]^\top \in \reals^N$, with $x_i$ denoting the value of the signal at the node $i$.
An example of a graph signal is given in \cref{fig:graph_and_signal}, where the height of the vertical bars represents the value of the signal at each node.
Then, because the graph signal $\bbx$ is defined on $\ccalG$, the core assumption of GSP is that either the values or the properties of $\bbx$ depend on the topology of $\ccalG$~\cite{sandryhaila2013discrete}.
For instance, consider a graph that encodes similarity.
If the value of $A_{ij}$ is high, then one will expect the signal values $x_i$ and $x_j$ to be akin to each other.
This rationale helps to explain the advantages of leveraging the topology of the graph when processing graph signals.
Some relevant examples of graph signal models that reflect the influence of the graph topology on $\bbx$ are provided in \cref{sec:models_signals}.

%%%%%%%%%%   FIGURE   %%%%%%%%%%
\begin{figure}[tb]
    \centering
    \includegraphics[width=.35\textwidth]{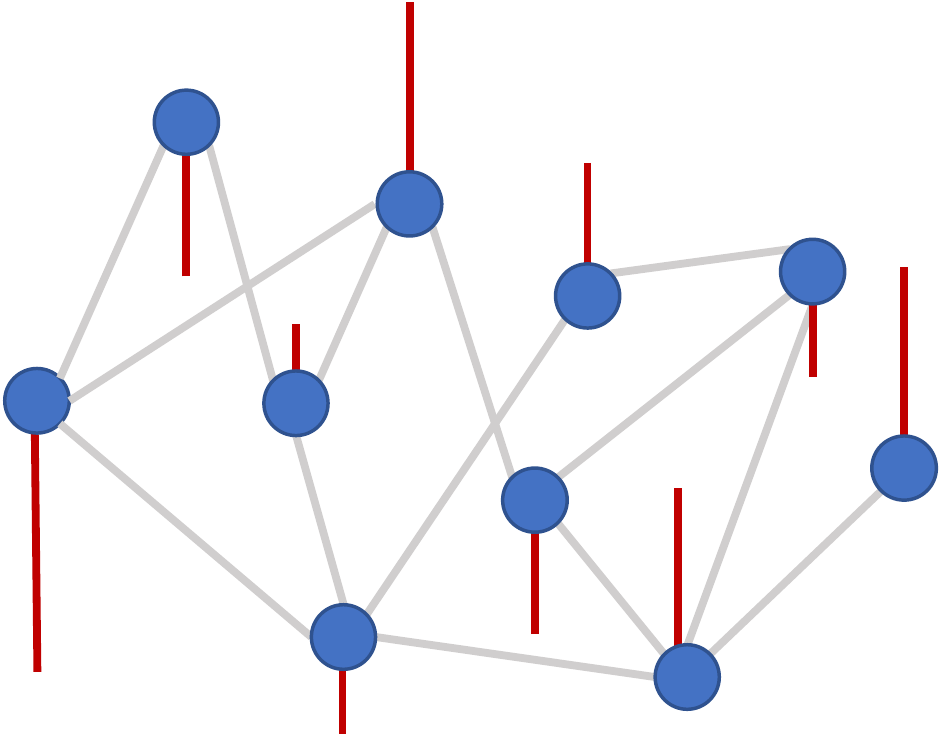}
    \caption{Depiction of a graph signal $\bbx$ and the underlying graph $\ccalG$. Nodes are represented in blue, the edges are the connections in gray, and the height of the red vertical bars represents the values of $\bbx$ at each node.}
    \label{fig:graph_and_signal}
\end{figure}
%%%%%%%%%%%%%%%%%%%%%%%%%%%%%%%%

% GSO
The last key element in the GSP framework is the so-called \emph{graph-shift operator} (GSO), a square matrix that captures the topology of the underlying graph 
$\ccalG$~\cite{sandryhaila2013discrete}.
The GSO is denoted by $\bbS \in \reals^{N \times N}$ and its entry $S_{ij}$ is allowed to be non-zero if and only if $i=j$ or $(i,j) \in \ccalE$.
Intuitively, $\bbS$ can be understood as a topology-aware local operator that can be applied to process graph signals.
There exist several options for selecting the GSO, with typical choices including the adjacency matrix $\bbA$, the graph combinatorial Laplacian $\bbL:=\diag(\bbA\mathbf{1}) - \bbA$, and its normalized variants~\cite{shuman2013emerging,sandryhaila2013discrete}.
Note that $\diag(\cdot)$ denotes the diagonal operator that transforms a vector into a diagonal matrix.
In this sense, using $\bbS$ instead of a specific choice for the GSO is particularly useful since it provides a higher level of abstraction and results in algorithms that may be applied to a wider range of scenarios.
When $\ccalG$ is assumed to be undirected, it follows that $\bbS$ is symmetric and it can be diagonalized as $\bbS = \bbV\bbLambda\bbV^\top$, where the orthonormal matrix $\bbV\in\reals^{N \times N}$ collects the eigenvectors of $\bbS$, and the diagonal matrix $\bbLambda=\diag(\bblambda)$ collects the eigenvalues $\bblambda\in\reals^N$.
On the other hand, when $\ccalG$ represents a directed graph we will assume that $\bbS$ is still diagonalizable and its decomposition will be given by $\bbS = \bbV\bbLambda\bbV^{-1}$.
Note that, in the directed case, the eigenvalues collected in $\bbLambda$ are likely to be complex numbers.

\section{Graph filters and filter identification}\label{sec:graph_filters}
Graph filters (GFs) are topology-aware linear operators whose inputs and outputs are graph signals.
More specifically, graph filters implement a linear transformation that can be expressed as a polynomial of the GSO of the form
\begin{equation}\label{eq:graph_filter}
    \bbH:=\sum_{r=0}^{R-1}h_r\bbS^r =\bbV\diag(\bbPsi\bbh)\bbV^{-1} = \bbV\diag(\tbh)\bbV^{-1},
\end{equation}
where $\bbH$ is the graph filter, and $\bbh := [h_0,...,h_{R-1}]^\top$ is the vector collecting the filter coefficients $h_i$.
The $N \times R$ Vandermonde matrix $\bbPsi$ defined as $\Psi_{ij} := \Lambda_{ii}^{j-1}$ represents the GFT for GFs, and thus, $\tbh := \bbPsi\bbh$ is the vector of size $N$ representing the frequency response of $\bbH$~\cite{sandryhaila2014discrete,segarra2017optimal}.
Since $\bbS^r$ encodes the $r$-hop neighborhood of the graph, graph filters can be used to diffuse input graph signals $\bbx$ across the graph as $\bby = \sum_{r=0}^{R-1} h_r\bbS^r\bbx = \bbH\bbx$, where $\bby$ is the result of diffusing the signal $\bbx$ across $R\!-\!1$ neighborhoods with $h_r$ being the coefficients of the linear combination.

% Description of GFI problem
A relevant problem in the context of GFs is that of GF identification.
Consider that we observe $M$ input and output pairs $\bbX := [\bbx_1,...,\bbx_M]$ and $\bbY := [\bby_1,...,\bby_M]$ whose relation is given by 
\begin{equation}\label{eq:rfi_observ_model}% OLD LABEL: eq:observation_model
    \bbY = \bbH\bbX + \bbW,
\end{equation}
with $\bbW$ being a zero-mean random matrix (typically assumed to have i.i.d. entries) that accounts for noisy measurements and model inaccuracies.
Leveraging \eqref{eq:rfi_observ_model}, the GF identification task amounts to using the input-output pairs to estimate $\bbH$ under the model in \eqref{eq:graph_filter}, which, if the GSO $\bbS$ is known, boils down to estimating the GF coefficients collected in $\bbh \in \reals^R$.
Hence, we can approach the GF identification task in the node domain by solving the convex problem
\begin{equation}\label{eq:fi_opt}
    \min_\bbh \left\|\bbY - \sum_{r=0}^{R-1}h_r\bbS^r\bbX \right\|_F^2.
\end{equation}
Leveraging the frequency definition of GFs in \eqref{eq:graph_filter}, we rewrite the \acrfull{ls} cost in \eqref{eq:fi_opt} and obtain its (closed-form) solution as
\begin{equation}\label{eq:solving_fi}
    \hbh = \argmin_\bbh \left\|\vvec(\bbY) - \big((\bbV^{-1}\bbX)^\top \odot \bbV\big)\bbPsi\bbh \right\|_2^2 = \bbXi^{\dagger}\vvec(\bbY),
\end{equation}
where $\vvec(\cdot)$ denotes the vectorization operation, $\bbV^{-1}\bbX$ is the frequency representation of the input signals (see \cref{sec:models_signals}), $\odot$ denotes the Khatri–Rao product, $\bbPsi$ is the GFT Vandermonde matrix, $\bbXi:=((\bbV^{-1}\bbX)^\top \odot \bbV)\bbPsi$, and $^{\dagger}$ is the pseudoinverse operator.

% Comments on the solution
From \eqref{eq:solving_fi} we observe that estimating $\bbH$ is straightforward under the assumptions of: i) $\bbXi$ being full rank (i.e., the inputs are sufficiently rich); and ii) $\bbS$ being perfectly known.
However, as discussed in the first chapter of this thesis, the assumption in ii) does not hold true in many practical settings. New formulations of \eqref{eq:solving_fi} that account for imperfect GSOs are addressed in \cref{chap:robust_filter_id}.

\section{Models for graph signals}\label{sec:models_signals}
There is a great diversity of models capturing different relations between the signals and the underlying graph.
Here, we introduce some popular models for graph signals, which will be leveraged in subsequent chapters. %, and show their application to inverse problems.

% Bandlimited graph signals
\vspace{3mm}\noindent\textbf{Bandlimited graph signals.}
The notion of bandlimited graph signals links the properties of a signal to those of the spectrum of the supporting graph.
To be specific, the frequency representation of the signal $\bbx$ is given by the $N$-dimensional vector $\tbx := \bbV^{-1}\bbx$, with $\bbV^{-1}$ acting as the GFT~\cite{sandryhaila2014discrete}. 
Then, a graph signal is said to be \emph{low-pass} bandlimited if $\tbx$ satisfies that $\tilde{x}_k=0$ for $k>K$, where $K\leq N$ is referred to as the bandwidth of $\bbx$.
If $\bbx$ is bandlimited with bandwidth $K$, it holds that
\begin{equation}\label{eq:bl_signals}
    \bbx=\bbV_K\tbx_K,
\end{equation}
with $\tbx_K=[\tilde{x}_1,\cdots,\tilde{x}_K]$ collecting the active frequency components, and $\bbV_K$ collecting the corresponding $K$ eigenvectors.
In other words, $\bbx$ lives in a subspace of dimension $K$ spanned by the eigenvectors $\bbV_K$.
Nonetheless, even though bandlimited graph signals are typically associated with low-pass signals, the non-zero elements in $\tbx$ are not constrained to its low-frequency components.
We might encounter high-pass bandlimited signals or signals whose active frequency components are scattered throughout the spectrum.
Furthermore, in relevant cases we might ignore the specific frequency components that are active, further challenging the solution of inverse problems dealing with bandlimited graph signals.

Interestingly, comparing the definition of $\tbx$ with the definition of $\tbh$ from \eqref{eq:graph_filter}, it follows that, in contrast with classical signal processing, the GFT for graph signals is different from the GFT for GFs.
Nonetheless, exploiting the frequency representations $\tbx$ and $\tbh$, we have that the output $\bby = \bbH\bbx$ in the frequency domain is given by
\begin{equation}\label{eq:output_freq}
    \tby = \diag(\bbPsi\bbh)\bbV^{-1}\bbx = \diag(\tbh)\tbx = \tbh \circ \tbx,
\end{equation}
with $\circ$ denoting the Hadamard (entry-wise) product.
Note that equation \eqref{eq:output_freq} is the counterpart of the convolution theorem for time signals~\cite{segarra2017optimal}.

% DSGS
\vspace{3mm}\noindent \textbf{Diffused sparse graph signals (DSGS).}
A graph signal is called a DSGS when it can be modeled as a signal with only a few non-zero entries which is then diffused through the graph.
Mathematically, given a GSO $\bbS$, a DSGS $\bbx$ with $S$ non-zero seeds can be written as
\begin{equation}\label{eq:diff_sparse_signals_def}
    \bbx = \bbH \bbs,\;\; \text{where}\;\; {\textstyle \bbH=\sum_{r=0}^{R-1}h_r\bbS^r}\;\;\text{and}\;\; \|\bbs\|_0 \leq S,
\end{equation} 
where $\bbs$ denotes the original sparse signal whose non-zero entries are referred to as \emph{seeding nodes}.
Clearly, signals in \eqref{eq:diff_sparse_signals_def} can be viewed as the state reached after the diffusion process modeled by $\bbH$ is over, and the sparse input $\bbs \in \reals^N$ has been spread throughout the graph.

It is worth noticing that bandlimited signals and DSGS are two generative models with a similar goal: providing a simpler representation of $\bbx$.
In this sense, the bandlimited model offers an alternative representation of $\bbx$ that is sparse in the \emph{frequency domain} while the DSGS model offers an alternative representation of $\bbx$ that is sparse in the \emph{node domain}.
As happened with the frequency components, the support of the seeding nodes may be known a priori or we may need to learn it through deconvolution schemes.

% Smooth
\vspace{3mm}\noindent\textbf{Smooth graph signals.}
A graph signal is considered smooth on $\ccalG$ if the signal value at two connected nodes is ``close'' or, equivalently, if the difference between the signal value at neighboring nodes is small.
A common approach to quantify the smoothness of a graph signal relies on the quadratic form~\cite{kalofolias2016learn}
\begin{equation}\label{eq:smoothness}
    \sum_{(i,j)\in\ccalE}{A_{ij}(x_{i}-x_{j})^2} = \bbx^\top\bbL\bbx,
\end{equation}
which quantifies how much the signal $\bbx$ changes with respect to the notion of similarity encoded in the weights of $\bbA$.
This measure will be referred to as \acrfull{lv} of $\bbx$.
Note that, if the goal is to obtain the mean $\mathrm{LV}$ of $M$ graph signals collected in the $N \times M$ matrix $\bbX=[\bbx_1,...,\bbx_M]$, this can be achieved by computing
\begin{equation}\label{E:localvariation_tracecovariance}
\frac{1}{M}\sum_{m=1}^M \bbx_m^\top\bbL\bbx_m=\frac{1}{M}\sum_{m=1}^M \tr(\bbx_m\bbx_m^\top\bbL)=\tr(\hbC_\bbx\bbL), 
\end{equation} 
where $\hbC_\bbx:=\frac{1}{M}\sum_{m=1}^M \bbx_m\bbx_m^\top=\frac{1}{M}\bbX\bbX^\top$ denotes the sample estimate of the covariance of $\bbX$.

When compared with the previous models for graph signals, it is clear that smoothness is a more lenient assumption.
For example, consider that we impose a maximum LV on $\bbx$ with the constraint $\bbx^\top\bbL\bbx \leq c$ for some $c > 0$.
Then, while bandlimited signals are constrained to live in the subspace spanned by $\bbV_K$, smooth signals are only constrained to lie within an ellipsoid.
Last but not least, more advanced notions of smoothness can be defined by considering $\|\bbx-\bbH\bbx\|_2^2$, where $\bbH$ represents a pre-specified low-pass GF whose filter taps/frequency response can be tailored to fit the notion of smoothness at hand. 

% Stationary
\vspace{3mm}\noindent\textbf{Stationary graph signals.}
The definition of graph stationarity connects the statistical properties of \emph{random} graph signals with the underlying graph.  
Formally, a zero-mean random graph signal $\bbx$ is said to be stationary on $\ccalG$ if its covariance matrix $\bbC_\bbx = \mathbb{E}[\bbx\bbx^\top]$ is a positive-semidefinite polynomial of the GSO \cite{marques2017stationary}.
When $\ccalG$ is undirected so that $\bbV\bbV^\top = \mathbf{I}$, a common example of stationary graph signals arises when $\bbx$ is the output of a linear graph-diffusion process whose input is a zero mean white signal $\bbw \in \reals^N$, i.e., when the covariance of $\bbw$ is $\mathbb{E}[\bbw\bbw^\top] = \mathbf{I}$ and $\bbx = \bbH\bbw$.
In this particular case, we have that the covariance of $\bbx$ is given by
\begin{equation}\label{eq:cov_x_st}
    \bbC_\bbx= \mathbb{E}[\bbx\bbx^\top]=\bbH\mathbb{E}[\bbw\bbw^\top]\bbH^\top=\bbH\bbH^\top=\bbH^2.
\end{equation}
Since the graph filter $\bbH$ is by definition a polynomial of the GSO $\bbS$, from the last equality in \eqref{eq:cov_x_st}, it readily follows that $\bbC_\bbx$ is a polynomial of $\bbS$ as well.
As a result, in the spectral domain, it holds that $\bbS$ and $\bbC_\bbx$ share the same eigenvectors, and moreover, we have that the matrices $\bbS$ and $\bbC_\bbx$ commute, i.e., $\bbC_\bbx\bbS = \bbS\bbC_\bbx$.
Finally, we emphasize that graph stationarity does not impose a deterministic condition on $\bbx$ but, instead, it imposes a condition on the covariance of the signal.

\section{Graph inverse problems: denoising and interpolation}\label{sec:denoising_and_interpolation}
% Intro
The models for graph signals discussed previously in \cref{sec:models_signals} have been shown to bear practical relevance in real-world datasets and are widely employed in inverse problems.
Here, we briefly describe traditional approaches to leverage the properties of some of those models when the observed signals are perturbed.

% Inverse problems - Denoising
In the context of \emph{graph signal denoising}, when the perturbation consists in the presence of an additive noise in the signal of interest, we are given the noisy observation $\bbx = \bbx_0 + \bbn$ and the goal is to recover the original signal $\bbx_0 \in \reals^N$.
If $\bbx_0$ is known to be bandlimited and the graph is undirected, an estimate of $\bbx_0$ is readily given by
\begin{equation}
    \hbx_0 = \bbV_K\bbV_K^\top\bbx,
\end{equation}
where, by projecting $\bbx$ onto the subspace spanned by $\bbV_K^\top$, we remove the components of the noise $\bbn$ orthogonal to $\bbV_K^\top$ while retaining all the energy of $\bbx_0$.
Differently, if $\bbx_0$ is known to be smooth and the noise is white and Gaussian, a popular approach is to solve an optimization problem of the form of
\begin{alignat}{2}
    \!\!&\! \hbx_0 = \mathrm{arg min}_{\check\bbx_0} \
    && \|\bbx-\check\bbx_0\|_2^2 + \alpha \check\bbx_0^\top\bbL\check\bbx_0.
\end{alignat}
Here, the estimate $\hbx_0$ is a smooth representation of the noisy observation $\bbx$, with the weight $\alpha$ controlling the trade-off between minimizing the similarity of $\hbx_0$ and $\bbx$ and the LV of $\hbx_0$.
Note that, if the noise is drawn from a different distribution, then a similarity metric other than the $\ell_2$ norm may be preferred. Along the same lines, if additional statistical information about the noise were available (e.g., the covariance of the noise), this can also be incorporated into the minimization problem.

% Inverse problems - Interpolation
Now, let us consider the problem of \emph{graph signal interpolation}. This is a relevant problem in setups where the graph signal has been corrupted with missing values or, alternatively, when only a sampled version of the signal is available due to the fact that only a subset of nodes has been observed.
To be specific, consider the sampling set $\ccalQ \subseteq \ccalV$ with cardinality $Q \leq N$ that collects the set of nodes that have been observed, and define the (fat) selection matrix $\bbPi_\ccalQ \in \{0,1\}^{Q \times N}$ whose elements satisfy: (i) $\sum_j \Pi_{\ccalQ,ij}=1$ for all $i$; and (ii) $\sum_i \Pi_{\ccalQ,ij}=1$ if $j\in \ccalQ$ and $\sum_i \Pi_{\ccalQ,ij}=0$ otherwise.
Then, if the original signal $\bbx$ is bandlimited, and assuming that the observed values correspond to observations at different nodes, we denote the perturbed signal with missing values as $\bbx_\ccalQ := \bbPi_\ccalQ\bbx = \bbPi_\ccalQ\bbV_K\tbx_K$.
Under these conditions, the original signal $\bbx$ can be readily recovered via
\begin{equation}
    \hbx = \bbV_K\tbx_K = \bbV_K(\bbPi_\ccalQ\bbV_K)^\dagger\bbx_\ccalQ,
\end{equation}
provided that the rank of the $Q \times K$ submatrix $\bbPi_\ccalQ\bbV_K$ is $K$. Nonetheless, in some settings, the missing values may be represented more accurately with alternative sampling schemes.
An equally valid, but less intuitive approach to sampling a graph signal, is to fix some node $i$, and consider the sampling of the signal seen by this node as the GSO is applied recursively.
In other words, consider that the signal has been locally diffused according to $\bbS$, as encoded in the matrix
\begin{equation}\label{eq:all_shifts_matrix}
    \bbZ := [\bbz^{(0)},\bbz^{(1)},...,\bbz^{(N-1)}] = [\bbx,\bbS\bbx,...,\bbS^{N-1}\bbx].
\end{equation}
Then, using the matrix $\bbZ$  and with $\bbe_i$ denoting the $i$-th canonical vector, the \emph{successively aggregated} signal at node $i$ is the $i$-th row of $\bbZ$, that is $\bbz_i:=(\bbe_i^T\bbZ)^T=\bbZ^T\bbe_i$.
Sampling is now reduced to the selection of $Q$ out of the $N$ elements of $\bbz_i$, that is $ \bbz_{\ccalQ,i} \ :=\ \bbPi_\ccalQ\bbz_i\ =\ \bbPi_\ccalQ \left(\bbZ^T\bbe_i\right)$.
Leveraging the results in \cite{marques2015sampling}, the signal $\bbx$ can be recovered from $\bbz_{\ccalQ,i}$ as
\begin{equation}
    \bbx = \bbV_K\big(\bbPi_\ccalQ\bbPsi^\top\diag(\bbupsilon_i)\big)^\dagger\bbz_{\ccalQ,i}, \;\;\;\mathrm{with}\;\;\;\bbupsilon_i:=[V_{i,1},...,V_{i,N}]^\top.
\end{equation}

\section{Graph learning}\label{sec:graph_learning}
% Short intro
Graph learning, also known as network topology inference, has developed swiftly in the last years and, currently, is among the most active research areas within GSP.
Given a set of graph signals (nodal observations) collected in the matrix $\bbX = [\bbx_1,\bbx_2,...,\bbx_M] \in \reals^{N \times M}$, which are typically assumed to be independent realizations of a random network process, the goal is to discover the topology of the graph encoded in the GSO by assuming that the observed signals $\bbX$ and the unknown graph are intimately connected.
\cref{fig:nti_example} illustrates the case where a graph learning algorithm is employed to learn the connection between the different regions of the brain based on the signals measured at each region.
Intuitively, the relation between $\bbX$ and $\ccalG$ will depend on the application at hand, with different relations between the observations and the unknown topology leading to different graph learning algorithms.
Here, we will provide a succinct summary of the most relevant approaches based on \cite{mateos2019connecting}.
The interested reader is referred there for additional details.

%%%%%%%%%%   FIGURE   %%%%%%%%%%
\begin{figure}[b]
    \centering
    \includegraphics[width=\textwidth]{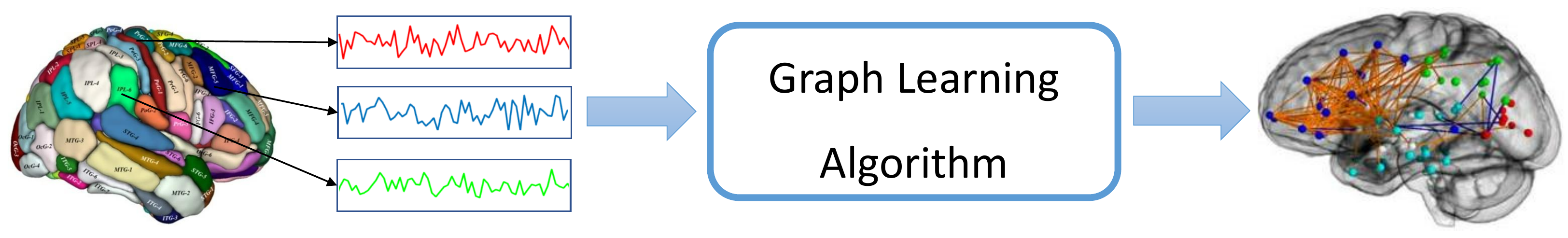}
    \caption{Application of a generic graph learning scheme. The input, represented on the left, is given by signals measured in the different regions of the brain (the nodes of the graph). Then, the output of the graph learning algorithm, on the right, are the inferred connections between the different regions, i.e., the estimated topology of the network.}
    \label{fig:nti_example}
\end{figure}
%brain images from:
%-https://francis.naukas.com/2013/12/08/francis-en-rosavientos-noticias-para-el-sabado/
%-https://www.elmundo.es/cataluna/2016/07/13/577e3fc3e2704eac368b4646.html
%%%%%%%%%%%%%%%%%%%%%%%%%%%%%%%%

% Correlation networks
One of the first methods to estimate the topology $\ccalG$ is given by \emph{correlation networks}, where the topology is obtained from the Pearson correlation of the i.i.d. random vectors collected in $\bbX$.
The Pearson correlation coefficient between variables $x_i$ and $x_j$ is denoted as $\rho_{ij}$ and can be computed from the entries of the covariance matrix $[\bbC_\bbx]_{ij}$.
Therefore, in the context of GSP, the GSO for correlation networks is usually set to the sampled covariance $\hbC_\bbx$, or a thresholded version to ensure a sparse matrix $\bbS$.

% Partial correlation/Graphical Lasso
While correlation networks are a simple alternative, high correlations may be due to latent network effects.
For example, the random variables $x_i$ and $x_j$ may be highly correlated not because the nodes $i$ and $j$ are connected but because of a third node $k$ that is influencing both of them.
In principle, such a confounding can be resolved by considering the \emph{partial correlation} coefficients
\begin{equation}
    \rho_{ij|\ccalV \setminus ij} := \frac{\mathrm{cov}(x_i,x_j | \ccalV \!\setminus\! ij)}{\sqrt{\mathrm{var}(x_i | \ccalV \!\setminus\! ij)\mathrm{var}(x_j | \ccalV \!\setminus\! ij)}}.
\end{equation}
Here, $\ccalV \!\setminus\! ij$ denotes the set of random variables except for those indexed by nodes $i$ and $j$.
The edge set in partial correlation networks is then defined analogously to their (unconditional) correlation network counterpart.  % by setting $A_{ij} = \rho_{ij|\ccalV \setminus ij}$.

% Graphical Lasso
Of particular interest is the case when each column of $\bbX$ is sampled independently from the same Gaussian distribution.
Under such an assumption, $ \rho_{ij|\ccalV \setminus ij} = 0$ implies that $x_i$ and $x_j$ are conditionally independent given the remaining variables in $\ccalV \setminus ij$.
The resulting partial correlation network is known as \emph{\acrfull{gmrf}} or \emph{Gaussian graphical model}~\cite{yuan2007model}.
Then, the key realization is that the partial correlation coefficients $\ccalV \setminus ij$, which capture the topology of the graph, can be expressed as the normalized entries of $\bbC_\bbx^{-1}$.
In the context of GMRFs, this important matrix is known as the \emph{precision} matrix. From the GSP perspective, the previous discussion implies that $\bbS = \bbC_\bbx^{-1}$.
In other words, the topology of the graph is encoded in the inverse covariance matrix.
Leveraging this observation, the notorious \emph{graphical Lasso} algorithm~\cite{friedman2008sparse} estimates $\bbS$ through the regularized \acrfull{ml} estimator
\begin{equation}
    \hbS = \argmax_{\bbS \succeq 0} \;\; \log\det (\bbS) - \tr(\hbC_\bbx\bbS) - \lambda\|\bbS\|_1,
\end{equation}
where the $\|\bbS\|_1$ denotes the $\ell_1$ norm of the vectorization of $\bbS$.

% Stationarity
Finally, we consider a network topology inference approach that builds upon the more lenient assumption of stationary graph signals.
First, in correlation networks, it was assumed that $\bbS = \bbC_\bbx$, and later on, in the graphical Lasso algorithm and other GMRF approaches the relation between the GSO and the covariance of the observed signals is constrained to $\bbS = \bbC_\bbx^{-1}$.  
In contrast, the assumption of stationary graph signals only implies that the mapping $\bbS \to \bbC_\bbx$ is given by a generic polynomial, hence including the previous scenarios as particular cases.
While there are different formulations for this graph learning approach, one particularly interesting for this thesis is given by
\begin{equation}
    \!\! \hbS = \argmin_{\bbS}  \; \|\bbS\|_1 \;\;\;\;\mathrm{s. \;t. } \;\;\; \|\bbS\hbC_\bbx - \hbC_\bbx\bbS \|_F^2 \leq \epsilon, \;\;\; \bbS \in \ccalS,
\end{equation}
which is formulated solely in the node domain thanks to the commutativity constraint $\|\bbS\hbC_\bbx - \hbC_\bbx\bbS \|_F^2$.
The optimization problem finds the sparsest GSO that commutes with $\hbC_\bbx$, with $\epsilon$ being a small positive parameter controlling the quality of the estimate $\hbC_\bbx$, and with $\ccalS$ collecting the requirements for $\bbS$ to be a specific type of GSO.
A typical example is the set of adjacency matrices
\begin{equation}\label{eq:A_set}
    \ccalS_\bbA := \{ A_{ij} \geq 0; \ \bbA = \bbA^\top; \ A_{ii} = 0; \ \bbA\mathbf{1}\geq\mathbf{1}\},
\end{equation}
where we require the GSO to have non-negative weights, be symmetric, and have no self-loops, and the last constraint rules out the trivial $0$ solution by imposing that every node has at least one neighbor.
Analogously, the set of combinatorial Laplacian matrices is
\begin{equation}\label{eq:L_set}
    \ccalS_\bbL := \{ L_{ij} \leq 0 \;\mathrm{for} \; i \neq j; \; \bbL=\bbL^\top; \;   \bbL \textbf{1} = \bb0; \; \bbL  \succeq 0 \},
\end{equation}
where we require the GSO to be a positive semidefinite matrix, have non-positive off-diagonal values, have positive entries on its diagonal, and have the constant vector as an eigenvector (i.e, the sum of the entries of each row to be zero).

% Link prediction and network tomography
So far the section has been focused on a graph learning setting that, in the network science parlance, is known as the \emph{network association} problem~\cite{kolaczyk2009book}.
While network association is the most widely considered approach in the context of graph learning, two relevant variations are:  the link prediction problem and the network tomography problem. \emph{Link prediction} is a simpler problem in which a subset of the edges of the graph is observed along with the signals.
This additional information can be incorporated into the previous framework by modifying the constraint set $\ccalS$.
In contrast, \emph{network tomography} is a more challenging task where the observed signals are perturbed and observations from only a subset of the nodes are available. 
Precisely, developing robust algorithms that address the latter problem leveraging several GSP assumptions is the subject of \cref{chap:nti_hidden}.

\section{Graph neural networks}\label{sec:gnn}
% Generic expression for GNN
Generically, we represent a GNN using a parametric nonlinear function $f_{\bbTheta}(\bbZ|\ccalG):\reals^{N^{(0)} \times F^{(0)}} \rightarrow \reals^N$ that depends on the graph $\ccalG$. 
The parameters of the architecture are collected in $\bbTheta$, and the matrix $\bbZ \in \reals^{N^{(0)} \times F^{(0)}}$ represents the input of the network.
Despite the many possibilities for defining a GNN, a broad range of such architectures recursively applies a graph-aware linear transformation followed by an entry-wise nonlinearity.
Then, a generic deep graph-based architecture $f_{\bbTheta}(\bbZ|\ccalG)$ with $L$ layers can be described as
\begin{align}
    \hbY^{(\ell)}&=\ccalT_{\bbTheta^{(\ell)}}^{(\ell)}\left\{ \bbY^{(\ell-1)}|\ccalG\right\}, \;\; 1 \leq\ell\leq L,  \label{eq:graph_aware_lin_trans_generic_GNN} \\
    Y^{(\ell)}_{ij}&=g^{(\ell)}\left( \hat{Y}^{(\ell)}_{ij} \right),  \;\; 1 \leq\ell\leq L,  \label{eq:nonlinear_trans} 
\end{align}
where $\ccalT^{(\ell)}_{\bbTheta^{(\ell)}}\{\cdot|\ccalG\}\colon \reals^{N^{(\ell-1)} \times F^{(\ell-1)}} \rightarrow \reals^{N^{(\ell)}\times F^{(\ell)}}$ is a \textit{graph-aware} linear transformation, $\bbY^{(0)}=\bbZ$ and $\bby=\bbY^{(L)}$ denote the input and output of the architecture, $\bbTheta^{(\ell)} \in \reals^{F^{(\ell-1)} \times F^{(\ell)}}$ are the parameters that define such a transformation, and $g^{(\ell)}\colon \reals\rightarrow\reals$ is a \textit{scalar} nonlinear transformation (e.g., the ReLU function), which is oftentimes omitted in the last layer.
Moreover, $N^{(\ell)}$ and $F^{(\ell)}$ represent the number of nodes and features at layer $\ell$, and $\bbTheta=\{\bbTheta^{(\ell)}\}_{\ell=1}^L$ collects all the parameters of the architecture.
The structure of the generic GNN specified by the recursion in \eqref{eq:graph_aware_lin_trans_generic_GNN}-\eqref{eq:nonlinear_trans} is depicted in \cref{fig:GNN_diagram}.
Finally, even though the output of $f_{\bbTheta}(\bbZ|\ccalG)$ are graph signals defined in $\reals^N$, which is the case of interest for \cref{chap:denoising}, it can be easily adapted to output graph signals with more than one feature.

%%%%%%%%%%   FIGURE   %%%%%%%%%%
\begin{figure}[tb]
    \centering
    \includegraphics[width=.96\textwidth]{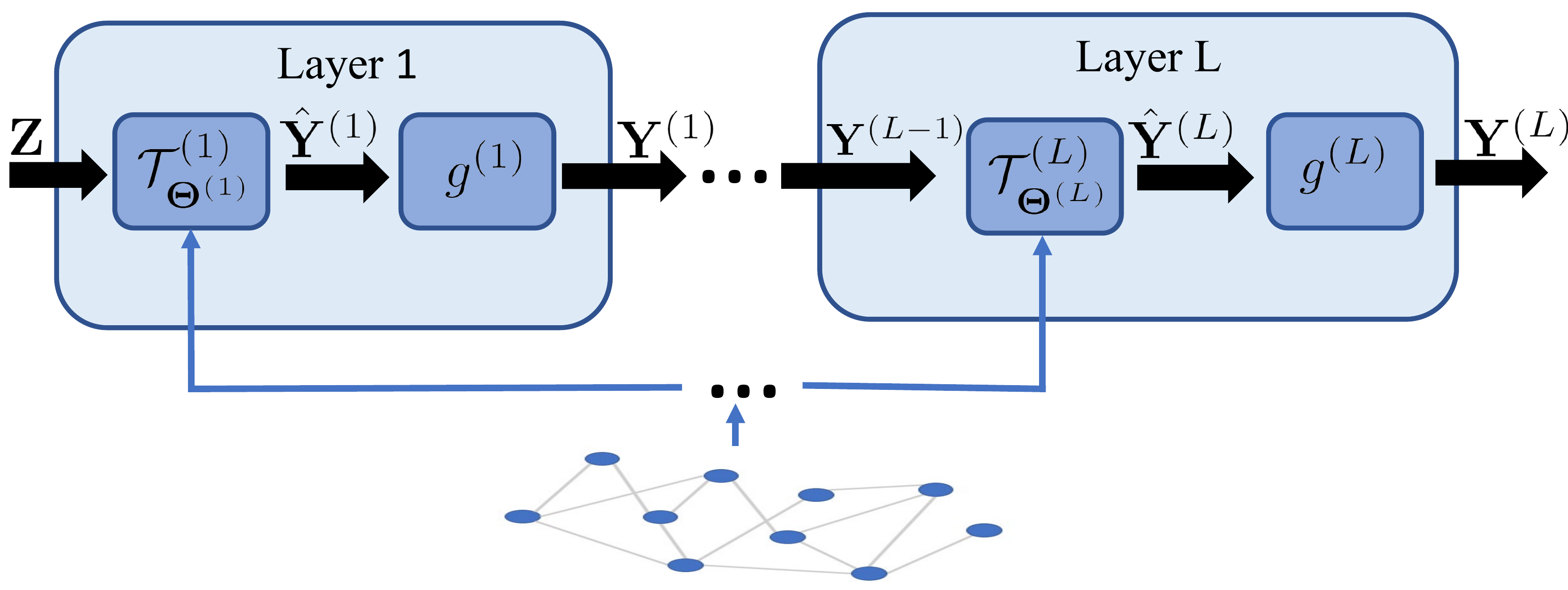}
    \caption{Block diagram of a generic GNN with $L$ layers. The inputs of the architecture are the matrix $\bbZ$ and the topology of the graph. As specified in \eqref{eq:graph_aware_lin_trans_generic_GNN}-\eqref{eq:nonlinear_trans}, each layer is composed of a learnable graph-aware linear transformation $\ccalT_{\bbTheta^{(\ell)}}^{(\ell)}\left\{ \bbY^{(\ell-1)}|\ccalG\right\}$ followed by a entry-wise non-linear transformation $g^{(\ell)}\left( \hat{Y}^{(\ell)}_{ij} \right)$.}
    \label{fig:GNN_diagram}
\end{figure}
%%%%%%%%%%%%%%%%%%%%%%%%%%%%%%%%

\section{Graph perturbations in GSP}\label{sec:perturbations_gsp}
The presence of noise in the observed topology represents a relevant but challenging problem that is yet to be studied in more depth by the GSP community.
Current works addressing this issue customarily represent the influence that these perturbations exert in the GSO via an additive term because of its tractability, but even then, the resulting models are non-trivial.
The source of this difficulty lies in the main tools used in GSP, which are  based either on the GFT (eigenvectors of the GSO) or on graph filters (polynomials of the GSO).
The challenges are then twofold: (i)~characterizing the impact that an additive matrix perturbation has on the eigenvectors and/or a polynomial of that matrix is highly nontrivial; and (ii)~even small perturbations on $\bbS$ may lead to great discrepancies in both the eigenvectors and the associated polynomials, as we show in~\cref{chap:robust_filter_id}.
Here, we present a succinct overview of relevant works considering the influence of noise in the observed topology to provide some context.

% Small perturbation in L --> focus: analysis
We start with the work presented in~\cite{ceci2020graph}, which analyzes how perturbations in the edges affect the spectrum of the combinatorial graph Laplacian $\bbL$.
The authors assume an additive perturbation model and define the perturbed Laplacian  as $\barbL := \bbL + \Delta\bbL$, with $\Delta\bbL$ denoting the perturbation matrix.
Assuming that all the eigenvalues of $\bbL$ have multiplicity one and that $\|\Delta\bbL\|_F \ll \|\bbL\|_F$, they perform a small perturbation analysis to quantify the influence of the perturbations in the eigenvalues and eigenvectors of $\bbL$.
Based on this result, and assuming that the perturbation of each edge is modeled as a random event characterized by a certain probability, a statistical analysis is carried out to characterize the mean and the variance of the perturbation of the eigenvalues.
Lastly,~\cite{ceci2020graph} studies the influence of the perturbations on the spectrum of BGS when the eigenvectors of $\bbL$ are used as the GFT.

% Graphon analysis --> focus: propose a random graph model
Differently,~\cite{miettinen2019modelling} investigated the influence of perturbations in the adjacency matrix leveraging results from the graphon theory.
The perturbed adjacency matrix is also defined based on an additive perturbation model as $\barbA = \bbA + \bbDelta_\epsilon \circ (\mathbf{1}_{N \times N} - 2\bbA)$, where $\mathbf{1}_{N \times N}$ is the matrix of all ones of size $N \times N$.
The perturbations in $\bbDelta_\epsilon$ are modeled as a random graph drawn from either an \acrfull{er} or a \acrfull{sbm}, and then, generalizations considering different probabilities for creating and destroying edges and dealing with weighted graphs are also proposed.
Finally, the model put forth in~\cite{miettinen2019modelling} is employed to analyze the influence of the perturbations in a polynomial of the GSO of order 2.

\vspace{3mm}
To close this chapter, we stress that, for ease of exposition, we presented a taxonomy where the fundamental concepts, tools, and problems in the GSP framework were clearly segregated.
Nonetheless, in practical settings we can encounter different combinations of the above problems and generative models giving rise to a rich gamut of GSP tasks.
This can be seen in \emph{blind deconvolution}, which is a graph filter identification problem when $\bbx$ is a DSGS and the seeding nodes in $\bbs$ are unknown, or in the problem approached in \cref{chap:robust_filter_id}, where we simultaneously estimate a graph filter and denoise the observed topology of the graph.

\chapter{Non-linear denoising of graph signals}\label{chap:denoising}
The first problem considered in the robust GSP framework proposed in this thesis involves the presence of noise in the observed graph signals.
As discussed in previous chapters, the presence of noise represents a pervasive type of perturbation capable of rendering the observed data useless when the signal-to-noise ratio is low.
As a result, (pre-)processing schemes that remove the noise from the observed signals are required.
It is worth recalling that, because the presence of noise in the signals does not affect the topology of the graph and results in tractable problems, there are several works addressing the denoising of graph signals.
In this sense, the approach described in this chapter, which encapsulates our work from~\cite{rey2019underparametrized,rey2021untrained}, is primarily concerned with incorporating the information encoded in the graph topology into non-linear architectures and, furthermore, providing a mathematical characterization of the denoising capabilities of the proposed architectures. 

Bearing the previous comments in mind, the chapter is organized as follows.
\cref{sec:denoise_intro} gives a brief overview of the architectures developed and summarizes the main contributions.
\cref{sec:GNN_for_inverse_problems} formally introduces the problem at hand and presents our general approach.
\cref{sec:conv_dec} and \cref{sec:ups_dec} detail the proposed architectures and provide the mathematical analysis for each of them.
Numerical experiments are presented in \cref{sec:experiments} and concluding remarks are provided in \cref{sec:conclusion}.

\section{Introduction}\label{sec:denoise_intro}
In order to develop a \emph{graph-aware non-linear architecture} capable of removing the noise from the observed signals, the goal of this chapter is twofold.
First, we explore different ways of incorporating the information encoded in the graph and propose new graph-based NN architectures to denoise graph signals.
Second, we provide theoretical guarantees for the denoising capabilities of this approach and show that such guarantees are directly influenced by the properties of the graph.
The mathematical analysis, performed on particular instances of these architectures, characterizes their denoising performance under specific assumptions for the original signal and its underlying graph.
In addition, we provide empirical evidence about the denoising performance of our method for scenarios more general  than those strictly covered by our theory, further illustrating the value of our graph-aware untrained architectures to denoise graph signals.

% Proposed method
The presented architectures are \textit{untrained} NNs, meaning that the parameters of the network are optimized using only the signal observation that we want to denoise, avoiding the dependency on a training set with multiple observed graph signals.
The underlying assumption behind this \textit{untrained} denoising architecture is that, due to the graph-specific structure incorporated into the different layers, when tuning the network parameters using stochastic gradient steps, the NNs are capable of learning (matching) the structure of the signal faster than that of the noise.
Hence, the denoising process is carried out separately for each individual observation by fitting the weights of the NN and stopping the updates after a few iterations.
This same phenomenon has been observed to hold true in non-graph deep learning architectures\cite{ulyanov2018deep,heckel2018deep}.
In the context of signal denoising, the consideration of an overparametrized graph-aware architecture along with early stopping avoids overfitting to the noise.

% Leveraging the graph topology
To incorporate the topology of the graph, the first architecture multiplies the input at each layer by a fixed (non-learnable) graph filter \cite{segarra2017optimal}, which can be seen as a generalization of the convolutional layer
in~\cite{kipf2016semi}.
The second architecture performs graph upsampling operations that, starting from a low-dimensional latent space, progressively increase the size of the input until it matches the size of the signal to denoise.
The sequence of upsampling operators are designed based on hierarchical clustering algorithms~\cite{jain1988algorithms, carlsson2018hierarchical, carlsson2013axiomatic, rey2019deep} so that, in contrast to~\cite{do2020graph}, matrix inversions are not required, avoiding the related numerical issues.

% Contributions
\vspace{3mm}
\noindent
\textbf{Contributions.}
In summary, the contributions of this chapter are the following: 
\begin{itemize}
    \item[(i)] We introduce two new overparametrized and untrained GNNs for solving graph-signal denoising problems.
    \item[(ii)] We characterize theoretically the denoising performance of each of the two architectures, improving our understanding of nonlinear architectures and the influence of incorporating graph structure into NNs.
    \item[(iii)] The proposed architectures are evaluated and compared to other denoising alternatives through numerical experiments carried out with synthetic and real-world data.
\end{itemize}

\section{GNNs for graph-signal denoising}\label{sec:GNN_for_inverse_problems}
% Sec. outline
We now formally introduce the problem of graph-signal denoising within the GSP framework, and present our approach to tackle it using untrained GNN architectures.
Given the graph $\ccalG$, let us consider the observed graph signal $\bbx \in \reals^N$, which is a noisy version of the original graph signal $\bbx_0$. With $\bbn\in \reals^N$ being a noise vector, the relation between $\bbx$ and $\bbx_0$ is
\begin{equation}\label{eq:noise_model}
    \bbx = \bbx_0 + \bbn.
\end{equation}
Then, the goal of graph-signal denoising is to remove as much noise as possible from the observed signal $\bbx$ to estimate the original signal $\bbx_0$, which is performed by exploiting the information encoded in $\ccalG$.

Recall that a traditional approach for the graph-signal denoising task is to solve an optimization problem of the form
\begin{alignat}{2}\label{eq:denoising_reg}
    \!\!&\! \hbx_0 = \mathrm{arg min}_{\check\bbx_0} \
    && \|\bbx-\check\bbx_0\|_2^2 + \alpha R(\check\bbx_0|\ccalG).
\end{alignat}
The first term promotes fidelity to the signal observations, the regularizer $R(\cdot|\ccalG)$ promotes denoised signals with desirable properties over the given graph $\ccalG$, and $\alpha>0$ controls the influence of the regularization.
Common choices for the regularizer include the quadratic Laplacian $R(\bbx|\ccalG)=\bbx^\top\bbL\bbx$~\cite{pang2017graph},
or regularizers involving high-pass graph filters $R(\bbx|\ccalG)=\|\bbH\bbx\|_2^2$ that foster smoothness on the estimated signal\cite{sandryhaila2014discrete,chen2014signal}.

%%%%%%%%%%%%%%%%%%%%%%%%%%%%%%%%%%%%%%%%%%%%%%%%%
%%%%%          Denoising algorithm          %%%%%
%%%%%%%%%%%%%%%%%%%%%%%%%%%%%%%%%%%%%%%%%%%%%%%%%
\begin{algorithm}[tb]
\SetKwInOut{Input}{Inputs}
\SetKwInOut{Output}{Outputs}
\Input{$\bbx$ and $\ccalG$}
\Output{$\hbx_0$ and $\hbTheta(\bbx)$}
\SetAlgoLined
Set $f_{\bbTheta}(\bbZ|\ccalG)$ as explained in \cref{sec:conv_dec} or \cref{sec:ups_dec} \\
Generate $\bbZ$ from iid zero-mean Gaussian distribution \\
Initialize $\bbTheta_{(0)}$ from iid zero-mean Gaussian \\
%Set the number of iteration $T$ small enough to avoid learning the noise \\
\For{$t=1$ \KwTo $T$}{
 Update $\bbTheta_{(t)}$ minimizing \eqref{eq:nonlinear_denoising} with SGD
}
%Minimize \eqref{eq:nonlinear_denoising} running SGD for $T$ iterations \\
$\hbTheta(\bbx)=\bbTheta_{(T)}$ \\
$\hbx_0=f_{\hbTheta(\bbx)}(\bbZ|\ccalG)$ \\
\caption{Proposed graph-signal denoising method}
\label{A:denoising}
\end{algorithm}

% Our approach to signal denoising
While those traditional approaches exhibit a number of advantages (including interpretability, mathematical tractability, and convexity),
they may fail to capture more complex relations between $\ccalG$ and $\bbx_0$, motivating the development of nonlinear graph-denoising approaches. 

As summarized in Algorithm~\ref{A:denoising}, in this chapter we advocate handling the graph-signal denoising task by employing an overparametrized GNN (denoted by $f_{\bbTheta}(\bbZ|\ccalG)$) as described in \eqref{eq:graph_aware_lin_trans_generic_GNN}-\eqref{eq:nonlinear_trans}.
The weights of the architecture, collected in $\bbTheta$, are learned by minimizing the loss function
\begin{align}
\label{eq:nonlinear_denoising}
\ccalL(\bbx,\bbTheta) = \frac{1}{2}\|\bbx-f_{\bbTheta}(\bbZ|\ccalG)\|_2^2,
\end{align}
applying \acrfull{sgd} in combination with early stopping to avoid overfitting the noise.
The entries of the parameters $\bbTheta$ and the input matrix $\bbZ$ are initialized at random using an i.i.d. zero-mean Gaussian distributions, and the weights learned after a few iterations of denoising the observation $\bbx$ are denoted as $\hbTheta(\bbx)$.
Note that $\bbZ$ is fixed to its random initialization.
Finally, the denoised graph signal estimate is computed as 
\begin{equation}
    \hbx_0=f_{\hbTheta(\bbx)}(\bbZ|\ccalG).    
\end{equation}

The intuition behind this approach is as follows: since the architecture is overparametrized it can in principle fit any signal, including noise.
However, as shown formally later, both empirically and theoretically, the proposed architectures fit graph signals faster than the noise and, therefore, with early stopping they fit most of the signal and little of the noise, enabling signal denoising.

{
\vspace{1mm}\noindent
\textbf{Remark 1.} 
    The proposed architectures are described as \textit{untrained} NNs because, when minimizing \eqref{eq:nonlinear_denoising}, the weights in $\bbTheta$ are learned to fit \emph{each observation} $\bbx$, with the denoised signal $\hbx_0$ being the output for those particular weights.
    This implies that each noisy-denoised signal pair $(\bbx,\hbx_0)$ is associated with a particular value of the weights $\bbTheta$, in contrasts with trainable NNs, where the weights $\bbTheta$ are first learned by fitting the signals in a \emph{training set} and later used (unchanged) to denoise signals that were not in the training set. }

% Link with following section 
Regarding the specific implementation of the untrained network $f_{\bbTheta}(\bbZ|\ccalG)$, there are multiple possibilities for selecting the linear and nonlinear transformations $\ccalT^{(\ell)}_{\bbTheta^{(\ell)}}$ and $g^{(\ell)}$ defined in equations~\eqref{eq:graph_aware_lin_trans_generic_GNN} and~\eqref{eq:nonlinear_trans}, respectively.
As customary in NNs dealing with signals defined in $\reals^N$, we select the $\relu$ operator, defined as $\relu(x)=\max(0,x)$, to be the entrywise nonlinearity $g^{(\ell)}$.
Then, we focus on the design of the linear transformation, which is responsible for incorporating the structure of the graph.
The two following sections postulate the implementation of two particular linear transformations $\ccalT^{(\ell)}_{\bbTheta^{(\ell)}}$ (each giving rise to a different GNN) and analyze the resulting architectures.

\section{Graph convolutional generator}\label{sec:conv_dec}
% Present the architecture --> convolutional layer
Our first architecture to address the graph-signal denoising task is a \acrfull{gcg} network that incorporates the topology of the graph into the NN pipeline via vertex-based graph convolutions.
Then, leveraging the fact that convolutions of a graph signal on the vertex domain can be represented by a graph filter $\bbH \in \reals^{N \times N}$ \cite{segarra2017optimal}, we define the linear transformation for the convolutional generator as %(cf.~\eqref{eq:graph_aware_lin_trans_generic_GNN})
\begin{equation}\label{eq:linear_trans_gcg}
      \ccalT^{(\ell)}_{\bbTheta^{(\ell)}}\{\bbY^{(\ell-1)}|\ccalG\} = \bbH\bbY^{(\ell-1)}\bbTheta^{(\ell)}.
\end{equation}
Remember that the $F^{(\ell-1)} \times F^{(\ell)}$ matrix $\bbTheta^{(\ell)}$ collects the learnable weights of the $\ell$-th layer, and the graph filter $\bbH$ is given by \eqref{eq:graph_filter}. {The coefficients $\{h_r\}_{r=0}^{R-1}$ are fixed a priori so that $\bbH$ promotes desired properties on the estimated signal.}
Using the linear transformation defined in \eqref{eq:linear_trans_gcg}, the output of the GCG with $L$ layers is given by the recursion
\begin{align}
      \bbY^{(\ell)} &=\relu(\bbH\bbY^{(\ell-1)}\bbTheta^{(\ell)}),\;\; \mathrm{for}\;\; \ell=1,...,L-1, \label{eq:gcg1}\\ 
      \bby^{(L)} &= \bbH\bbY^{(L-1)}\bbTheta^{(L)}, \label{eq:gcg2}
\end{align}
where $\bbY^{(0)}=\bbZ$ {denotes the random input} and the $\relu$ is not applied in the last layer of the architecture.
% Advantages of the proposed linear transformation
With the proposed linear transformation, the GCG learns to combine the features within each node by fitting the weights of the matrices $\bbTheta^{(\ell)}$ while the graph filter $\bbH$ interpolates the signal by mixing features from $R-1$ neighborhoods.

% Differences wrt GCNN
Even though the proposed GCG exploits graph convolutions to incorporate the graph topology into the architecture, it is intrinsically different from other GCNNs.
The linear transformation proposed in \cite{kipf2016semi}, arguably one of the most popular implementations of GCNNs, is given by
\begin{equation}\label{eq:gcnn_kipf}
    \ccalT^{(\ell)}_{\bbTheta^{(\ell)}}\{\bbY^{(\ell-1)}|\ccalG\} = (\bbA+\bbI)\bbY^{(\ell-1)}\bbTheta^{(\ell)}.
\end{equation}
Recalling the definition of graph filters in \eqref{eq:graph_filter}, it is evident that \eqref{eq:gcnn_kipf} is a particular case of our proposed linear transformation, obtained by setting the generative graph filter to $\bbH=\bbA+\bbI$, a low-pass graph filter of degree one.
In addition to representing a more general scenario, \eqref{eq:gcg1} endows the GCG with two main advantages.
First, the graph filter $\bbH$ allows us to incorporate prior information on the signals to denoise, making our GCG architecture more suitable to denoise a (high-) low-frequency signal by employing a (high-) low-pass filter.
Second, in \eqref{eq:gcnn_kipf} there is an equivalence between the depth of the network and the radius of the considered neighborhood, so that gathering information from nodes that are $R$ hops apart requires a GNN with $R$ layers. In contrast, with the architecture considered in \eqref{eq:gcg1}, the same can be achieved by considering a GCG with $L$ layers and a graph filter $\bbH$
of degree $R/L$~\cite{segarra2017optimal}, reducing the number of learnable parameters and bypassing some of the well-known over-smoothing problems associated with \eqref{eq:gcnn_kipf}~\cite{chen2020measuring}.

% Sec. outline
Next, we adopt some simplifying assumptions to provide theoretical guarantees on the denoising capability of the GCG (\cref{sec:analysis_GCG}). Then, we rely on numerical evaluations to demonstrate that the results also hold in more general settings (\cref{sec:analyze_deep_gcg}).

\subsection{Guaranteed denoising with the GCG}\label{sec:analysis_GCG}
% Present 2-layer GCG 
To formally prove that the proposed architecture can successfully denoise the observed graph signal $\bbx$, we consider a two-layer GCG given by
\begin{equation}\label{eq:2layer_gen}
    f_{\bbTheta}(\bbZ|\ccalG) = \relu(\bbH\bbZ\bbTheta^{(1)})\bbtheta^{(2)},
\end{equation}
where $\bbTheta^{(1)}\in\reals^{F\times F}$ and $\bbtheta^{(2)}\in\reals^{F}$ are the learnable coefficients. With $F$ denoting the number of features, we consider the overparametrized regime where $F \geq 2N$, and analyze the behavior and performance of denoising with the untrained network defined in~\eqref{eq:2layer_gen}. 

We start by noting that scaling the $i$-th entry of $\bbtheta^{(2)}$ is equivalent to scaling the $i$-th column of $\bbTheta^{(1)}$, so that, without loss of generality, we can set the weights to $\bbtheta^{(2)}=\bbb$, where $\bbb$ is a vector of size $F$ with half of its entries set to $1/\sqrt{F}$ and the other half to $-1/\sqrt{F}$.
Furthermore, since $\bbZ$ is a random matrix of dimension $N \times F$, the column space of $\bbZ$ spans $ \reals^N$, and hence, minimizing over $\bbZ\bbTheta^{(1)}$ is equivalent to minimizing over $\bbTheta\in\reals^{N\times F}$.
With these considerations in place, the optimization over \eqref{eq:nonlinear_denoising} can be performed replacing the two-layer GCG described in \eqref{eq:2layer_gen} by its simplified form
\begin{equation}\label{eq:2layer_gcg_simp}
    f_{\bbTheta}(\bbH) = f_{\bbTheta}(\bbZ|\ccalG) = \relu(\bbH\bbTheta)\bbb.
\end{equation}
Note that we replaced $f_{\bbTheta}(\bbZ|\ccalG)$ with  $f_{\bbTheta}(\bbH)$ since the graph influence is modeled by the graph filter $\bbH$, and the influence of the matrix $\bbZ$ is absorbed by the learnable weights $\bbTheta$.

% Need for relating eigenvectors of the squared Jacobian with those of the graph
The denoising capability of the two-layer architecture is related to the eigendecomposition of its expected squared Jacobian\cite{heckel2019denoising}.
However, to understand which signals can be effectively denoised with the proposed architecture, we need to connect the spectral domain of the expected squared Jacobian with the spectrum of the graph, given by the eigenvectors of the adjacency matrix.

% Derivation of the expected squared Jacobian
To that end, we next compute the expected squared Jacobian of the two-layer architecture in \eqref{eq:2layer_gcg_simp}.
Denote as $\ccalJ_{\bbTheta}(\bbH)\in~\reals^{N \times NF}$ the Jacobian matrix of $f_{\bbTheta}(\bbH)$  with respect to $\bbTheta$, which is given by 
\begin{align}\label{eq:jacobian_gcg}
    \ccalJ^\top_{\bbTheta}(\bbH)
    =
    \begin{bmatrix}
    \mathrm{b}_1 \bbH^\top \diag( \relu'(\bbH \bbtheta_1) ) \\
    \vdots \\
    \mathrm{b}_F \bbH^\top \diag( \relu'(\bbH \bbtheta_F) ) \\
    \end{bmatrix}
    \in \reals^{NF \times N},
\end{align}
where $\bbtheta_i$ represents the $i$-th column of $\bbTheta$, and $\relu'$ is the derivative of the $\relu$, which is the Heaviside step function.
Then, define the $N\times N$ expected squared Jacobian matrix as
\begin{equation}\label{eq:sqJac_step1}
    \sqJacob := \mathbb{E}_{\bbTheta}[\ccalJ_{\bbTheta}(\bbH)\ccalJ^\top_{\bbTheta}(\bbH)] = {\sum_{i=1}^F b_i^2 \mathbb{E}\left[\relu'(\bbH \bbtheta_i)\relu'(\bbH\bbtheta_i)^\top\right] \circ \bbH\bbH^\top. }
\end{equation}
Moreover, from the work in \cite[Sec. 3.2]{daniely2016toward}, we note that $\mathbb{E}\left[\relu'(\bbH \bbtheta_i)\relu'(\bbH\bbtheta_i)^\top\right]$ is in fact the so-called dual activation of the step function.
Therefore, combining the expression for the dual activation of the step function from \cite[Table~1]{daniely2016toward} with \eqref{eq:sqJac_step1}, we obtain that
\begin{equation}\label{eq:expected_jac}
    \sqJacob = 0.5 \left( \mathbf{1} \mathbf{1}^\top - {\pi}^{-1} \arccos(\bbC^{-1} \bbH^2 \bbC^{-1})\right) \circ \bbH \bbH^\top,
\end{equation}
where $\circ$ represents the Hadamard (entry-wise) product, $\arccos(\cdot)$ is computed entry-wise, $\bbh_i$ represents the $i$-th column (row) of $\bbH$, $\bbC=\diag([\|\bbh_1\|_2,...,\|\bbh_N\|_2])$ is a normalization term so that $\bbC^{-1} \bbH^2 \bbC^{-1}$ is the autocorrelation of the graph filter $\bbH$.

% Eigendecomposition of J
Since $\sqJacob$ is symmetric and positive (semi) definite, it has an eigendecomposition $\sqJacob=\bbW\bbSigma\bbW^\top$.
Here, the columns of the orthonormal matrix $\bbW=[\bbw_1,\ldots,\bbw_N]$ are the $N$ eigenvectors, and the nonnegative eigenvalues in the diagonal matrix $\bbSigma$ are assumed to be ordered as $\sigma_1 \geq \sigma_2 \geq ... \geq \sigma_N$.

% Motivate relevance of the result
After defining the two-layer GCG  $f_{\bbTheta}(\bbH)$ and its expected square Jacobian $\sqJacob$, we formally analyze its performance when denoising bandlimited graph signals.
This is particularly relevant given the importance of (approximate) bandlimited graph signals both from analytical and practical points of view~\cite{djuric2018cooperative}.
For the sake of clarity, we first introduce the main result (\cref{theorem_denoising_gcg}) and then we detail a key intermediate result (Lemma~\ref{lemma_eigs_gcg}) that provides additional insight.

Formally, consider the $K$-bandlimited graph signal $\bbx_0$ as described in \eqref{eq:bl_signals}, and let the architecture $f_{\bbTheta}(\bbH)$ have a sufficiently large number of features $F$:
\begin{equation}\label{bound_on_F}
    F \geq \left( \frac{\sigma_1^2}{\sigma_N^2}\right)^{26} \xi^{-8}N,\;\;\text{with}\;\xi \in (0,(2\log (2N/\phi))^{-\frac{1}{2}})
\end{equation}
being an error tolerance parameter for some prespecified $\phi$.
Then, for a specific set of graphs {with minimum number of nodes $N_{\epsilon,\delta}$} that is introduced later in the section (cf. Ass.~\ref{A:sbm}), if we solve \eqref{eq:nonlinear_denoising} running gradient descent with a step size $\eta\!\leq\frac{1}{\sigma_1^2}$, the following result holds (see Appendix~\ref{proof_theorem_denoising_gcg}).

\begin{theorem}\label{theorem_denoising_gcg}
    Let $f_{\bbTheta}(\bbH)$ be the network defined in equation~\eqref{eq:2layer_gcg_simp}, and assume it is sufficiently wide, i.e., it satisfies condition \eqref{bound_on_F} for some error tolerance parameter $\xi$. 
    Let $\bbx_0$ be a $K$-bandlimited graph signal spanned by the eigenvectors $\bbV_K$, and let $\bbw_i$ and $\sigma_i$ be the $i$-th eigenvector and eigenvalue of $\sqJacob$. 
    Let $\bbn$ be the noise present in $\bbx$, set $\phi$ and $\epsilon$ to small positive numbers, and let the conditions from Ass.~\ref{A:sbm} hold.
    Then, for any $\epsilon$, $\delta$, there exists some $N_{\epsilon,\delta}$ such that if $N>N_{\epsilon,\delta}$,
    the error for each iteration $t$ of gradient descent with stepsize $\eta$ used to fit the architecture is bounded as
    \begin{align}\label{eq_bound_theorem_fitting_eigs_Jacobian_for_t}
        & \|\bbx_0-f_{\bbTheta_{(t)}}(\bbH)\|_2 \leq  \left((1-\eta\sigma_K^2)^t+\delta(1-\eta\sigma_N^2)^t\right)\|\bbx_0\|_2 \nonumber \\
        & +\xi\|\bbx\|_2 +\sqrt{\textstyle\sum_{i=1}^N((1-\eta\sigma_i^2)^t-1)^2(\bbw_i^\top\bbn)^2},
    \end{align}
    with probability at least $1-e^{-F^2}-\phi-\epsilon$.
\end{theorem}

% Interpret and comment on the theorem
As explained next, the fitting (denoising) bound provided by the theorem first decreases and then increases with the number of iterations $t$. To be more precise, let us analyze separately each of the three terms in the right hand side of \eqref{eq_bound_theorem_fitting_eigs_Jacobian_for_t}. 
The first term captures the part of the signal $\bbx_0$ that is fitted after $t$ iterations while accounting for the misalignment of the eigenvectors $\bbV_K$ and $\bbW_K$.
This term decreases with $t$ and, since $\delta$ can be made arbitrary small (cf. Lemma~\ref{lemma_eigs_gcg}), vanishes for moderately low values of $t$.
The second term is an error term that is negligible if the network is sufficiently wide.
{Therefore,  $\xi$ can be chosen to be sufficiently small by designing the architecture according to the condition in~\eqref{bound_on_F}.}
Finally, the third term, which depends on the noise present in each of the spectral components of the squared Jacobian $(\bbw_i^\top\bbn)^2$, grows with $t$. More specifically, if the $\sigma_i$ associated with a spectral component is very small, the term $(1-\eta\sigma_i^2)$ is close to $1$ and, hence, the noise power in the $i$-th frequency will be small. Only when $t$ grows very large the coefficient $(1-\eta\sigma_i^2)^t$ vanishes and the $i$-th frequency component of the noise is fitted. As a result, if the filter $\bbH$ is designed such that eigenvalues of the squared Jacobian satisfy that $\sigma_K \gg \sigma_{K+1}$, then there will be a range of moderate-to-high values of $t$ for which: i) the first term is zero and ii) only the $K$ strongest components of the noise have been fitted, so that the third term can be approximated as $\sqrt{\sum_{i=1}^K (\bbw_i^\top\bbn)^2 }$. Clearly, as $t$ grows larger, the coefficient $((1-\eta\sigma_i^2)^t-1)$ will also be close to one for $i>K$, meaning that additional components of the noise will be fitted as well, deteriorating the performance of the denoising architecture. This implies that if the optimization algorithm is stopped before $t$ grows too large, the original signal is fitted along with the noise that aligns with the signal, but not the noise present in other components.

In other words, \cref{theorem_denoising_gcg} not only characterizes the performance of the two-layer GNN, but also illustrates that, if early stopping is adopted, our overparametrized architecture is able to effectively denoise the bandlimited graph signal.
This result is related to the error bound for denoising images presented in~\cite{heckel2019denoising}, where $\bbx_0$ is assumed to lie in the span of $\bbW_K$.
However, when dealing with graphs, it is unclear which signals would satisfy this requirement.
Motivated by this, we assume that $\bbx_0$ is a bandlimited signal (i.e., lies in the span of $\bbV_K$), which is a natural condition employed in many applications.

% Problem with E
As a consequence, a critical step to attain \cref{theorem_denoising_gcg} is to relate the eigenvectors of $\sqJacob$ with those of the adjacency matrix $\bbA$, denoted as $\bbV$.
To achieve this, we assume that $\bbA$ is random and provide high-probability bounds between the leading eigenvectors of $\bbA$ and $\sqJacob$.
More specifically, consider a graph $\ccalG$ drawn from a SBM~\cite{newman2018networks} with $K$ communities.
Also, denote by $\ccalM(\ccalbA)$ the SBM with expected adjacency matrix ${\ccalbA}=\mathbb{E}[\bbA]$, and by $\beta_{min}$ the minimum expected degree $\beta_{min}:=\mathrm{min}_i[\ccalbA\textbf{1}]_i$.
Given some $\rho > 0$, we define as $\ccalM_N(\beta_{min},\rho)$ the class of SBMs $\ccalM(\ccalbA)$ with $N$ nodes for which 
$\beta_{min}=\omega (\mathrm{ln}(N/\rho))$, where $\omega(\cdot)$ denotes the (conventional) asymptotic dominance.
Then, the condition of $\ccalG$ being drawn from this SBM whose expected minimum degree increases with $N$ is formally expressed in the following assumption.

\begin{assumption}\label{A:sbm}
    The model $\ccalM(\ccalbA)$ from which $\bbA$ is drawn satisfies $\ccalM(\ccalbA)\in\ccalM_N(\beta_{min},\rho)$, with $\beta_{min} = \omega (\mathrm{ln}(N/\rho))$.
\end{assumption}
We note that, as discussed in~\cite{schaub2020blind}, the minimal degree condition considered in Ass.~1 ensures that nodes belonging to the same community also belong to the same connected component with high probability, which is required to relate $\bbA$ and $\ccalbA$.
Under these conditions, the following result holds.
\begin{lemma}\label{lemma_eigs_gcg}
    Let the matrix $\sqJacob$ be defined as in \eqref{eq:expected_jac}, set $\epsilon$ and $\delta$ to small positive numbers, and denote by $\bbV_K$ and $\bbW_K$ the $K$ leading eigenvectors in the respective eigendecompositions of $\bbA$ and $\sqJacob$. Under Ass.~\ref{A:sbm}, there exists an orthonormal matrix $\bbQ$ and an integer $N_{\epsilon,\delta}$ such that, for $N > N_{\epsilon,\delta}$, the bound
$$\| \bbV_K - \bbW_K \bbQ \|_{\text{F}} \leq \delta,$$
holds with probability at least $1-\epsilon$.
\end{lemma}
%
% Comments on Th. 1
The proof is provided in Appendix~\ref{proof_lemma_eigs_gcg}, and it leverages Ass.~1 to relate the eigenvectors $\bbV_K$ and $\bbW_K$ based on the eigenvectors of the expected values of $\bbA$ and $\sqJacob$.

For a given $K$, Lemma~\ref{lemma_eigs_gcg} bounds the difference between the subspaces spanned by the $K$ leading eigenvectors of $\bbA$ and $\sqJacob$ when graphs are big enough, a result that is key in obtaining Theorem~\ref{theorem_denoising_gcg}. 
Moreover, the lemma shows that if the lower bound $N_{\epsilon,\delta}$ increases, then the error encoded $\delta$ becomes arbitrary small.
Also note that, if a larger value of $K$ is considered, then the minimum required  graph size $N_{\epsilon,\delta}$ will also be larger.
An inspection of~\eqref{eq:expected_jac} reveals that the result in Lemma~\ref{lemma_eigs_gcg} is not entirely unexpected.
Indeed, since $\bbH$ is a polynomial in $\bbA$, so is $\bbH^2$.
This implies that $\bbV$ are also the eigenvectors of $\bbH^2$, and because $\bbH^2$ appears twice on the right hand side of~\eqref{eq:expected_jac}, a relationship between the eigenvectors of $\sqJacob$ and $\bbV$ can be anticipated.
However, the presence of the Hadamard product and the (non Lipschitz continuous) nonlinearity $\arccos$ renders the exact analysis of the eigenvectors a challenging task.
Consequently, we resorted to a stochastic framework in deriving Lemma~\ref{lemma_eigs_gcg}.

\begin{figure*}[!t]
	\centering
	\begin{subfigure}{0.45\textwidth}
		\centering
		    \includegraphics[width=1\textwidth]{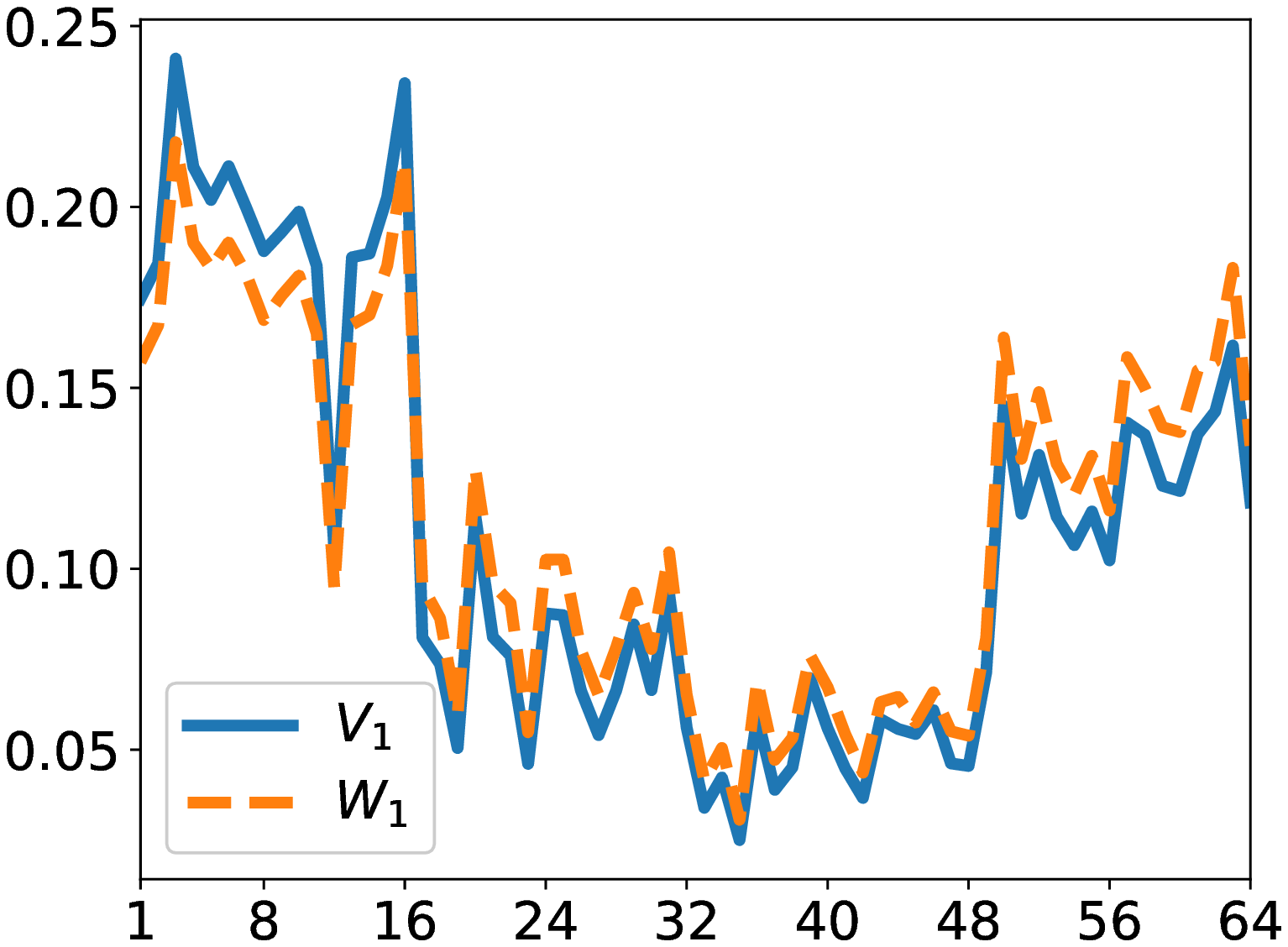}
	\end{subfigure}
	\begin{subfigure}{0.45\textwidth}
		\centering
		    \includegraphics[width=1\textwidth]{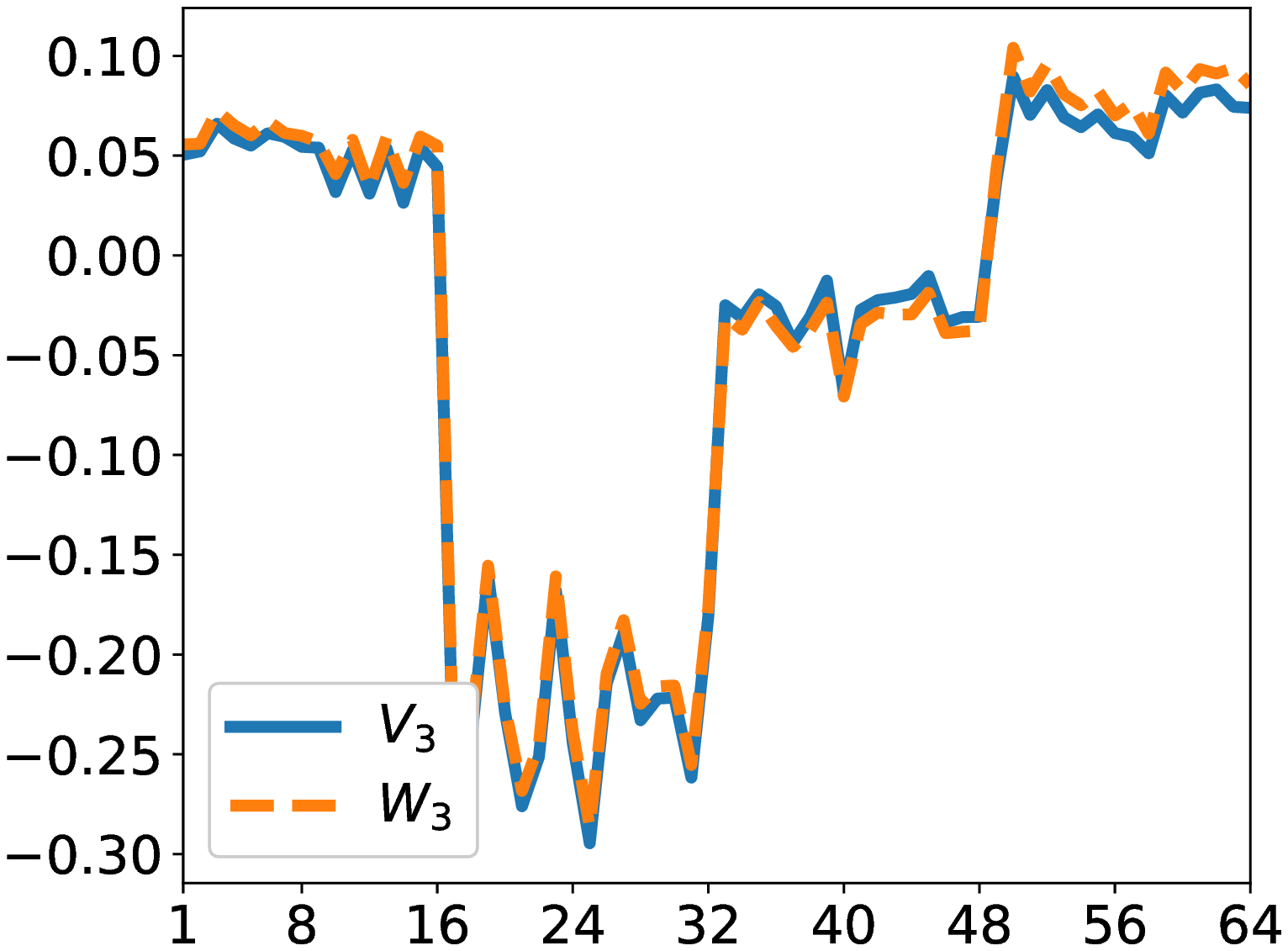}
	\end{subfigure}
	\begin{subfigure}{0.45\textwidth}
		\centering
		    \includegraphics[width=1\textwidth]{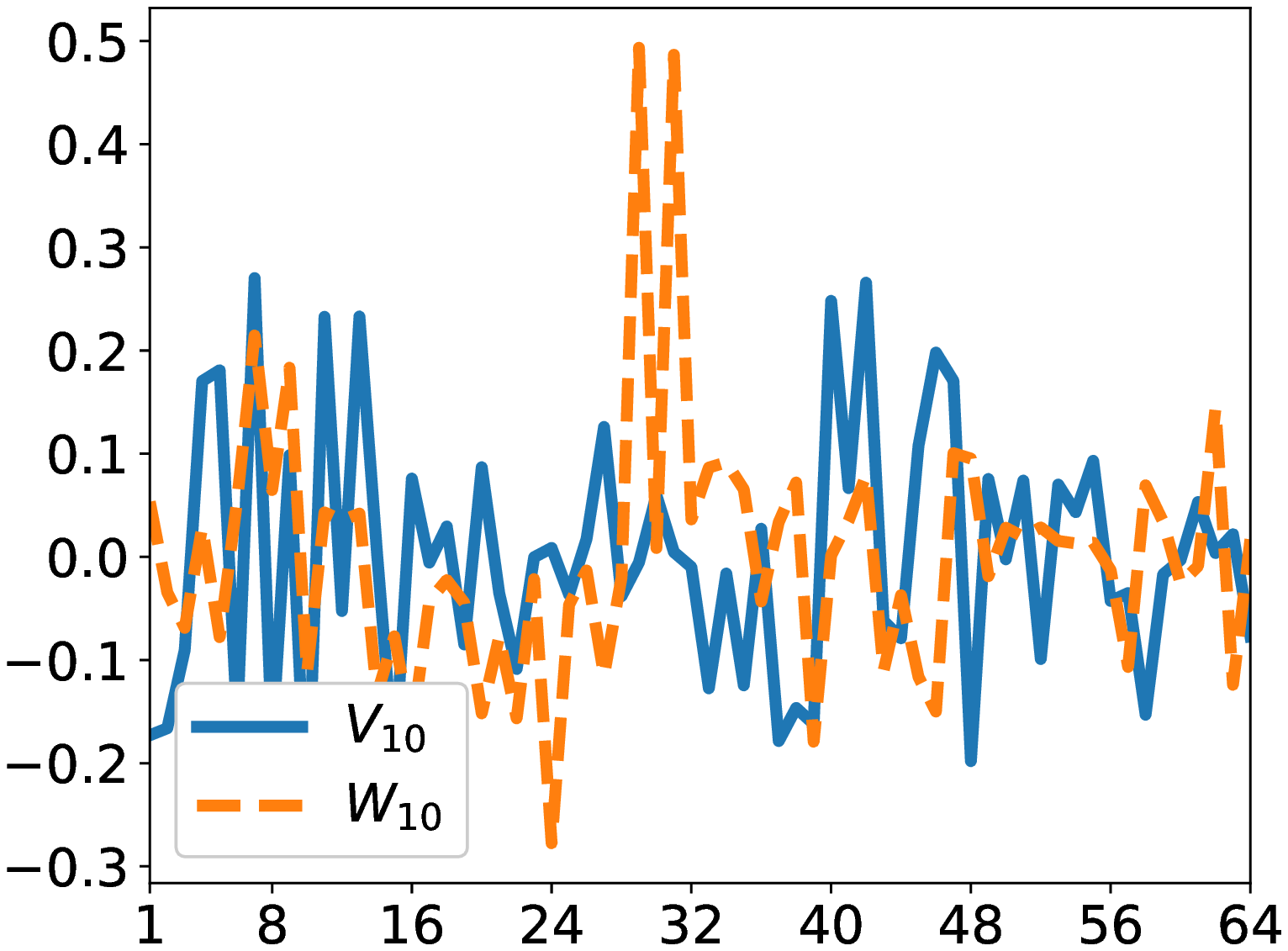}
	\end{subfigure}
	\begin{subfigure}{0.45\textwidth}
		\centering
		    \includegraphics[width=1\textwidth]{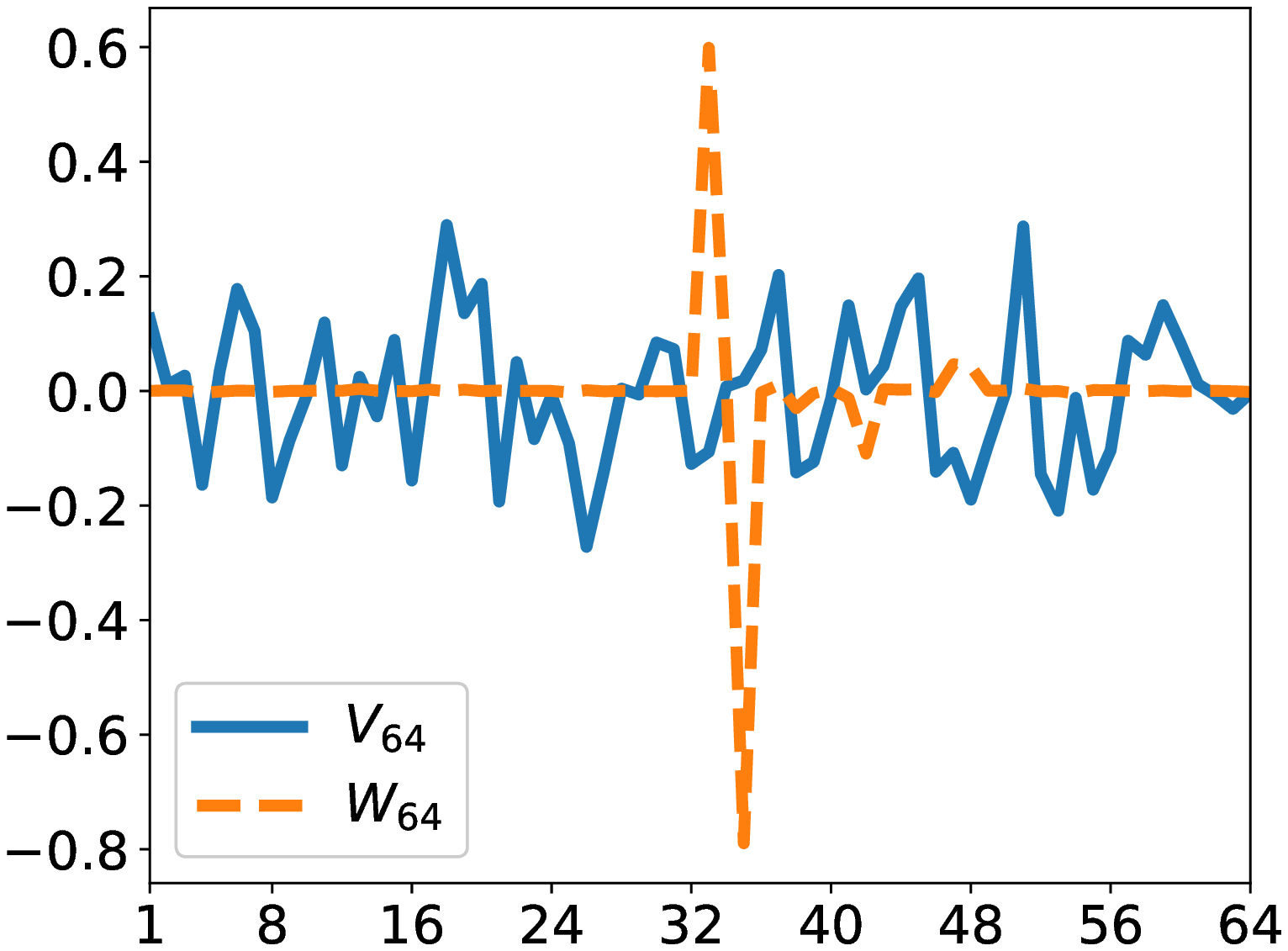}
	\end{subfigure}
	\caption{Comparison between the eigenvectors of the matrices $\bbA$ and $\sqJacob$ for an SBM graph with $N=64$ nodes and $K=4$ communities, and for a GCG of $L=5$ layers. From left to right, the figures represent the first, third, tenth, and last eigenvectors.} \label{fig:generalizing_deep}
\end{figure*}

\begin{figure}
    \centering
    \includegraphics[width=0.45\textwidth]{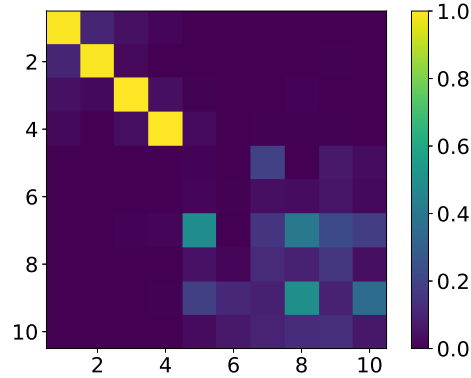}
    \caption{Heatmap representation of the matrix product $\bbV_K^\top\bbW_K$. The low values of the off-diagonal entries illustrate the orthogonality between both sets of eigenvectors. These eigenvectors are the same as those depicted in \cref{fig:generalizing_deep}.}
    \label{fig:orthogonality_WV}
\end{figure}

\begin{figure*}[!t]
	\centering
	%%% A %%%
	\begin{subfigure}{0.32\textwidth}
		\centering
		    \includegraphics[width=1.04\textwidth]{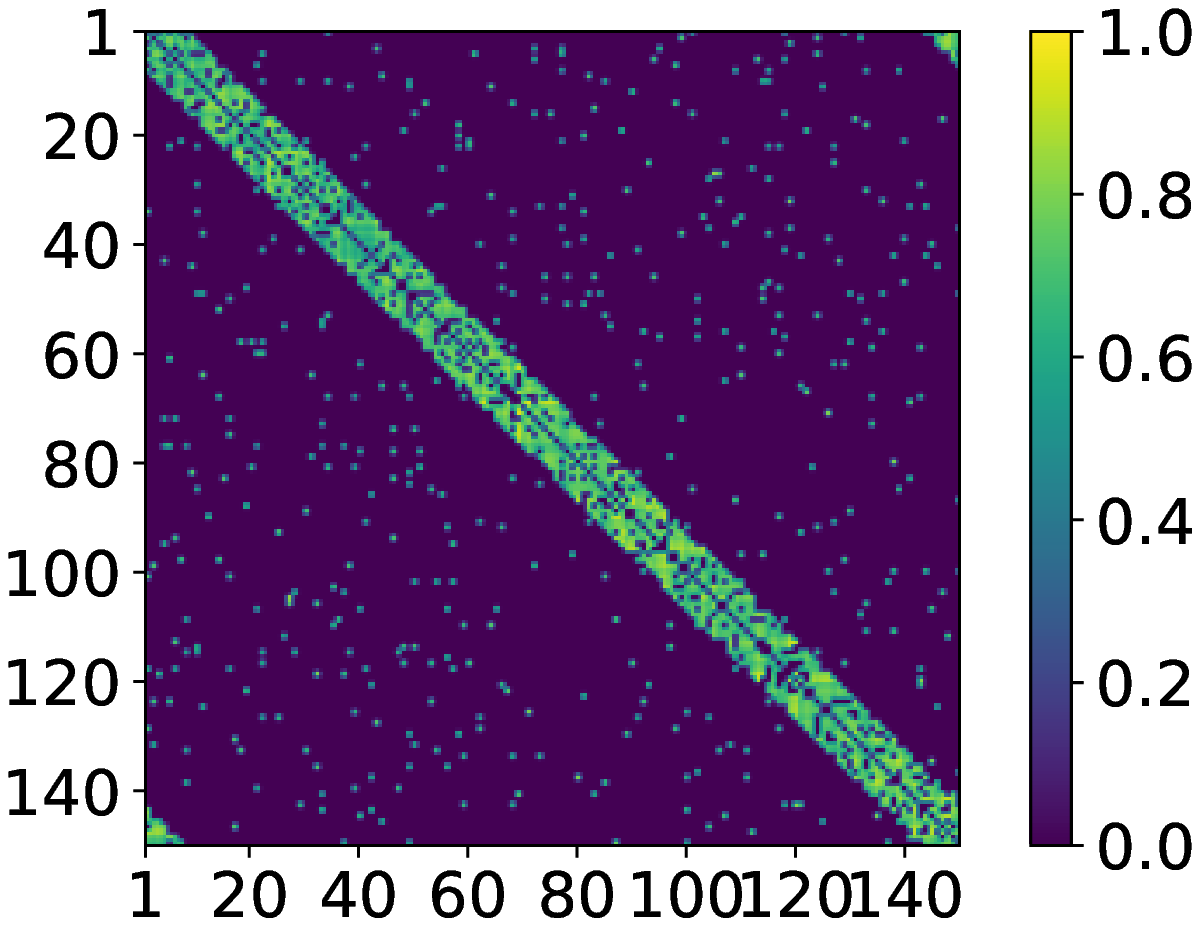}
	\end{subfigure}
	\begin{subfigure}{0.32\textwidth}
		\centering
		    \includegraphics[width=1.04\textwidth]{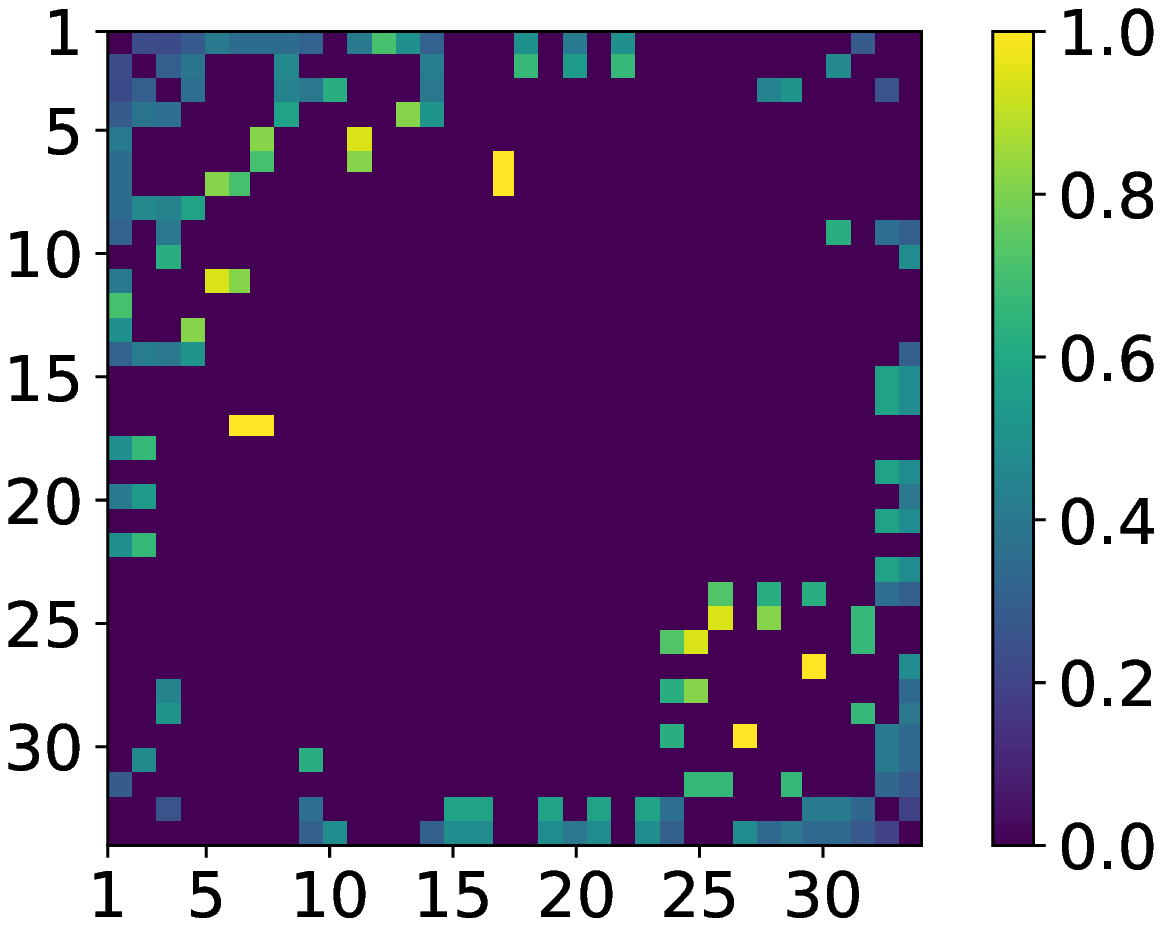}
	\end{subfigure}
	\begin{subfigure}{0.32\textwidth}
		\centering
		    \includegraphics[width=1.04\textwidth]{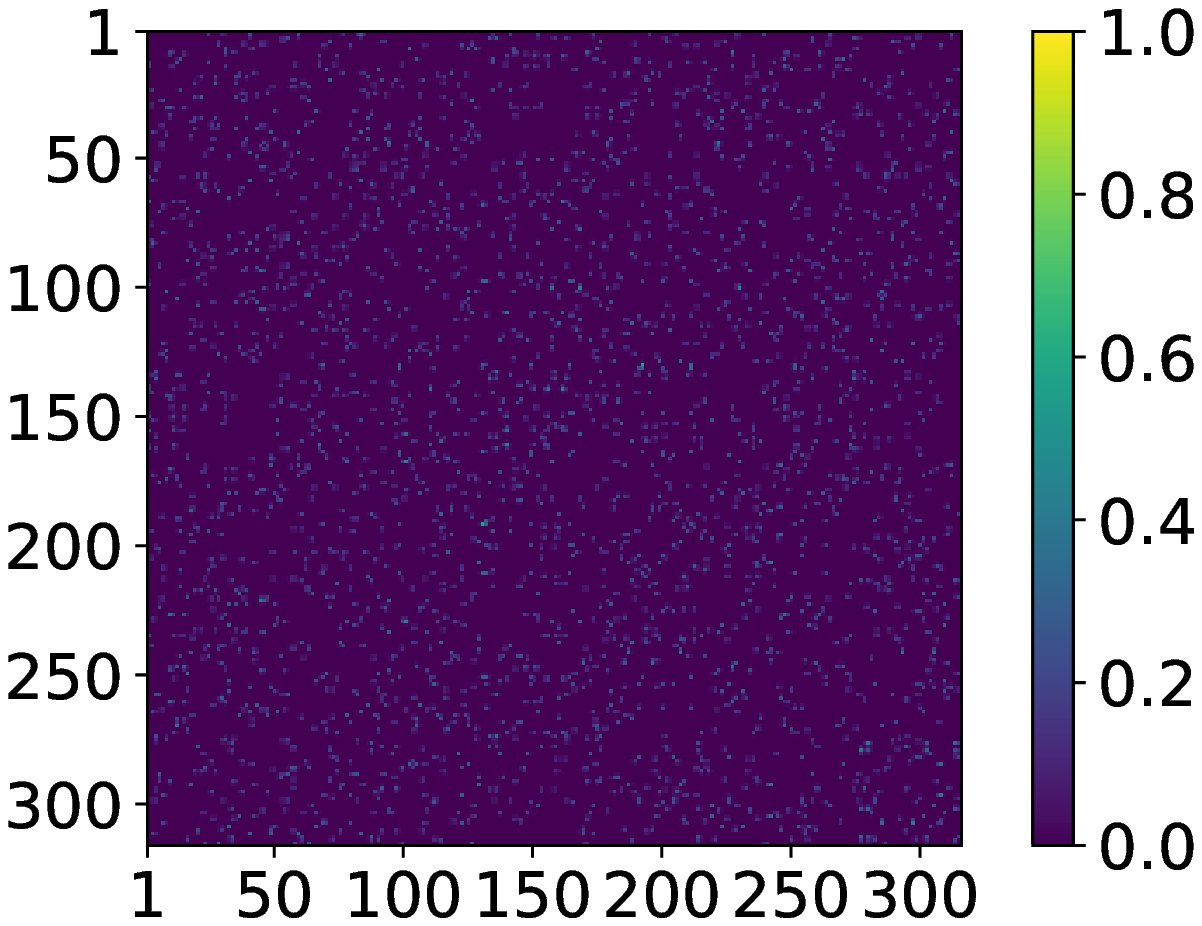}
	\end{subfigure}
	
	%%% HH %%%
	\begin{subfigure}{0.32\textwidth}
		\centering
		    \includegraphics[width=1.04\textwidth]{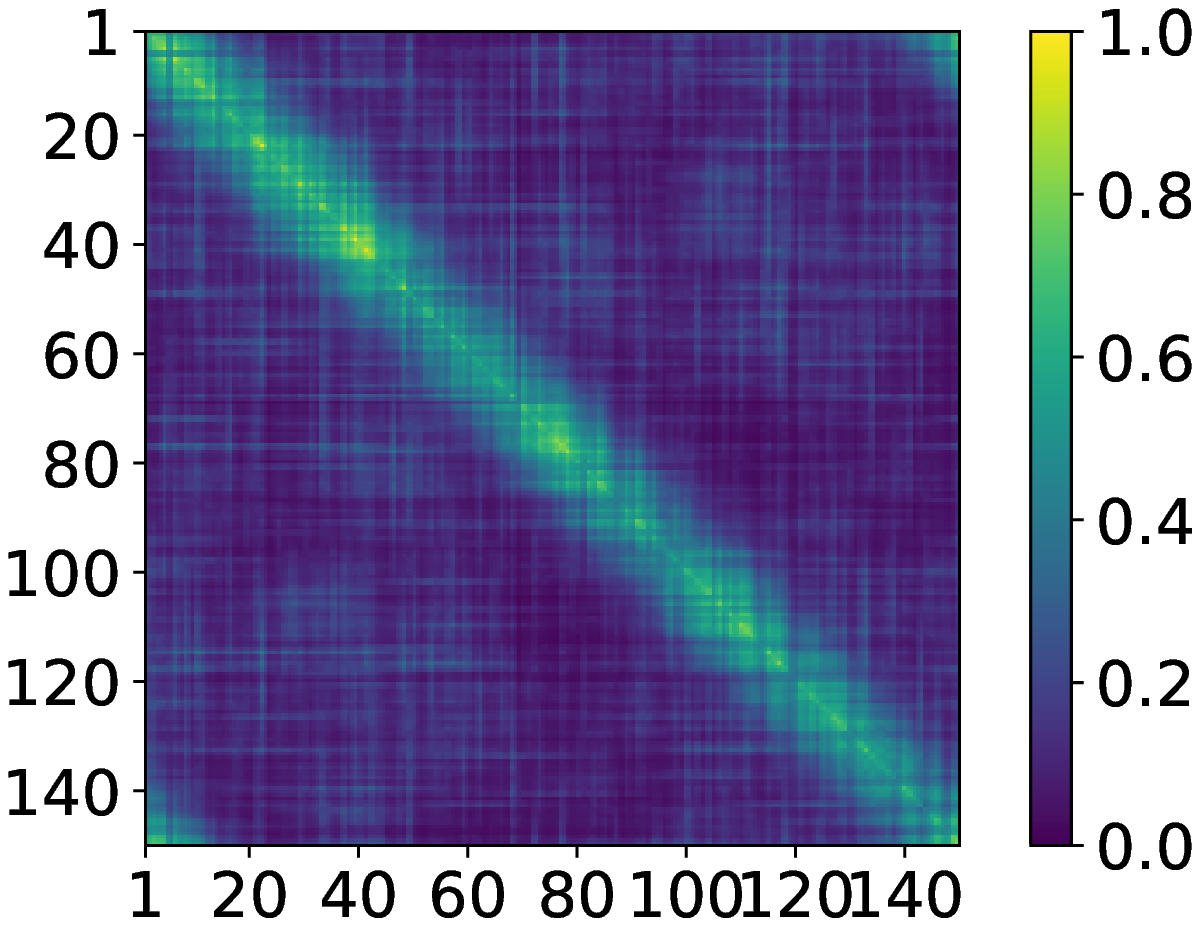}
	\end{subfigure}
	\begin{subfigure}{0.32\textwidth}
		\centering
		    \includegraphics[width=1.04\textwidth]{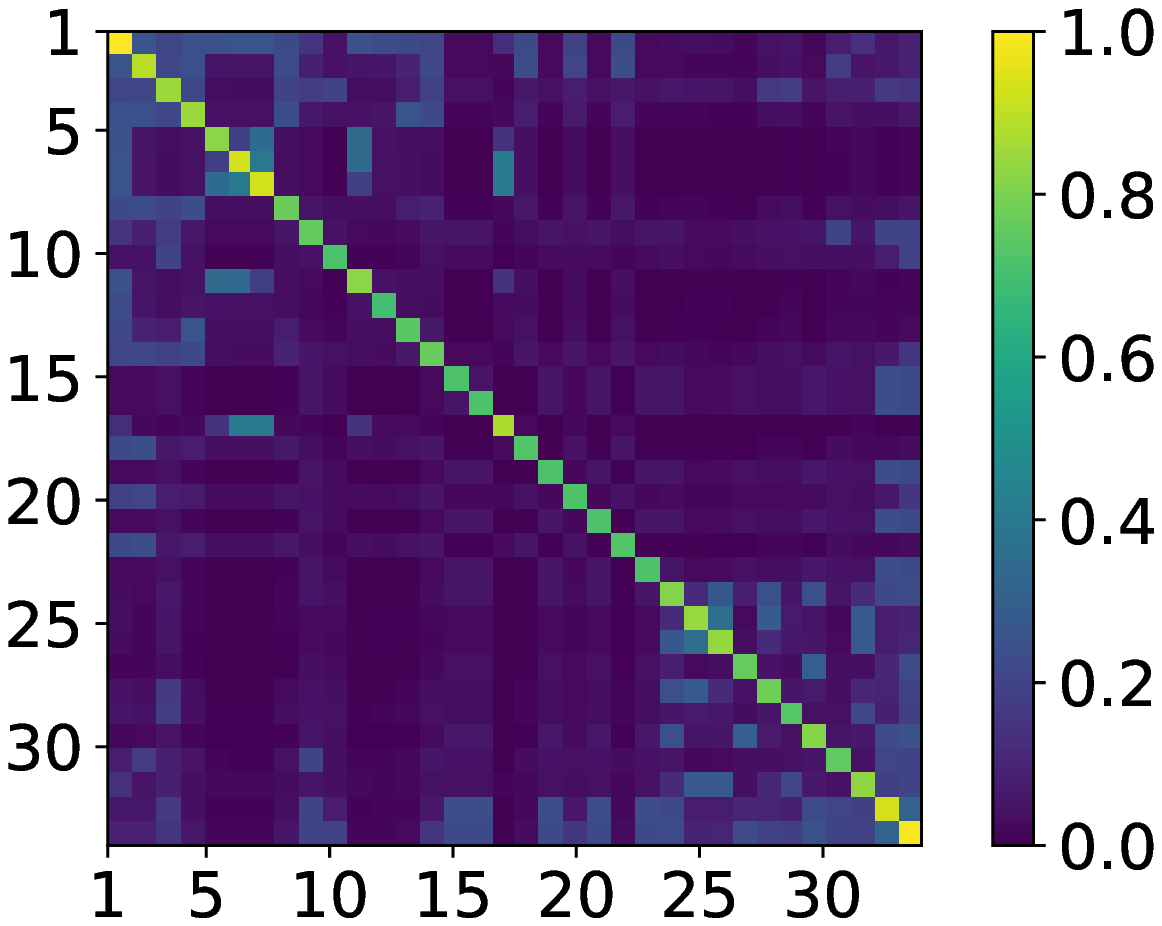}
	\end{subfigure}
	\begin{subfigure}{0.32\textwidth}
		\centering
		    \includegraphics[width=1.04\textwidth]{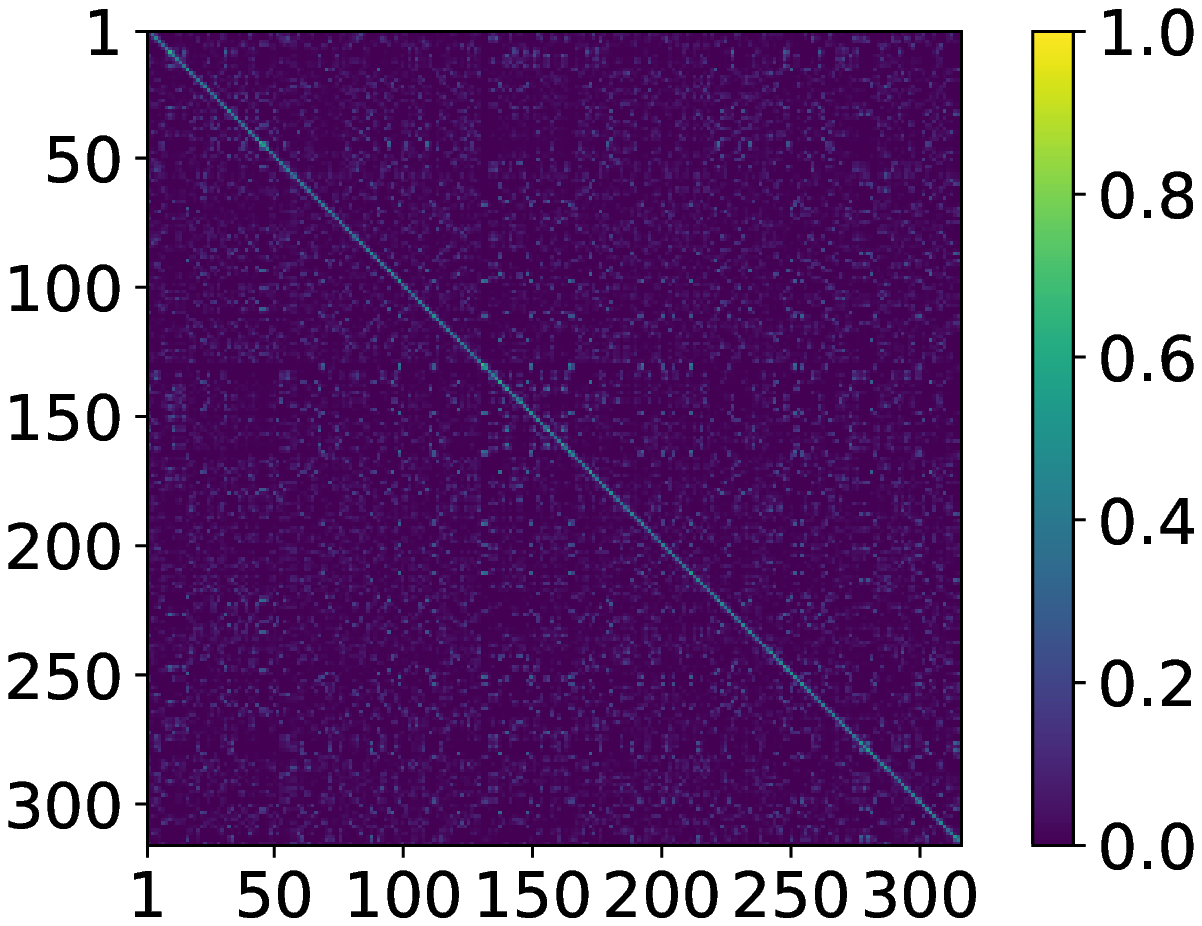}
	\end{subfigure}
	
	%%% E %%%
	\begin{subfigure}{0.32\textwidth}
		\centering
		    \includegraphics[width=1.04\textwidth]{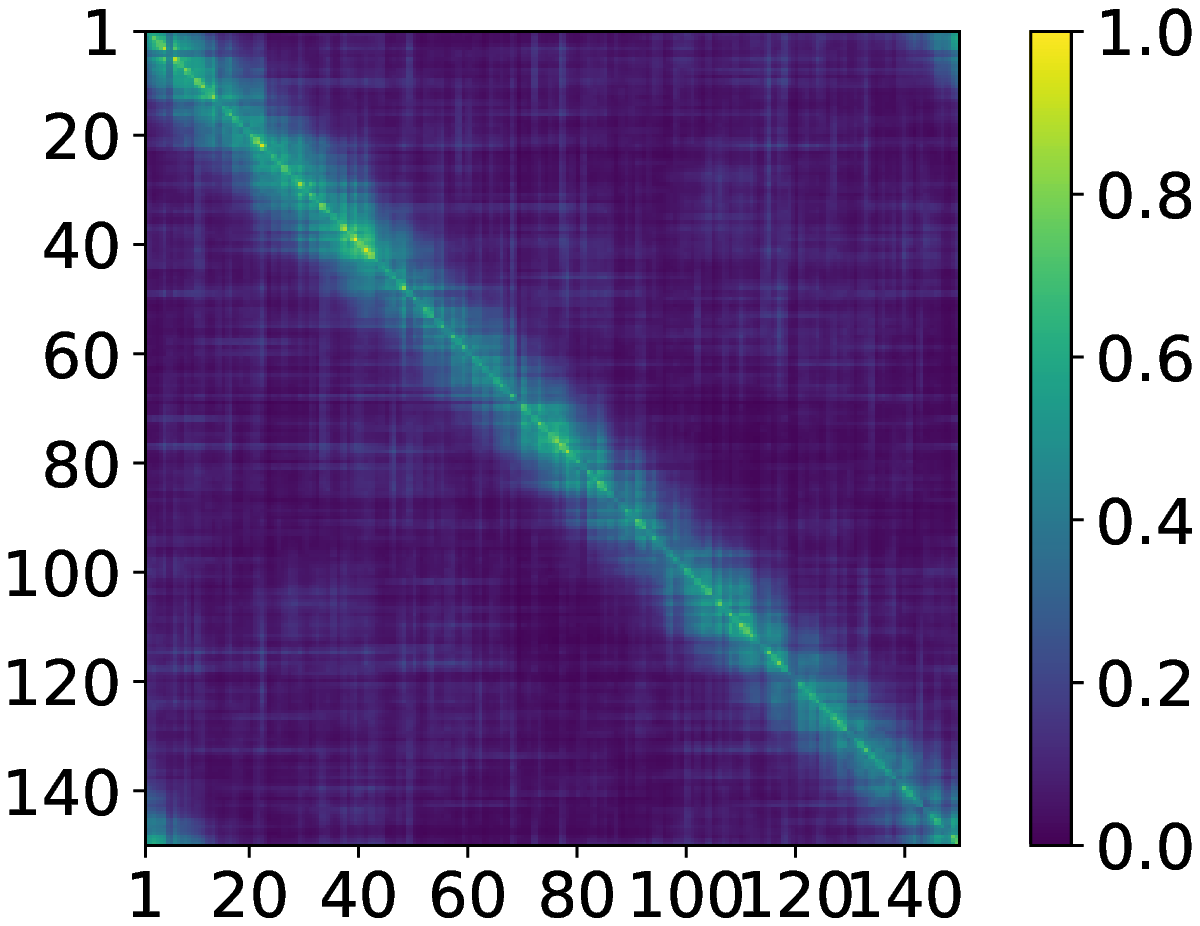}
	\end{subfigure}
	\begin{subfigure}{0.32\textwidth}
		\centering
		    \includegraphics[width=1.04\textwidth]{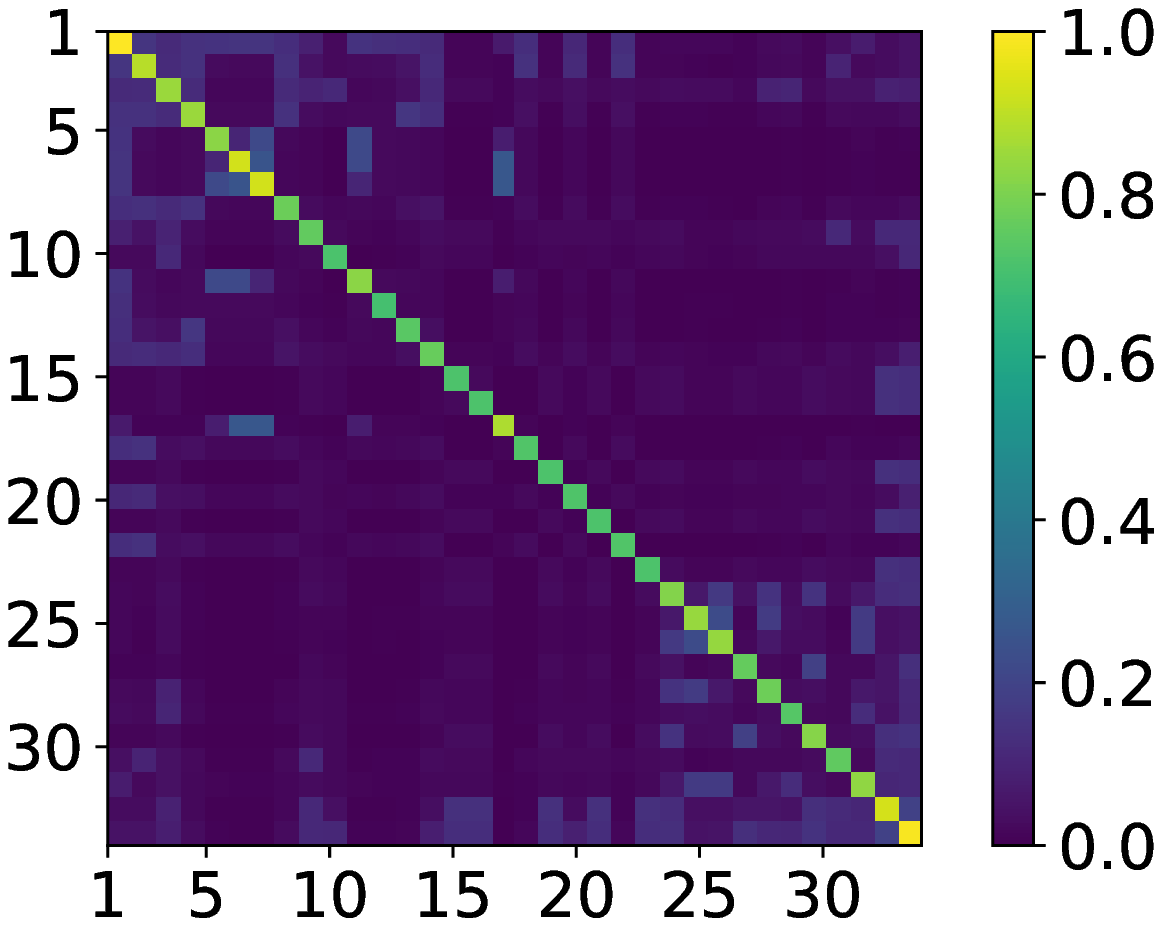}
	\end{subfigure}
	\begin{subfigure}{0.32\textwidth}
		\centering
		    \includegraphics[width=1.04\textwidth]{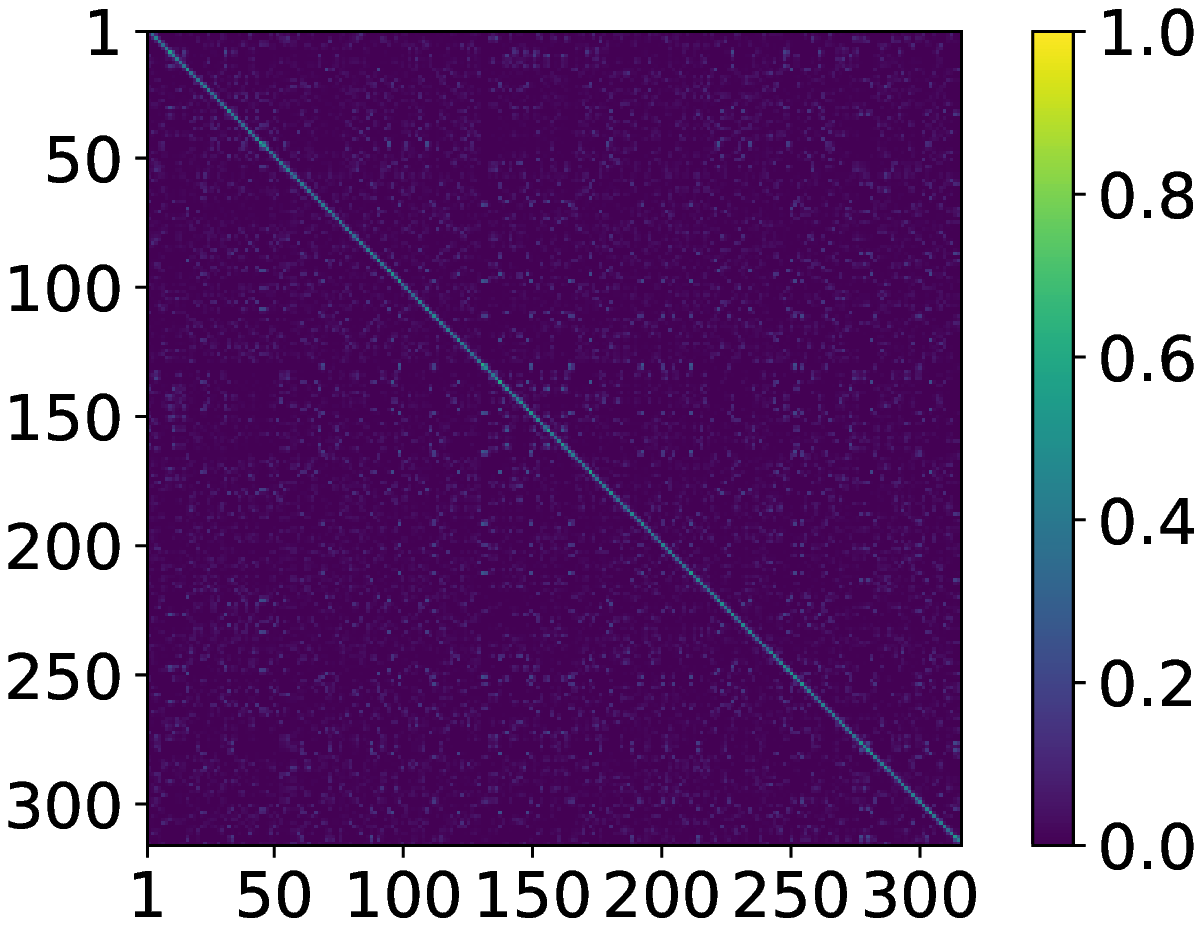}
	\end{subfigure}
	
	%%% Orthogonality %%%
	\begin{subfigure}{0.32\textwidth}
		\centering
		    \includegraphics[width=1.04\textwidth]{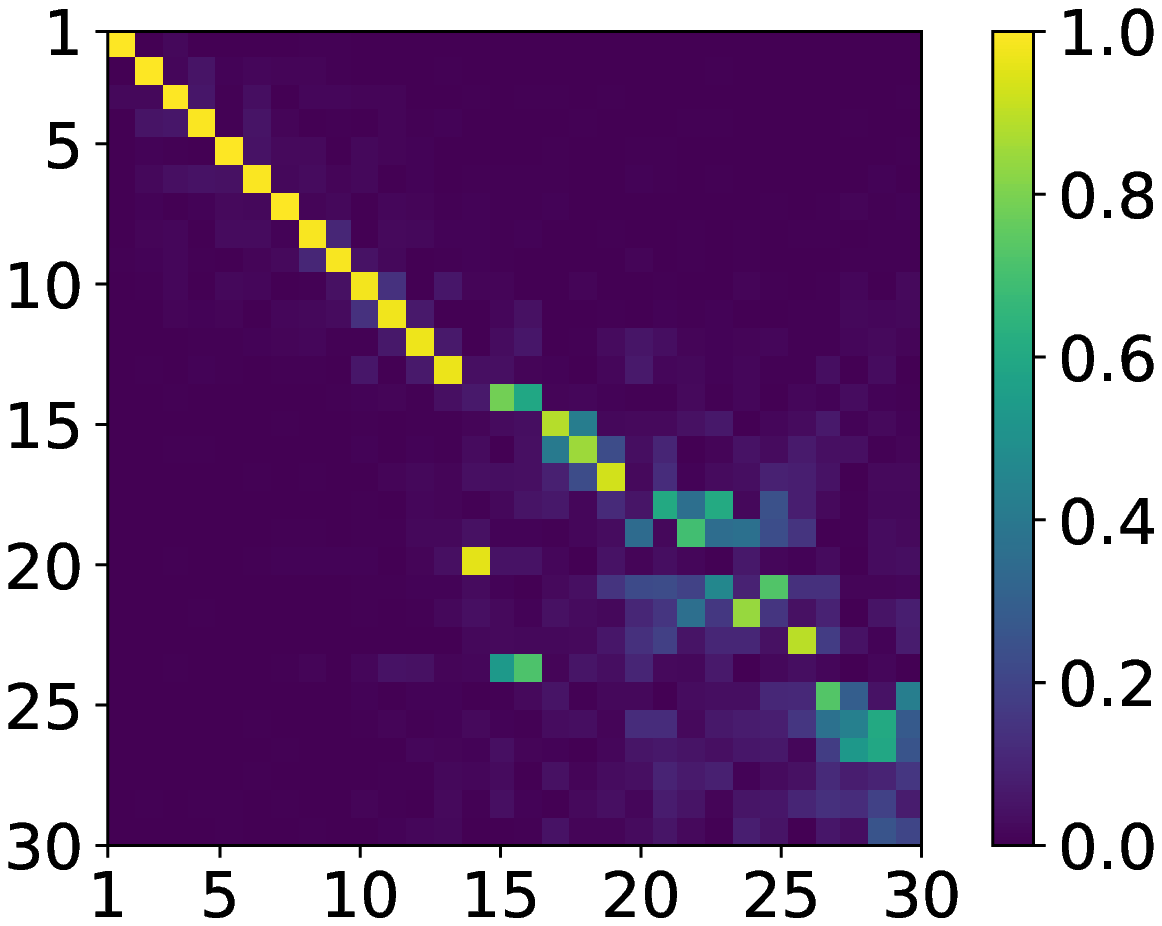}
	\end{subfigure}
	\begin{subfigure}{0.32\textwidth}
		\centering
		    \includegraphics[width=1.04\textwidth]{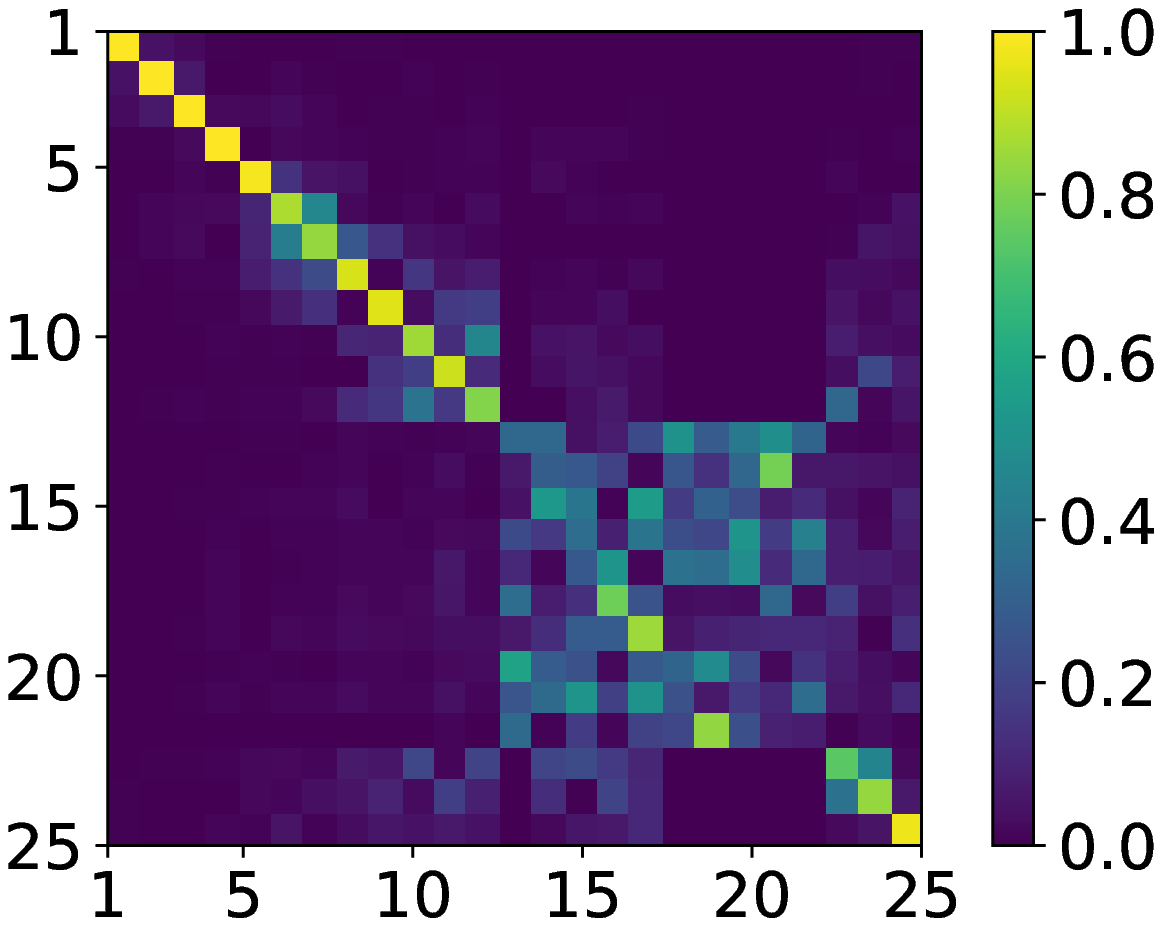}
	\end{subfigure}
	\begin{subfigure}{0.32\textwidth}
		\centering
		    \includegraphics[width=1.04\textwidth]{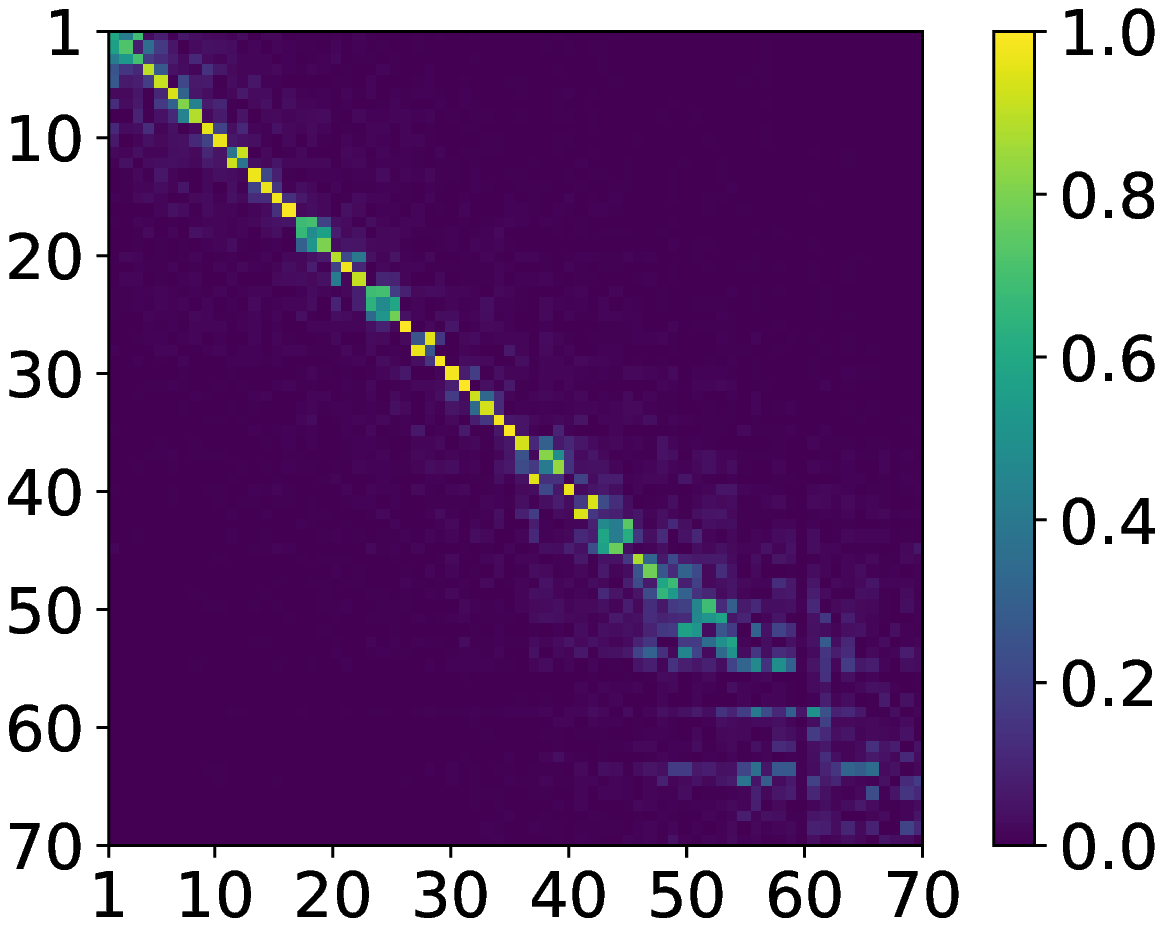}
	\end{subfigure}
	\caption{Illustrating the matrices $\bbA$, $\bbH^2$, $\sqJacob$, and $\bbV_K^\top\bbW_K$, shown in rows 1, 2, 3, and 4, respectively, for different types of graphs. The column 1, 2, and 3 present a SW graph, the Zachary's Karate graph, and the weather stations graph. The graph filter $\bbH^2$ is created as a square graph filter with coefficients drawn from a uniform distribution and set to unit $\ell_1$ norm. For each graph (column), it can be seen that the matrices $\bbA$, $\bbH^2$, and $\sqJacob$ are related, and that the leading eigenvectors $\bbV_K$ and $\bbW_K$ are close to orthogonal.} \label{fig:geralizing_graphs}
\end{figure*}
%%%%%%%%%%%%%%%%%%%%%%%%%%%%%%%%

\subsection{Numerical inspection of the deep GCG spectrum}\label{sec:analyze_deep_gcg}
% Introduction
While for convenience, the previous section focused on analyzing the GCG architecture with $L = 2$ layers, in practice we often work with a larger number of layers.
In this section, we provide numerical evidence showing that the relation between matrices $\bbA$ and $\sqJacob$ described in Lemma~\ref{lemma_eigs_gcg} also holds when $L > 2$.

% Intractability of deep architectures --> numerical evaluation
To that end, \cref{fig:generalizing_deep} shows the pairs of eigenvectors $\bbv_i$ and $\bbw_i$ for the indexes $i=\{1,3,10,64\}$, for a given graph $\ccalG$ drawn from an SBM with $N=64$ nodes and 4 communities.
The GCG is composed of $L=5$ layers and, to obtain the eigenvectors of the squared Jacobian matrix, the Jacobian is computed using the \textit{autograd} functionality of PyTorch.  
The nodes of the graph are sorted by communities, i.e., the first $N_1$ nodes belong to the first community and so on.
It can be clearly seen that, even for moderately small graphs, the leading eigenvectors of $\bbA$ and $\sqJacob$ are almost identical, becoming more dissimilar as the eigenvectors are associated with smaller eigenvalues.
It can also be observed how leading eigenvectors have similar values for entries associated with nodes within the same community.
Moreover, \cref{fig:orthogonality_WV} depicts the matrix product $\bbV^\top\bbW$, where it is observed that the $K=4$ leading eigenvectors of both matrices are orthonormal.
The presented numerical results strengthen the argument that the analytical results obtained for the two-layer case can be extrapolated to deeper architectures. 

% Generalizing to other graphs
Another key assumption of Lemma~\ref{lemma_eigs_gcg} is that $\ccalG$ is drawn from the SBM described in $\ccalM_N(\beta_{min},\rho)$.
This assumption facilitates the derivation of a bound relating the spectra of $\bbA$ and $\sqJacob$ (i.e., the subspaces spanned by the eigenvectors $\bbV_K$ and $\bbW_K$). 
However, the results reported in \cref{fig:geralizing_graphs} suggest that such a relation exists for other type of graphs, even though its analytical characterization is more challenging.
The figure has 12 panels (3 columns and 4 rows).
Each of the columns corresponds to a different graph, namely: 1)~a realization of a \acrfull{sw} graph~\cite{watts1998collective} with $N=150$ nodes, 2)~the Zachary's Karate graph~\cite{zachary1977information} with $N=34$ nodes, and 3)~a graph of $N=316$ weather stations across the United States\footnote{Data extracted from the National Centers for Environmental Information. Available at https://www.ncei.noaa.gov/data/global-summary-of-the-day}.
Each of the three first rows correspond to an $N\times N$ matrix, namely: 1) the normalized adjacency matrix $\bbA$, 2) $\bbH^2$, the squared version of a low pass graph filter and whose coefficients are drawn from a uniform distribution and set to unit $\ell_1$ norm, and 3) the squared Jacobian matrix $\sqJacob$. 
Although we may observe some similarity between $\bbA$ and $\sqJacob$, the relation between $\sqJacob$ and the graph $\ccalG$ becomes apparent when comparing the matrices $\bbH^2$ and $\sqJacob$.
The matrix $\bbH$ is a random graph filter used in the linear transformation of the convolutional generator $f_{\bbTheta}(\bbH)$, and it is clear that the vertex connectivity pattern of $\sqJacob$ is related to that of $\bbH^2$.
Since $\sqJacob$ and $\bbH^2$ are closely related and we know that the eigenvectors of $\bbH^2$ and those of $\bbA$ are the same, we expect $\bbW$ (the eigenvectors of $\sqJacob$) and $\bbV$ (the eigenvectors of $\bbA$) to be related as well. To verify this, the fourth row of \cref{fig:geralizing_graphs} represents $\bbV_K^\top\bbW_K$, i.e., the pairwise inner products of the $K$ leading eigenvectors of $\bbA$ and those of $\sqJacob$. It can be observed that the $K$ leading eigenvectors are close to orthogonal, which means that the relation observed in the vertex domain carries over to the spectral domain and $\bbV_K$ and $\bbW_K$ expand the same subspace.
These results suggest that a deep GCG could be able to denoising signals living in the subspace spanned by $\bbV_K$.
However, because the bound in \cref{theorem_denoising_gcg} assumed a 2-layer GCG, we address this hypothesis numerically in \cref{sec:experiments}.

% Introduction of graph decoder
To summarize, the presented results illustrate that the analytical characterization provided in \cref{sec:analysis_GCG}, which considered a 2-layer GCG operating over SBM graphs, carries over to more general setups.

\section{Graph upsampling decoder}\label{sec:ups_dec}
% Present the new architecture --> upsampling layer
The GCG architecture presented in \cref{sec:conv_dec} incorporated the topology of $\ccalG$ via the vertex-based convolutions implemented by the graph filter $\bbH$.
In this section, we introduce the \acrfull{gdec} architecture.
{In contrast to the GCG and other GCNNs, this novel graph-aware denoising NN incorporates the topology of $\ccalG$ via a (nested) collection of graph upsampling operators~\cite{rey2019underparametrized}.} 
Specifically, we propose the linear transformation for the GDec denoiser to be given by
\begin{equation}\label{eq:graph_decoder}
      \ccalT_{\bbTheta^{(\ell)}}^{(\ell)}\{\bbY^{(\ell-1)}|\ccalG\} = \bbU^{(\ell)}\bbY^{(\ell-1)}\bbTheta^{(\ell)},
\end{equation}
where $\bbU^{(\ell)} \in \reals^{N^{(\ell)} \times N^{(\ell-1)}}$, with $N^{(\ell)}\geq N^{(\ell-1)}$, are graph upsampling matrices to be defined soon.
Note that, compared to \eqref{eq:linear_trans_gcg}, the graph filter $\bbH$ is replaced with the upsampling operator $\bbU^{(\ell)}$ that \emph{depends} on $\ell$.
Adopting the proposed linear transformation, the output of the GDec with $L$ layers is given by the recursion
\begin{align}
      \bbY^{(\ell)}\! &=\relu(\bbU^{(\ell)}\bbY^{(\ell-1)}\bbTheta^{(\ell)}),\;\; \mathrm{for}\; \ell=1,...,L\!-\!1, \label{eq:gd1}\\ 
      \bby^{(L)} \!&= \bbU^{(L)}\bbY^{(L-1)}\bbTheta^{(L)}, \label{eq:gd2}
\end{align}
where the $\relu$ is also removed from the last layer.

% Comment GDec
Similar to the GCG, the proposed GDec learns to combine the features within each node.
{However, the interpolation of the signals in this case is determined by the graph upsampling operators $\{\bbU^{(\ell)}\}_{\ell=1}^L$, rather than by employing convolutions.}
The size of the input $N^{(0)}$ is now a design parameter that will determine the implicit degrees of freedom of the architecture.
Note that, from the GSP perspective, the input feature matrix $\bbY^{(\ell-1)} \in \reals^{N^{(\ell-1)} \times F^{(\ell-1)}}$ represents $F^{(\ell-1)}$ graph signals, each of them defined over a graph $\ccalG^{(\ell-1)}$ with $N^{(\ell-1)}$ nodes.
Therefore, even though the input $\bbY^{(0)}=\bbZ$ is still a random white matrix across rows and columns, since $N^{(\ell)} \geq N^{(\ell-1)}$, the dimensionality of the input is progressively increasing. 

% Differentiate from GCG
{A closer comparison with the GCG reveals that the smaller dimensionality of the input $\bbZ$ endows the GDec architecture with fewer degrees of freedom, rendering the architecture more robust to noise.
Not only that, but the graph information is now included via the graph upsampling operators $\bbU^{(\ell)}$ instead of relying on graph filters.}
Clearly, the method used to design the graph upsampling matrices, which is the subject of the next section, will have an impact on the type of graph signals that can be efficiently denoised using the GDec architecture.

\subsection{Graph upsampling operator from hierarchical clustering}\label{sec:upsampling_operator}

\begin{figure}
    \centering
    \includegraphics[width=0.6\textwidth]{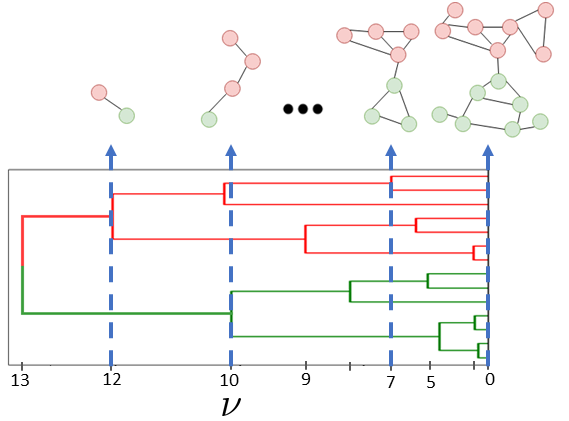}
    \caption{Dendrogram of an agglomerative hierarchical clustering algorithm and the resulting graphs with 2, 4, 7 and 14 nodes.}
    \label{fig:dendrogram}
\end{figure}

% Need for upsampling operator --> challenge
Regular upsampling operators have been successfully used in NN architectures to denoise signals defined on regular domains \cite{heckel2019denoising}. While the design of upsampling operators in regular grids is straightforward, when the signals at hand are defined on irregular domains the problem becomes substantially more challenging. The approach that we put forth in this chapter is to use agglomerative hierarchical clustering methods~\cite{jain1988algorithms,carlsson2013axiomatic,carlsson2018hierarchical} to design a graph upsampling operator that leverages the graph topology.
These methods take a graph as an input and return a dendrogram; see \cref{fig:dendrogram}.
A dendrogram can be interpreted as a rooted-tree structure that shows different clusters at the different levels of resolution $\nu$. 
At the finest resolution ($\nu=0$) each node forms a cluster of its own.
Then, as $\nu$ increases, nodes start to group together (agglomerate) in bigger clusters and, when the resolution becomes large (coarse) enough, all nodes end up being grouped in the same cluster.

% Construction of matrix P and A
By cutting the dendrogram at $L+1$ resolutions, including $\nu=0$, we obtain a collection of node sets with parent-child relationships inherited by the refinement of clusters.
Since we are interested in performing graph upsampling, note that the dendrogram is interpreted from left to right.
This can be observed in the example shown in \cref{fig:dendrogram}, where the three red nodes in the second graph ($\nu=10$, layer $\ell=1$) are children of the red parent in the coarsest graph ($\nu=12$, layer $\ell=0$).
{In this sense, the graph upsampling operator is given by the inverse operation of the clustering algorithm.}
We leverage these parent-children relations to define the membership matrices $\bbP^{(\ell)} \in \{0,1\}^{N^{(\ell)} \times N^{(\ell-1)}}$, where the entry $P_{ij}^{(\ell)}=1$ only if the $i$-th node in layer $\ell$ is the child of the $j$-th node in layer $\ell-1$.
{Moreover, we can further exploit the dendrogram to obtain  coarser-resolution versions of the original graph $\ccalG$.
To that end, note that the clusters at layer $\ell$ can be interpreted as nodes of a graph $\ccalG^{(\ell)}$ with $N^{(\ell)}$ nodes and adjacency matrix $\bbA^{(\ell)}$.}
There are several ways of defining $\bbA^{(\ell)}$ based on the original adjacency matrix $\bbA$. While our architecture does not focus on a particular form, in the simulations we set $A^{(\ell)}_{ij} \neq 0$ only if, in the original graph $\ccalG$, there is at least one edge between nodes belonging to the cluster $i$ and nodes from cluster $j$.
In addition, the weight of the edge depends on the number of existing edges between the two clusters.

% Construction of U 
With the definition of the membership matrix $\bbP^{(\ell)}$ and the adjacency matrix $\bbA^{(\ell)}$, the upsampling operator of the $\ell$-th layer is given by
\begin{equation}\label{eq:upsampling_operator}
    \bbU^{(\ell)} = \left(\gamma\bbI+\left(1-\gamma\right)\bbA^{(\ell)}\right)\bbP^{(\ell)},
\end{equation}
where $\gamma\in[0, 1]$ is a pre-specified constant. 
Notice that $\bbU^{(\ell)}$ first copies the signal value from the parents to the children by applying the matrix $\bbP^{(\ell)}$, and then every child performs a convex combination between this value and the average signal value of its neighbors.
This design promotes that nodes descending from the same parent have similar (related) values, which conveys a notion (prior) of smoothness on the targeted graph signals.
As we show in \cref{sec:experiments}, the implicit smoothness prior results in a better performance when denoising smooth signals but, on the other hand, makes the architecture more sensitive to model mismatch.
Therefore, when dealing with high-frequency signals, a worth-looking approach left as a future research direction is to rely on algorithms that cluster the nodes considering not only the topology of $\ccalG$ but also the properties of the graph signals.

% Implications of the design of U
Because the membership matrices $\bbP^{(\ell)}$ are designed using a clustering algorithm over $\ccalG$, and the matrices $\bbA^{(\ell)}$ capture how strongly connected the clusters of layer $\ell$ are in the original graph, these two matrices are responsible for incorporating the information of $\ccalG$ into the upsampling operators $\bbU^{(\ell)}$.
Furthermore, we remark that the upsampling operator $\bbU^{(\ell)}$ can be reinterpreted as the application of $\bbP^{(\ell)}$ followed by the application of a graph filter
\begin{equation} \label{eq:upsampling_filter_second_step}
\tbH^{(\ell)}=\gamma\bbI+(1-\gamma) \bbA^{(\ell)},    
\end{equation} 
which sets the filter coefficients as $h_0=\gamma$ and $h_1=1-\gamma$.

\subsection{Guaranteed denoising with the GDec}
% Present 2 layer decoder
As we did for the GCG, our goal is to theoretically characterize the denoising performance of the GNN architecture defined by \eqref{eq:gd1}-\eqref{eq:upsampling_operator}.
To achieve that goal, we replicate the approach implemented in \cref{sec:analysis_GCG}.
We first derive the matrix $\sqJacob$ and provide theoretical guarantees when denoising a $K$-bandlimited graph signal with the GDec.
Then, to gain additional insight, we detail the relation between the subspace spanned by the eigenvectors $\bbW$ and the spectral domain of $\bbA$.
This relation is key in deriving the theoretical analysis.

We start by introducing the 2-layer GDec
\begin{equation}
    f_{\bbTheta}(\bbZ|\ccalG) = \relu(\bbU\bbZ\bbTheta^{(1)})\bbtheta^{(2)}.
\end{equation}
Upon following a reasoning similar to that provided after \eqref{eq:2layer_gcg_simp}, instead of employing the previous architecture we can optimize \eqref{eq:nonlinear_denoising} over its simplified version
\begin{equation}\label{eq:2layer_gd_simp}
   f_{\bbTheta}(\bbU) = f_{\bbTheta}(\bbZ|\ccalG) = \relu(\bbU\bbTheta) \bbb.
\end{equation}
%
%where we use the fact that $\bbU=\tbH\bbP$.
An important difference with respect to the GCG presented in \eqref{eq:2layer_gcg_simp} is that the matrix $\bbTheta$ has a dimension of $N^{(0)} \times F$, so it spans $\reals^{N^{(0)}}$ instead of $\reals^{N}$.
Since $N^{(0)} < N$, the smaller subspace spanned by the weights of the GDec renders the architecture more robust to fitting noise, but, on the other hand, the number of degrees of freedom to learn the graph signal of interest are reduced. As a result, the alignment between the targeted graph signals and the low-pass vertex-clustering architecture becomes more important.

% Introduce E, assumtions and Th. 3
The expected squared Jacobian $\sqJacob=\mathbb{E}_{\bbTheta}[\ccalJ_{\bbTheta}(\bbU) \ccalJ^\top_{\bbTheta}(\bbU)]$ is obtained following the procedure used to derive~\eqref{eq:expected_jac}, arriving at the expression
\begin{equation}\label{eq:expected_jac_dec}
    \sqJacob = 0.5 \left( \mathbf{1} \mathbf{1}^\top - \frac{1}{\pi} \arccos(\tbC^{-1} \bbU\bbU^\top \tbC^{-1})\right) \circ \bbU\bbU^\top,
\end{equation}
where $\bbu_i$ represents the $i$-th row of $\bbU$, and $\tbC=\diag([\|\bbu_1\|_2,...,\|\bbu_N\|_2])$ is a normalization matrix.

Then, let $\bbx_0$ be a $K$-bandlimited graph signal and let $f_{\bbTheta}(\bbU)$ have a number of features $F$ satisfying \eqref{bound_on_F}.  
If we solve \eqref{eq:nonlinear_denoising} running gradient descent with a step size $\eta\leq\frac{1}{\sigma_1^2}$, the following result holds.

\begin{theorem}\label{theorem_denoising_gd}
    Let $f_{\bbTheta}(\bbU)$ be the network defined in equation~\eqref{eq:2layer_gd_simp}.
    Consider the conditions described in \cref{theorem_denoising_gcg} and let $N^{(0)}$ match the number of communities $K$ (see Ass.~\ref{A:sbm}).
    Then, for any $\epsilon$, $\delta$, there exists some $N_{\epsilon,\delta}$ such that if $N>N_{\epsilon,\delta}$, then the error for each iteration $t$ of gradient descent with stepsize $\eta$ used to fit the architecture is bounded as \eqref{eq_bound_theorem_fitting_eigs_Jacobian_for_t}, with  probability at least $1-e^{-F^2}-\phi-\epsilon$.
\end{theorem}

% Proof
The proof of the theorem is analogous to the one provided in Appendix~\ref{proof_theorem_denoising_gcg} but exploiting Lemma~\ref{lemma_eigs_gd} instead of Lemma~\ref{lemma_eigs_gcg}.
Lemma~\ref{lemma_eigs_gd} is fundamental in attaining \cref{theorem_denoising_gd} and is presented later in the section.

% Interpretation of Th. 2 and comparison with Th. 1
\cref{theorem_denoising_gd} formally establishes the denoising capability of the GDec when $\bbx_0$ is a $K$-bandlimited graph signal and  $K=N^{(0)}$ matches the number of communities in the SBM graph.
When compared with the GCG, the smaller dimensionality of the input $\bbZ$, and thus the smaller rank of the matrix $\bbTheta$, constrains the learning capacity of the architecture, making it more robust to the presence of noise.
However, this additional robustness also implies that the architecture is more sensitive to model mismatch, since its capacity to learn arbitrary signals is smaller.
Intuitively, the GDec represents an architecture tailored for a more specific family of graph signals than the GCG.
Moreover, employing the GDec instead of the GCG has a significant impact on the relation between the subspaces spanned by $\bbV_K$ and $\bbW_K$.

To establish the new relation between $\bbV_K$ and $\bbW_K$, assume that the adjacency matrix is drawn from an SBM $\ccalM(\ccalbA)$ with $K$ communities such that $\ccalM(\ccalbA)\in\ccalM_N(\beta_{min},\rho)$, so that the SBM follows Ass. \ref{A:sbm}.
In addition, set the size of the latent space to the number of communities so $N^{(0)} = K$.
Under this setting, the counterpart to Lemma~\ref{lemma_eigs_gcg} for the case where $f_{\bbTheta}(\bbU)$ is a GDec architecture follows.
\begin{lemma}\label{lemma_eigs_gd}
    Let the matrix $\sqJacob$ be defined as in \eqref{eq:expected_jac_dec}, set $\epsilon$ and $\delta$ to small positive numbers, and denote by $\bbV_K$ and $\bbW_K$ the $K$ leading eigenvectors in the respective eigendecompositions of $\bbA$ and $\sqJacob$. Under Ass.~\ref{A:sbm}, there exist an orthonormal matrix $\bbQ$ and an integer $N_{\epsilon,\delta}$ such that for $N > N_{\epsilon,\delta}$ the bound
$$\| \bbV_K - \bbW_K \bbQ \|_{\text{F}} \leq \delta,$$
holds with probability at least $1-\epsilon$.
\end{lemma}
%

% Intuition behind lemma_eigs_gd
Lemma~\ref{lemma_eigs_gd} asserts that the difference between the subspaces spanned by $\bbV_K$ and $\bbW_K$ becomes arbitrarily small as the size of the graph increases. 
The proof is provided in Appendix~\ref{proof_lemma_eigs_gd} and the intuition behind it arises from the fact that the upsampling operator can be understood as $\bbU=\tbH\bbP$, where $\tbH$ is a graph filter of the specific form described in \eqref{eq:upsampling_filter_second_step}.
Remember that $\bbP$ is a binary matrix encoding the cluster in the layer $\ell-1$ to which the nodes in the layer $\ell$ belong.
Since we are only considering two layers, and we have that $N^{(0)}=K$, the matrix $\bbP$ is encoding the node-community membership of the SBM graph and, hence, the product $\bbP\bbP^\top$ is a block matrix with constant entries matching the block pattern of $\ccalbA$.
As shown in the proof, this property can be leveraged to bound the eigendecomposition of $\bbA$ and $\sqJacob$.

\subsection{Analyzing the deep GDec}
% Effect of additional layers in the learning capability
The deep GDec composed of $L>2$ layers can be constructed following the recursion presented in \eqref{eq:gd1} and \eqref{eq:gd2}.
In this case, by stacking more layers we perform the upsampling of the input signal in a progressive manner and, at the same time, we add more nonlinearities, which helps alleviating the rank constraint related to the input size $N^{(0)}$.
In the absence of nonlinear functions, the maximum rank of the weights would be $N^{(0)}$, and thus, only signals in a subspace of size $N^{(0)}$ could be learned.
By properly selecting the number of layers and the input size when constructing the network, we can obtain a trade-off between the robustness of the architecture and its learning capability.

% Effect of deeper decoder on smoothness
In addition, the effect of adding more layers is also reflected on the smoothness assumption inherited from the construction of the upsampling operator.
Adding more layers is related to less smooth signals, since the number of nodes in $\ccalG$ with a common parent, and thus, with similar values, is smaller. 

% Introduce numerical results
We note that numerically illustrating that the bound between $\bbV_K$ and $\bbW_K$ holds true for the deep GDec, and that its denoising capability is not limited to signals defined over SBM graphs provide results similar to those in \cref{sec:analyze_deep_gcg}.
Therefore, instead of replicating the previous section, we directly illustrate the performance of the deep GDec under more general settings in the following section, where we present the numerical evaluation of the proposed architectures.

\section{Numerical results}\label{sec:experiments}
This section presents different experiments to numerically validate the theoretical claims introduced in the chapter, and to illustrate the denoising performance of the GCG and the GDec.
The experiments are carried out using synthetic and real-world data, and the proposed architectures are compared to other graph-signal denoising alternatives.
The code for the experiments and the architectures is available on GitHub\footnote{\url{https://github.com/reysam93/Graph_Deep_Decoder}}.
For hyper-parameter settings and implementation details the interested reader is referred to the online available code.

\subsection{Denoising capability of graph untrained architectures}
The goal of the experiment shown in Figures~\ref{fig:experiments1a} and~\ref{fig:experiments1b} is to illustrate that the proposed graph untrained architectures are capable of learning the structured original signal $\bbx_0$ faster than the noise, which is one of the core claims of the chapter.
To that end, we generate an SBM graph with $N=64$ nodes and $K=4$ communities, and define 3 different signals: (i)~``Signal'': a piece-wise constant signal $\bbx_0$ with the value of each node being the label of its community; (ii)~``Noise'': zero-mean white Gaussian noise $\bbn$ with unit variance; and (iii)~``Signal + Noise'': a noisy observation $\bbx=\bbx_0+\bbn$ where the noise has a normalized power of $0.1$.
\cref{fig:experiments1a} and \ref{fig:experiments1b} show the \acrfull{nmse}, with the error for each realization being $\|\bbx_0-\nolinebreak \hbx_0\|_2^2/\|\bbx_0\|_2^2$.
The mean is computed for 100 realizations of the noise as the number of epochs increases when the different signals are fitted by the 2-layer GCG and the 2-layer GDec, respectively.
It can be seen how, in both cases, the error when fitting the noisy signal $\bbx$ decreases for a few epochs until it reaches a minimum, and then starts to increase.
This is because the proposed untrained architectures learn the signal $\bbx_0$ faster than the noise, but if they fit the observation for too many epochs, they start learning the noise as well and, hence, the NMSE increases.
As stated by \cref{theorem_denoising_gcg} and \cref{theorem_denoising_gd}, this result illustrates that, if early stopping is applied, both architectures are capable of denoising the observed graph signals without a training step. 
It can also be noted that, under this setting, the GDec learns the signal $\bbx_0$ faster than the GCG and, at the same time, is more robust to the presence of noise.
This can be seen as a consequence of GDec implicitly making stronger assumptions about the smoothness of the targeted signal.

%%%%%%%%%%   FIGURE   %%%%%%%%%%
\begin{figure}[t]
	\centering
	\begin{subfigure}{0.49\textwidth}
		\centering
        \includegraphics[width=\textwidth]{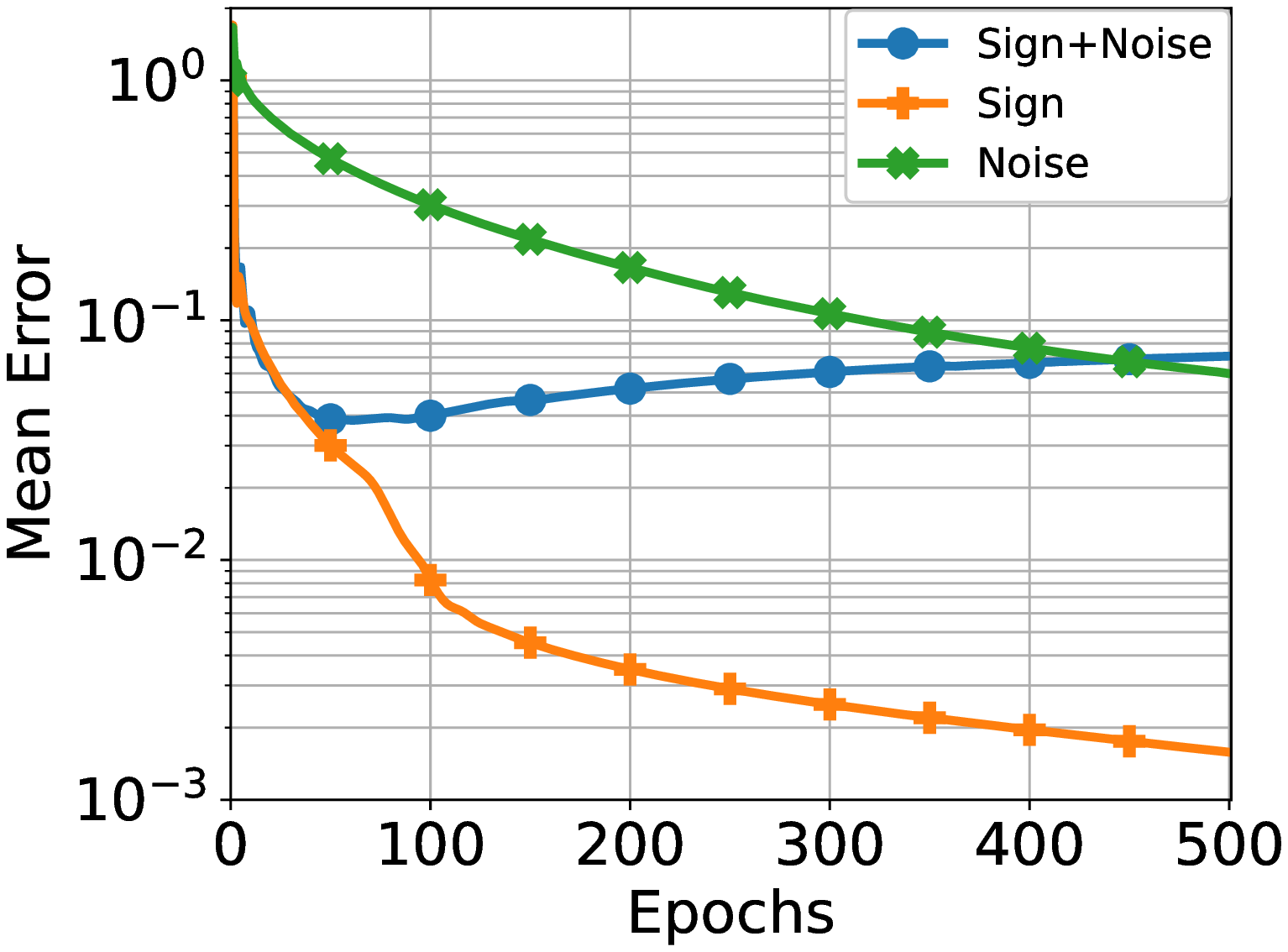}
        \captionsetup{justification=centering,margin=4.2cm}
        \caption{}\label{fig:experiments1a}
	\end{subfigure}
	\begin{subfigure}{0.49\textwidth}
		\centering
		\includegraphics[width=\textwidth]{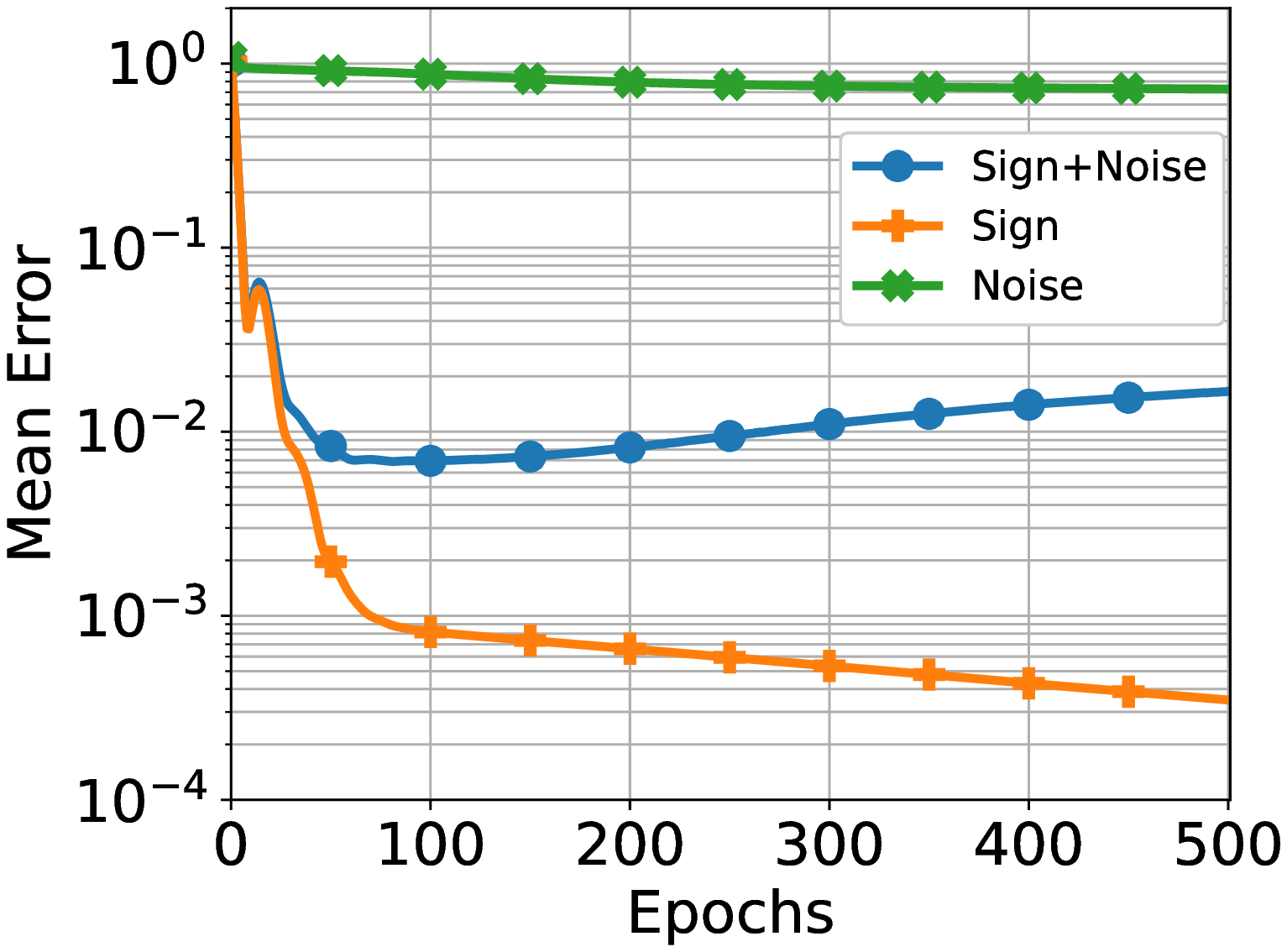}
		\captionsetup{justification=centering,margin=4.2cm}
		\caption{}\label{fig:experiments1b}
	\end{subfigure}
	\caption{a)~Error of the 2-layer GCG when fitting a piece-wise constant signal, noise, and a noisy signal, as a function of the number of epochs. The graph is drawn from an SBM with 64 nodes and 4 communities, and the normalized noise power is $P_n=0.1$.
	b)~Counterpart of a) but for the 2-layer GDec architecture.}
\end{figure}

\begin{figure}[t]
    \centering
    \includegraphics[width=0.55\textwidth]{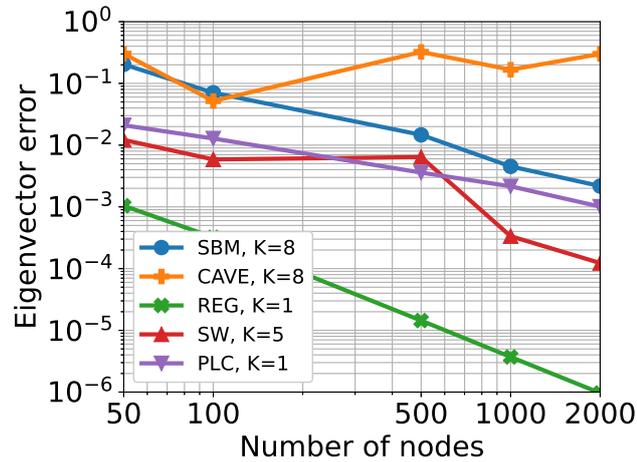}
	\caption{Mean distance between the $K$ leading eigenvectors of the adjacency matrix and $\sqJacob$ as a function of the graph size for several graph models.}\label{fig:experiments1c}
\end{figure}
%%%%%%%%%%%%%%%%%%%%%%%%%%%%%%%%

The goal of the second test case is two-fold.
First, it illustrates that the result presented in Lemma~\ref{lemma_eigs_gcg} is not constrained to the family of SBM (as specified by Ass.~1), but can be generalized to other families of random graphs as well.
In addition, it measures the influence of the number of nodes in the discrepancies between $\bbV_K$ and $\bbW_K$.
To that end, \cref{fig:experiments1c}  contains the mean eigenvector similarity measured as $\frac{1}{K}\|\bbV_K-\bbW_K\bbQ\|_F$ as a function of the number of nodes in the graph.
The eigenvector similarity is computed for 50 realizations of random graphs and the presented error is the median of all the realizations. 
The random graph models considered are: the SBM (``SBM''), the connected caveman graph (``CAVE'')\cite{watts1999networks}, the regular graph whose fixed degree increases with its size (``REG''), the small world graph (``SW'')\cite{watts1998collective}, and the power law cluster graph model (``PLC'')\cite{holme2002growing}.
The second term in the legend denotes the number of leading eigenvectors taken into account in each case, which depends on the number of active frequency components of the specific random graph model.   
We can clearly observe that for most of the random graph models, the eigenvector error goes to 0 as $N$ increases and, furthermore, the error is below $10^{-1}$ even for moderately small graphs.
This illustrates that, although the conditions assumed for Lemma~\ref{lemma_eigs_gcg} and Lemma~\ref{lemma_eigs_gd} focus on the specific setting of the SBM, the results can be applied to a wider class of graphs.
Here, the regular graphs are particularly interesting since most classical signals may be interpreted as signals defined over regular graphs.
As a result, this empirical evidence motivates the extension of the proposed theorems to more general settings as a future line of work.

\subsection{Denoising synthetic data}
We now proceed to comment on the denoising performance of the proposed architectures with synthetic data.
The usage of synthetic signals allows us to study how the properties of the noiseless signal influence the quality of the denoised estimate.

%%%%%%%%%%   FIGURE   %%%%%%%%%%
% \begin{figure*}[!tb]
% 	\centering
% 	%%% SW graph %%%
% 	\begin{subfigure}{0.32\textwidth}
% 		\centering
%         \includegraphics[width=1.05\textwidth]{figs/denoising/sig_inf_H_norm_1500epochs.eps}
%         \caption{}\label{fig:experiments2a}
% 	\end{subfigure}
% 	\begin{subfigure}{0.32\textwidth}
% 		\centering
% 		\includegraphics[width=1.05\textwidth]{figs/denoising/models_SBM8_BL.eps}
% 		\caption{}\label{fig:experiments2b}
% 	\end{subfigure}
% 	\begin{subfigure}{0.32\textwidth}
% 		\centering
% 		\includegraphics[width=1.05\textwidth]{figs/denoising/models_SBM8_DW.eps}
% 	    \caption{}\label{fig:experiments2c}
% 	\end{subfigure}
% 	\caption{Median MSE when denoising a graph signal as a function of the number of epochs. a)~The 2-layer GCG is used to denoise different families of signals. b)~Performance comparison between total variation, Laplacian regularization, bandlimited models, the 2-layer GCG, the deep GCG, and the deep GDec, when the signals are bandlimited.
% 	c)~Counterpart of b) for the case where signals are diffused white.}
% 	\label{fig:experiments2}
% \end{figure*}

\begin{figure}[t]
    \centering
        \includegraphics[width=.55\textwidth]{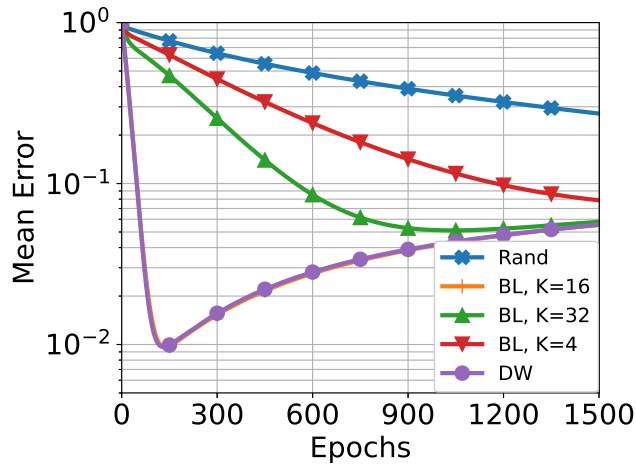}
        \caption{Median NMSE when the 2-layer GCG is used to denoise different families of graph signals.}\label{fig:experiments2a}
\end{figure}
\begin{figure}[t]
	\centering
	\begin{subfigure}{0.49\textwidth}
	    \centering
		\includegraphics[width=\textwidth]{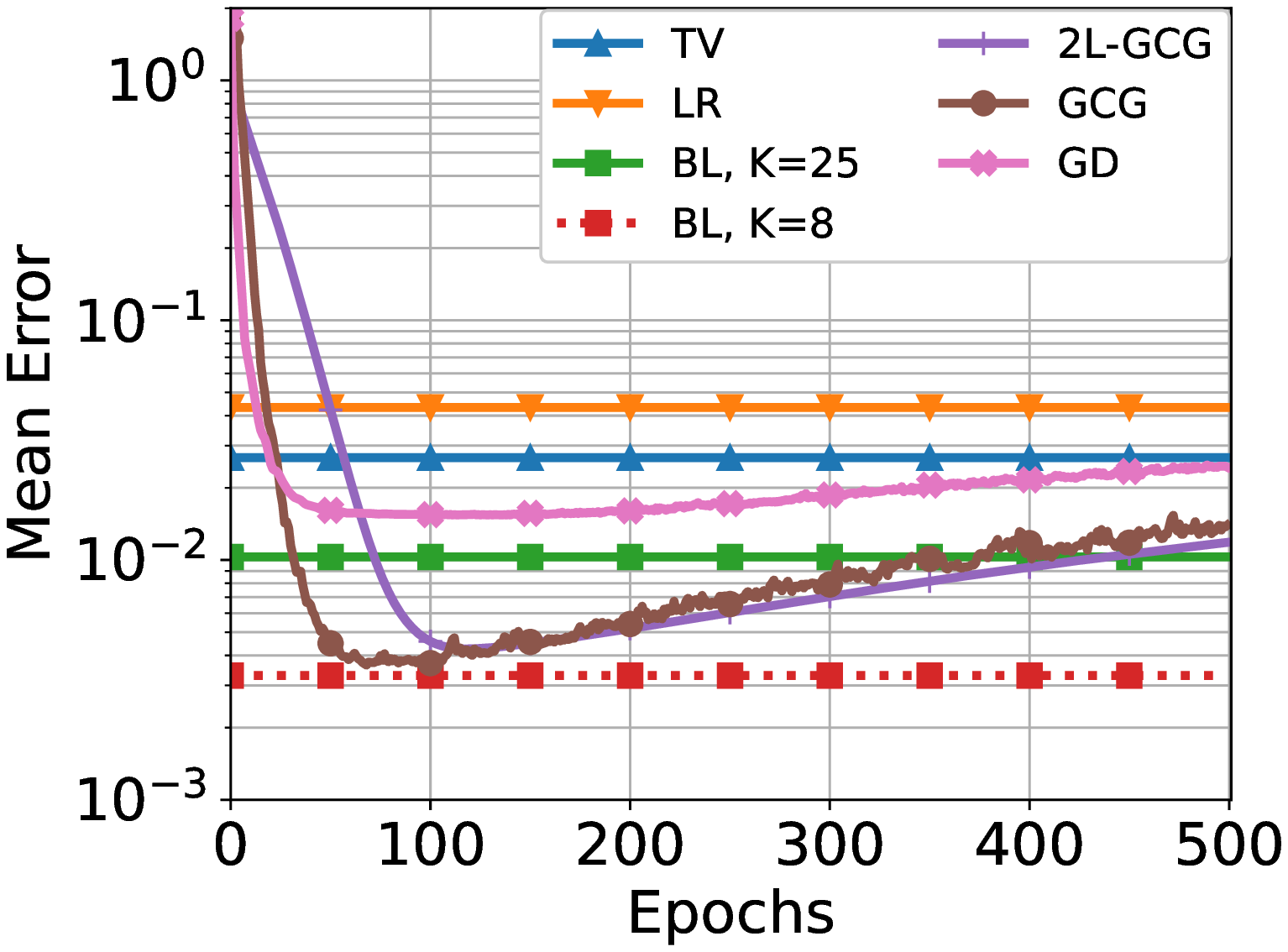}
		\captionsetup{justification=centering,margin=4.2cm}
		\caption{}\label{fig:experiments2b}
	\end{subfigure}
	\begin{subfigure}{0.49\textwidth}
		\centering
		\includegraphics[width=1.05\textwidth]{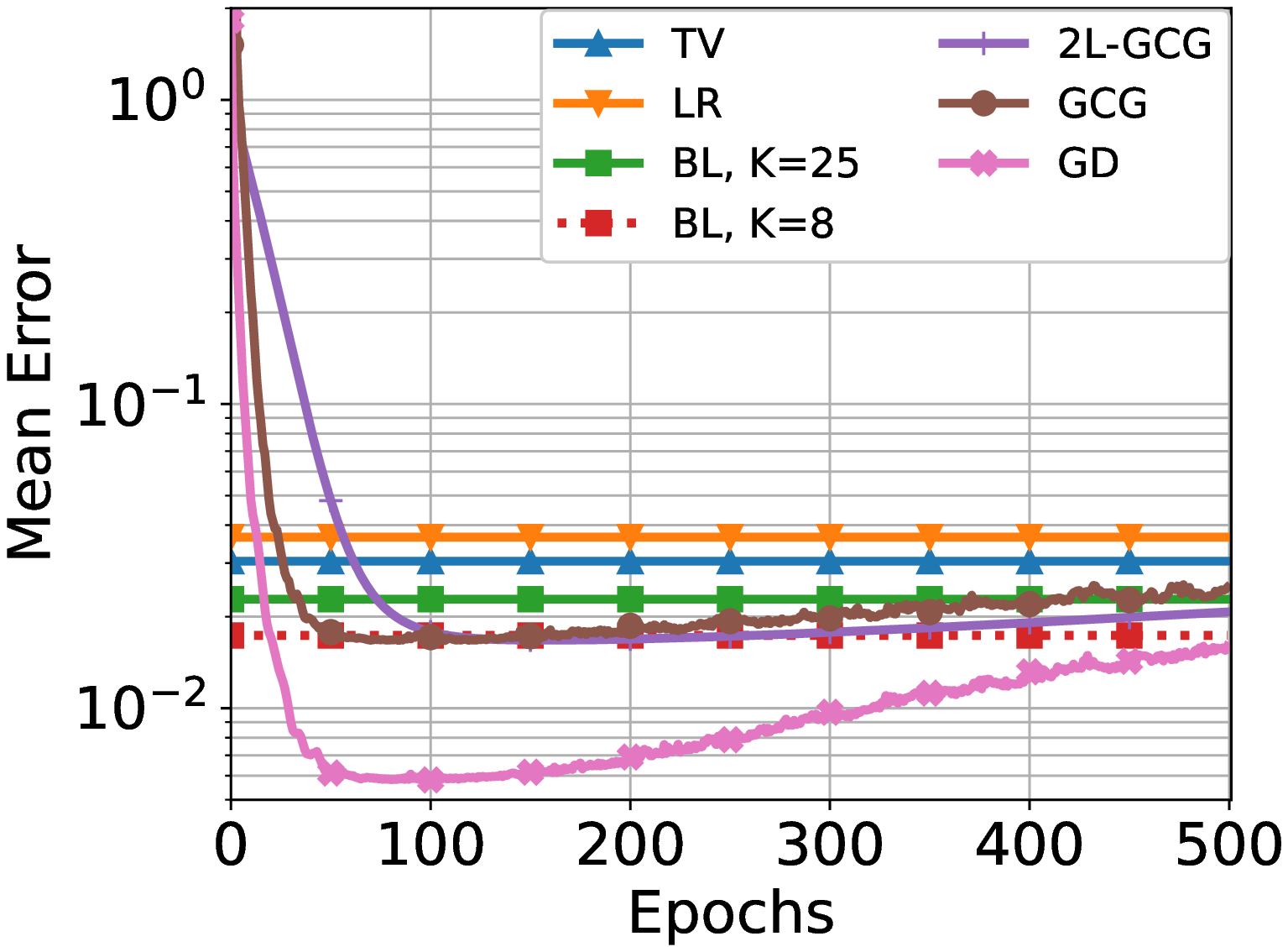}
		\captionsetup{justification=centering,margin=4.2cm}
	    \caption{}\label{fig:experiments2c}
	\end{subfigure}
	\caption{Median MSE when denoising a graph signal as a function of the number of epochs. a)~Performance comparison between total variation, Laplacian regularization, bandlimited models, the 2-layer GCG, the deep GCG, and the deep GDec, when the signals are bandlimited.
	b)~Counterpart of a) for the case where signals are diffused white.}
\end{figure}
%%%%%%%%%%%%%%%%%%%%%%%%%%%%%%%%

The first experiment,  shown in \cref{fig:experiments2a}, studies the error of the denoised estimate obtained with the 2-layer GCG as the number of epochs increases.
The reported error is the NMSE of the estimated signal $\hbx_0$, and the figure shows the mean values of 100 realizations of graphs and graph signals.
The normalized power of the noise present in the data is $0.1$.
Graphs are drawn from an SBM with $N=64$ nodes and 4 communities, and the graph signals are generated as: (i)~a  zero-mean  white Gaussian noise with unit variance (``Rand'');
(ii)~a bandlimited graph signal (cf. \ref{eq:bl_signals}) using the $K$ leading eigenvectors of $\bbA$ as base (``BL''); and (iii)~a diffused white (``DW'') signal created as $\bby = \mathrm{med}(\bbH\bbw|{\ccalG})$, where $\bbw$ is a white vector whose entries are sampled from $\ccalN(0,1)$, $\bbH$ is a low-pass graph filter, and $\mathrm{med}(\cdot|{\ccalG})$ represents the graph-aware median operator such that the value of the node $i$ is the median of its neighborhood~\cite{tay2020time, segarra2017designmedian}.
The results in \cref{fig:experiments2a} show that the best denoising error is obtained when the signal is composed of just a small number of eigenvectors, and the performance deteriorates as the bandwidth (i.e., the number of leading eigenvectors that span the subspace where the signal lives) increases, obtaining the worst result when the signal is generated at random.
This result is aligned with the theoretical claims since it is assumed that the signal $\bbx_0$ is bandlimited.
It is also worth noting that the architecture also achieves a good denoising error with the ``DW'' model, showcasing that the GCG is also capable of denoising other types of smooth graph signals.

Next, \cref{fig:experiments2b} compares the performance of the 2-layer GCG (``2L-GCG''), the deep GCG (``GCG'') and the deep GDec (``GDec'') with the baseline models introduced in \cref{sec:GNN_for_inverse_problems}, which are the total variation (``TV'')\cite{chen2014signal}, Laplacian regularization (``LR'')\cite{pang2017graph}, and bandlimited model (``BL'')\cite{chen2015discrete}.
In this setting, the graphs are SBM with 256 nodes and 8 communities, and the signals are bandlimited with a bandwidth of 8.
Since the ``BL'' model with $K=8$ captures the actual generative model of the signal $\bbx_0$, it achieves the best denoising performance.
However, it is worth noting that the GCG obtains a similar result, outperforming the other alternatives.
On the other hand, the ``LR'' obtains an error noticeably larger than that of ``BL'' and ``GCG'', highlighting that, even though ``BL'' and ``LR'' are related models their different assumptions lead to different performances.
Moreover, the benefits of using the deep GCG instead of the 2-layer architecture are apparent, since it achieves a better performance in fewer epochs.

On the other hand, \cref{fig:experiments2c} illustrates a similar experiment but with the graph signals generated as ``DW''.
Under this setting, it is clear that the GDec outperforms the other alternatives.
These results showcase the benefits of employing a nonlinear architecture relative to classical denoising approaches.
Furthermore, this experiment corroborates that the GDec is more robust to the presence of noise when the signals are aligned with the prior implicitly captured by the architecture.

\subsection{Denoising real-world signals}
Finally, we assess the performance of the proposed architectures in several real-world datasets.
To the baselines considered in the previous experiments, we add the following competitive denoising algorithms: graph trend filtering (``GTF'')~\cite{wang2015trend}, a graph-aware median operator (``MED'')~\cite{tay2020time}, a GCNN (``GCNN'') implemented as in~\cite{kipf2016semi}, a graph attention network (``GAT'')~\cite{velivckovic2017graph}, a Kron reduction-based autoencoder (``K-GAE'')~\cite{dorfler2012kron}, and the graph unrolling sparse coding architecture (``GUSC'') in~\cite{chen2021graph}.
Moreover, we consider the following noise distributions: (i) zero-mean Gaussian distribution, which is the noise model typically assumed for sensor measurements in signal processing; (ii) uniform distribution on some interval $[0, a]$, where $a\in\reals_+$ is chosen accordingly to the desired noise power; and (iii) Bernoulli distribution to model errors in binary signals.
Next, we describe the selected datasets and analyze the achieved results, which are summarized in Table~\ref{T:real_data}.

\vspace{3mm}
\noindent\textbf{Temperature.}
We consider a network of 316 weather stations distributed across the United States~\cite{sandryhaila2014discrete}.
Graph signals represent daily temperature measurements in the first three months of the year 2003.
The graph $\ccalG$ represents the geographical distance between weather stations and is given by the 8-nearest neighbors graph.
The first and second rows of Table~\ref{T:real_data} list the NMSE when the noise is drawn from a Gaussian and a uniform distribution, respectively.
In both cases, the noise has a normalized power of 0.3.
It is clear that the GDec architecture outperforms the alternatives in both scenarios.
Furthermore, we can observe that the GCG achieves a better performance than GCNN, showcasing the benefits of being able to use a more general graph filter.

\vspace{3mm}
\noindent\textbf{S\&P 500.}
In this experiment, we have 189 nodes representing stocks belonging to 6 different sectors of the S\&P 500 with the graph signals representing the prices of those stocks at particular time instants.
We follow \cite{cardoso2020algorithms} to estimate the graph $\ccalG$ assuming that the signals are drawn from a multivariate Gaussian distribution and are smooth on $\ccalG$.
We consider the noise specifications described in the previous dataset and provide the NMSE in the third and fourth rows of Table~\ref{T:real_data}.
It is worth noting that considering Gaussian noise in this dataset constitutes a more challenging denoising problem than using uniform noise.
A plausible explanation is that the graph is estimated assuming that the data follows a Gaussian distribution, and hence, it is harder to separate the Gaussian noise from the true signals. 
In the presence of Gaussian noise, the GCG and the GDec outperform the other 8 alternatives.
However, when the noise follows a uniform distribution, the best performance is obtained by the GCG and the GCNN, with GDec being the third best.
In addition, we observe that traditional methods yield an error that is considerably larger than that incurred by the proposed architectures.
This is aligned with our initial intuition about linear and quadratic methods being more limited when the actual relation between $\bbx_0$ and $\ccalG$ is more intricate, as is the case for financial data.

\vspace{3mm}
\noindent\textbf{Cora.}
Lastly, we consider the Cora citation network dataset~\cite{kipf2016semi}.
Nodes represent different scientific documents and edges capture citations among them.
Like in \cite{chen2021graph}, we consider the 7 class labels as binary graph signals encoding if the particular node belongs to that class.
For each signal, we consider 25 realizations of Bernoulli noise that randomly flips 30\% of the binary values of the signals, resulting in a total of 175 noisy graph signals.
With the error rate denoting the proportion of labels correctly recovered after the denoising process, Table~\ref{T:real_data} shows the error metric averaged over all the signals.
Moreover, since the graph is formed by several connected components, we report two results: the error rate when the whole graph is considered (fifth row) and the error rate when only the largest connected component is considered (sixth row).
It can be seen that the GCG yields the best performance in both cases.

\begin{table*}[t]
\begin{center}
\caption{Denoising error of several datasets with different types of random noise}\label{T:real_data}
\resizebox{1\columnwidth}{!}{
    \begin{tabular}{ c c || c c c c || c c c c c || c c}
    \toprule
    \begin{tabular}{c}
        DATASET  \\
        (METRIC)
    \end{tabular} & METHOD & BL & TV & LR & GTF & MED & GCNN & GAT & K-GAE & GUSC & GCG & GDec \\ \midrule 
    %\multirow{2}{*}{TEMPERATURE (NMSE)}
    TEMPERATURE & Gaussian & 0.062  & 0.117  & 0.095  & 0.066 & 0.053 & 0.123 & 0.045 & 0.134 & 0.044 & 0.056 & \textbf{0.035} \\
     (NMSE)   & Uniform  & 0.063  & 0.117  & 0.094  & 0.064 & 0.053 & 0.118 & 0.047 & 0.136 & 0.049 & 0.057 & \textbf{0.036} \\ \midrule
    S\&P 500 & Gaussian & 0.350 & 0.238 & 0.231 & 0.239 & 0.319 & 0.252 & 0.199 & 0.354 & 0.203 & \textbf{0.188} & \textbf{0.188} \\
     (NMSE) & Uniform & 0.216 & 0.246 & 0.161 & 0.298 & 0.340 & \textbf{0.091} & 0.222 & 0.273 & 0.127 & \textbf{0.094} & 0.121 \\ \midrule
    CORA & Whole $\ccalG$ & 0.154 & 0.142 & 0.115 & 0.126 & 0.167 & 0.099 & 0.141 & 0.135 & 0.099 & \textbf{0.093} & 0.121 \\
     (ERROR RATE) & Conn. comp. & 0.151 & 0.141 & 0.105 & 0.116 & 0.165 & 0.093 & 0.139 & 0.135 & 0.094 & \textbf{0.088} & 0.125 \\ \bottomrule
    \end{tabular}
}
\end{center}
\end{table*}

\section{Conclusion}\label{sec:conclusion}
In this chapter, we faced the relevant task of graph-signal denoising. 
To approach this problem, we presented two overparametrized and untrained GNNs and provided theoretical guarantees on the denoising performance of both architectures when denoising $K$-bandlimited graph signals under some simplifying assumptions.
Moreover, we numerically illustrated that the proposed architectures are also capable of denoising graph signals in more general settings.  
The key difference between the two architectures resided in the linear transformation that incorporates the information encoded in the graph.
The GCG employs fixed (non-learnable) low-pass graph filters to model convolutions in the vertex domain, promoting smooth estimates.
On the other hand, the GDec relies on a nested collection of graph upsampling operators to progressively increase the input size, limiting the degrees of freedom of the architecture, and providing more robustness to noise.
In addition to the aforementioned analysis, we tested the validity of the proposed theorems and evaluated the performance of both architectures with real and synthetic datasets, showcasing a better performance than other classical and nonlinear methods for graph-signal denoising.
Finally, we consider extending the results from \cref{theorem_denoising_gcg} and \cref{theorem_denoising_gd} to more general scenarios as an interesting future line of work.

\newpage
\section{Appendix: Proof of Theorem~\ref{theorem_denoising_gcg}}\label{proof_theorem_denoising_gcg}
Let $\bbx_0$ be a $K$ bandlimited graph signal as described in \eqref{eq:bl_signals}, which is spanned by the $K$ leading eigenvectors of the graph $\bbV_K$, with $\tbx_0$ denoting its frequency representation.
Let $\bbQ$ be an orthonormal matrix that aligns the subspaces spanned by $\bbV_K$ and $\bbW_K$, and denote as $\barbx_0=\bbW_K\bbQ\tbx_0$ the bandlimited signal using $\bbW_K$ as basis and whose frequency response is also $\tbx_0$.
Note that $\barbx_0$ can be interpreted as recovering $\bbx_0$ from its frequency response using $\bbW_K$ in lieu of $\bbV_K$.
Also, note that $\bbx_0-\barbx_0=(\bbV_K-\bbW_K\bbQ)\tbx_0$ represents the error between the signal $\bbx_0$ and its approximation inside the subspace spanned by $\bbW_K$.
With these definitions in place, in~\cite[Th. 3]{heckel2019denoising} the authors showed that error when denoising a signal $\bbx = \bbx_0 + \bbn$ is bounded with probability at least $1-e^{-F^2}-\phi$ by
\begin{align}\label{eq:bound_theorem_eigs_J}
    & \|\bbx_0-f_{\bbTheta_{(t)}}(\bbZ|\ccalG)\|_2 \leq  \|\bbXi\bbx_0\|_2+\xi\|\bbx\|_2\\ \nonumber
    & +\sqrt{\textstyle \sum_{i=1}^N((1-\eta\sigma_i^2)^t-1)^2(\bbw_i^\top\bbn)^2},
\end{align}
with $\bbXi:=\bbW(\bbI_N-\eta\bbSigma^2)^t\bbW^\top$, and $\bbI_N$ the $N\times N$ identity matrix.
However, note that the bound provided for $\|\bbXi\bbx_0\|_2$ in \cite{heckel2019denoising} requires that $\bbx_0$ lies in the subspace spanned by $\bbW_K$, which is not the case.
As a result, we further bound this term as
\begin{align}\label{eq:bound_basis_mismatch}
    \|\bbXi  \bbx_0 \|_2 &= \|\bbXi(\bbx_0+\barbx_0-\barbx_0)\|_2 \nonumber\\
    & \stackrel{(i)}{=} \|\bbXi_K\barbx_0+\bbXi(\bbV_K-\bbW_K\bbQ)\tbx_0\|_2 \nonumber \\
    &\stackrel{(ii)}{\leq}\|\bbXi_K\barbx_0\|_2+\|\bbXi(\bbV_K-\bbW_K\bbQ)\tbx_0\|_2 \nonumber \\
    &\stackrel{(iii)}{\leq}\|\bbXi_K\|_2\|\barbx_0\|_2+ \|\bbXi\|_2\|\bbV_K-\bbW_K\bbQ\|_F\|\tbx_0\|_2 \nonumber \\
    &\stackrel{(iv)}{\leq}\left(\|\bbXi_K\|_2 + \delta\|\bbXi\|_2\right)\|\bbx_0\|_2 \nonumber \\
    &\stackrel{(v)}{=} \left((1-\eta\sigma_K^2)^t+\delta(1-\eta\sigma_N^2)^t\right)\|\bbx_0\|_2.
\end{align}
Here, $\bbXi_K:=\bbW_K(\bbI_K-\eta\bbSigma_K^2)^t\bbW_K^\top$, and $\bbSigma_K$ represents a diagonal matrix containing the first $K$ leading eigenvalues $\sigma_k$.
We have that $(i)$ follows from $\barbx_0$ being bandlimited in $\bbW_K$, so $\bbXi\barbx_0=\bbXi_K\barbx_0$.
Then, $(ii)$ follows from the triangle inequality, and $(iii)$ from the $\ell_2$ norm being submultiplicative and using the Frobenius norm as an upper bound for the $\ell_2$ norm.
In $(iv)$ we apply the result of Lemma~\ref{lemma_eigs_gcg}, which holds with probability at least $1-\epsilon$ because $N>N_{\epsilon,\delta}$, and the fact that, since both $\bbW_K$ and $\bbV_K$ are orthonormal matrices, we have that $\|\bbx_0\|_2=\|\barbx_0\|_2=\|\tbx_0\|_2$.
We obtain $(v)$ from the largest eigenvalues present in $\bbXi_K$ and $\bbXi$.

Finally, the proof concludes by combining \eqref{eq:bound_basis_mismatch} and \eqref{eq:bound_theorem_eigs_J}.

\section{Appendix: Proof of Lemma \ref{lemma_eigs_gcg}}\label{proof_lemma_eigs_gcg}
Define $\ccaltbA$ as $\ccaltbA:=\mathbb{E}[\tbA]=\mathbb{E}[\bbD]^{-\frac{1}{2}}\ccalbA\mathbb{E}[\bbD]^{-\frac{1}{2}}$ and let $\sqJacob$ be given by \eqref{eq:expected_jac}. 
Denote by $\ccalbH$ a graph filter defined as a polynomial of the expected adjacency matrix $\ccaltbA$, and let $\barsqJacob$ be the expected squared Jacobian using the graph filter $\ccalbH$, i.e.,
\begin{equation}\label{eq:expected_E}
    \barsqJacob = 0.5 \left( \mathbf{1} \mathbf{1}^\top - \frac{1}{\pi} \arccos(\ccalbC^{-1} \ccalbH^2 \ccalbC^{-1})\right) \circ \ccalbH^2,
\end{equation}
where $\ccalbC$ is the counterpart of $\bbC$ in \eqref{eq:expected_jac}, but using $\ccalbH$ instead of $\bbH$.
Given the following eigendecompositions $\tbA=\bbV\bbLambda\bbV^\top$, $\sqJacob=\bbW\bbSigma\bbW^\top$, $\ccaltbA=\barbV\barbLambda\barbV^\top$, and $\barsqJacob=\barbW\barbSigma\barbW^\top$, for arbitrary orthonormal matrices $\bbT$ and $\bbR$, we have that
\begin{equation}\label{eq:main_triangle_ineq}
\| \bbV_K - \bbW_K \bbQ \|_{\text{F}} \leq  \| \bbV_K - \bar{\bbV}_K \bbT \|_{\text{F}} + \| \bar{\bbV}_K \bbT - \bar{\bbW}_K \bbR \|_{\text{F}} + \| \bar{\bbW}_K \bbR - \bbW_K \bbQ \|_{\text{F}}.
\end{equation}
To prove the theorem, we bound the three terms on the right hand side of \eqref{eq:main_triangle_ineq}.

\vspace{1mm}
\noindent \emph{Bounding $\| \barbV_K \bbT - \barbW_K \bbR \|_{\text{F}}$.}
From the definition of an SBM, it follows that $\ccalbA=\mathbb{E}[\bbA]=\bbB\bbOmega\bbB^\top$, where $\bbB \in \{0,1\}^{N\times K}$ is an indicator matrix encoding the community to which each node belongs, and $\bbOmega$ is a $K \times K$ matrix encoding the link probability between the communities of the graph.
Therefore, $\ccaltbA$ and $\barsqJacob$ are both block matrices whose blocks coincide with the communities in the SBM.
This implies that the eigenvectors associated with non-zero eigenvalues must span the columns of $\bbB$. 
Hence, the leading eigenvectors must be related by an orthonormal transformation, from where it follows that, given $\bbT$, we can always find $\bbR$ such that
\begin{equation}\label{eq:first_bound}
\| \bar{\bbV}_K \bbT - \bar{\bbW}_K \bbR \|_{\text{F}} = 0.
\end{equation}

\vspace{1mm}
\noindent \emph{Bounding $\| \bbV_K - \bar{\bbV}_K \bbT \|_{\text{F}}$.} 
Under Ass.~\ref{A:sbm}, as it is shown in~\cite{schaub2020blind}, with probability at least $1-\rho$ we have that
\begin{equation}\label{eq:second_bound_aux}
    \| \tilde{\bbA} - \ccaltbA\| \leq 3 \sqrt{\frac{3 \ln(4N/\rho)}{\beta_{\min}}}.
\end{equation}
Then, we combine the concentration \eqref{eq:second_bound_aux} with the Davis-Kahan results~\cite[Th. 2]{yu2015useful}, which bound the distance between the subspaces spanned by the population eigenvectors ($\barbV_K$) and their sample version ($\bbV_K$).
Denoting as $\bar{\lambda}_i$ the $i$-th eigenvalue collected in $\barbLambda$, i.e. $\bar{\lambda}_i=\bar{\Lambda}_{ii}$, we obtain that there exists an orthonormal matrix $\bbT$ such that
\begin{equation}\label{eq:second_bound}
    \|\bbV_K-\barbV_K\bbT\|_F \leq \frac{\sqrt{8K}}{\bar{\lambda}_{K}-\bar{\lambda}_{K+1}}\|\tbA-\ccaltbA\|_F \leq \frac{3\sqrt{8K}}{\bar{\lambda}_{K}}\sqrt{\frac{3 \ln(4N/\rho)}{\beta_{\min}}},
\end{equation}
where we note that, since $\ccaltbA$ follows an SBM, then $\bar{\lambda}_i=0$ for all $i>K$.

Since $\beta_{\min} = \omega(\ln (N/\rho))$, we obtain that 
\begin{equation}\label{eq:second_bound_lim}
\| \bbV_K - \bar{\bbV}_K \bbT \|_{\text{F}} \to 0, \qquad \text{as } N \to \infty.
\end{equation}

\vspace{1mm}
\noindent \emph{Bounding $\| \bar{\bbW}_K \bbR - \bbW_K \bbQ \|_{\text{F}}$.}
If we show that $\|\sqJacob - \barsqJacob\| \to 0$ as $N \to \infty$, we can then mimic the procedure in~\eqref{eq:second_bound_aux} and~\eqref{eq:second_bound} to show that the difference between the leading $K$ eigenvectors of $\sqJacob$ and $\barsqJacob$ also vanishes. 
Hence, we are left to show that $\|\sqJacob - \barsqJacob\| \to 0$ as $N \to \infty$. From the definitions of $\sqJacob$ and $\barsqJacob$, it follows that
\begin{align}\label{eq:third_bound_aux}
\|\sqJacob - \barsqJacob\| \leq \; &0.5 \| \bbH^2 - \ccalbH^2 \| + \frac{1}{2\pi} \| \arccos(\ccalbC^{-1} \ccalbH^2 \ccalbC^{-1}) \circ \ccalbH^2 \\ \nonumber
&-\arccos(\bbC^{-1} \bbH^2 \bbC^{-1}) \circ \bbH^2 \|.
\end{align}
To bound the difference between the sampled and expected filters, we have that
\begin{align}\label{eq:third_bound_aux_2}
&\| \bbH^2 - \ccalbH^2 \| = \left\| \left(\sum_{\ell=0}^L h_\ell \tilde{\bbA}^\ell \right)^2 - \left(\sum_{\ell=0}^L h_\ell \ccaltbA^\ell\right)^2 \right\| \\ \nonumber
&= \left\| \sum_{\ell=0}^{2L} \alpha_\ell (\tilde{\bbA}^\ell - \ccaltbA^\ell) \right\| \leq \sum_{\ell=0}^{2L} \alpha_\ell \left\| \tilde{\bbA}^\ell - \ccaltbA^\ell\right\|,
\end{align}
for suitable coefficients $\alpha_\ell$ and recalling that $L=2$. Then,
we can then leverage the fact that $\|\tilde{\bbA}\| = \|\ccaltbA\| = 1$ to see that $\left\| \tilde{\bbA}^\ell - \ccaltbA^\ell\right\| \leq \ell \left\| \tilde{\bbA} - \ccaltbA\right\|$. We thus get that
\begin{equation}\label{eq:third_bound_aux_3}
\| \bbH^2 - \ccalbH^2 \| \leq \sum_{\ell=0}^{2L} \ell \alpha_\ell  \left\| \tilde{\bbA} - \ccaltbA\right\| \to 0, \quad \text{as } N \to \infty,
\end{equation}
where the limiting behavior follows from~\eqref{eq:second_bound_aux}. Finally, to bound the second term in~\eqref{eq:third_bound_aux}, we first note that the argument of the norm can be re-written as $\arccos(\bbC^{-1} \bbH^2 \bbC^{-1}) \circ (\ccalbH^2 - \bbH^2) +  (\arccos(\ccalbC^{-1} \ccalbH^2 \ccalbC^{-1}) - \arccos(\bbC^{-1} \bbH^2 \bbC^{-1})) \circ \ccalbH^2$. The limit in~\eqref{eq:third_bound_aux_3} ensures that the first of these two terms vanishes. 
Similarly, it follows that $\| \ccalbC^{-1} \ccalbH^2 \ccalbC^{-1} - \bbC^{-1} \bbH^2 \bbC^{-1} \| \to 0$ which, combined with the fact that $\arccos$ is a uniformly continuous function, we can always find an $N_{\delta'}$ such that $\| \arccos(\ccalbC^{-1} \ccalbH^2 \ccalbC^{-1}) - \arccos(\bbC^{-1} \bbH^2 \bbC^{-1})\| \leq \delta'$ with high probability. Combining this result with~\eqref{eq:third_bound_aux_3} and applying the Davis-Kahan Theorem as done to obtain~\eqref{eq:second_bound} we get that
\begin{equation}\label{eq:third_bound_lim}
\| \bar{\bbW}_K \bbR - \bbW_K \bbQ \|_{\text{F}} \to 0, \qquad \text{as } N \to \infty.
\end{equation}

\vspace{3mm}
Replacing~\eqref{eq:first_bound}, \eqref{eq:second_bound_lim}, and \eqref{eq:third_bound_lim} into~\eqref{eq:main_triangle_ineq} our result follows.

\section{Appendix: Proof of Lemma \ref{lemma_eigs_gd}}\label{proof_lemma_eigs_gd}

Recall that $\ccaltbA=\mathbb{E}[\tbA]$, and define  $\ccaltbH:=\gamma\bbI+(1-\gamma)\ccaltbA$ as the specific graph filter introduced in \cref{sec:upsampling_operator} as a polynomial of $\ccaltbA$.
Let $\sqJacob$ be given by equation \eqref{eq:expected_jac_dec}, and denote by $\barsqJacob$ the expected squared Jacobian using the graph filter $\ccalbH$, i.e.,
\begin{align}\label{eq:expected_E_dec}
    \barsqJacob = 0.5 \left( \mathbf{1} \mathbf{1}^\top - \frac{1}{\pi} \arccos(\ccaltbC^{-1}  \ccalbU\ccalbU^\top\ccaltbC^{-1})\right) \circ \ccalbU\ccalbU^\top
 \end{align}
with $\ccalbU=\ccaltbH\bbP$ and where the matrix $\ccaltbC$ is the counterpart of $\tbC$ in \eqref{eq:expected_jac_dec}, but using $\ccalbU$ in lieu of $\bbU$.
Given the eigendecompositions $\tbA=\bbV\bbLambda\bbV^\top$, $\sqJacob=\bbW\bbSigma\bbW^\top$, $\ccaltbA=\barbV\barbLambda\barbV^\top$, and $\barsqJacob=\barbW\barbSigma\barbW^\top$, analogously to Lemma~\ref{lemma_eigs_gcg}, we bound the difference between $\bbV_K$ and $\bbW_K$ by bounding the three terms in the right hand side of \eqref{eq:main_triangle_ineq}.

\vspace{1mm}
\noindent\textit{Bounding $\|\bar{\bbV}_K\bbT-\bar{\bbW}_K\bbR\|$}. 
We have that $\ccalbU\ccalbU^\top=\ccaltbH\bbP\bbP^\top\ccaltbH^\top$.
Since $\bbP$ is a binary matrix indicating to which community belongs each node, $\bbP\bbP^\top$ is a block diagonal matrix that captures the structure of the communities of the SBM.
Then, because $\ccaltbH$ is also block matrix with the same block pattern that the SBM, it turns out that the matrix $\barsqJacob$ is also a block matrix whose blocks coincide with the communities in the SBM graph.
Therefore, the rest of the bound is analogous to that in Lemma~\ref{lemma_eigs_gcg}.

\vspace{1mm}
\noindent\textit{Bounding $\|\bbV_K-\bar{\bbV}_K\bbT\|$}. The relation between $\bbA$ and $\ccalbA$ is the same as in Lemma~\ref{lemma_eigs_gcg} so the bound provided in \eqref{eq:second_bound_lim} holds.

\vspace{1mm}
\noindent\textit{Bounding $\|\bar{\bbW}_K\bbR-\bbW_K\bbQ\|$}. To derive this bound we show that $\|\bbU\bbU^\top-\ccalbU\ccalbU^\top\|=\|\tbH\bbP\bbP^\top\tbH^\top-\ccaltbH\bbP\bbP^\top\ccaltbH^\top\|$ goes to 0 as $N$ grows.
From \eqref{eq:third_bound_aux_3} we have that $\|\bbH-\ccalbH\|\to0$, as $N\to\infty$, and hence, $\|\tbH-\ccaltbH\|\to0$, as $N\to\infty$.
Therefore, it can be seen that
\begin{equation}
    \|\bbU\bbU^\top-\ccalbU\ccalbU^\top\|  \to 0, \quad \text{as } N \to \infty,
\end{equation}
with $\|\bbU\bbU^\top-\ccalbU\ccalbU^\top\|$ vanishing as $N$ grows. The remainder of the derivation of the bound is analogous to that for \eqref{eq:third_bound_lim}.
\chapter{Signal interpolation of diffused sparse signals}\label{chap:interpolation}
We continue with GSP setups where signals are corrupted by perturbations, looking at a scenario where the (network aggregated) data has missing (limited) values.
A natural and systematic approach in these setups is to model the observed values as a \emph{sampled} signal, and then, propose an \emph{interpolation} algorithm that reconstructs the original signal.
Key to the success of the interpolation is to exploit the structure of both the graph signals and the assumed sampling scheme employed to represent the missing values.
This is precisely the strategy put forth in our work in~\cite{rey2019sampling}, which is also the main subject of this chapter.

The outline of the chapter is as follows.
After briefly highlighting our contributions and remembering some concepts fundamental for this work in \cref{sec:sampling_intro}, \cref{sec:SamplingBandlimited} presents the main results, including the recovery schemes with known and unknown seeds, as well as unknown diffusing filter.
The effect of noise, the design of the sampling matrix, and the consideration of more than one sampling node are also briefly analyzed, and next, a gamut of illustrative numerical results, including some showcasing practical relevance in real datasets, are presented in \cref{sec:NumExper}.

\section{Introduction}\label{sec:sampling_intro}
In this chapter, we are concerned with the recovery of perturbed graph signals with a non-negligible proportion of missing values.
Furthermore, we assume that the observations are gathered using an AGSS scheme, so that the (sampling and reconstruction) results in~\cite{marques2015sampling,chen2016signal} can be leveraged.
In addition, because dealing with a large number of missing values (or, equivalently, having access to only a few sampled values) is a non-trivial problem, we assume that the original (unperturbed) signals studied in this framework are DSGS.
In this sense, recall that DSGS are a class of signals that can be modeled as a sparse graph signal, i.e., a signal that is zero everywhere except in a few seeding nodes, which is then diffused through the network via a GF.
First, we aim at reconstructing the signal assuming that the seeding nodes are known, but ultimately, we consider that the seeds are also unknown and our goal consists in \textit{reconstructing both} the unperturbed signal and the seeds from the available observations.
It is worth noticing that the AGSS is a sampling method for graph signals introduced in \cite{marques2015sampling} where nodes successively aggregate the values of the signal in their neighborhood. 
Under this setting, the recovery can be guaranteed even if observations are gathered at a single node and the sampling collection process can be implemented distributively.

\vspace{3mm}
\noindent
\textbf{Contributions.}
The main contribution of this chapter is the generalization of the AGSS, which was originally proposed for BGS, to DSGS. Additionally, we generalize existing results for support identification and blind deconvolution to setups where observations are collected using an AGSS. The algorithms presented in this chapter are relevant also for distributed estimation and source localization. Sampling and recovery using as input the signal value at a subset of nodes were discussed in \cite{anis2014towards,chen2015discrete,tsitsvero2016signals} for BGS, and in \cite{ramirez2017graph} for DSGS. Aggregation and space-shift sampling (a generalization of the AGSS considering multiple nodes) for BGS were investigated in \cite{marques2015sampling}. Blind deconvolution and filter identification for DSGS in a centralized setup with access to the full signal (no sampling) were investigated in \cite{segarra2016blind}.

\subsection{Successively aggregated graph signals}
Fundamental to the approach put forth in this chapter is the AGSS, which was introduced in \cref{sec:denoising_and_interpolation}, and the definition of successively aggregated graph signals $\bbz_i$.
Here, we briefly refresh these concepts and highlight their local nature.

Consider the signal $\bby = \bbH\bbx$, which is the diffusion of $\bbx$ across the graph $\ccalG$ as modeled by the GF $\bbH$ with coefficients $\bbh$, and let us define the $r$-th shifted signal $\bbz^{(r)}:=\bbS^r\bbx$ and further define the $N\times N$ matrix
\begin{equation}\label{eq:shift_matrix_sampling}
    \bbZ\,:=\,[\bbz^{(0)},\bbz^{(1)},\ldots,\bbz^{(N-1)}]
    \,=\,[\bbx,\bbS\bbx,\ldots,\bbS^{N-1}\bbx],
\end{equation}
Intuitively, $\bbZ$ groups the signal $\bbx$ and the result of the first $N-1$ applications of the GSO.
It is then clear that: a) the output of the graph filter can be found as $\bby = \bbH \bbx = \bbZ \bbh$, with $\bbh$ being zero-padded if $R<N$; b) since $\bbS$ is a local operator, the $r$-th column of $\bbZ$ can be found locally from the $(r-1)$-th diffusion step as $\bbz^{(r)}=\bbS\bbz^{(r-1)}$; and c) as one moves right-wise in \eqref{eq:shift_matrix_sampling}, the columns of $\bbZ$ can be viewed as the evolution of a process which is diffused linearly according to the local structure codified in $\bbS$.

Then, in the AGSS we consider the $i$-th node fixed and sample the signals observed at this node as the GSO $\bbS$ is applied recursively.
In other words, as the signal has been locally diffused according to $\bbS$, as described in \eqref{eq:shift_matrix_sampling}.
Then, with $\bbe_i$ denoting the $i$-th canonical vector and leveraging the definition of the matrix $\bbZ$, the successively aggregated signal observed by node $i$ is the $i$-th row of \textbf{$\bbZ$}, that is $\bbz_i:=(\bbe_i^T\bbZ)^T=\bbZ^T\bbe_i$.
Under this setting, sampling is reduced to the selection of $Q$ out of the $N$ elements of $\bbz_i$, which can be done as follows
\begin{equation}
    \bbz_{\ccalQ,i} \ :=\ \bbPi_{\ccalQ}\bbz_i\ =\ \bbPi_{\ccalQ} \left(\bbZ^T\bbe_i\right).
\end{equation}

\section{Aggregation Sampling of DSGS}\label{sec:SamplingBandlimited}
This section considers the problem of reconstructing DSGS from \textit{local} observations obtained from AGSS. In contrast with BGS, the interest when dealing with DSGS can be either in recovering $\bbx$ (signal reconstruction) or in recovering $\bbs$ (distributed source localization and estimation). For that reason, we start by analyzing the case where $\bbH=\bbI$ and $\bbx=\bbs$. After that, we discuss the more general case where the nodes sample the signal $\bbx=\bbH\bbs$, both for known and unknown $\bbH$. A collaborative setting where more than one node collects the samples closes the section. 
To help readability, a summary of the setups considered is provided in Table \ref{T:setting_summary}.

\begin{table}[htpb]
	\centering
	%\resizebox{0.8\textwidth}{!}{%
		\begin{tabular}{|c|c|c|c|c|c|}
			\hline
			\textbf{\begin{tabular}[c]{@{}c@{}}Scenario \end{tabular}} & \textbf{\begin{tabular}[c]{@{}c@{}}Sparse \\ support\end{tabular}} & \textbf{$\bbH$} & \textbf{$\bbPi_\ccalQ$} & \textbf{\begin{tabular}[c]{@{}c@{}}Obs.\\ matrix\end{tabular}} & \textbf{Equation} \\ \hline
			Sparse recovery & Known & \begin{tabular}[c]{@{}c@{}}Known\\ $\bbH=\bbI$\end{tabular} & Fixed & $\bbTheta$ & \eqref{eq:optimal_estimator} \\ \hline
			Active sampling & Known & \begin{tabular}[c]{@{}c@{}}Known\\ $\bbH=\bbI$\end{tabular} & Flexible & $\bbTheta$ & \eqref{eq:optimal_sampling} \\ \hline
			\begin{tabular}[c]{@{}c@{}}Blind sparse\\ recovery\end{tabular} & Unknown & \begin{tabular}[c]{@{}c@{}}Known\\ $\bbH=\bbI$\end{tabular} & Fixed & $\bbTheta$ & \eqref{eq:sampled_shifted_signal_unkown} \\ \hline
			Diffused recovery & Known & Known & Fixed & $\bbXi$ & 	\begin{tabular}[c]{@{}c@{}}\eqref{eq:optimal_estimator}, replacing  \\ $\bbTheta$  with $\bbXi$  \end{tabular} \\ \hline
			\begin{tabular}[c]{@{}c@{}}Diffused active\\ recovery\end{tabular} & Known & Known & Flexible & $\bbXi$ & \begin{tabular}[c]{@{}c@{}}\eqref{eq:optimal_sampling}, replacing  \\ $\bbTheta$ with $\bbXi$  \end{tabular}  \\ \hline
			\begin{tabular}[c]{@{}c@{}}Blind diffused\\ recovery\end{tabular} & Unknown & Known & Fixed & $\bbXi$ & \begin{tabular}[c]{@{}c@{}}\eqref{eq:sampled_shifted_signal_unkown}, replacing  \\ $\bbTheta$ with $\bbXi$  \end{tabular}  \\ \hline
			\begin{tabular}[c]{@{}c@{}}Blind\\ deconvolution\end{tabular} & Unknown & Unknown & Fixed & $\bbPhi$ & \eqref{eq:optimal_blind_deconv} \\ \hline
		\end{tabular}
	%}
	\caption{Compilation of the settings considered in \cref{sec:SamplingBandlimited}.}
	\label{T:setting_summary}
\end{table}

\subsection{Aggregating the sparse input}
A critical aspect to analyze the recovery $\bbx=\bbs$ from its aggregated samples is to write the relationship between the sampled signal $\bbz_{\ccalQ}$ and the sparse input $\bbs$. For the ease of exposition, we do that in the form of a lemma. 
\begin{lemma} Given the GSO $\bbS\in \reals^{N\times N}$, the sampling matrix $\bbPi_{\ccalQ}$ and a sparse input $\bbx=\bbs$, the shifted signal $\bbz_i$ and its sampled version $\bbz_{\ccalQ,i}$ can be expressed as
	\begin{equation}\label{eq:shifted_signal}
	    \bbz_{\ccalQ,i}=\bbPi_{\ccalQ}\bbz_i\;\;\text{and}\;\;\bbz_i=\boldsymbol{\Psi}^T\diag(\boldsymbol{\upsilon}_i)\bbU \bbs:=\bbTheta_i\bbs.
	\end{equation}
\end{lemma}
For the expression above, note that we have defined the observation matrix $\bbTheta_i:=\boldsymbol{\Psi}^T\diag(\boldsymbol{\upsilon}_i)\bbU$, where we recall that $\bbPsi$ is the GFT for filters, $\bbU := \bbV^{-1}$ is the GFT for signals, and $\bbupsilon_i=[V_{i,1},...,V_{i,N}]^T$. While nontrivial, \eqref{eq:shifted_signal} can be derived after substituting \eqref{eq:all_shifts_matrix} into \eqref{eq:diff_sparse_signals_def}, or by a minor modification of the proof in \cite[Lemma 2]{marques2015sampling}. Interestingly, \eqref{eq:shifted_signal} reveals that $\bbz_{\ccalQ,i}$ depends on how strongly the seeds express each of the frequencies (represented by  $\bbU\bbs$), how strongly the sampling node senses each of the frequencies (represented by the frequency pattern $\bbupsilon_i$) and the spectral effect of the diffusion (powers of the GSO) captured in the Vandermode matrix $\bbPsi$.

To proceed with the recovery of the sparse input, let us denote as $\ccalS$ the support of $\bbs$, define $\bbs_\ccalS:=\bbPi_\ccalS\bbs \in \reals^S$ as the vector collecting the non-zero values of $\bbs$, and suppose for now that this support is known. Then, we have that
\begin{equation}\label{eq:aggre_samples_from_sparse_input} \bbz_{\ccalQ,i}=\bbPi_{\ccalQ}\bbTheta_i \bbs= \bbPi_{\ccalQ}\bbTheta_i \bbPi_{\ccalS}^T\bbs_\ccalS.
\end{equation} 
This setup would correspond to the case where the indexes of the seeding nodes (identity of the influencers or location of the sources) is known, but their particular values are not. Clearly, for the recovery problem in \eqref{eq:aggre_samples_from_sparse_input} being identifiable, a necessary condition is $Q$, the number of observations in $\bbz_{\ccalQ,i}$, being no less than $S$, the number of unknowns in $\bbs_\ccalS$. Consider first the extreme case $Q=S$. It is then clear that the sparse signal $\hbs^{(i)}$, with the superscript $(i)$ denoting the index of the sampling node, can be recovered as
\begin{equation}
\hbs^{(i)} = \bbPi_\ccalS^T \hbs_\ccalS^{(i)}\;\;\text{and}\;\;
\hbs_\ccalS^{(i)} = (\bbPi_{\ccalQ}\bbTheta_i \bbPi_{\ccalS}^T)^{-1} \bbz_{\ccalQ,i},
\end{equation} 
provided that the inverse exists, which will depend on the particular set of rows $\ccalQ$ and columns $\ccalS$. For the standard case of $Q>S$ and the observations being  corrupted by additive noise, the \emph{observed} signal $\bby_i$ is given by $\bby_i=\bbz_i+\bbw_i$. Assuming that the noise $\mathbf{w}_i$ is zero-mean, independent of $\bbx$, and with known covariance matrix $\mathbf{R}_{w}^{(i)}:=\EE[\mathbf{w}_i\mathbf{w}_i^H]$, the best linear unbiased estimator of the sparse inputs is \cite{kay1993fundamentals}
\begin{align}
    \!\!\!\!\hbs^{(i)} \!=\! \bbPi_\ccalS^T \hbs_\ccalS^{(i)}\;\;\text{and}\;\;
    \hbs_\ccalS^{(i)} \!=\! (\bbM_\ccalQ)^\dagger (\bbR_{w,\ccalQ}^{(i)})^{-1/2}\bby_{\ccalQ,i},\;\text{with}\label{eq:optimal_estimator}
    \end{align}
    \begin{align}
    \!\!\bbR_{w,\ccalQ}^{(i)}\!:=\!\bbPi_{\ccalQ}\bbR_w^{(i)}\bbPi_{\ccalQ}^T\;\,\text{and}\;\, \bbM_\ccalQ\!:=\!(\bbR_{w,\ccalQ}^{(i)})\!^{-1/2}\bbPi_\ccalQ \bbTheta_i \bbPi_{\ccalS}^T,\nonumber
\end{align} 
provided that $\bbR_{w,\ccalQ}^{(i)}$ is non-singular. Note that if the noise is Gaussian, the estimator in \eqref{eq:optimal_estimator} attains the Cram\'er-Rao bound. 

Using \eqref{eq:optimal_estimator}, the error covariance matrix is \cite{kay1993fundamentals}
\begin{eqnarray}
\bbR_e^{(i)}:=\EE[(\bbs_\ccalS-\hbs_\ccalS^{(i)})(\bbs_\ccalS-\hbs_\ccalS^{(i)})^T]=(\bbM_\ccalQ^T\bbM_\ccalQ )^{\dagger},\label{eq:cov_error_est_time_noise}
\end{eqnarray}
which depends on the noise model, the spectrum of the GSO, the seed nodes, the node taking the observations, and the sample-selection scheme adopted. \eqref{eq:cov_error_est_time_noise}  can then be used to assess the performance of the estimation. The particular error metric depends on the application at hand \cite{pukelsheim1993optimal}. The most commonly used is the \acrfull{mse}, which corresponds to minimizing $\mathrm{trace}(\mathbf{R}_e^{(i)})$, with other popular metrics including the spectral norm $\lambda_{\max}(\mathbf{R}_e^{(i)})$ and the log determinant $\log \det(\mathbf{R}_e^{(i)})$.

\subsubsection{\textbf{Active sampling}} Critical for the error performance is the design of a good sampling matrix. This requires solving first the optimal scheme for a fixed node $i$ (which is a relevant problem by itself) and then selecting the best node (provided that the system operating conditions yield such a possibility). Except for trivial cases, optimizing the sampling set is NP-hard. In particular, when the interest is in the MSE, the optimal sampling matrix $\bbPi_\ccalQ^{(i)*}$ corresponds to
\begin{flalign}
    &\min_{\bbPi_{\ccalQ}} \quad \mathrm{trace}\Big((\bbPi_{\ccalS}\bbTheta_i^T\bbPi_\ccalQ^T(\bbPi_{\ccalQ}\bbR_w^{(i)}\bbPi_{\ccalQ}^T)^{-1}\bbPi_\ccalQ \bbTheta_i \bbPi_{\ccalS}^T)^{\dagger}\Big) \nonumber \\
    &\,\,\mathrm{s. t.}  \;\;\Pi_{\ccalQ,ij}\in \{0,1\},\;\;{\textstyle \sum_j} \Pi_{\ccalQ,ij}=1,\;\; \|\bbPi_\ccalQ\|_0=Q,  \label{eq:optimal_sampling}
\end{flalign}

which is a fractional high-order polynomial minimization over binary variables. While convex relaxations to approximate \eqref{eq:optimal_sampling} are available, greedy schemes for related problems have been shown to work well \cite{chamon2018greedysampling}, and are advocated here.

\subsubsection{\textbf{Blind (sparse) recovery}}
In many relevant applications (e.g., those dealing with inverse problems) the seeds in $\ccalS$ are unknown. In that case, the noiseless recovery requires solving
\begin{align}\label{eq:sampled_shifted_signal_unkown}
    \bbs^*:=\arg\min_{\bbs} \;
    ||\bbs||_0 \;\;\mathrm{s. t.} \;\; \bbz_{\ccalQ,i}= \bbPi_\ccalQ\bbTheta_i \bbs, 
\end{align}
and then setting $\hbx^{(i)}=\hbs^{(i)}=\bbs^*$. Leveraging results from sparse recovery, it can be shown that identifiability needs $Q\geq 2S$ and the observation matrix $\bbTheta_i$ to be full spark \cite{donoho2003spark}. The problem in \eqref{eq:sampled_shifted_signal_unkown}, which does not account for noise,
is NP-hard due to the presence of the $\ell_0$ norm. A standard approach is to relax the equality with a LS cost and relax the $\ell_0$ norm with a (weighted) $\ell_1$ regularizer. This yields
$$\bbs_1^{(i)}:=\arg \min_{\bbs}  \|{(\bbR_{w,\ccalQ}^{(i)})}\!^{-1/2}\big(\bby_{\ccalQ,i} - \bbPi_\ccalQ \bbTheta_i \bbs\big) \|_2^2 + \gamma ||\bbs||_1,$$
where $(\mathbf{R}_{w,\ccalQ}^{(i)})^{-1/2}$ accounts for the colored noise and $\gamma$ is the regularization parameter.

\subsection{Aggregating the diffused sparse input}
We now analyze the recovery of $\bbx=\bbH\bbs$ from its aggregated samples, assuming first that the filter $\bbH$ is known. The main difference with respect to the previous case is that now the relation between $\bbs$ and $\bbz_{\ccalQ,i}$ is 
\begin{equation}\label{eq:shifted_signal_diffused}
\bbz_{\ccalQ,i}=\bbPi_{\ccalQ}\bbz_i\;\;\text{and}\;\;\bbz_i=\boldsymbol{\Psi}^T\diag(\boldsymbol{\upsilon}_i\circ\tbh)\bbU \bbs:=\boldsymbol{\Xi}_i\bbs,
\end{equation}
where $\bbXi_i:=\boldsymbol{\Psi}^T\diag(\boldsymbol{\upsilon}_i\circ\tbh)\bbU$ and $\circ$ denotes the entry-wise product. Compared with \eqref{eq:shifted_signal}, we notice that the observations depend not only on the frequency pattern of the sampling node $\bbupsilon_i$, but also on the frequency response of the diffusing filter $\tbh$. Intuitively, nodes with a frequency pattern more aligned with that of the diffusing filter (so that $|\mathrm{det}(\diag(\boldsymbol{\upsilon}_i\circ\tbh))|$ is large) are more likely to give rise to a better reconstruction in the presence of noise.

Apart from replacing the observation matrix $\bbTheta_i$ with $\bbXi_i$, another important difference stems from the particular error to minimize. Since in this case $\bbx$ and $\bbs$ are different, depending on the application the focus can be on estimating $\bbs$ and minimizing the MSE associated with $\bbR_{e,s}^{(i)}:=\EE[(\bbs_\ccalS-\hbs_\ccalS^{(i)})(\bbs_\ccalS-\hbs_\ccalS^{(i)})^T]$ or on estimating $\bbx$ and minimizing the MSE associated with $\bbR_{e,x}^{(i)}:=\EE[(\bbx-\hbx^{(i)})(\bbx-\hbx^{(i)})^T]$. This requires to premultiply (or not) the error terms in the objectives of the optimizations presented in the previous sections. Such a dichotomy was not an issue for BGS since the frequency coefficients are a byproduct and the ultimate  goal was to recover $\bbx$. By contrast, for DSGS both $\bbx$ (signal reconstruction) and $\bbs$ (source localization) are meaningful on their own.

\subsection{Blind deconvolution} 
There may be scenarios where the diffusing filter $\bbH=\sum_{r=0}^{R-1}h_r\bbS^r$ is unknown. The recovery problem in this case is considerably more challenging, but it can be tackled provided that $R$, the order of the diffusing filter, is sufficiently small. To be specific, after some manipulations, the expression $\bbx=\bbH \bbs$ with $\bbH$ being a graph filter can be written as $\bbx = \bbV (\bbPsi^T \odot \bbU^T )^T \text{vec}(\bbs \bbh^T)$, where $\odot$ stands for the Khatri-Rao product (cf. \cite{segarra2016blind}).
We can then relate the samples $\bbz_{\ccalQ,i}$ with the unknown $\bbs$ and $\bbh$ as 
\begin{equation}
\label{eq:bd_agss_dsgs}
\bbz_{\ccalQ,i} = \bbPi_{\ccalQ}\boldsymbol{\Psi}^T\diag(\boldsymbol{\upsilon}_i) (\bbPsi^T \odot \bbU^T )^T \text{vec}(\bbs \bbh^T),
\end{equation}
which is a system of $Q$ \textit{bilinea}r equations. If the support $\ccalS$ is known, the number of unknowns is $R+S$. If it is not, the number is $R+N$ and the constraint $\|\bbs \|_0\leq S$ must be added, further complicating the problem. The resultant problem can be handled by an alternating scheme that iterates between optimizing $\bbs$ given $\bbh$ and optimizing $\bbh$ for the new $\bbs$. More sophisticated approaches include lifting techniques that define the lifted variable $\bbSigma := \bbs\bbh^T$, the observation matrix $\bbPhi_i:=$ $\boldsymbol{\Psi}^T\diag(\boldsymbol{\upsilon}_i) (\bbPsi^T\! \odot \! \bbU^T )^T\!$, and find an approximation to
\begin{equation}\label{eq:optimal_blind_deconv}
\min_{\bbSigma} \|\bbz_{\ccalQ,i} - \bbPi_{\ccalQ}\bbPhi_i \text{vec}(\bbSigma)\|_2^2 + \gamma_1\text{rank}(\bbSigma) + \gamma_2 \|\bbSigma\|_{2,0},
\end{equation}
where $\gamma_1$ and $\gamma_2$ are regularization parameters and $\|\bbSigma\|_{2,0}$ is defined as the number of non-zero rows of $\bbSigma$. The estimates $\hbs^{(i)}$ and $\hbh^{(i)}$ are then found the main left and right singular vectors of $\bbSigma^*$, which are subject to an inherent scaling ambiguity. See \cite{segarra2016blind} for further justification and suitable relaxations.

In the case of \acrfull{ss}, the expression analogous to \eqref{eq:bd_agss_dsgs} is
\begin{equation}
    \label{eq:bd_ss_dsgs}
    \bbx_{\ccalQ} = \bbPi_{\ccalQ} \bbV (\bbPsi^T \odot \bbU^T )^T \text{vec}(\bbs \bbh^T),
\end{equation}
which is also a bilinear system similar to the previous one. Note that, as in the BGS case (cf. \cref{sec:models_signals}), the role of  $\bbV$ in SS is taken by $\boldsymbol{\Psi}^T\diag(\boldsymbol{\upsilon}_i)$ in AGSS. 

\subsection{Space-shift sampling of diffused sparse signals}\label{Ss:conventional sampling}
In many setups, access to more than one sampling node is available. This is useful to robustify the recovery and reduce the number of required samples per node, which is convenient because the conditioning number of the Vandermonde matrix $\bbPsi$ (one of the factors in $\bbTheta_i$) worsens as the samples per node increase. The resultant sampling scheme is referred to as space-shift sampling \cite{marques2015sampling}. To particularize it to the setup at hand, define the vectorized version of $\mathbf{Z}$ as $\bar{\bbz}$ and then the  $N^2 \times N$ matrix $\bar{\bbUpsilon}:=[\diag(\bbupsilon_1), \ldots, \diag(\boldsymbol{\upsilon}_N)]^T\diag(\tbh)$. With these conventions, $\bar{\bbz}$ can be written as
$\bar{\bbz}=(\mathbf{I}\otimes\bbPsi^T) \bar{\bbUpsilon} \bbU \bbs$, where $\otimes$ stands for the Kronecker product. The sampled version in this case is given as $\bar{\bbz}_{\ccalQ}=\bar{\bbPi}_\ccalQ \bar{\bbz}$, where $\bar{\bbPi}_\ccalQ$ is a selection matrix of size $Q \times N^2$. The results in the previous sections can be applied to this case as well provided that $\bbTheta_i$ and $\bbPi_\ccalQ$ are replaced with $\bar{\bbTheta}:=(\mathbf{I}\otimes\bbPsi^T) \bar{\bbUpsilon} \bbU$ and $\bar{\bbPi}_\ccalQ$.

%%%%%%%%%%%%%%%   FIGURE: RANDOM NETWORK   %%%%%%%%%%%%%%%%%%%%%
\begin{figure}[t]
	\centering
	\includegraphics[width=.5\textwidth]{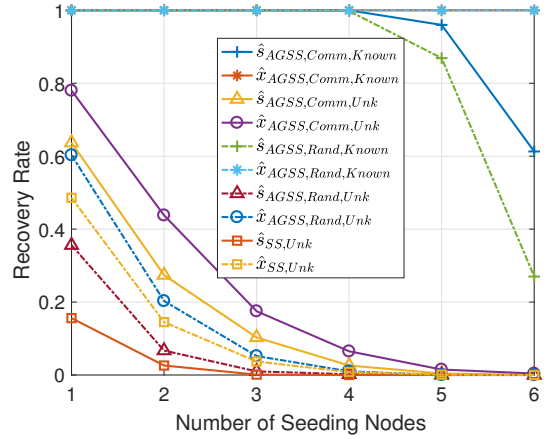}
	\caption{Recovery rate of DSGS in SBM graphs. Signals are recovered via the $\ell_1$-norm relaxation using the Laplacian as the GSO. 500 simulations with different graphs: $N=50,B=5,Q=8,\bbR_{w,\ccalQ}^{(i)}=10^{-5}\bbI$.} 
	\label{fig:node_influence}
\end{figure}
%%%%%%%%%%%%%%%%%%%%%%%%%%%%%%%%%%%%%%%%%%%%%%%%%%%%%%%%%%%%%%%%%%%

%%%%%%%%%%%%%%%   FIGURE: REAL DATA   %%%%%%%%%%%%%%%%%%%%%
\begin{figure}[t]
	\centering
	\includegraphics[width=.5\textwidth]{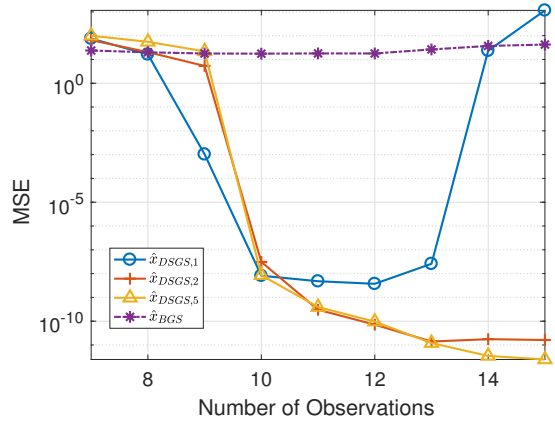}
	\caption{Median MSE of recovered signals defined over 95 real-world graphs using a blind diffused recovery scheme. }
	\label{fig:real_data}
\end{figure}
%%%%%%%%%%%%%%%%%%%%%%%%%%%%%%%%%%%%%%%%%%%%%%%%%%%%%%%%%%%%%%%%%%%

%%%%%%%%%%%%%%%   FIGURE: BANDWIDTH   %%%%%%%%%%%%%%%%%%%%%
\begin{figure}[t]
	\centering
	\includegraphics[width=.5\textwidth]{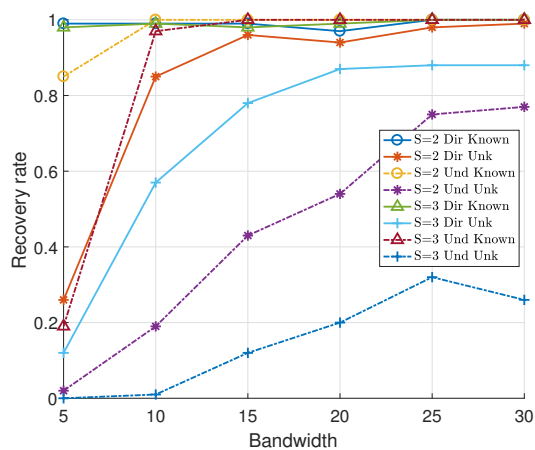}
	\caption{Recovery rate of DSGS in directed and undirected SBM graphs for varying filter bandwidth length. Signals are recovered via the pseudoinverse (known $\bbPi_\mathcal{S}$) and the $\ell_1$-norm relaxation (unknown $\bbPi_\mathcal{S}$) using the adjacency matrix as the GSO. Directed graphs with non-diagonalizable adjacency matrices are discarded. 100 graph-realizations for each type of graph, selecting the sampling node $i$ from all the $N$ nodes as the one leading to the smallest $\ell_2$-norm of $(\bbs-\hat{\bbs})$. The remaining parameters are: $N=30,B=2,p_{b}=0.25,p_{bb'}=0.05,Q=4,\bbR_{w,\ccalQ}^{(i)}=10^{-5}\bbI$.}
	\label{fig:bandwidth_directivity}
\end{figure}
%%%%%%%%%%%%%%%%%%%%%%%%%%%%%%%%%%%%%%%%%%%%%%%%%%%%%%%%%%%%%%%%%%%

%%%%%%%%%%%%%%%   FIGURE: DIRECTIVITY   %%%%%%%%%%%%%%%%%%%%%
\begin{figure}[t]
	\centering
	\includegraphics[width=.5\textwidth]{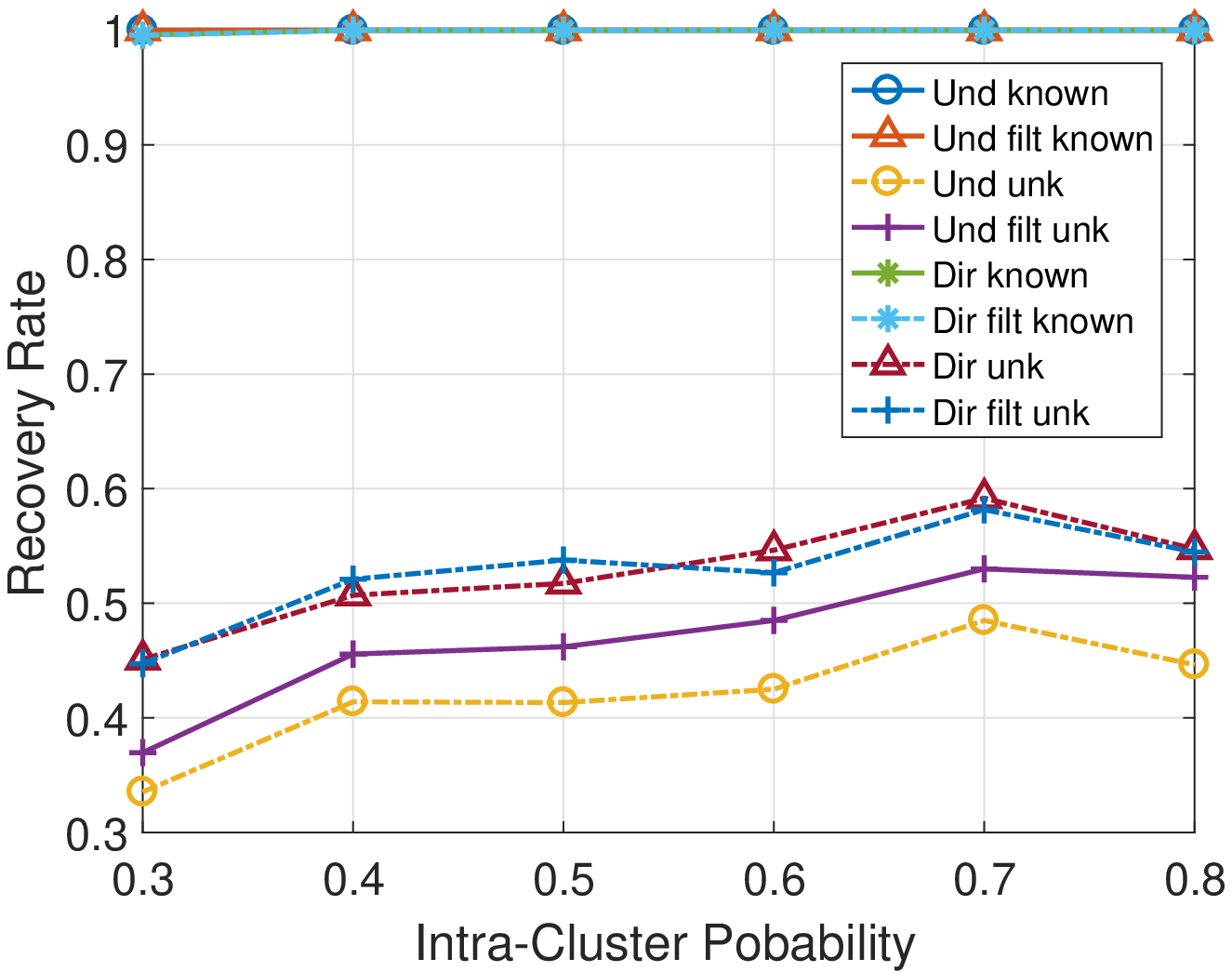}
	\caption{Recovery rate of sparse and DSGS in directed and undirected SBM graphs for different probabilities of inter-cluster and intra-cluster links. Signals are recovered via the pseudoinverse (known $\bbPi_\mathcal{S}$) and the $\ell_1$-norm relaxation (unknown $\bbPi_\mathcal{S}$) using the adjacency matrix as the GSO. The tested intra-cluster probability are $[0.3,0.4,0.5,0.6, 0.7,0.8]$, as shown in the horizontal axis. The corresponding  inter-cluster probabilities are $[0.1,0.15,0.2, 0.25,0.3]$. For each point in the figure, 100 graph-realizations are considered and, for each of those realizations, the $N=30$ nodes are tested, so that the rates shown correspond to averages across $3000$ trials. The remaining parameters are: $B=3,Q=9,\bbR_{w,\ccalQ}^{(i)}=10^{-6}\bbI$.}
	\label{fig:density_influence}
\end{figure}
%%%%%%%%%%%%%%%%%%%%%%%%%%%%%%%%%%%%%%%%%%%%%%%%%%%%%%%%%%%%%%%%%%%

%%%%%%%%%%%%%%%%%%%%%%%%%%%%%%%%%%%%%%%%%%
\section{Numerical experiments}\label{sec:NumExper}
%%%%%%%%%%%%%%%%%%%%%%%%%%%%%%%%%%%%%%%%%%

Short simulations to illustrate and gain intuition about some of the results presented are shown here. 

\vspace{3mm}\noindent \textbf{Test case 1.} First, we consider a  SBM graph with $N$ nodes and $B$ communities with $N_b\!=\!N/B$ nodes each \cite{decelle11}. Edges exist with probability $p_{bb}=0.4$ if the incident nodes are in the same community and with probability $p_{bb'}=0.1$ if they are not. The remaining parameters are given in the caption of \cref{fig:node_influence}. The index of the seeding nodes is chosen uniformly at random and the seed value is drawn from a \acrfull{zmuvg}. The filter taps have length $R=6$ and each of them is drawn from a ZMUVG. \cref{fig:node_influence} depicts the recovery rate, defined as the proportion of simulations for which the seeds are correctly identified and the $\ell_2$-norm of the error is less than $0.1$, as the number of seeds $S$ increases. All sampling nodes are considered, and the median error is reported. 
The $10$ scenarios (lines) in the figure consider if: 1) $\bbPi_\mathcal{S}$ is known or not (``Known''/``Unk''); 2) the sampling node $i$ is in the same community than one of the seeds or is a random node (``Comm''/``Rand''); 3) the sampled signal is either $\bbs$ or  $\bbx$ (``$\hat{s}$''/``$\hat{x}$''); and 4) the sampling scheme is AGSS or SS (``AGSS''/``SS'') . The results confirm that recovery is harder as $S$ increases, that blind schemes are not able to recover the signal if $S\!>\!Q/2\!=\!4$, and that knowledge of $\bbPi_\mathcal{S}$ notably facilitates the recovery. We also observe that if the other two criteria are fixed, AGSS always outperform SS, confirming that AGSSs are more robust and less sensitive to the sampling configuration \cite{marques2015sampling}. Similarly, ``Comm''  is always better than ``Rand''. This is not surprising since the (diffused) seed values reach the sampling node faster if the node belongs to the same community. This also explains why recovering $\bbx$ seems to be always easier than recovering (non-diffused) signal $\bbs$.
 
\vspace{.05cm}\noindent \textbf{Test case 2.} \cref{fig:real_data} tests our schemes in the D\&D protein structure database \cite{dobson2003distinguishing}, where nodes account for amino acids, links capture similarity, and signals are the expression level of the amino acids. We assume that the data can be accurately modeled as DSGS and try to recover the full signal following the blind diffused recovery scheme (label ``DSGS,1''), its space-shift counterpart with 2 and 5 nodes (``DSGS,2'', ``DSGS,5''), and AGSS modeling the data not as DSGS but as bandlimited (``BGS''). The median error of all graphs is reported and all sampling nodes are considered, selecting the 25th error percentile. The main observations are: 1) BGS yields the worst performance, pointing out that the DSGS model is a good fit for the information in the D\&D database;  and 2) for the ``DSGS,1'' the median MSE increases when the number of observations is high. This stems from the conditioning number of $\bbPsi$ as explained in \cref{Ss:conventional sampling}. In contrast,  the ``DSGS,2'' and ``DSGS,5'' schemes are more robust. 
 
\vspace{.05cm}\noindent \textbf{Test case 3.} We test our schemes in the ETEX dataset \cite{nodop1998field}, which contains $\{\bby_t\}_{t=0}^{29}$ graph signals whose nodes correspond to different locations and $t$ represents time. We use as GSO the adjacency of the geographical graph \cite{thanou2017learning}, the seed is set as $\bbs = \bby_0$, and the signal to be sampled and recovered is $\bbx=\bby_t$ for all $t>0$. Using $Q=16$ samples and the same approach than in the second test case, we run the experiment for $29$ different signals (one per $t,t>0$), obtaining MSE of \textbf{$3 \cdot 10^{-5}$} and \textbf{$1.5 \cdot10^{-5}$} for ``DSGS,1'' and ``DSGS,2'' respectively.  
 
\vspace{.05cm}\noindent \textbf{Test case 4.} In \cref{fig:bandwidth_directivity} we analyze the impact on the recovery of two factors: a) the bandwidth of the diffusing filter and b) the directivity of the supporting graph. To this end, let us consider a bandpass filter $\tilde{\bbh}$ whose non-zero band consists of $W$ elements randomly drawn from a ZMUVG. Moreover, we consider two types of SBM graphs: one where links are directed (denoted as ``Dir'' in the figure) and antoher one with undirected links (``Und''). For this test case, the adjacency matrix is chosen as GSO. The remaining parameters are detailed in the caption of \cref{fig:bandwidth_directivity}. The plotted results reveal that successful interpolation from DSGS samples is more amenable in directed than undirected graphs. The additional information about the edge direction contained in the GSO, central during both filtering and the AGSS, helps identifying the seeds in $\mathcal{S}$. Furthermore, lower $W$ hinders the diffusion of the seeds, making the recovery harder. Indeed, in the extreme case of $W=0$ the factor $\diag(\boldsymbol{\upsilon}_i\circ\tbh)$ in \eqref{eq:shifted_signal_diffused} renders the observations zero.

\vspace{.05cm}\noindent \textbf{Test case 5.} \cref{fig:density_influence} studies the impact of the density of the graph on the recoverability of the signals. In this experiment, both the intra-cluster and the inter-cluster probability vary in the same direction, as explained in the caption of the figure. The results show that the recovery rate tends to improve when the graphs are denser. With a higher link probability, the chances that any node is close to the seeds increases, so that they can access the information related with the non-zero elements of the sampled signal. As a result, the fraction of nodes able to recover the signal increases. In addition, the plot confirms that directed graphs (``Dir'') have a better recovery rate than undirected graphs (``Und''), which is consistent with the results presented for the test case 3.

\section{Conclusion}\label{sec:conclusion_sampling}
Here, we considered the presence of missing values in the (aggregated) observed data and approached the reconstruction of the original signal in the context of sampling and interpolation of graph signals.
Assuming that the observed values were gathered through an AGSS, first, we proposed a convex optimization problem to interpolate sparse signals with either known or unknown support of the seeding nodes.
Later on, we moved on to the more general case where the signals were DSGS and contemplated the signal interpolation when the diffusing filter $\bbH$ was known, and then, the blind sparse recovery case where $\bbH$ was unknown.
Finally, we studied the case where the aggregated observations were collected at more than one sampling node, and we evaluated the proposed interpolation algorithms over synthetic and real-world datasets. 

\chapter{Robust graph filter identification}\label{chap:robust_filter_id}
% Why this perturbations
In this chapter, we shift our attention from perturbations involving the observed graph signals to perturbations involving the edges of the observed graph, a prevalent type of uncertainty that is especially relevant when the networks are inferred from a set of nodal observations.
Nonetheless, we may also encounter imperfections in the topology when networks are physical entities due to errors in the observation process.
Regardless of their source, the presence of errors in the observed topology is critical for most GSP applications, so it is essential that they are properly accounted for in order to build a robust GSP framework.
In this sense, we note that the polynomial definition of GFs renders them particularly sensitive to the presence of these perturbations, since even errors affecting only a few edges can lead to large discrepancies when high-order polynomials are involved.
As a result, this chapter picks up our work in~\cite{rey2021robust,rey2022robust} and addresses the relevant problem of GF identification from input-output pairs from a robust perspective.
Even though we also consider the presence of noise in the observed graph signals, the proposed analysis is primarily concerned with perturbations involving the edges of the graph.
Next, we provide some highlights about the method presented in this chapter and a brief summary of the resulting contributions.

\section{Introduction}
On top of its theoretical interest, the task of GF identification is practically relevant to, e.g., understanding the dynamics of network diffusion processes \cite{segarra2017optimal,segarra2016blind,djuric2018cooperative}, as well as explaining the structure of real-world datasets~\cite{rey2019sampling,zhu2020estimating,he2022detecting}.
Motivated by this, the work described in this chapter investigates the problem of estimating a GF from input-output signal pairs assuming that both the signals and the supporting graph have errors.
The proposed approach is formulated in the vertex domain, avoiding the numerical instability of computing large polynomials and, at the same time, bypassing the challenges associated with robust \emph{spectral} graph theory.
To that end, we recast the robust estimation as a joint optimization problem where the GF identification objective is augmented with a graph-denoising regularizer, so that, on top of the desired GF, we also obtain an enhanced estimate of the supporting graph.
The joint formulation leads to a non-convex bi-convex optimization problem, for which a provably-convergent efficient (alternating minimization) algorithm able to find an approximate solution is developed. Furthermore, to address scenarios where multiple GFs are present (e.g., when dealing with vector \acrfull{ar} spatio-temporal processes or in setups where nodes collect multi-feature vector measurements), we generalize our framework so that multiple GFs, all defined over the same graph, are jointly identified. 

% Challenging task - brief review
Despite their theoretical and practical relevance, the number of robust GSP works is limited, due in part to the challenges emanating from the presence of graph perturbations.
Initial works modeling the influence of perturbation in the spectrum of the graph Laplacian~\cite{ceci2020graph}, and proposing a graphon-based perturbation model~\cite{miettinen2019modelling} were previously commented on in \cref{sec:perturbations_gsp}.
More recently, \cite{ceci2020_semtls} combines SEM with TLS to jointly infer the GF and the perturbations when the observed data is explained by a SEM.
A different robust alternative is presented in~\cite{natali2020topology}, where the support of the graph is assumed to be known and the goal is to estimate the weights of the network topology and the coefficients of the GF.
The resultant problem is non-convex and the authors adopt a \acrfull{scp} approach to solve it.
Finally, the presence of perturbations has also been considered in non-linear GSP tasks.
An alternative definition of GFs robust to perturbations is proposed in~\cite{tenorio2021robust}, and the transferability of GFs when employed in graph neural networks is studied in~\cite{levie2019transferability,levie2021transferability,ruiz2021graph}.

% Contributions and outline
\vspace{3mm}
\noindent\textbf{Contributions and outline.}
After analyzing the influence of edge perturbations in polynomial GFs and stating the robust GF identification problem in \cref{sec:perturbed_graph_filters}, our main contributions are:
\begin{enumerate}
    \item We formulate a non-convex optimization problem to jointly estimate the graph and the GF, develop an alternating optimization algorithm to solve it, and prove its convergence to a stationary point (\cref{sec:rfi}).
    \item We consider a generalization where several GFs are jointly estimated by exploiting the fact that they are polynomials of the same GSO (\cref{sec:rfi_joint}).
    \item We propose an efficient implementation of the GF identification algorithm to handle graphs with a large number of nodes (\cref{sec:efficient_rfi}). 
\end{enumerate}
The effectiveness of the proposed algorithms is evaluated numerically in \cref{sec:experiments_filter_id}, and some concluding remarks are provided in \cref{sec:conclusion_filter_id}. Last but not least, while we focus on GF identification from input-output pairs, the approach put forth in this chapter can be generalized to other GSP tasks, which is a research path we plan to pursue in the near future.

%%%%%%%%%%%%%%%%%%%%%%%%%%%%%%%%%%%%%%%%%%%%%%%%%%%%%%%%%%%%%%%%%%%%%%%%%%%%%%%%%%%%%%%%%%%%%%%%%%%%%%%%%%%%%%%%%
%SECTION
%%%%%%%%%%%%%%%%%%%%%%%%%%%%%%%%%%%%%%%%%%%%%%%%%%%%%%%%%%%%%%%%%%%%%%%%%%%%%%%%%%%%%%%%%%%%%%%%%%%%%%%%%%%%%%%%%
\section{GF identification with imperfect graph knowledge}\label{sec:perturbed_graph_filters}
% Perturbed GSO
This section introduces and discusses the problem of estimating a GF $\bbH=\sum_{r=0}^{N-1}h_r\bbS^r$ from noisy input-output signal pairs $(\bbX\in \reals^{N \times M}, \bbY \in \reals^{N \times M})$ assuming that we have access to an \emph{imperfect} GSO $\barbS \in \reals^{N \times N}$, which can be modeled as 
\begin{equation}\label{eq:additive_GSO_perturbation_model}
\barbS = \bbS + \bbDelta,
\end{equation}
where $\bbS \in \reals^{N \times N}$ represents the true GSO and $\bbDelta \in \reals^{N \times N}$ is a \emph{perturbation matrix}. 
Before discussing models for the perturbation matrix, we find illustrative to demonstrate the impact of $\bbDelta$ on the GSP problem at hand.

%%%%%%%%%%   FIGURE   %%%%%%%%%%
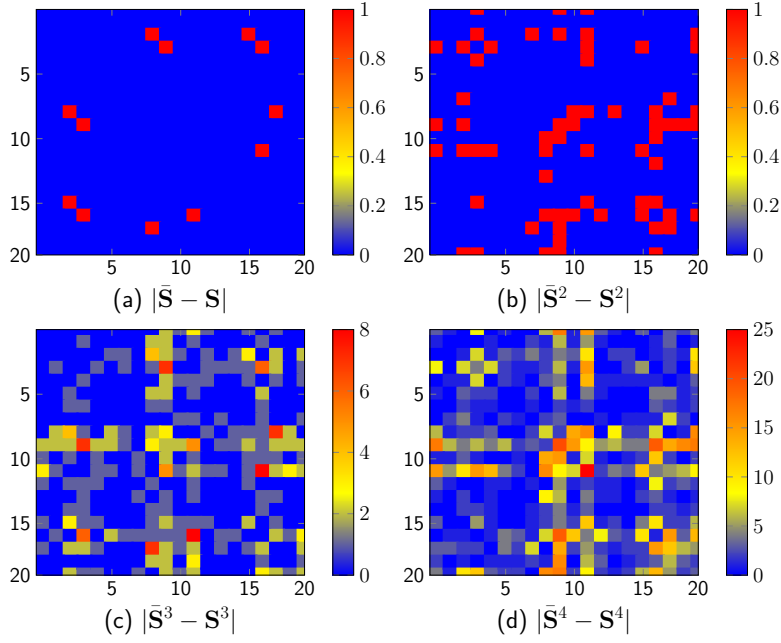
\begin{figure}[!t]
	\centering
	\begin{tikzpicture}[baseline,scale=.55]
\begin{groupplot}[
    table/col sep=space,
    width=8cm,
    height=7.5cm,
    group style={group size=2 by 2,
        horizontal sep=3cm,
        vertical sep=1.8cm,},
    enlargelimits=false,
    colorbar,
    % point meta max=15,
    % colorbar style={ytick={0,5,...,15}},
    xtick={5,10,15,19},
    xticklabels={5,10,15,20},
    ytick={5,10,15,19},
    yticklabels={5,10,15,20},
    x label style={font=\LARGE},
    tick label style={font=\Large}
    ]

    \pgfplotstableread{data/A-Apert_k1.csv}\matrixA
    \pgfplotstableread{data/A-Apert_k2.csv}\matrixB
    \pgfplotstableread{data/A-Apert_k3.csv}\matrixC
    \pgfplotstableread{data/A-Apert_k4.csv}\matrixD
    
\nextgroupplot[xlabel={(a) $|\barbS - \bbS|$},]
    \addplot [
        matrix plot,
        % nodes near coords=\coordindex, mark=*,
        point meta=explicit,
        mesh/cols=20,  % only if the matrix is not square
    ] table [meta=value] {\matrixA};

\nextgroupplot[xlabel={(b) $|\barbS^2 - \bbS^2|$}]
    \addplot [
        matrix plot,
        point meta=explicit,
        mesh/cols=20,  % only if the matrix is not square
    ] table [meta=value] {\matrixB};

\nextgroupplot[xlabel={(c) $|\barbS^3 - \bbS^3|$}]
    \addplot [
        matrix plot,
        point meta=explicit,
        mesh/cols=20,  % only if the matrix is not square
    ] table [meta=value] {\matrixC};
    
\nextgroupplot[xlabel={(d) $|\barbS^4 - \bbS^4|$},
    colorbar style={ytick={0,5,...,25}}]
    \addplot [
        matrix plot,
        point meta=explicit,
        mesh/cols=20,  % only if the matrix is not square
    ] table [meta=value] {\matrixD};

\end{groupplot}
\end{tikzpicture}
	\caption{Absolute error for different powers of the matrix $\bbS$ and its perturbed version $\barbS$. The true GSO is the adjacency matrix of an Erd\H{o}s-Rényi graph with link probability of 0.15, and $\barbS$ is perturbed by creating and destroying links independently with a probability of $0.05$.}\label{fig:pert_err_example}
\end{figure}
%%%%%%%%%%   END FIGURE   %%%%%%%%%%

% Impact of perturbations in GF
As pointed out in the introduction, the presence of uncertainties in the topology of $\ccalG$ is particularly relevant when dealing with GFs.
Indeed, due to the polynomial definition of $\bbH$, %[see \eqref{eq:graph_filter}]
even small perturbations can lead to significant errors when $\barbS$ (and not $\bbS$) is used as the true GSO.
To see this more clearly, \cref{fig:pert_err_example} provides an example that illustrates how the errors encoded in $\bbDelta$ propagate for different matrix powers, demonstrating that the discrepancies between $\barbS^r$ and $\bbS^r$ increase swiftly as the power $r$ grows.
More rigorously, let $C$ be a positive constant such that $\|\bbS\|\leq C$ and $\|\barbS\| \leq C$, and define %$\bbH:=\sum_{r=0}^{N-1}h_r\bbS^{r}$ and
$\barbH:=\sum_{r=0}^{N-1}h_r\barbS^{r}$. 
%Then, it follows that the error generated by the perturbations is upper-bounded by
Then, the error generated by the perturbations is upper-bounded by
\begin{equation}\label{eq:err_H_pert}
    \|\barbH - \bbH\| \leq \sum_{r=1}^{N-1}|h_r|\|\barbS^r - \bbS^r\| \leq \sum_{r=1}^{N-1}|h_r|rC^{r-1}\|\bbDelta\|,
\end{equation}
where the last inequality follows from~\cite[Lemma 3]{levie2019transferability}.
In words, the maximum difference between the true $\bbH$ and the perturbed $\barbH$ increases \emph{exponentially} with the degree of the GF.

% Need for a robust approach
From the previous discussion, it is not surprising that the imperfect knowledge of the graph topology is also relevant when estimating the filter coefficients.
In fact, ignoring the errors in $\bbDelta$ and attempting to estimate $\bbh$ solving \eqref{eq:solving_fi} when $\barbS$ is used in lieu of the true (unknown) $\bbS$ leads to a poor solution, as we illustrate numerically in \cref{sec:experiments_filter_id}.
Motivated by this, we approach the GF identification problem from a robust perspective by taking into account the imperfect knowledge of the GSO.
The resultant robust estimation task is formally stated next.

\begin{problem}\label{P:problem_statement}
    Let $\ccalG$ be a graph with $N$ nodes, let $\bbS \in \reals^{N \times N}$ be the true (unknown) GSO associated with $\ccalG$, and let $\barbS \in \reals^{N \times N}$ be the perturbed (observed) GSO. Moreover, let $\bbX\in \reals^{N \times M}$ and $\bbY \in \reals^{N \times M}$ be a pair of matrices collecting $M$ observed input and output signals defined over $\ccalG$ and related by the model in \eqref{eq:rfi_observ_model}.
    Our goal is to use the triplet $(\bbX,\bbY,\barbS)$ to: i) learn the GF $\bbH$ that best fits the model in \eqref{eq:rfi_observ_model} and ii) recover an enhanced estimation of $\bbS$.
    To that end, we make the following assumptions: \\
    (\textbf{AS1}) $\bbH$ is a polynomial of $\bbS$ [cf. \eqref{eq:graph_filter}]. \\
    (\textbf{AS2}) $\bbS$ and $\barbS$ are close according to some metric $d(\bbS,\barbS)$, i.e., the observed perturbations are ``small'' in some sense.
\end{problem}

% Comment problem and assumptions
On top of the previous two assumptions, we also consider that the norm of the noise observation matrix $\bbW$ in \eqref{eq:rfi_observ_model} is small, which is a workhorse assumption in this type of problems. Similar to standard GF identification approaches, (\textbf{AS1}) limits the degrees of freedom of the linear operator in \eqref{eq:rfi_observ_model}. However, the fact of the true $\bbS$ being unknown adds  uncertainty to the problem and, as a result, additional signal observations are required to achieve an identification performance comparable to the one obtained when $\barbS=\bbS$. Regarding the recovery of the true GSO, (\textbf{AS2}) accounts for the hypothesis that $\barbS$ is a perturbed observation of $\bbS$ and, hence, matrices $\bbS$ and $\barbS$ are not extremely different. Note that this guarantees that ``some'' information about the true GSO is available, so that (\textbf{AS1}) can be effectively leveraged. While not exploited in our formulation, additional assumptions constraining the GSO could also be incorporated into the problem. Finally, the metric $d(\cdot, \cdot)$ employed to quantify the similarity between $\bbS$ and $\barbS$ should depend on the model for the perturbation $\bbDelta$, a subject that is briefly discussed next.  %, as is detailed next.

\subsection{Modeling graph perturbations}\label{sec:graph_pert}
The development and analysis of graph perturbation models that combine practical relevance and analytical tractability constitutes an interesting yet challenging open line of research~\cite{miettinen2018graph,miettinen2019modelling}.
Due to its flexibility and tractability, here we consider an additive perturbation model [cf. \eqref{eq:additive_GSO_perturbation_model}], so that the focus is constrained to understanding the structural (statistical) properties of matrix  $\bbDelta = \barbS - \bbS$.

% Creation/Destruction of links
Consider first the case where perturbations only \emph{create or destroy links} independently.
If $\ccalG$ is an \emph{unweighted graph}, a simple approach is to consider perturbations modeled as independent Bernoulli variables with possibly different creation/destruction probabilities.
In this case, the entries of $\bbDelta$ would be
\begin{equation}
    \Delta_{ij} = \left\{\hspace{-2mm} \begin{array}{rl}
        1 &  \mathrm{if} \; \mathrm{link} \; (i,j) \; \mathrm{is} \; \mathrm{created},\\
        -1 & \mathrm{if} \; \mathrm{link} \; (i,j) \; \mathrm{is} \; \mathrm{destroyed},\\
        0 & \mathrm{otherwise}.
    \end{array} \right.
\end{equation}
Since $\bbDelta$ models the creation and destruction of links, it is worth noting that $\Delta_{ij} = 1$ only if $S_{ij} = 0$ and $\Delta_{ij} = -1$ only if $S_{ij} = 1$.
In the more general case of $\ccalG$ being a \emph{weighted graph}, $\Delta_{ij} = -S_{ij}$ destroys an existing link while $\Delta_{ij} = z$ creates a new link.
Here, $z$ is a random variable sampled from a particular distribution (typically mimicking the weight distribution of the true $\bbS$).
When facing this type of perturbations, a suitable distance function is the $\ell_0$ norm
\begin{equation}
    d(\bbS,\barbS) = \| \bbS - \barbS \|_0,
\end{equation}
with the $\ell_1$ norm $\| \bbS - \barbS \|_1$ being a prudent convex relaxation.

Alternatively, rather than creating or destroying links, perturbations may represent uncertainty over the edge weights.
This entails the support of matrix $\bbDelta$ matching that of $\bbS$ and $\barbS$, and the non-zero entries of $\bbDelta$ being sampled from a distribution that models the observation noise. 
For example, if the noise is zero-mean, Gaussian and white, it holds that $\Delta_{ij} \sim \ccalN(0,\sigma^2)$ when $\bbS_{ij} \neq 0$ and $\Delta_{ij}=0$ when $\bbS_{ij} = 0$.
Under this setting, an appropriate distance metric is given by
\begin{equation}
    d(\bbS, \barbS) = \| \bbS_\ccalE - \barbS_\ccalE\|_2^2,
\end{equation}
where $\bbS_\ccalE$ and $\barbS_\ccalE$ only select the non-zero entries (edges) in $\bbS$ and $\barbS$.
Additionally, one can have setups where the two types of perturbations are simultaneously present.
That is, perturbations may create and destroy links while the actual value of the existing links is also uncertain.
In such a case, a combination of $\ell_1$  and $\ell_2$ norms like in elastic nets~\cite{zou2005regularization} is adequate.  

% Comment on non independent errors
The models previously described only consider the perturbation of edges in an independent fashion.
However, there may be scenarios where the perturbations are correlated.
Consider for example a communication network.
If the power supply of a node stalls, the signal-to-noise ratio of all its links will be poor, and hence, links involving that node will be more likely to fail.
Perturbations dependent across links can be modeled by means of a multivariate correlated Bernoulli distribution, an Ising model, or more sophisticated random graph models~\cite{dai2013multivariate}.
When prior information about the dependence of the perturbations is available, it can be incorporated into the function $d(\bbS,\barbS)$ to better extract the information encoded in $\barbS$.

%%%%%%%%%%%%%%%%%%%%%%%%%%%%%%%%%%%%%%%%%%%%%%%%%%%%%%%%%%%%%%%%%%%%%%%%%%%%%%%%%%%%%%%%%%%%%%%%%%%%%%%%%%%%%%%%%
%SECTION
%%%%%%%%%%%%%%%%%%%%%%%%%%%%%%%%%%%%%%%%%%%%%%%%%%%%%%%%%%%%%%%%%%%%%%%%%%%%%%%%%%%%%%%%%%%%%%%%%%%%%%%%%%%%%%%%%
\section{Robust GF identification}\label{sec:rfi}
% Sec. intro - non-convex problem
This section presents the optimization problem and the proposed algorithm to estimate $\bbH$ and $\bbS$ under the setting described in \cref{P:problem_statement}.
Given the matrices $\bbX$, $\bbY$, and $\barbS$, we approach the robust GF identification task by means of the following non-convex optimization
\begin{alignat}{2}\label{eq:rfi_nonconvex}
    \!&\! \hbH, \hbS = \argmin_{\bbH, \bbS} && \;\; \|\bbY-\bbH\bbX\|_F^2 + \lambda d(\bbS, \barbS) + \beta \|\bbS\|_0 \nonumber \\
    \!&\! \hspace{1.2cm}\mathrm{\;\;s. \;to: } && \;\; \bbS \in \ccalS, \;\; \bbS\bbH = \bbH\bbS,
\end{alignat}
where $\mathrm{s.\;to}$ stands for $\mathrm{subject\;to}$. The first term in the objective promotes the linear input-output relation in \eqref{eq:rfi_observ_model}, encouraging the norm of $\bbW=\bbY - \bbH\bbX$ to be small. The use of the Frobenius norm is well-justified when the observation noise is Gaussian and white, but other types of noise could be accommodated by using a different norm. The second term incorporates the assumption (\textbf{AS2}) as a regularizer to obtain an estimate $\hbS$ that is related to the given GSO $\barbS$. The $\ell_0$ norm in the third term accounts for the fact of $\bbS$ being sparse. Clearly, if additional information about $\bbS$ is available, it can be incorporated into \eqref{eq:rfi_nonconvex}, either as a regularizer (e.g., a statistical prior quantifying the log-likelihood of a class of GSOs) or as a constraint that \emph{must} be satisfied (e.g., the GSO being symmetric). The latter is indeed the role of $\bbS\in\ccalS$ in \eqref{eq:rfi_observ_model}, with $\ccalS$ representing a (desired) family of GSOs such as the set of adjacency matrices with no self-loops ($\ccalS$ is the set of matrices with non-negative entries whose diagonal entries are zero) or the set of combinatorial graph Laplacians (matrices with non-positive off-diagonal entries and zero row-sum). Finally, the (key) constraint $\bbS\bbH = \bbH\bbS$ captures the fact of $\bbH$ being a polynomial of $\bbS$ and not of $\barbS$ (\textbf{AS1}). Note first that the constraint is pertinent, if $\bbH$ is a polynomial of $\bbS$, then $\bbH$ and $\bbS$ have the same eigenvectors and, as a result, their product commutes \cite{segarra2017optimal}. More importantly for the GF-identification at  hand, when the GSO is perfectly known the model $\bbH=h_0\bbI+h_1\bbS+...+h_{N-1}\bbS^{N-1}$ is linear in the unknown $\bbh$. As a result, a formulation that estimates $\bbh$ directly (as carried out in classical non-robust approaches) is well-motivated. However, when both $\bbh$ and $\bbS$ are unknown, the model $\bbH=h_0\bbI+h_1\bbS+...+h_{N-1}\bbS^{N-1}$ is highly non-linear in $\bbS$, challenging the development of a tractable solution that jointly estimates $\bbh$ and $\bbS$. Our formulation bypasses this problem by recasting the optimization variables as $\bbH$ and $\bbS$, leading to the (more tractable) bilinear constraint in \eqref{eq:rfi_nonconvex}. Nonetheless, if estimating $\bbh$ is the ultimate goal, this can be readily achieved from $\hbH$ and $\hbS$ as
\begin{equation}
    \hbh = \Big(\vvec(\bbI), \vvec(\hbS),..., \vvec(\hbS^{N-1})\Big)^\dagger \vvec(\hbH).
\end{equation}

% Formulation advantages and selected distance
The approach put forth in \eqref{eq:rfi_nonconvex} has two main advantages. 
First, while most works formulate the recovery of the GF in the spectral domain, our formulation operates in the vertex domain.
Working on the spectral domain would imply finding the Vandermonde GFT matrix $\bbPsi$. Since this matrix involves high-order polynomials of the eigenvalues of the GSO, it is also prone to numerical instability and error accumulation~\cite{djuric2018cooperative}.
Even if approaches that bypass this issue by estimating the graph-frequency response $\tbh = \bbPsi\bbh$ in lieu of $\bbh$ are adopted, the estimation would still be challenging since they require computing the eigenvectors $\bbV$, which are known to be highly sensitive to errors in the GSO (especially those associated with small eigenvalues)\cite{segarra2015stability,ceci2020graph}. On top of this, characterizing the spectral errors and incorporating those to the optimization is not a trivial task. 
%In contrast, we bypass the need to compute the spectral decomposition by postulation a formulation that works entirely in the vertex domain. 
The second advantage emanates from casting the true GSO $\bbS$ as an explicit optimization variable. As already explained, this approach is robust to error accumulation and facilitates the incorporation of the (additive) effect of the perturbations into the optimization. An additional benefit is that we obtain a denoised version (enhanced estimation) of the true GSO, which can be practically relevant in most real-world applications. 

In a nutshell, in the context of robust GF identification, choosing a formulation that: i)~works entirely in the vertex domain, ii)~considers $\bbS$ as an explicit optimization variable, and iii)~codifies the GF structure via the constraint $\bbH\bbS=\bbS\bbH$, exhibits multiple advantages. However, it must be noted that the number of optimization variables is larger than in classical approaches (adding computational complexity) and that the bilinear filtering constraint $\bbH\bbS=\bbS\bbH$, while more tractable than its polynomial counterpart, is still non-convex. Alternatives to deal with these issues are discussed in later sections.

\subsection{Alternating minimization for robust GF identification}
This section presents a systematic efficient approach to find an approximate solution to \eqref{eq:rfi_nonconvex}. Since the goal is to design specific algorithms, from this section onwards, we particularize the GSO distance to $d(\bbS,\barbS)=\|\bbS-\barbS\|_0$, so that, according to the discussion in \cref{sec:graph_pert}, the focus is on graph perturbations that create and destroy links. Apart from its practical relevance, the reason for choosing the $\ell_0$ norm as a distance is also motivated by its more intricate (challenging) structure. Indeed, the algorithms presented next can be easily adapted to (more tractable) distances associated with alternative perturbation models. 
Having clarified this, the main obstacle to solving \eqref{eq:rfi_nonconvex} is its lack of convexity, which emanates from two different \emph{sources}: \emph{(s1)} the $\ell_0$ norms in the objective, and \emph{(s2)} the bilinear constraint involving $\bbS$ and $\bbH$.
Next, we explain the strategy adopted to deal with them and find a solution to \eqref{eq:rfi_nonconvex} by solving a succession of convex problems.
% Dealing with non-convexity
\begin{itemize}%[leftmargin=2.5mm]
\item Regarding the $\ell_0$ norm in \emph{(s1)}, a workhorse approach is to replace it with its convex surrogate, the $\ell_1$ norm. However, it is possible to exploit more sophisticated (non-convex) alternatives that typically lead to sparser solutions.
The one chosen in this work is to approximate the $\ell_0$ norm of a generic matrix $\bbZ \in \reals^{I \times J}$ using the logarithmic penalty
\begin{equation}\label{eq:log_sparsity}
     \| \bbZ \|_0 \approx r_\delta(\bbZ) := \sum_{i=1}^I\sum_{j=1}^J \log(|Z_{ij}| + \delta),
\end{equation}
where $\delta$ is a small positive constant~\cite{candes2008enhancing}.
The non-convexity of the logarithm can be handled efficiently by relying on a \acrfull{mm} approach~\cite{sun2016majorization}, which considers an iterative linear approximation leading to an iterative re-weighted $\ell_1$ norm.
It is worth noting that, since we will consider an iterative algorithm to deal with the bilinearity of \eqref{eq:rfi_nonconvex}, the iterative nature of the re-weighted $\ell_1$ norm will not impose a significant computational burden.
Details on the exact form of this sparse regularizer will be provided soon, when describing the estimation of $\bbS$. 
\item To deal with the bilinear terms in \emph{(s2)}, we adopt an alternating optimization approach \cite{gorski2007biconvex} resulting in an iterative algorithm where the optimization variables $\bbH$ and $\bbS$ are updated in \emph{two separate iterative steps}.
At each step, we optimize over one of the optimization variables with the other remaining fixed, resulting in two simpler problems that can be solved efficiently. The details about the specific steps will be provided shortly.
\end{itemize}

% New objective function
Taking into account these considerations, the first task to implement our approach is to rewrite the problem in \eqref{eq:rfi_nonconvex} as
% \begin{alignat}{2}\label{eq:rfi_nonconvex_rew}
%     \!&\!\min_{\bbS\in \ccalS, \bbH} && \;\;\|\bbY-\bbH\bbX\|_F^2 + \lambda r_{\delta_1}(\bbS - \barbS) + \beta r_{\delta_2}(\bbS)   \nonumber     \\
%     \!&\! && \;\;+ \gamma \|\bbS\bbH - \bbH\bbS\|_F^2, %\nonumber \\
%     %\!&\! \mathrm{\;\;s. \;to: } && \;\;\bbS \in \ccalS,
% \end{alignat}
\begin{equation}\label{eq:rfi_nonconvex_rew}
    \min_{\bbS\in \ccalS, \bbH} \|\bbY \!-\!\bbH\bbX\|_F^2 \!+\! \lambda r_{\delta_1}\!(\bbS \!-\! \barbS) \!+\! \beta r_{\delta_2}\!(\bbS) \!+\! \gamma \|\bbS\bbH \!-\! \bbH\bbS\|_F^2,
\end{equation}
where we recall that $r_{\delta}(\cdot)$ was introduced in \eqref{eq:log_sparsity}. Note that: i) the logarithmic penalty has also been used to promote sparsity in the term $\bbS - \barbS$ since we selected the $\ell_0$ norm as the distance between $\bbS$ and $\barbS$, and ii) the constraint $\bbS\bbH = \bbH\bbS$ was relaxed and rewritten as a regularizer, a formulation more amenable to an alternating optimization approach.

% Iterative algorithm
The next task is to solve \eqref{eq:rfi_nonconvex_rew} by means of an iterative algorithm that blends techniques from alternating optimization and MM approaches.
Specifically, for a maximum of $t_{max}$ iterations, we run the following two steps at each iteration $t = 0,...,t_{max}-1$. 

\vspace{2mm}
\noindent \textbf{Step 1: GF Identification.}
We estimate the block of $N^2$ variables collected in $\bbH$ while the current estimate of the GSO, denoted as $\bbS^{(t)}$, remains fixed.
This results in the convex optimization problem
\begin{alignat}{2}\label{eq:step1_filterid}
\!\!&\bbH^{(t+1)} =  \arg \min_{\bbH} && \|\bbY\!-\!\bbH\bbX\|_F^2 \!+\! \gamma \|\bbS^{(t)}\bbH \!-\! \bbH\bbS^{(t)}\|_F^2,
\end{alignat}
an LS minimization whose closed-form solution is 
\begin{equation}\label{eq:step1_closed_form}
    \vvec(\bbH^{(t+1)}) \!= \big(\bbX\bbX^\top \!\!\otimes\! \bbI \!+\! \gamma (\bbS\bbS^\top \!\!\oplus\! \bbS^\top\bbS \!-\! \bbS^\top \!\!\otimes\! \bbS^\top \!\!-\! \bbS \!\otimes\!  \bbS)\big)^{-1} (\bbX \!\otimes \bbI)\vvec(\bbY).
\end{equation} 
Here, $\otimes$ is the Kronecker product, $\oplus$ is the Kronecker sum, and $\bbI$ is the identity matrix of size $N \times N$.
Also note that \eqref{eq:step1_closed_form} omitted the iteration superscript in $\bbS^{(t)}$ to alleviate notation. 

\vspace{2mm}
\noindent \textbf{Step 2: Graph Denoising.}
Following an MM scheme, we optimize an upper bound of \eqref{eq:rfi_nonconvex_rew} where the logarithmic penalties are linearized.
Then, we estimate the block of $N^2$ variables collected in $\bbS$ while the current estimate of the GF $\bbH^{(t+1)}$ remains fixed.
This yields
\begin{equation}\label{eq:step2_graph_denoising}
    \bbS^{(t+1)} =  \arg\min_{\bbS \in \ccalS} \sum_{i=1}^N\sum_{j=1}^N \big(\lambda \bar{\Omega}_{ij}^{(t)}|S_{ij}-\bar{S}_{ij}|+\beta \Omega_{ij}^{(t)}|S_{ij}|\big) + \gamma \|\bbS\bbH^{(t+1)} - \bbH^{(t+1)}\bbS\|_F^2,
\end{equation}
where $\bar{\bbOmega}^{(t)}$ and $\bbOmega^{(t)}$ are computed in an entry-wise fashion based on the GSO estimate from the previous iteration as
\begin{align}\label{eq:log_weights}
    &\bar{\Omega}_{ij}^{(t)} = \frac{1}{|S_{ij}^{(t)} - \bar{S}_{ij}| + \delta_1},
    &\Omega_{ij}^{(t)} = \frac{1}{|S_{ij}^{(t)}| + \delta_2}.
\end{align}

%%%%%%%%%%  ALGORITHM  %%%%%%%%%%
\begin{algorithm}[tb]
\SetKwInput{Input}{Input}
\SetKwInOut{Output}{Output}
\Input{$\bbX$, $\bbY$, $\barbS$}
\Output{$\hbH$, $\hbS$.}
\SetAlgoLined
Initialize $\bbS^{(0)}$ as $\bbS^{(0)} = \barbS$. \\
\For{$t=0$ \KwTo $t_{max}-1$}{
    Compute $\bbH^{(t+1)}$ by solving \eqref{eq:step1_closed_form} fixing $\bbS^{(t)}$. \\
    Update $\bbOmega^{(t)}$ and $\bar{\bbOmega}^{(t)}$ as in \eqref{eq:log_weights}. \\
    Compute $\bbS^{(t+1)}$ by solving \eqref{eq:step2_graph_denoising} using $\bbH^{(t+1)}$, $\bbOmega^{(t)}$, and $\bar{\bbOmega}^{(t)}$.
    \\
}
$\hbH = \bbH^{(t_{max})},\; \hbS = \bbS^{(t_{max})}$.
\caption{Robust GF identification with graph denoising.}
\label{A:rfi_alg}
\end{algorithm}
%%%%%%%%%%  END ALGORITHM  %%%%%%%%%%

% Comments on the algorithm - convergence
The overall alternating algorithm is summarized in Algorithm~\ref{A:rfi_alg}, where a fixed number of iterations is considered.
The algorithm starts by initializing the GSO as  $\bbS^{(0)} = \barbS$ (although other options could also be appropriate), and then, it iterates between Steps 1 and 2 for a fixed number of epochs (or until some stopping criterion is met).
In this regard, a key feature of the algorithm is that it is guaranteed to converge to a stationary point, as is formally stated next.

\begin{theorem}\label{thm1}
    Denote as $f(\bbH,\bbS)$ the objective function in \eqref{eq:rfi_nonconvex_rew}, and let $\ccalZ^*$ be the set of stationary points of $f$.
    Let $\bbz^{(t)}=[\vvec(\bbH^{\!(t)})^{\!\top}\!, \vvec(\bbS^{\!(t)})^{\!\top}\!]^{\!\top}$ represent the solution provided by the iterative algorithm \eqref{eq:step1_closed_form}-\eqref{eq:step2_graph_denoising} after $t$ iterations. Assuming that i)~the GSO does not have repeated eigenvalues and ii)~every row of $\tbX\!=\!\bbV^{\!-1}\bbX$ has at least one nonzero entry, then $\bbz^{\!(t)}$ converges to a stationary point of $f$ as $t$ goes to infinity, i.e.,
        \begin{equation}
            \lim_{t\to\infty} \mathsf{d}(\bbz^{(t)}~|\ccalZ^*) = 0, \nonumber
        \end{equation}
    with $\mathsf{d}(\bbz~|\ccalZ^*) := \min_{\bbz^* \in \ccalZ^*} \|\bbz-\bbz^*\|_2$.
\end{theorem}

The proof relies on the convergence results shown in~\cite[Th. 1b]{hong2015unified} and the details are provided in Appendix~\ref{A:proof_thm1}.
Note that the convergence of the algorithm was not self-evident since the original optimization problem in \eqref{eq:rfi_nonconvex_rew} is non-convex and Step 2 is minimizing an upper-bound of the original objective function. The sufficient conditions in i) and ii) guarantee that every graph frequency is excited so that the GF is identifiable and \eqref{eq:step1_filterid} has a unique solution, which is a requirement for convergence (see \cref{thm3} in Appendix~\ref{A:proof_thm1} for details). Clearly, condition ii) is fulfilled even for $M=1$ if all the entries of the vector $\tbx=\bbV^{-1}\bbx$ are nonzero. Alternatively, when $M>1$ and ii) is satisfied, condition i) can be relaxed.

% Sensitivity of gamma
Another relevant element in the proposed algorithm is the weight $\gamma$.
If $\gamma$ is set to a value that is too large, the GF estimated in the first iteration $\bbH^{(1)}$ will be an (almost exact) polynomial of $\bar{\bbS}$ so that the algorithm will converge quickly to the same solution as that of the non-robust design [cf. \eqref{eq:rfi_nonconvex} with $\bbS=\bar{\bbS}$].
On the other hand, if $\gamma$ is too close to zero the two problems decouple and the solution converges quickly to that of the two separated problems [cf. \eqref{eq:step1_filterid} and \eqref{eq:step2_graph_denoising} with $\gamma=0$].
As a result, the value of the parameter must be chosen carefully. In this context, schemes that start with a small $\gamma$ to encourage the exploration during the warm-up phase, and then increase $\gamma$ as the iteration index grows to guarantee that the final $\hbH$ is a polynomial of $\hbS$ are a suitable alternative for the setup at hand.

Finally, one drawback of the proposed robust GF identification algorithm is that the optimization problems in \eqref{eq:step1_filterid} and \eqref{eq:step2_graph_denoising} may be slow when dealing with large graphs.
However, we will mitigate this issue by introducing an efficient implementation that reduces the computational complexity of the overall algorithm (see \cref{sec:efficient_rfi}).

\subsection{Leveraging stationary observations}\label{sec:stationary_observations}
The alternating convex approximation in Algorithm~\ref{A:rfi_alg} exploits the fact that $\bbX$ and $\bbY$ are linearly related via $\bbH$, which is a polynomial of $\bbS$.
However, in setups where the perturbations in $\barbS$ are very large, obtaining accurate estimates of $\bbS$ and $\bbh$ from $\hbH$ may still be challenging. One alternative to overcome this issue is to leverage the additional structure potentially present in our data.
Indeed, as detailed in the introduction, it is common to consider setups where the signals exhibit additional properties depending on the supporting graph, with notable examples including graph-bandlimited signals~\cite{shuman2013emerging,sandryhaila2014discrete}, diffused sparse graph signals~\cite{segarra2016blind,rey2019sampling}, or graph stationary signals~\cite{marques2017stationary,shafipour2020online,buciulea2022learning}. Clearly, incorporating such additional information into the optimization problem would enhance its estimation performance.

This section explores this path, restricting our attention to the case where the observed signals are stationary on $\ccalG$. The motivation for this decision is that, due to the tight connection between graph-stationary signals and GFs (see \cref{chap:preliminaries}), the formulation in \eqref{eq:rfi_nonconvex_rew} and Algorithm~\ref{A:rfi_alg} require relatively minor modifications to incorporate the assumption of $\bbX$ and $\bbY$ being stationary on $\bbS$, leaving the incorporation of additional signal models as future work. To formulate the updated problem, recall that the covariance matrix of a stationary graph signal can be expressed as a polynomial of the GSO. Therefore, incorporating stationarity calls for modifying \eqref{eq:rfi_nonconvex_rew} as
\begin{alignat}{2}\label{eq:rfi_nonconvex_st}
    \!&\!\min_{\bbS \in \ccalS, \bbH} && \;\!\|\bbY\!-\!\bbH\bbX\|_F^2 + \lambda r_{\delta_1}\!(\bbS \!-\! \barbS) + \beta r_{\delta_2}\!(\bbS)+ \gamma \|\bbS\bbH \!-\! \bbH\bbS\|_F^2   \nonumber     \\
  %  \!&\! && \;+ \gamma \|\bbS\bbH - \bbH\bbS\|_F^2 \nonumber \\
    \!&\! \mathrm{\;\;s. \;to: } && \; \|\bbC_\bby\bbS\!-\!\bbS\bbC_\bby\!\|_F^2\!\leq\! \epsilon_\bby,\,\!\|\bbC_\bbx\bbS\!-\!\bbS\bbC_\bbx\!\|_F^2\!\leq\!\epsilon_\bbx,
\end{alignat}
where $\bbC_\bby$ and $\bbC_\bbx$ denote the covariance matrices of $\bbY$ and $\bbX$, respectively.
If the covariances are perfectly known, then the corresponding parameters $\epsilon_\bby$ and $\epsilon_\bbx$ are set to zero.
Alternatively, if the $\bbC_\bby$ and $\bbC_\bbx$ are the sample estimates of the true covariances, then the values of $\epsilon_\bby$ and $\epsilon_\bbx$ must be selected based on the quality of the estimators (accounting, e.g., for the number of available observations $M$).

The constraints in \eqref{eq:rfi_nonconvex_st} capture the graph-stationarity assumption by promoting the commutativity with the true GSO.
Therefore, such constraints are considered in the graph denoising step [cf. \eqref{eq:step2_graph_denoising}].
In addition, since $\bbC_\bby$, $\bbC_\bbx$ and $\bbH$ are all polynomials of $\bbS$, the equalities $\bbC_\bby \bbH=\bbH\bbC_\bby$ and $\bbC_\bbx \bbH=\bbH\bbC_\bbx$ must hold as well, so it is also possible to augment the GF identification step [cf. \eqref{eq:step1_filterid}] with the corresponding constraints.
While in the interest of brevity, we do not spell out all the possible formulations here, the impact of several of these alternatives is numerically analyzed in \cref{sec:experiments_filter_id}. Finally, it is important to note that, since the stationarity constraints are quadratic and convex, the convergence described in \cref{thm1} also holds true for the iterative algorithm associated with \eqref{eq:rfi_nonconvex_st}.

%%%%%%%%%%%%%%%%%%%%%%%%%%%%%%%%%%%%%%%%%%%%%%%%%%%%%%%%%%%%%%%%%%%%%%%%%%%%%%%%%%%%%%%%%%%%%%%%%%%%%%%%%%%%%%%%%
%SECTION: MULTIPLE FILTERS
%%%%%%%%%%%%%%%%%%%%%%%%%%%%%%%%%%%%%%%%%%%%%%%%%%%%%%%%%%%%%%%%%%%%%%%%%%%%%%%%%%%%%%%%%%%%%%%%%%%%%%%%%%%%%%%%%
\section{Joint robust identification of multiple GFs}\label{sec:rfi_joint}
% Intro
In \cref{sec:rfi}, we approached the problem of identifying a single GF $\bbH$ defined over a single graph $\ccalG$.
However, in a variety of situations we encounter multiple processes (signals) over the same graph $\ccalG$.
Consider for example a network of weather stations measuring the temperature, humidity, and wind speed.
Each of these measurements corresponds to observations of a different process, all of them taking place over a common graph.
Intuitively, since all the GFs are related by the underlying graph $\ccalG$, we propose a \emph{joint} GF identification approach that exploits this relationship to enhance the quality of the estimation.
We focus first on the case where the input-output signals associated with each GF (graph process) are observed separately. Later in the section, we address a slightly more involved case where the GFs model the (AR) dynamics of a time-varying graph signal and, as a result, the observed signals are intertwined.

% New definitions
Consider a set of $K$ unknown GFs $\{\bbH_k\}_{k=1}^K$, all represented by $N \times N$ matrices and defined over the  graph $\ccalG$.
To be consistent with \cref{P:problem_statement}, we assume that: i) the true $\bbS$ is unknown and only the perturbed version $\barbS$ is available; ii) all $\bbH_k$ are polynomials of the \emph{same} GSO $\bbS$; and iii) for each $k$, matrices $\bbX_k \in \reals^{N \times M_k}$ and $\bbY_k \in \reals^{N \times M_k}$ collect the observed input and output graph signals and are related via
\begin{equation}\label{eq:rfi_observ_model_multi}
    \bbY_k = \bbH_k\bbX_k + \bbW_k,
\end{equation}
with $\bbH_k=\sum_{r=0}^{N-1}h_{r,k}\bbS^r$ and $\bbW_k$ being a white random matrix capturing observation noise and model inaccuracies.
Then, we aim at estimating the GFs $\{\bbH_k\}_{k=1}^K$ in a joint fashion while taking into account the inaccuracies in the topology of $\ccalG$.
This is summarized in the following problem statement.
\begin{problem}\label{P:multiple_filter_ir}
    Let $\ccalG$ be a graph with $N$ nodes, let $\bbS \in \reals^{N \times N}$ be the true (unknown) GSO associated with $\ccalG$, and let $\barbS \in \reals^{N \times N}$ be the perturbed (observed) GSO.
    Moreover, let $\bbX_k\in \reals^{N \times M_k}$ and $\bbY_k \in \reals^{N \times M_k}$ be the matrices collecting the $M_k$ observed input and output graph signals associated with $k=1,...,K$ network processes, all defined over $\ccalG$ and adhering to the model in \eqref{eq:rfi_observ_model_multi}.
    Our goal is to use $\{\bbX_k\!\}_{k=1}^K$, $\{\bbY_k\!\}_{k=1}^K$, and $\barbS$ to learn the $K$ GFs $\{\bbH_k\}_{k=1}^K$ that best fit the data, along with an enhanced estimation of $\bbS$.
    To that end, we make the following assumptions: \\
    (\textbf{AS2}) $\bbS$ and $\barbS$ are close according to some metric $d(\bbS,\barbS)$, i.e., the observed perturbations are ``small'' in some sense. \\
    (\textbf{AS3}) Every $\bbH_k$ is a polynomial of $\bbS$. \\
\end{problem}
Assumption (\textbf{AS2}), which was also considered in \cref{P:problem_statement}, promotes the tractability of the problem by ensuring that $\bbS$ and $\barbS$ are related. As discussed in \cref{sec:graph_pert}, the distance function $d(\cdot, \cdot)$ must be selected depending on the perturbation model at hand. 
(\textbf{AS3}) captures the key fact that all the matrices $\bbH_k$ are GFs of the \emph{same} GSO, establishing a link that can be leveraged via a joint estimation (optimization) of the $K$ GFs. 
Implementing an approach similar to that in \cref{sec:rfi} (i.e., working on the vertex domain, considering the true GSO as an explicit optimization variable, accounting for the GF structure via a commutativity constraint, and assuming that the graph perturbations create and destroy links), the multi-filter counterpart to \eqref{eq:rfi_nonconvex_rew} that codifies \cref{P:multiple_filter_ir} is
\begin{alignat}{2}\label{eq:joint_rfi_noncvx_rew}
    \!&\! \min_{\bbS\in \ccalS, \{\bbH_k\}_{k=1}^K} && \sum_{k=1}^K  \alpha_k\|\bbY_k-\!\bbH_k\bbX_k\|_F^2 \!+\! \lambda r_{\delta_1}(\bbS - \barbS) \nonumber \\
    \!&\! &&  +\beta r_{\delta_2}(\bbS) + \sum_{k=1}^K \gamma \|\bbS\bbH_k \!\!-\! \bbH_k\bbS\|_F^2. %\nonumber \\
    %\!&\! \mathrm{\;\;s. \;to: } && \bbS \in \ccalS.
\end{alignat}
Ideally, the value of the positive weight $\alpha_k$ must be selected based on the norm of $\bbW_k$ (e.g., prior information on the noise level and the number of signal pairs $M_k$). If none is available, then $\alpha_k=1$ for all $k$.
Equally important, the fact of pursuing a joint optimization implies that each $\bbH_k$ contributes with a regularization term $\|\bbS\bbH_k - \bbH_k\bbS\|_F^2$ promoting the commutativity of the $k$-th GF with the \emph{single} $\bbS$. Intuitively, having the same $\bbS$ in all these terms couples the optimization across $k$ and contributes to reduce the uncertainty over $\bbS$, leading to enhanced estimates of both $\bbS$ and $\{\bbH_k\}_{k=1}^K$.
As a result, the joint GF identification approach is expected to provide better results than estimating each $\bbH_k$ separately by solving $K$ instances of \eqref{eq:rfi_nonconvex_rew}.
We validate this hypothesis numerically via the experiments in \cref{sec:experiments_filter_id}.

Following a motivation similar to that in the previous section, we deal with the non-convex minimization in \eqref{eq:joint_rfi_noncvx_rew} designing an alternating optimization algorithm that breaks the bilinear terms $\bbS\bbH_k$ and $\bbH_k\bbS$, and approximates the logarithmic terms with a linear upper-bound. The resulting algorithm solves iteratively the following two subproblems for $t=1,...,t_{\max}$ iterations.

\vspace{2mm}
\noindent \textbf{Step 1: Multiple GF Identification.}
Given the current estimate $\bbS^{(t)}$, we solve the optimization problem in \eqref{eq:joint_rfi_noncvx_rew} with respect to each $\bbH^{(k)}$.
This yields
\begin{alignat}{2}\label{eq:joint_filterid}
    \!\!\bbH_k^{(t+1)\!} \!\!=\!\argmin_{\bbH_k} \alpha_k \|\! \bbY_k \!\!-\!\! \bbH_k\bbX_k\!\|_{\!F}^{\!2}  \!\!+\! \gamma \|\! \bbS^{\!(t)\!}\bbH_k \!\!-\!\! \bbH_k\bbS^{\!(t)} \!\|_{\!F}^{\!2},\!
\end{alignat}
whose closed-form solution can be found using \eqref{eq:step1_closed_form} replacing $\gamma$ with $\gamma/\alpha_k$, $\bbX$ with $\bbX_k$, and $\bbY$ with $\bbY_k$.
Note that since the only coupling across GFs is via the GSO, \eqref{eq:joint_filterid} estimates each $\bbH_k^{(t+1)}$ separately from the other GFs, solving $K$ LS problems (each with $N^2$ unknowns). Furthermore, if multiple processors are available, \eqref{eq:joint_filterid} can be run in parallel across $k$.

\vspace{2mm}
\noindent \textbf{Step 2: Graph Denoising.}
Given the current estimates of the GFs $\{\bbH_k^{(t+1)}\}_{k=1}^K$, we follow an MM scheme that, minimizing a linear upper-bound of the logarithmic penalties, yields the estimate of the GSO via
\begin{align}\label{eq:joint_graph_denoising}
    \bbS^{(t+1)} =  \argmin_{\bbS \in \ccalS} &\sum_{ij=1}^N \big( \lambda \bar{\Omega}_{ij}^{(t)}|S_{ij}-\bar{S}_{ij}|+\beta \Omega_{ij}^{(t)}|S_{ij}| \big) \nonumber \\
     &+ \sum_{k=1}^K\gamma \|\bbS\bbH_k^{(t+1)} - \bbH_k^{(t+1)}\bbS\|_F^2,
\end{align}
where $\bbOmega$ and $\bar{\bbOmega}$ are obtained as in \eqref{eq:log_weights}.

The solution to \cref{P:multiple_filter_ir} is simply given by $\hbS = \bbS^{(t_{max})}$ and $\hbH_k = \bbH_k^{(t_{max})}$ for every $k$. Similar to \eqref{eq:rfi_nonconvex_rew}, convergence to a stationary point of \eqref{eq:joint_rfi_noncvx_rew} is guaranteed, as formally stated next.

\begin{corollary}\label{thm2}
    Denote as $f(\{\bbH_k\}_{k=1}^K, \bbS)$ the objective function in \eqref{eq:joint_rfi_noncvx_rew}.
    If the vector $\bbz^{(t)} = [\vvec(\bbH_1^{(t)})^\top,...,\vvec(\bbH_K^{(t)})^\top, \vvec(\bbS)^\top]^\top$ represents the solution provided by the iterative algorithm  \eqref{eq:joint_filterid}-\eqref{eq:joint_graph_denoising} after $t$ iterations and every $\bbX_k$ excites all graph frequencies, then $\bbz^{(t)}$ converges to a stationary point of $f$ as the number of iterations $t$ goes to infinity.
\end{corollary}
The key to prove \cref{thm1}, which established the convergence to a stationary point for the robust estimation of a single GF, was to show that the optimization problem in \eqref{eq:rfi_nonconvex_rew} and the proposed algorithm satisfied the conditions in \cite[Th. 1b]{hong2015unified}. The formulation we put forth for the multi-filter case resembles closely that of the single-filter case, and, as a result, it is not difficult to show that those conditions also hold true for the problem in \eqref{eq:joint_rfi_noncvx_rew} (see Appendix~\ref{A:proof_thm1} for details). 

The discussion and formulations in \cref{sec:stationary_observations} dealing with incorporating additional information about the input-output signals into the optimization are also pertinent for the setup in this section. The details of such a formulation are omitted for brevity, but it will be explored in the experimental section.

\subsection{Joint GF identification for time series}
A slightly different, practically relevant, setup where multiple GFs need to be estimated is that of graph-based multivariate time series.
In that setup, each variable is associated with a node of the graph and the multiple graph-signal observations correspond to different instants of a time-varying graph signal. AR and \acrfull{ma} modeling of time series has a long tradition, with common approaches to decrease the degrees of freedom including limiting the memory of the series and assuming that matrices of coefficients relating different time instants are low rank \cite{reinsel2003elements}.
In the context of graph signals and network processes, a natural approach is to constrain the matrices of coefficients to be GFs, all defined over the same graph \cite{mei2017causal,isufi2019forecasting}.
This section introduces a variation of the problem in \eqref{eq:joint_rfi_noncvx_rew} tailored to this setup. 

To introduce the multiple-graph identification problem formally, let $\bbX_\kappa$ and $\bbY_\kappa$ denote a collection of $M_\kappa$ graph signals corresponding to measurements of a network process for $\kappa=1,...,\kappa_{max}$ time instants.
%Suppose now that the network process $\bbY_\kappa$ can be accurately modeled by an AR dynamics with memory $K$. This implies that, at every instant $\kappa$, the observations $\bbY_\kappa$ satisfy the equation
Suppose now that $\bbY_\kappa$ can be accurately modeled by an AR dynamics with memory $K$ so, at every instant $\kappa$, the observations $\bbY_\kappa$ satisfy the equation
\begin{equation}\label{eq:ar_observation}
    \bbY_\kappa = \sum_{k=1}^K \bbH_k\bbY_{\kappa-k} + \bbX_\kappa,\;\mathrm{with}\;\bbH_k=\sum_{r=0}^{N-1}h_{r,k}\bbS^r, 
\end{equation}
where $\bbX_\kappa$ is the exogenous input, and the GF $\bbH_k$ models the influence that the signal observations from the time instant $\kappa-k$ exert on the (current) signal at time $\kappa$.
%Although we only consider the influence of the exogenous input $\bbX_\kappa$ at the current time $\kappa$ for simplicity, modifying \eqref{eq:ar_observation} to account for additional time instants of $\bbX$ is straightforward.

Suppose now that: i) we have access to an estimated (imperfect) graph $\barbS$, ii) the value of the graph signals at different time instants is available, and iii) our goal is to estimate the set of matrices (GFs) $\{\bbH_k\}_{k=1}^K$ in \eqref{eq:ar_observation} that describe the dynamics of the multivariate time series. This can be accomplished as
\begin{alignat}{2}\label{eq:joint_rfi_noncvx_rew_ar}
    \!&\! \min_{\bbS\in \ccalS, \{\bbH_k\}_{k=1}^K}  \sum_{\kappa=K+1}^{\kappa_{max}}\Big\|\bbY_\kappa-\bbX_\kappa -\!\sum_{k=1}^K\bbH_k\bbY_{\kappa-k}\Big\|_F^2   \nonumber \\
    \!&\! \hspace{0.6cm} +\! \lambda r_{\delta_1}(\bbS - \barbS) \!+\! \beta r_{\delta_2}(\bbS) \!+\! \sum_{k=1}^K \gamma \|\bbS\bbH_k \!\!-\! \bbH_k\bbS\|_F^2. \!
    %\!&\! \mathrm{\;\;s. \;to: } && \bbS \in \ccalS.
\end{alignat}
The main difference relative to \eqref{eq:joint_rfi_noncvx_rew} is in the first term, which accounts for the new observation model [cf. \eqref{eq:rfi_observ_model_multi} vs. \eqref{eq:ar_observation}]. Note that we assume that the exogenous input $\bbX_\kappa$ is observed. If that were not the case, it would suffice to remove $\bbX_\kappa$ from the objective (possibly updating the Frobenius norm in case statistical knowledge about $\bbX_\kappa$ were available). Albeit the differences, the problem in \eqref{eq:joint_rfi_noncvx_rew_ar} is closely related to \eqref{eq:joint_rfi_noncvx_rew}, with the sources of non-convexities being the same. As a result, we approach its solution with a modified version of Algorithm~\ref{A:rfi_alg} which, at each iteration $t$, runs two steps. In the first one, we estimate each of the $K$ GFs by solving
\begin{alignat}{2}\label{eq:joint_filterid_time}
    \!& \bbH_k^{(t+1)} \!&&\!=\!\argmin_{\bbH_k}  \!\!\!\sum_{\kappa=K+1}^{\kappa_{max}}\!\!     \Big\|\bbY_\kappa\!-\!\bbX_\kappa\!-\!\bbH_k\bbY_{\kappa-k}\!-\!\!\!\sum_{k'<k}\!\!\bbH_{k'}^{(t+1)}\bbY_{\kappa-k'} \nonumber\\
            \!\!& \! &&\! \!-\!\! \sum_{K\geq k'>k}\!\!\bbH_{k'}^{(t)}\bbY_{\kappa-k'} \Big\|_{\!F}^{\!2}  +\!\! \sum_{k=1}^K\!\gamma \Big\| \bbS^{(t)}\bbH_k - \bbH_k\bbS^{(t)} \Big\|_{\!F}^{\!2},\!\!
\end{alignat}
which is different from the previous GF identification step [cf. \eqref{eq:joint_filterid}]. In contrast, the graph-denoising step in \eqref{eq:joint_graph_denoising} remains the same. Note that \eqref{eq:joint_filterid_time} updates each GF separately in a cyclic way by solving an LS problem with $N^2$ unknowns. Alternative implementations include using $\bbH_{k'}^{(t)}$ in lieu of $\bbH_{k'}^{(t+1)}$ for all $k'< k$ (so that a parallel implementation is enabled) as well as considering a single LS problem with $KN^2$ unknowns. 

Finally, it is worth emphasizing that the formulation introduced in this section can be used as a starting point to design more general robust schemes for multivariate time series defined over a graph. Dealing with both AR and MA matrices, assuming that the memory of the system is not known, having only partial/statistical information on the exogenous input, and observing the signals at only a subset of nodes are all examples of setups of interest. Since our goal in this section was to demonstrate the relevance of a robust multiple GF formulation in the context of multivariate time series, to facilitate exposition we restricted our discussion to the relatively simple case in \eqref{eq:ar_observation},  but many other setups (including those previously listed) will be subject of our future work.

%%%%%%%%%%%%%%%%%%%%%%%%%%%%%%%%%%%%%%%%%%%%%%%%%%%%%%%%%%%%%%%%%%%%%%%%%%%%%%%%%%%%%%%%%%%%%%%%%%%%%%%%%%%%%%%%%
%SECTION: FAST ALGORITHM
%%%%%%%%%%%%%%%%%%%%%%%%%%%%%%%%%%%%%%%%%%%%%%%%%%%%%%%%%%%%%%%%%%%%%%%%%%%%%%%%%%%%%%%%%%%%%%%%%%%%%%%%%%%%%%%%%
\section{Efficient implementation of the robust GF identification algorithm}\label{sec:efficient_rfi}
The algorithms proposed up to this point are able to find a solution to the robust GF identification problem in polynomial time. However, their computational complexity scales with the number of nodes as $N^7$. To facilitate the deployment in setups where $N$ is large, this section puts forth an efficient implementation that reduces the number of operations.

The new algorithm %(labeled as Algorithm~\ref{A:efficient_rfi_alg}, and whose pseudocode can be found in the following page) 
(summarized in Algorithm~\ref{A:efficient_rfi_alg}) preserves the core structure of Algorithm~\ref{A:rfi_alg}, with an outer loop that, at each iteration, runs two steps: one involving the estimation of the GF(s) and another one dealing with the estimation of the GSO.
The main difference is that now, instead of finding the exact solution to those two problems, we obtain an approximate solution. While the details, which are step-dependent, will be specified in the next paragraphs, the overall idea is that for each of the steps we run a few simple (gradient/proximal) iterations. Although Algorithm~\ref{A:efficient_rfi_alg} involves two nested loops, the complexity of the problems in the inner loop is cut down significantly, so that the overall computational overhead is reduced. 

To be specific, we describe next the two steps that, at each iteration of the outer loop $t \!=\! 0,...,t_{max}\!-\!1$, Algorithm~\ref{A:efficient_rfi_alg} runs.

\vspace{2mm}
\noindent \textbf{Step 1: Efficient GF Identification.}
Solving the GF-identification step with the closed-form solution presented in \eqref{eq:step1_closed_form} involves inverting a matrix of size $N^2 \times N^2$, which requires $\ccalO(N^6)$ operations.
To explain our alternative implementation, let $f_1(\bbH|\bbS^{(t)})$ denote the objective function in \eqref{eq:step1_filterid}. Since $f_1$ is strictly convex and smooth, it can be efficiently optimized using a gradient descent approach~\cite{boyd2004convex}.

To that end, for each iteration $t$ of the outer loop, we define the inner iteration index  $\tau$ as well as the sequence of variables $\bchkH^{(\tau)}$ with $\tau = 0,...,\tau_{max_1}$, which is initialized as $\bchkH^{(0)} = \bbH^{(t)}$.
With this notation at hand, at each iteration $\tau=0,...,\tau_{max_1}-1$ of the inner loop, we update $\bchkH^{(\tau+1)}$ via
\begin{equation}
    \bchkH^{(\tau+1)} = \bchkH^{(\tau)} - \mu \nabla f_1(\bchkH^{(\tau)}|\bbS^{(t)}).
\end{equation}
Here, $\mu > 0$ is the step size and $\nabla f_1$ denotes the gradient of $f_1$ with respect to $\bbH$, which is given by
\begin{equation}
 \nabla f_1\!(\bbH|\bbS^{(t)}) \!= \!2\Big(\bbH\bbX\bbX^\top \!\!-\!\! \bbY\bbX^\top \Big)\!\! 
    +\!2\gamma\Big( \!\bbS^{(t)^{ \!\top}}  \!\!(\bbS^{(t)}\bbH \!-\! \bbH\bbS^{(t)}) \!-\! (\bbS^{(t)}\bbH \!-\! \bbH\bbS^{(t)})\bbS^{(t)^\top} \!\Big).
\end{equation}
When the $\tau_{max_1}$ gradient updates are computed, we conclude the GF-identification step by setting $\bbH^{(t+1)} = \bchkH^{(\tau_{max_1})}$.

Since each gradient calculation involves the multiplication of $N \times N$ matrices, the resultant computational complexity is $\ccalO(\tau_{max_1}N^3)$, which may go down to $\ccalO(\tau_{max_1}N^{2.4})$ if an efficient multiplication algorithm is employed~\cite{coppersmith1987matrix}.
For large values of $N$, this complexity is substantially smaller than that required to find the inverse of an $N^2 \times N^2$ matrix.

\vspace{2mm}
\noindent \textbf{Step 2: Efficient graph denoising.}
Since the optimization in \eqref{eq:step2_graph_denoising} involves $N^2$ variables (the entries in $\bbS$), using an off-the-shelf convex solver incurs a computational complexity of $\ccalO(N^7)$ \cite{boyd2004convex}.
Inspired by the Lasso regression algorithm~\cite{hastie2015statistical}, we optimize individually over each entry $S_{ij}$ in an iterative manner.
The main idea is running multiple rounds of $N^2$ efficient scalar optimizations rather than dealing with a single but demanding $N^2$-dimensional problem. To provide the details of the scheme developed to estimate $\bbS$, we need to specify the set of constraints $\ccalS$ and introduce some definitions.
Let us focus on the set of adjacency matrices $\ccalS_\ccalA := \{ \bbS | S_{ij} \geq 0,\; S_{ii} = 0 \}$ and define the vectors $\bbs := \vvec(\bbS)$, vector $\barbs := \vvec(\barbS)$, and the matrix $\bbSigma^{(t)} := \bbH^{(t+1)^\top} \oplus -\bbH^{(t+1)}$.
With these definitions in place, the minimization in \eqref{eq:step2_graph_denoising} is equivalent to solving
\begin{alignat}{2}\label{eq:efficient_step2}
    \!&\!\min_{\bbs} && \sum_{i=1}^{N^2} \left(  \lambda \bar{\omega}^{(t)}_i|s_i - \bar{s}_i| + \beta \omega^{(t)}_i s_i \right) + \gamma \|\bbSigma^{(t)}\bbs\|_2^2, \nonumber \\
    \!&\! \mathrm{\;\;s. \;to: } && \;\; \bbs \geq 0, \;\; \bbs_\ccalD = 0,
\end{alignat}
where $\bbs_\ccalD$ collects the elements in the diagonal of $\bbS$, and the vectors $\bar{\bbomega}^{(t)}$ and $\bbomega^{(t)}$ are computed according to \eqref{eq:log_weights} but with $\barbs^{(t)}$ and $\bbs^{(t)}$ in lieu of $\barbS^{(t)}$ and $\bbS^{(t)}$. The constraint $\bbs_\ccalD = 0$, implies that only the $N^2-N$ elements of $\bbs$ representing the off-diagonal entries of $\bbS$ need to be optimized. The key point to find those $N^2-N$ values is to leverage that the non-differentiable part of the cost in \eqref{eq:efficient_step2} is separable across $s_i$, postulate $N^2-N$ scalar optimization problems (coupled via the $\ell_2$ term in the cost), and address the optimization  following a projected cyclic coordinate descent scheme.

To define clearly the operation of Step 2 at each iteration $t$ of the outer loop, we need to introduce some notation. First, let us denote as $\tau$ the iteration index for the inner loop, define the sequence of variables  $\bchks^{(\tau)}$ where $\tau=0,...,\tau_{max_2}$, and initialize the sequence as $\bchks^{(0)} = \bbs^{(t)}$. Moreover, with $\ell \not\in \ccalD$ denoting an index of the off-diagonal elements of the GSO, let $\bbsigma_\ell \in \reals^{N^2}$ denote the associated $\ell$-th column of $\bbSigma^{(t)}$, $\omega_\ell\geq 0$ and $\bar{\omega}_\ell\geq 0$ the associated entries of $\bbomega^{(t)}$ and $\bar{\bbomega}^{(t)}$, and $\check{s}_\ell^{(\tau)}\in\reals$ the associated entry of $\bchks^{(\tau)}$ (note that dependence on $t$ was omitted to facilitate readability). Then, at every iteration $\tau=0,...,\tau_{max_2}-1$ of the inner loop, Algorithm~\ref{A:efficient_rfi_alg} optimizes over each $\check{s}_\ell$ separately in a cyclic (successive) way.
The advantage of this approach is that the solution to the \emph{scalar} optimization over $\check{s}_\ell$ is given in closed form by the following projected soft-thresholding operation
\begin{equation}\label{eq:soft-thresholding}
    \check{s}_\ell^{(\tau+1)} = \left\{\hspace{-2mm} \begin{array}{cl}
        \left( - \bar{\lambda}_\ell + u^{(\tau)}_\ell  \right)^+ 
        &  \mathrm{if} \; \bar{s}_\ell <  - \bar{\lambda}_\ell+u^{(\tau)}_\ell,\\
        \left(\bar{\lambda}_\ell + u^{(\tau)}_\ell  \right)^+
        & \mathrm{if} \; \bar{s}_\ell >   \bar{\lambda}_\ell+u^{(\tau)}_\ell, \\
        \bar{s}_\ell & \mathrm{otherwise},
    \end{array} \right.
\end{equation}
\begin{equation}
    \mathrm{with} \;\;\; \bar{\lambda}_\ell = \frac{\lambda \bar{\omega}_\ell}{\gamma \bbsigma_\ell^\top\bbsigma_\ell}\;\; 
    \mathrm{and} \;\;\; u^{(\tau)}_\ell = \frac{-\beta \omega_\ell - \gamma\bbsigma_\ell^\top \bbr_\ell^{(\tau)} }{\gamma \bbsigma_\ell^\top\bbsigma_\ell} . \nonumber
\end{equation}
Here, $(\cdot)^+$ denotes the operation $(x)^+ = \max(0,x)$, and
\begin{equation}\label{eq:r_l}
    \bbr_\ell^{(\tau)} := \sum_{j < \ell} \bbsigma_j \check{s}_j^{(\tau+1)} + \sum_{j > \ell} \bbsigma_j \check{s}_j^{(\tau)}.
\end{equation}
Note that \eqref{eq:soft-thresholding} is a soft-thresholding operation with respect to the term $|s_i - \bar{s}_i|$.
Also, the constraints in $\ccalS_\ccalA$ are satisfied due to the projection operator $(\cdot)^+\!:=\!\max\{\cdot,\!0\}$, and because we do not optimize over the elements of the diagonal of $\bbS$.

At first sight, computing each $\check{s}_\ell$ requires roughly $N^2$ operations, so estimating the whole vector $\bbs$ would entail a computational complexity of $\ccalO(N^4)$.
However, a closer inspection of the vectors $\bbsigma_\ell$ reveals that no more than $2N$ of their entries are non-zero because $\bbsigma_\ell$ are the columns of the Kronecker sum of two $N \times N$ matrices.
We exploit this sparsity and reduce the number of operations required to compute each $s_\ell$ to approximately $2N$, rendering the final computational complexity of the graph denoising step to $\ccalO(2\tau_{max_2}N^3)$.

%%%%%%%%%%  ALGORITHM  %%%%%%%%%%
\begin{algorithm}[tb]
\SetKwInput{Input}{Input}
\SetKwInOut{Output}{Output}
\Input{$\bbX$, $\bbY$, $\barbS$}
\Output{$\hbH$, $\hbS$.}
\SetAlgoLined
Initialize $\bbH^{(0)}$ and $\bbS^{(0)}$ \\
$\barbs = \vvec(\barbS)$ \\
\For{$t=0$ \KwTo $t_{max}-1$}{
    \tcp{GF-identification step}
    $\bchkH^{(0)} = \bbH^{(t)}$ \\
    \For{$\tau=0$ \KwTo $\tau_{max_1}-1$}{
        $\bchkH^{(\tau+1)} = \bchkH^{(\tau)} + \mu \nabla f_1(\bchkH^{(\tau)}|\bbS^{(t)})$ \\
    }
    $\bbH^{(t+1)} = \bchkH^{(\tau_{max_1})}$ \\
    
    %\DontPrintSemicolon \;
    \vspace{.2cm}
    \tcp{Graph denoising step}
    $[\bbsigma_1,...,\bbsigma_{N^2}] =\bbH^{(t+1)^\top} \oplus \bbH^{(t+1)}$ \\ 
    $\bchks^{(0)} = \vvec(\bbS^{(t)})$ \\
    Update $\bar{\bbomega}^{(t)}$, $\bbomega^{(t)}$ via \eqref{eq:log_weights} using $\barbs$ and $\bchks^{(0)}$ \\
    
    \For{$i=0$ \KwTo $\tau_{max_2}-1$}{
    \For{$\ell \not\in \ccalD$}{
    Obtain $\bbr_\ell^{(\tau)}$ via \eqref{eq:r_l} \\
    Obtain $\check{s}_\ell^{(\tau+1)}$ via \eqref{eq:soft-thresholding} using $\bbsigma_\ell$, $\bbr_\ell^{(\tau)}$, $\omega_\ell$, $\bar{\omega}_\ell$ \\
    }
    }
    $\bbS^{(t+1)} = \mathrm{unvec}(\bchks^{(\tau_{max_2})})$
}
$\hbH = \bbH^{(t_{max})},\; \hbS = \bbS^{(t_{max})}$.
\caption{Reduced-complexity robust GF identification.}
\label{A:efficient_rfi_alg}
\end{algorithm}
%%%%%%%%%%  END ALGORITHM  %%%%%%%%%%

%The pseudocode describing all the steps in Algorithm~\ref{A:efficient_rfi_alg} is provided at the top of this page. 
The pseudocode describing the efficient implementation of Steps~1 and~2 is provided in Algorithm~\ref{A:efficient_rfi_alg}.
The summary is as follows. We postulate a nested algorithm with two loops. The outer loop runs $t_{\max}$ iterations. The inner loop runs two steps: Step~1, with $\tau_{max_1}$ iterations, and Step 2, with $\tau_{\max_2}$ iterations. While the complexity for Algorithm~\ref{A:rfi_alg} scaled as $\ccalO(t_{\max} N^7)$, with $t_{\max}$ being typically small, the overall computational complexity of Algorithm~\ref{A:efficient_rfi_alg} is roughly $\ccalO(t_{\max}(\tau_{max_1}+\tau_{max_2})N^3)$, which is encouraging, since $2N^2$ variables are optimized and it scales with $N$ significantly better than Algorithm~\ref{A:rfi_alg}. Solving Steps 1 and 2 optimally requires setting large values for $\tau_{\max_1}$ and $\tau_{\max_2}$. Nonetheless, we observe that in most tested setups the approach of setting small values for $\tau_{\max_1}$ and $\tau_{\max_2}$ (at the cost of setting a slightly higher value for $t_{\max}$) typically yields a faster convergence. Finally, implementations where the number of iterations is not fixed but selected based on some convergence criterion are also sensible alternatives. 
%the original one

We close the section noting that we developed Algorithm~\ref{A:efficient_rfi_alg} for the setting described in \cref{P:problem_statement} because the notation was simpler and facilitated the discussion. Nonetheless, an analogous approach may be followed for the joint estimation of $K$ GFs (cf. \cref{sec:rfi_joint}), resulting in an algorithm with complexity per GF similar to that for Algorithm~\ref{A:efficient_rfi_alg}.

%%%%%%%%%%%%%%%%%%%%%%%%%%%%%%%%%%%%%%%%%%%%%%%%%%%%%%%%%%%%%%%%%%%%%%%%%%%%%%%%%%%%%%%%%%%%%%%%%%%%%%%%%%%%%%%%%
%SECTION: NUMERICAL RESULTS
%%%%%%%%%%%%%%%%%%%%%%%%%%%%%%%%%%%%%%%%%%%%%%%%%%%%%%%%%%%%%%%%%%%%%%%%%%%%%%%%%%%%%%%%%%%%%%%%%%%%%%%%%%%%%%%%%
\section{Numerical results}\label{sec:experiments_filter_id}
This section discusses several numerical experiments to gain insights and assess the performance of the robust GF identification algorithms.
Unless specified otherwise, for a variable of interest $\bbTheta$, we report its normalized estimation error defined as
\begin{equation}\label{eq:rel_err}
    nerr(\hbTheta, \bbTheta) := \frac{\| \hbTheta - \bbTheta \|_F^2}{\|\bbTheta\|_F^2},    
\end{equation}
where $\hbTheta$ and $\bbTheta$ denote the estimated and the true value, respectively.
The code implementing our algorithms and the experiments presented next is available on GitHub\footnote{\url{https://github.com/reysam93/graph_denoising}}.
The interested reader is referred there for additional details and tests.

\subsection{Synthetic experiments}
We start by evaluating our algorithms with synthetic data, which is key to gain intuition.
Unless otherwise stated, graphs are sampled from ER random graph model with a link probability of $p=0.2$ and $N=20$ nodes; $\barbS$ is obtained by randomly creating and destroying 10\% of the links in $\bbS$; $M=50$ signals $\bbX$ and $\bbY$ are generated according to \eqref{eq:rfi_observ_model},  with the columns of $\bbX$ being drawn from a multivariate Gaussian distribution $\ccalN(\mathbf{0},\bbI)$, so the signals $\bbY$ are stationary on $\bbS$; signals in $\bbY$ are corrupted with white Gaussian noise with a normalized power of $\eta_\bbW=0.05$; and the reported error corresponds to the median of $nerr$ across 64 realizations of graphs and graph signals.

\vspace{2mm}
\noindent\textbf{Test case 1.} The first experiment evaluates the influence of perturbations as the order of the GF $R$ increases.
The number of observed pairs of signals considered is $M=100$ and 10\% of the edges in $\bbS$ are perturbed.
Results are reported in \cref{fig:exps1a}, where the x-axis shows $R$ and the y-axis $nerr(\hbh,\bbh)$.
The algorithms considered are: (i) the GF identification algorithm that ignores perturbations [see \eqref{eq:solving_fi}], denoted as ``FI''; (ii) the robust GF identification algorithm from Algorithm~\ref{A:rfi_alg} (``RFI''); (iii) a variation of ``RFI'' where the reweighted $\ell_1$ norm is replaced by the standard $\ell_1$ norm (``RFI-$\ell_1$''); and (iv) the robust GF identification algorithm accounting for the stationarity of $\bbY$ (``RFI-ST'').
First, we observe that the error of the ``FI'' algorithm, while small for low values of $R$, increases rapidly as $R$ grows. 
This is aligned with the discussion of high-order polynomials in \cref{sec:perturbed_graph_filters} and illustrates the merits of the robust algorithms.
Moreover, ``RFI-ST'' presents the best performance illustrating the importance of exploiting additional structure when it is available.
Finally, comparing the error of ``RFI'' and ``RFI-$\ell_1$'' showcases the benefits of replacing the $\ell_1$ norm with its reweighted version.

%%%%%%%%%%%%%%%  FIRST SET OF FIGURES   %%%%%%%%%%%%%%%
\begin{figure}[tb]
    \centering
    \begin{tikzpicture}[baseline,scale=1]
\begin{semilogyaxis}[
    width=.5\textwidth,
    xlabel={Filter order},
    xmin={2},
    xmax={6},
    ylabel={$nerr(\hbh,\bbh)$},
    ymin={1e-5},
    ymax={5},
    ytick={1e-5,1e-4,1e-3,1e-2,1e-1,1},
    grid=major,
    legend style={
        at={(1,0)},
        anchor=south east}
    ]
    \pgfplotstableread{data/1-FilterOrder_decay05.csv}\timetable
    
    \addplot[orange, dotted, mark=Mercedes star] table [x={Filter Order}, y=FI] {\timetable};
    \addplot[blue, solid, mark=*] table [x={Filter Order}, y=ITER-REW-NONST] {\timetable};
    \addplot[red, solid, mark=+] table [x={Filter Order}, y=ITER-NONST] {\timetable};
    \addplot[green!80!black, solid, mark=x] table [x={Filter Order}, y=ITER-REW-ST-REAL] {\timetable};
    
    \legend{FI, RFI, RFI-$\ell_1$, RFI-st}
    
\end{semilogyaxis}
\end{tikzpicture}
    \caption{Comparison of the error when estimating $\hbh$ via robust and non-robust algorithms in the presence of perturbations as the order of the GF increases.}
    \label{fig:exps1a}
\end{figure}
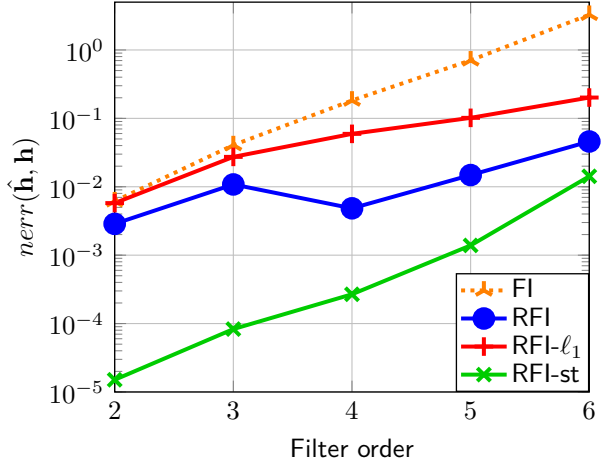
%%%%%%%%%%%%%%%%%%%%%%%%%%%%%%%%%%%%%%%%%%%

\vspace{2mm}
\noindent\textbf{Test case 2.}
The next experiment tests the influence of different types of perturbations in the robust and non-robust GF identification algorithms.
Figs.~\ref{fig:exps1b_c}a and~\ref{fig:exps1b_c}b illustrate the error of the estimated GF $\hbH$ and the denoised GSO $\hbS$ 
as the ratio of perturbed links in $\barbS$ increases.
Graphs are sampled from the SW~\cite{watts1998collective} random graph model and $\barbS$ is obtained by creating new links, destroying existing links, or simultaneously creating and destroying links, which are respectively denoted as ``C'', ``D'', and ``C/D'' in the legend.
Since the non-robust ``FI'' algorithm does not perform graph denoising we show the error $nerr(\barbS, \bbS)$, denoted as ``$\barbS$'' in \cref{fig:exps1b_c}b.
Furthermore, because the number of perturbed links is fixed, the error of $\barbS$ is the same for the considered perturbations and it is only plotted once.
From the figures, we observe that destroying links is the most harmful perturbation, especially when the focus is on $\hbS$.
This may be explained because destroying links is prone to produce non-connected graphs.
Nevertheless, the results show the resilience of the ``RFI'' algorithm, which provides low-error estimates $\hbH$ and $\hbS$ even when more than 20\% of the links are perturbed.

%%%%%%%%%%%%%%%  FIRST SET OF FIGURES   %%%%%%%%%%%%%%%
% \begin{figure*}[tb]
% 	\centering
% 	\begin{subfigure}{0.32\textwidth}
% 		\centering
% 		 \input{figs/filter_id/filter_order}
% 	\end{subfigure}
% 	\begin{subfigure}{0.32\textwidth}
% 		\centering
% 		\input{figs/filter_id/pert_type_errH}
% 	\end{subfigure}
% 	\begin{subfigure}{0.32\textwidth}
% 		\centering
% 		\input{figs/filter_id/pert_type_errS}
% 	\end{subfigure}
% 	\caption{Assessing the performance of the robust GF identification algorithm and the impact of perturbations in the topology.
% 	(a) shows the error of estimating $\hbh$ as the order of the GF increases; (b) and (c) respectively show the error of estimating $\hbH$ and $\hbS$ using a robust or a non-robust approach for several types of perturbations.}\label{fig:exps1}
% \end{figure*}
\begin{figure}[tb]
    \begin{subfigure}{0.47\textwidth}
		\centering
		\begin{tikzpicture}[baseline,scale=1]

\begin{semilogyaxis}[
    %table/col sep=semicolon,
    xlabel={(a) Proportion of perturbed links},
    xmin={0.05},
    xmax={0.25},
    xtick={.05, .1, .15, .2, .25},
    xticklabels={0.05, 0.1, 0.15, 0.2, 0.25},
    ylabel={$nerr(\hbH,\bbH)$},
    ymin={1e-3},
    ymax={0.4},
    grid=major,
    legend style={
        at={(0,0.5)},
        anchor=west},
    legend columns=2,
    ]
    
    \pgfplotstableread{data/pert_type_rf_errH.csv}\errRFtable
    \pgfplotstableread{data/pert_type_rfi_errH.csv}\errRFItable
    
    \addplot[green!80!black, dotted, mark=+] table [x=perts, y=RF-creat] {\errRFtable};
    \addplot[red, dotted, mark=x] table [x=perts, y=RF-dest] {\errRFtable};
    \addplot[blue, dotted, mark=*] table [x=perts, y=RF-creat-dest] {\errRFtable};
    \addplot[green!80!black, mark=+] table [x=perts, y=RFI-creat] {\errRFItable};
    \addplot[red, mark=x] table [x=perts, y=RFI-dest] {\errRFItable};
    \addplot[blue, mark=*] table [x=perts, y=RFI-creat-dest] {\errRFItable};
    
    % \legend{FI-c, RFI-c, FI-d, RFI-d, FI-c/d, RFI-c/d}
    \legend{FI-C, FI-D, FI-C/D, RFI-C, RFI-D, RFI-C/D}
    
\end{semilogyaxis}
\end{tikzpicture}
	\end{subfigure}
	\begin{subfigure}{0.47\textwidth}
		\centering
		\begin{tikzpicture}[baseline,scale=1]

\begin{semilogyaxis}[
    %table/col sep=semicolon,
    xlabel={(b) Proportion of perturbed links},
    xmin={0.05},
    xmax={0.25},
    xtick={.05, .1, .15, .2, .25},
    xticklabels={0.05, 0.1, 0.15, 0.2, 0.25},
    ylabel={$nerr(\hbS,\bbS)$},
    ymin={7e-5},
    ymax={1},
    grid=major,
    legend style={
        at={(0,0)},
        anchor=south west},
    legend columns=2,
    ]
    
    \pgfplotstableread{data/pert_type_rf_errS.csv}\errRFtable
    \pgfplotstableread{data/pert_type_rfi_errS.csv}\errRFItable
    
    \addplot[orange, dotted, mark=triangle*] table [x=perts, y=RF-creat] {\errRFtable};
    \addplot[green!80!black, mark=+] table [x=perts, y=RFI-creat] {\errRFItable};
    \addplot[red, mark=x] table [x=perts, y=RFI-dest] {\errRFItable};
    \addplot[blue, mark=*] table [x=perts, y=RFI-creat-dest] {\errRFItable};
    
    \legend{$\barbS$, RFI-C, RFI-D, RFI-C/D}
    % \legend{FI-C, FI-D, FI-C/D, RFI-C, RFI-D, RFI-C/D}
    
\end{semilogyaxis}
\end{tikzpicture}
	\end{subfigure}
	\caption{Assessing the performance of the robust GF identification algorithm and the impact of perturbations in the topology. (a) and (b) respectively show the error of estimating $\hbH$ and $\hbS$ using a robust or a non-robust approach for several types of perturbations.}\label{fig:exps1b_c}
\end{figure}
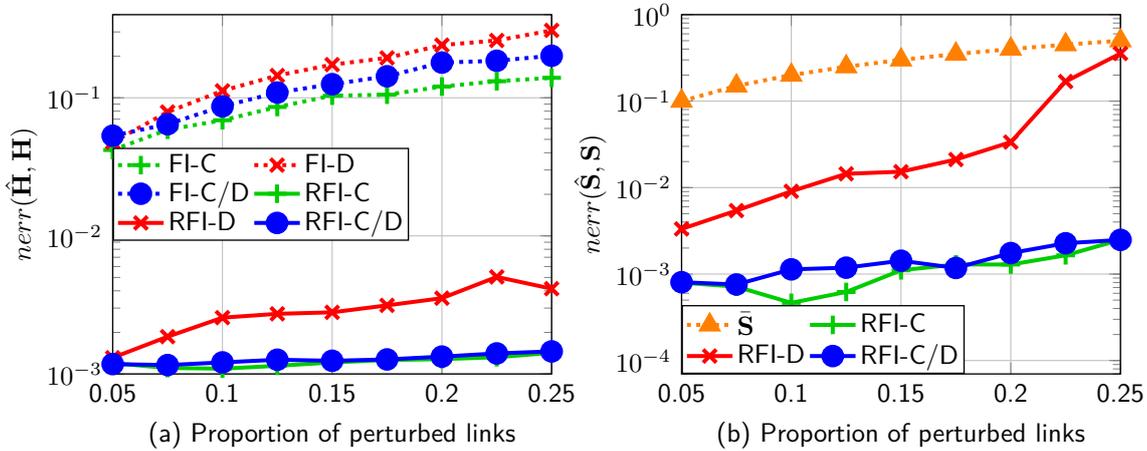
%%%%%%%%%%%%%%%%%%%%%%%%%%%%%%%%%%%%%%%%%%%

\vspace{2mm}
\noindent\textbf{Test case 3.}
Next, we compare the performance of our algorithms with other robust alternatives.
\cref{fig:exps2a} reports, for each algorithm, $nerr(\hbH,\bbH)$ as the ratio of perturbed links increases.
The baselines considered are the TLS-SEM algorithm from \cite{ceci2020_semtls}, and LLS-SCP from \cite{natali2020topology}.
We note that the TLS-SEM algorithm is tailored to graph signals following a SEM of the form 
\begin{equation}\label{eq:sem_model}
    \bbY=\bbA\bbY+\bbX = (\bbI-\bbA)^{-1}\bbX,
\end{equation}
where the observations at the $i$-th node are represented by the values of the neighbors of $i$ and an exogenous input.
As a result, the TLS-SEM algorithm may not be well suited to deal with signals generated according to the more general model in \eqref{eq:rfi_observ_model}.
Taking this into account, to offer a more favorable comparison we consider two types of graph signals: (i) signals generated according to \eqref{eq:sem_model}, denoted as ``SEM''; and (ii) signals generated according to \eqref{eq:rfi_observ_model}, denoted as ``H''.
It is worth noting that the ``SEM'' can be considered as a particular case of the model ``H'' when the GF $\bbH_{SEM} = (\bbI-\bbA)^{-1}$ is employed.

Looking at the results in \cref{fig:exps2a} we observe the following.
When the ``SEM'' model is considered, TLS-SEM (denoted as ``TLS'') obtains the best performance when the perturbation probability is small, and then, the performance of ``TLS'' and that of the ``RFI'' algorithm become comparable.
This illustrates that our algorithm is especially suitable to deal with a large number of perturbed links.
On the other hand, when the ``H'' model is considered, we observe that the ``RFI'' algorithm consistently outperforms the baselines in the presence of perturbations.
The good performance of the ``RFI'' algorithm on both signal models highlights the flexibility of the proposed formulation since it considers more lenient assumptions than the other alternatives.

%%%%%%%%%%%%%%%   FIGURES   %%%%%%%%%%%%%%%
% \begin{figure*}[tb]
% 	\centering
% 	\begin{subfigure}{0.32\textwidth}
% 		\centering
% 		 \input{figs/filter_id/compare_models}
% 	\end{subfigure}
% 	\begin{subfigure}{0.32\textwidth}
% 		\centering
% 		\input{figs/filter_id/eff_alg_time}
% 	\end{subfigure}
% 	\begin{subfigure}{0.32\textwidth}
% 		\centering
% 		\input{figs/filter_id/eff_alg_err_H}
% 	\end{subfigure}
% 	\caption{Comparing the performance of several robust GF identification algorithms.
% 	(a) shows the error of $\hbH$ when estimated with the proposed algorithm and with other baselines as the ratio of perturbed links increases.
% 	Different graph-signal models are considered.
% 	(b) and (c) respectively show the running time and error of $\hbH$ using Algorithm~\ref{A:rfi_alg} and Algorithm~\ref{A:efficient_rfi_alg} as the number of nodes increases.
% 	Different values for the maximum number of iterations of the inner loops are considered.}	\label{fig:exps2}
% \end{figure*}

\begin{figure}[tb]
    \centering
     \begin{tikzpicture}[baseline,scale=1]
\begin{semilogyaxis}[
    width=.5\textwidth,
    xlabel={Proportion of perturbed links},
    xmin={0},
    xmax={0.3},
    ylabel={$nerr(\hbH,\bbH)$},
    ymin={1e-5},
    ymax={1},
    ytick={1e-5,1e-4,1e-3,1e-2,1e-1,1},
    grid=major,
    legend style={
        at={(1,0)},
        anchor=south east},
    legend columns=2]
    
    \pgfplotstableread{data/4-ComparisonTLS-LLS.csv}\timetable
    
    \addplot[blue, solid, mark=*] table [x={Adjacency perturbation}, y=RFI-iter-H] {\timetable};
    \addplot[blue, dotted, mark=*] table [x={Adjacency perturbation}, y=RFI-iter-SEM] {\timetable};
    \addplot[red, solid, mark=x] table [x={Adjacency perturbation}, y=TLS-SEM-H] {\timetable};
    \addplot[red, dotted, mark=x] table [x={Adjacency perturbation}, y=TLS-SEM-SEM] {\timetable};
    \addplot[green!80!black, solid, mark=+] table [x={Adjacency perturbation}, y={LLS-SCP (Natali)-H}] {\timetable};
    \addplot[green!80!black, dotted, mark=+] table [x={Adjacency perturbation}, y={LLS-SCP (Natali)-SEM}] {\timetable};
    
    \legend{{RFI, H}, {RFI, SEM}, {TLS, H}, {TLS, SEM}, {SCP, H}, {SCP, SEM}}

\end{semilogyaxis}
\end{tikzpicture}
    \caption{Normalized error of $\hbH$ when estimated with the proposed algorithm and with other baselines as the ratio of perturbed links increases.
	Different graph-signal models are considered.}
    \label{fig:exps2a}
\end{figure}
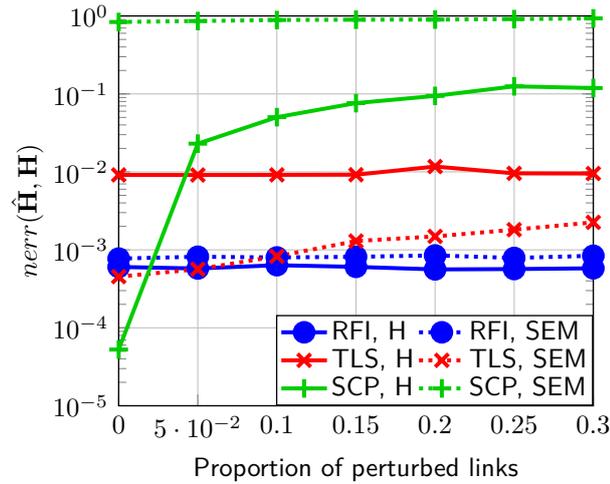
\begin{figure}[tb]
    \begin{subfigure}{0.47\textwidth}
		\centering
		\begin{tikzpicture}[baseline,scale=1]

\begin{semilogyaxis}[
    table/col sep=comma,
    xlabel={(a) Number of nodes},
    xmin={20},
    xmax={100},
    ylabel={time (s)},
    ymax={1e5},
    ytick={1,1e1,1e2,1e3,1e4,1e5},
    grid=major,
    legend style={
        at={(0,1)},
        anchor=north west,
        },
    ]
    
    \pgfplotstableread{data/eff_alg_mean_times.csv}\timetable
    
    \addplot[blue, solid, mark=*] table [x=x, y=alg1] {\timetable};
    \addplot[red, loosely dotted, mark=x] table [x=x, y=alg2] {\timetable};
    \addplot[orange, dotted, mark=x] table [x=x, y=alg3] {\timetable};
    \addplot[green!80!black, densely dotted, mark=x] table [x=x, y=alg4] {\timetable};
    
    \legend{Stand-5, Eff-5-10, Eff-5-25, Eff-5-50}
    
\end{semilogyaxis}
\end{tikzpicture}
	\end{subfigure}
	\begin{subfigure}{0.47\textwidth}
		\centering
		\begin{tikzpicture}[baseline,scale=1]

\begin{axis}[
    table/col sep=comma,
    xlabel={(b) Number of nodes},
    xmin={20},
    xmax={100},
    ylabel={$nerr(\hbH,\bbH)$},
    ymin={0},
    ymax={0.4},
    grid=major,
    legend style={
        at={(0,1)},
        anchor=north west}
    ]
    
    \pgfplotstableread{data/eff_alg_mean_err_H.csv}\timetable
    
    \addplot[blue, solid, mark=*] table [x=x, y=alg1] {\timetable};
    \addplot[red, loosely dotted, mark=x] table [x=x, y=alg2] {\timetable};
    \addplot[orange, dotted, mark=x] table [x=x, y=alg3] {\timetable};
    \addplot[green!80!black, densely dotted, mark=x] table [x=x, y=alg4] {\timetable};
    
    \legend{Stand-5, Eff-5-10, Eff-5-25, Eff-5-50}
    
\end{axis}
\end{tikzpicture}
	\end{subfigure}
	\caption{Comparing the performance of several robust GF identification algorithms.
	(a) and (b) respectively show the running time and error of $\hbH$ using Algorithm~\ref{A:rfi_alg} and Algorithm~\ref{A:efficient_rfi_alg} as the number of nodes increases.
	Different values for the maximum number of iterations of the inner loops are considered.}\label{fig:exps2b_c}
\end{figure}
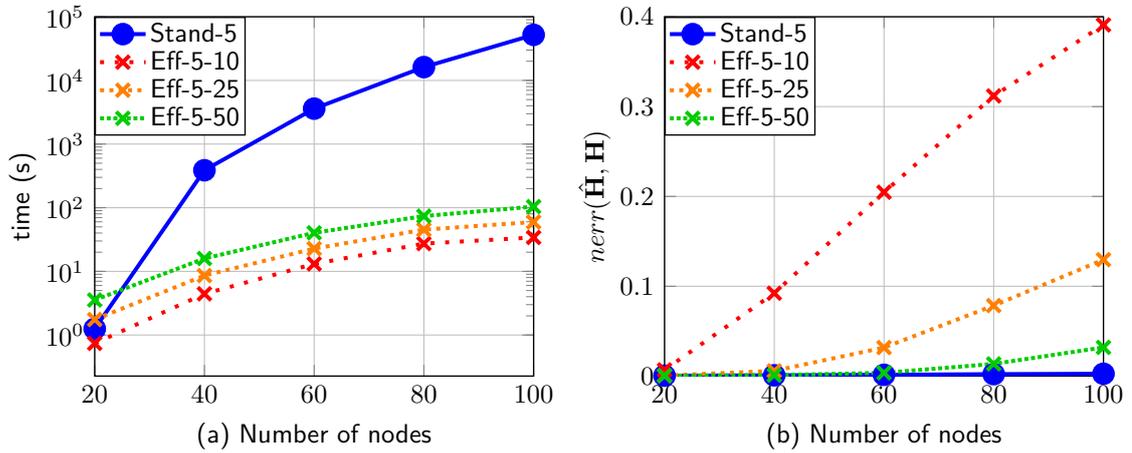
%%%%%%%%%%%%%%%%%%%%%%%%%%%%%%%%%%%%%%%%%%%

\vspace{2mm}
\noindent\textbf{Test case 4.}
Now, we compare the performance of the standard and the efficient implementation of the robust identification algorithm, as described in Algorithms~\ref{A:rfi_alg} and~\ref{A:efficient_rfi_alg}.
The results are shown in Figs.~\ref{fig:exps2b_c}a and~\ref{fig:exps2b_c}b, where the figures depict the running time measured in seconds and $nerr(\hbH,\bbH)$ as $N$ increases.
The legend identifies first the algorithm employed, then the number of iterations of the outer loop ($t_{max}$), and finally the iterations of the inner loops (with $\tau_{{max_1}}=\tau_{{max_2}}$).
As expected, \cref{fig:exps2b_c}a shows that Algorithm~\ref{A:efficient_rfi_alg} is remarkably faster than Algorithm~\ref{A:rfi_alg} even with medium-sized graphs, achieving a running time $10^3$ times smaller when $N=100$.
On the other hand, in \cref{fig:exps2b_c}b we observe that ``Eff-5-50'' has an error that is close to the standard implementation (``Stand-5'') even though it is considerably faster.
Furthermore, the trade-off between speed and estimation accuracy is also evident.
``Eff-5-10'' is the fastest implementation but the quality of its estimated GF may not be enough for graphs with more than 40 nodes.

\vspace{2mm}
\noindent\textbf{Test case 5.} 
The last experiment with synthetic data studies the benefits of the joint GF estimation.
All the GFs are polynomials of the same $\bbS$, and for each $\bbH_k$ we consider $M_k = 15$ noisy observations with $\eta_\bbw = 0.01$.
\cref{fig:exp_jointfi} shows the results, with the y-axis being the normalized error averaged across the $K$ graphs, i.e., $\frac{1}{K}\sum_{k=1}^K nerr(\hbH_k, \bbH_k)$, and the x-axis representing $K$.
We compare the performance of estimating the GFs jointly (marked as ``J'' in the legend) or separately for the three algorithms (``RFI-$\ell_1$'', ``RFI'', and ``RFI-st'') described in Test case~1.
Note that ``RFI-J'' corresponds to the formulation in \eqref{eq:joint_rfi_noncvx_rew}.
The first thing we observe from the results in \cref{fig:exp_jointfi} is that the error decreases as $K$ increases when a joint algorithm is employed.
This is aligned with the discussion in \cref{sec:rfi_joint} and illustrates the benefit of exploiting the common structure.
In addition, algorithms accounting for the stationary of $\bbY$ outperform the non-stationary alternatives even though we only have $M=15$ signals to estimate the covariance $\hbC_\bby$.

%%%%%%%%%%%%%%%   MORE FIGURES   %%%%%%%%%%%%%%%
\begin{figure}[tb]
	\centering
	\centering
		\begin{tikzpicture}[baseline,scale=1]

\begin{semilogyaxis}[
    %table/col sep=semicolon,
    width=0.5\textwidth,
    xlabel={Number of filters},
    xmin={1},
    xmax={5},
    ylabel={$\sum_{k=1}^K nerr(\hbH^{(k)},\bbH^{(k)})/K$},
    ymin={10^-2.5},
    ymax={1e-1},
    ytick={10^-2.5, 1e-2, 10^-1.5, .1},
    grid=major,
    legend style={
        at={(0,0)},
        anchor=south west},
    ]
    
    \pgfplotstableread{data/5-Joint-RFI.csv}\timetable
    
    \addplot[red, dashed, mark=+] table [x={Number of filters}, y=RFI-iter] {\timetable};
    \addplot[red, mark=+] table [x={Number of filters}, y=RFI-iter-J] {\timetable};
    \addplot[blue, dashed, mark=*] table [x={Number of filters}, y=RFI-iter-rew] {\timetable};
    \addplot[blue, mark=*] table [x={Number of filters}, y=RFI-iter-rew-J] {\timetable};
    \addplot[green!80!black, dashed, mark=x] table [x={Number of filters}, y=RFI-iter-Cyest] {\timetable};
    \addplot[green!80!black, mark=x] table [x={Number of filters}, y=RFI-iter-Cyest-J] {\timetable};

    % \addplot[red, solid, mark=+] table [x={Number of filters}, y=ITER-NONST] {\timetable};
    % \addplot[blue, solid, mark=*] table [x={Number of filters}, y=ITER-REW-NONST] {\timetable};
    % \addplot[green!80!black, solid, mark=x] table [x={Number of filters}, y=ITER-REW-ST-REAL] {\timetable};
    
    \legend{RFI-$\ell_1$, RFI-$\ell_1$-J, RFI, RFI-J, RFI-st, RFI-st-J}
    
\end{semilogyaxis}
\end{tikzpicture}
	\caption{Error performance when estimating $K$ GFs using the separate and joint approach for different values of $K$. %The reported error is the median across 64 (graph and perturbation) realizations.
	}\label{fig:exp_jointfi}
\end{figure}
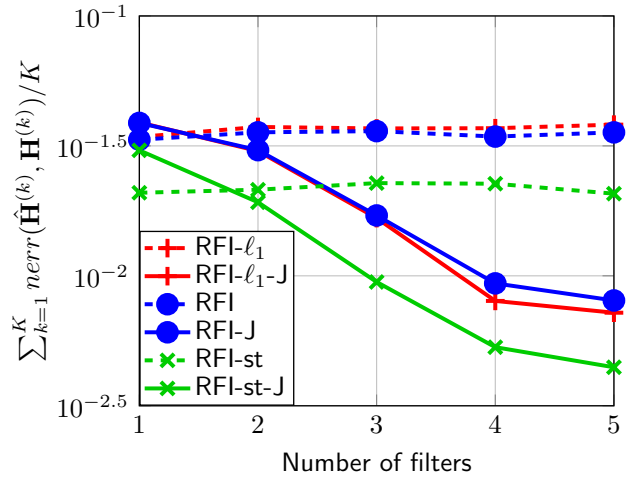
%%%%%%%%%%%%%%%%%%%%%%%%%%%%%%%%%%%%%%%%%%%

\subsection{Real-world datasets}
To close the numerical evaluation, we test our robust GF identification algorithms over two real-world datasets.

\vspace{3mm}
\noindent\textbf{Weather station network.} This test case evaluates the ability of our algorithms to predict the temperature measured by a network of stations using the data from previous days. The data comes from the ``Global Summary of the Day'' dataset of the National Centers for Environmental Information\footnote{\url{https://www.ncei.noaa.gov/data/global-summary-of-the-day/archive/}} and we used daily temperature measurements from $N=17$ stations in California during 2017 \& 2018. Specifically, with $\bby_\kappa \in \reals^N$ collecting the measurements of the 17 stations at day $\kappa$, we consider an AR model without exogenous inputs, so that $\bby_\kappa \approx \sum_{k=1}^K \bbH_k \bby_{\kappa-k}$. The data samples were divided into two subsets, the first one (training) was used to obtain the GFs $\bbH_k$ and the second one (evaluation) was used to assess the performance and the generalization power of the GFs obtained. Also, the data is normalized so that the signal at each station for all time samples has unitary norm.

The underlying $\ccalG$ was constructed as the unweighted 5-nearest neighbors graph, using the geographical distance between stations. Since temperature relations across stations are likely to be due to a range of factors (including, e.g., altitude), the considered adjacency (based only on geographical positions) may be imperfect, rendering our robust algorithms better suited for this task.

The estimation performance of the different algorithms is shown in Table~\ref{T:tempData}. Since in this case the ground-truth GF is not known, we use the signal denoising error $nerr(\bby_{\kappa},\hby_{\kappa})$ to assess the quality of the schemes. In this specific experiment, the error is measured over all samples (both training and test subsets), to see a clear downward trend when increasing the number of training samples, or equivalently, the train-test split (TTS) value.
The algorithms evaluated are ``LS'', ``LS-GF'' (which postulates a GF with coefficients $\hbh = \argmin_\bbh \| \bbY - \sum_r h_r \bbS^r \bbX \|_F^2$), ``TLS'', ``RFI'' (which assumes an AR(1) process) and ``AR(3)-RFI''. Two values of TTS (0.25 and 0.50) and two prediction horizons (1 and 3) are considered. 
The main observation is that ``AR(3)-RFI'' yields the best performance in all settings. Additionally, the results for TTS=0.25 demonstrate the benefits of considering the underlying graph in the low-sample regime, since even ``LS-GF'', which relies on the  imperfect $\bar{\bbS}$, outperforms ``LS''. On the other hand, ``LS-GF'' does not seem to improve its prediction as TTS increases, while our two algorithms yield a lower prediction error.

\begin{table}[!t]
	\centering
	%\setlength{\tabcolsep}{2pt}
% \begin{tabular}{c||c|c||c|c|}
% \multirow{2}{*}{Models} & \multicolumn{2}{c}{1-Step} & \multicolumn{2}{c}{3-Step} \\
% & TTS = 0.25 & TTS = 0.5 & TTS = 0.25 & TTS = 0.5 \\ \hline
% \!\!LS & 6.9e-3 & 3.1e-3 & 2.1e-2 & 9.1e-3 \\
% \!\!LS-GF & 3.3e-3 & 3.3e-3 & 8.4e-3 & 8.5e-3 \\
% \!\!TLS-SEM & 4.0e+1 & 3.7e-2 & 6.8e-1 & 5.5e-2 \\
% \!\!RFI & 3.4e-3 & 3.1e-3 & 8.5e-3 & 7.5e-3 \\
% \!\!VAR-RFI & \textbf{3.2e-3} & \textbf{2.8e-3} & \textbf{7.8e-3} & \textbf{6.9e-3} \\
% \end{tabular}

\begin{tabular}{c||c|c||c|c|}
\multirow{2}{*}{Models} & \multicolumn{2}{c}{1-Step} & \multicolumn{2}{c}{3-Step} \\
& TTS=0.25 & TTS = 0.5 & TTS=0.25 & TTS = 0.5 \\ \hline
\!\!LS & $6.9 \cdot 10^{-3}$ & $3.1 \cdot 10^{-3}$ & $2.1 \cdot 10^{-2}$ & $9.1 \cdot 10^{-3}$ \\
\!\!LS-GF & $3.3 \cdot 10^{-3}$ & $3.3 \cdot 10^{-3}$ & $8.4 \cdot 10^{-3}$ & $8.5 \cdot 10^{-3}$ \\
\!\!TLS-SEM & $4.0 \cdot 10^{1}$ & $3.7 \cdot 10^{-2}$ & $6.8 \cdot 10^{-1}$ & $5.5 \cdot 10^{-2}$ \\
\!\!RFI & $3.4 \cdot 10^{-3}$ & $3.1 \cdot 10^{-3}$ & $8.5 \cdot 10^{-3}$ & $7.5 \cdot 10^{-3}$ \\
\!\!AR(3)-RFI & $\mathbf{3.2 \cdot 10^{-3}}$ & $\mathbf{2.8 \cdot 10^{-3}}$ & $\mathbf{7.8 \cdot 10^{-3}}$ & $\mathbf{6.9 \cdot 10^{-3}}$ \\
\end{tabular}
	\caption{Performance of the algorithms in predicting the temperature for 2 prediction horizons (1 and 3) and 2 values (25\% and 50\%) of train-test split (TTS). The metrics shown are the average of the normalized error at each timestep $\frac{1}{M} \sum_{\kappa=1}^M nerr(\hby_{\kappa},\bby_{\kappa})$ for all samples.} \label{T:tempData}
    \vspace{0.1cm}
\end{table}

\vspace{3mm}
\noindent\textbf{Air quality station network.}
We consider an experimental setup (AR model, graph creation method...) similar to that for the weather station data but, in this case, we use 2018 \& 2019 data from the United States Environmental Protection Agency\footnote{\url{https://www.epa.gov/outdoor-air-quality-data}} to predict the ozone levels in a network of 17 outdoor stations in California.
The stations chosen were those with at least 330 measurements each year for a selection of pollutants, and missing data was filled via first-order interpolation.

The goal here is to analyze how the prediction horizon affects the prediction error.
The value of TTS chosen was 0.5, i.e. evaluation data represented 50\% of the samples.
\cref{fig:exp_airQual} shows the performance of the algorithms when predicting ozone levels.
As a baseline, ``LS-Eval (LB)'' shows the error measured on the evaluation data when obtaining the GF also using evaluation data, therefore representing a lower bound for the LS error using AR models of order 1. Also, ``Copy-Prev-Day'' represents the error obtained by the ``identity GF'', which copies the previous day's measurement.
As in the previous example, the best performing algorithm is ``AR(3)-RFI'', whose performance is close to the baseline, followed by ``RFI''.

%%%%%%%%%%%%%%%   MORE FIGURES   %%%%%%%%%%%%%%%
\begin{figure}[t]
	\centering
	\centering
	\begin{tikzpicture}[baseline,scale=1]

\begin{semilogyaxis}[
    %table/col sep=semicolon,
    width=0.5\textwidth,
    xlabel={Time horizon used for prediction},
    xmin={1},
    xmax={5},
    xtick={1, 2, 3, 4, 5},
    xticklabels={1, 2, 3, 4, 5},
    ylabel={$\sum_{\kappa=1}^M nerr(\hby_\kappa,\bby_\kappa) / M$},
    ymin={1e-2},
    ymax={0.1},
    ytick={1e-2, 10^-1.5, .1},
    grid=major,
    legend style={
        at={(1,0)},
        anchor=south east},
    legend columns=2,
    ]
    
    \pgfplotstableread{data/airQuality-NSteps.csv}\errNSteps
    
    \addplot[green!80!black, dotted, mark=+] table [x=N-Steps, y=LS-Perfect-(LB)] {\errNSteps};
    \addplot[blue, mark=*] table [x=N-Steps, y=Least-Squares] {\errNSteps};
    \addplot[red, mark=+] table [x=N-Steps, y=Copy-Prev-Day] {\errNSteps};
    \addplot[orange, mark=+] table [x=N-Steps, y=TLS-SEM] {\errNSteps};
    \addplot[purple, mark=*] table [x=N-Steps, y=Least-Squares-GF] {\errNSteps};
    \addplot[cyan, mark=x] table [x=N-Steps, y=RGFI] {\errNSteps};
    \addplot[black, mark=x] table [x=N-Steps, y=VAR-RGFI] {\errNSteps};
    
    % \legend{FI-c, RFI-c, FI-d, RFI-d, FI-c/d, RFI-c/d}
    \legend{{LS-Eval (LB)},LS,Copy-Prev-Day,TLS-SEM,LS-GF,RFI,AR(3)-RFI}
    
\end{semilogyaxis}
\end{tikzpicture}
	\caption{Performance of the algorithms predicting ozone levels in the AirData station network, as the time horizon of the prediction increases.} \label{fig:exp_airQual}
\end{figure}
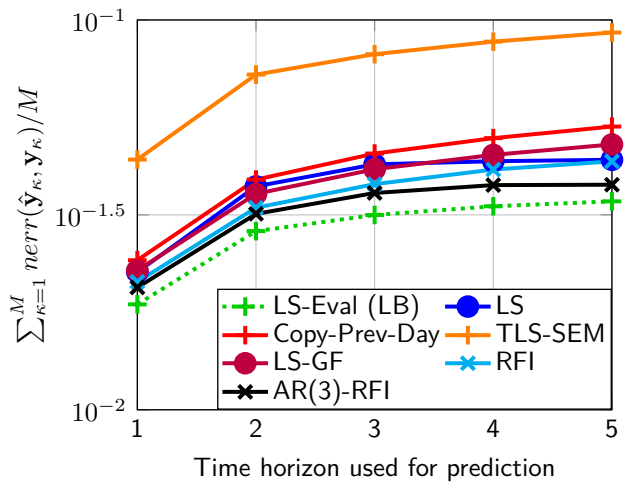
%%%%%%%%%%%%%%%%%%%%%%%%%%%%%%%%%%%%%%%%%%%

%%%%%%%%%%%%%%%%%%%%%%%%%%%%%%%%%%%%%%%%%%%%%%%%%%%%%%%%%%%%%%%%%%%%%%%%%%%%%%%%%%%%%%%%%%%%%%%%%%%%%%%%%%%%%%%%%
%SECTION: CONCLUDING REMARKS
%%%%%%%%%%%%%%%%%%%%%%%%%%%%%%%%%%%%%%%%%%%%%%%%%%%%%%%%%%%%%%%%%%%%%%%%%%%%%%%%%%%%%%%%%%%%%%%%%%%%%%%%%%%%%%%%%
%
\section{Concluding remarks}\label{sec:conclusion_filter_id}
This chapter put forth a framework dealing with estimation problems in GSP where the information about (the links of) the supporting graph is uncertain. 
Specifically, we addressed the problem of estimating a GF (i.e., a polynomial of the GSO) from input and output graph signals under the key assumption that only a perturbed version of the true GSO was available. 
In contrast to the majority of existing approaches that operate on the spectral domain, we recast the true graph as an additional estimation variable and formulated an optimization problem that \emph{jointly} estimated the GF and the true (unknown) GSO. We focused first on the case where only one GF needed to be estimated and, then, shifted to (multi-feature and AR graph signal) setups where multiple GFs have to be jointly identified. 
The formulated optimizations operated completely in the vertex domain and bypassed the problem of computing high-order polynomials, avoiding the challenges of dealing with the influence of perturbations in the graph spectrum as well as the numerical instability and error propagation associated with  high-order matrix polynomials. While non-convex, upon blending techniques from alternating optimization and MM, the proposed algorithm was shown to be capable to find a stationary point in polynomial time. This algorithm was later modified so that the scaling of the computational complexity with respect to the number of nodes in the graph is reduced. Future work includes delving into the robust estimation of ARMA time-varying graph signals, consideration of additional graph perturbation models, and application of our robust estimation framework to other GSP problems, to name a few.

%%%%%%%%%%%%%%%%%%%%%%%%%%%%%%%%%%%%%%%%%%%%%%%%%%%%%%%%%%%%%%%%%%%%%%%%%%%%%%%%%%%%%%%%%%%%%%%%%%%%%%%%%%%%%%%%%
%%%%%%%%%%%%%%%%%%%%%%%%%%%%%%%%%%%%%%%%%%%%%%%%%%%%%%%%%%%%%%%%%%%%%%%%%%%%%%%%%%%%%%%%%%%%%%%%%%%%%%%%%%%%%%%%%
%APPENDICES AND REFERENCES
%%%%%%%%%%%%%%%%%%%%%%%%%%%%%%%%%%%%%%%%%%%%%%%%%%%%%%%%%%%%%%%%%%%%%%%%%%%%%%%%%%%%%%%%%%%%%%%%%%%%%%%%%%%%%%%%%
%%%%%%%%%%%%%%%%%%%%%%%%%%%%%%%%%%%%%%%%%%%%%%%%%%%%%%%%%%%%%%%%%%%%%%%%%%%%%%%%%%%%%%%%%%%%%%%%%%%%%%%%%%%%%%%%%
\newpage
\section{Appendix: Proof of Theorem~\ref{thm1}}\label{A:proof_thm1}
The proof relies on the results presented in \cite[Th. 1b]{hong2015unified}, so it suffices to show that our formulation and algorithm fulfill the required conditions in \cite{hong2015unified}. To that end, recall that $f(\bbz)$ is the objective function in \eqref{eq:rfi_nonconvex_rew}, and let $\bbz_1:=\vvec(\bbH)$ and $\bbz_2:=\vvec(\bbS)$ denote the $B=2$ blocks of variables considered in our algorithm.
Moreover, at each step, the function $f(\bbz)$ is approximated by $u_1(\bbz_1)$ and $u_2(\bbz_2)$, corresponding to the objective functions in \eqref{eq:step1_filterid} and \eqref{eq:step2_graph_denoising}.
Then, to ensure the convergence of our iterative algorithm the following conditions are required.

% \noindent\textit{(\textbf{C1}) Each of the approximation functions $u_b(\bbz_b)$ must be a global upper bound of $f(\bbz)$ and the first order behavior of $u_b(\bbz_b)$ and $f(\bbz)$ must be the same.}
\noindent\textit{(\textbf{C1}) Each function $u_b(\bbz_b)$ must be a global upper bound of $f(\bbz)$ and the first-order behavior of $u_b(\bbz_b)$ and $f(\bbz)$ must be the same.}

\noindent\textit{(\textbf{C2}) $f(\bbz)$ must be regular (cf. \cite{hong2015unified}) at every point in $\ccalZ^*$.} 

\noindent \textit{(\textbf{C3}) The level set $\ccalZ^{(0)} = \{\bbz \; | \; f(\bbz) \leq f(\bbz^{(0)}) \}$ is compact.}

\noindent\textit{(\textbf{C4}) At least one of the problems in \eqref{eq:step1_filterid} and \eqref{eq:step2_graph_denoising} must have a unique solution.}

\noindent
Next, we address each of the four conditions separately, proving that our approach satisfies all of them.

Condition (\textbf{\textit{C1}}) requires the surrogate functions $u_b(\bbz_b)$ to be global upper bounds of $f(\bbz)$.
For the first block ($b=1$), it is easy to see that $u_1(\bbz_1)=f(\bbz)$ when the block $\bbz_2$ remains constant, so it satisfies the requirements.
Regarding $u_2(\bbz_2)$, we approximate $f(\bbz)$ with the first-order Taylor series of the logarithmic penalty.
Because the $\log$ is a concave differentiable function, it follows that its Taylor series of order one constitutes a global upper bound. 
Moreover, because $u_2(\bbz_2)$ is a first-order Taylor series approximation of $f(\bbz)$, it also follows that the first-order behavior of $f(\bbz)$ and $u_2(\bbz_2)$ is the same.
Therefore, $u_2$ also satisfies the requirement, and hence, (\textbf{\textit{C1}}) is fulfilled.

To prove (\textit{\textbf{C2}}), according to \cite{hong2015unified}, a function $f(\bbz)$ is regular if its non-smooth components are separable across the different blocks of variables.
To show this, we decompose $f$ as $f = g_A+g_B$, with functions $g_A$ and $g_B$ being defined as

\begin{align}
    &g_A(\bbH, \bbS) = \|\bbY - \bbH \bbX\|_F^2 + \gamma \|\bbH \bbS - \bbS \bbH\|_F^2, \nonumber \\
    &g_B(\bbS) = \lambda \sum_{i,j=1}^N\log(|S_{ij}\!-\!\bar{S}_{ij}|+\delta_2) + \beta\sum_{i,j=1}^N\log(|S_{ij}|+\delta_1). \nonumber
\end{align}
Since $g_A$ is a smooth function and the non-smooth function $g_B$ only depends on the variables on the second block, $\bbz_2=\vvec(\bbS)$, it follows that $f(\bbz)$ is a regular function for all feasible points.

Next, we show that the level set $\ccalZ^{(0)} = \{\bbz \; | \; f(\bbz) \leq f(\bbz^{(0)}) \}$ is compact as required by (\textit{\textbf{C3}}).
We start by noting that the entries of $\bbS$ are continuous subsets of $\reals$, (e.g., $S_{ij}\in\reals_+$ when $\bbS=\bbA$), and that $\bbH\in\reals^{N\times N}$, so $f(\bbz)$ is continuous.
Moreover, $f(\bbz)\leq f(\bbz^{(0)})$ implies that the functions $\|\bbY-\bbH\bbX\|_F^2$ and $\log(|S_{ij}|+\delta_1)$ are all bounded, rendering the domain of $f(\bbz)$ bounded. 
It follows then that the level set $\ccalZ^{(0)}$ is compact.

Finally, we need to prove that either \eqref{eq:step1_filterid} or \eqref{eq:step2_graph_denoising} has a unique solution, so that (\textit{\textbf{C4}}) is fulfilled. \cref{thm3} (see below) states that, under the two conditions required by \cref{thm1} (i.e., $\bbS$ does not have repeated eigenvalues, and the graph signals $\bbX$ excite every graph frequency), the solution to \eqref{eq:step1_filterid} is unique. This confirms that (\textit{\textbf{C4}}) is satisfied, concluding the proof. 

\begin{proposition}\label{thm3}
Let $\bbH\in\reals^{N\times N}$, $\bbS=\bbV\diag(\bblambda)\bbV^{-1}\in\reals^{N\times N}$, and $\bbX\in\reals^{N\times M}$ be the GF, the GSO, and the input signals in \eqref{eq:step1_filterid}. Then, \eqref{eq:step1_filterid} has a unique solution w.r.t. $\bbH$ if the following conditions are satisfied:
\begin{enumerate}
	\item $\lambda_i \neq \lambda_{i'}$, for all $i\neq i'$ and $(i,i')\in\{1,...,N\}^2$.
	\item Every row of $\tbX = \bbV^{-1} \bbX$ has at least one non-zero entry. 
\end{enumerate}
\end{proposition}

\begin{proof}
To simplify exposition, we focus first on the (most restrictive) setup of having only $M=1$ input-output pair. Defining $\hbh:=\vvec(\bbH)$, we can reformulate \eqref{eq:step1_filterid} as
\begin{equation}
	\text{min}_{\hbh \in \reals^{N^2}} \gamma\| (\bbI \otimes \bbS - \bbS^\top  \otimes \bbI) \hbh \|_2^2 + \| \bby - (\bbx^\top \otimes \bbI) \hbh \|_2^2,
	\label{eq:step1_filterid_vec}
\end{equation}
where lowercase symbols $\bby$ and $\bbx$ are used to emphasize that the output and input signals are a single $N$-dimensional vector. %
Upon defining $\bbD:= \bbI \otimes \bbS - \bbS^\top \otimes \bbI$, and $\bbE := \bbx^\top \otimes \bbI$, solving \eqref{eq:step1_filterid_vec} is equivalent to solving
\begin{equation}\label{eq:step1_filterid_single_ls_reformulation}
	\text{min}_{\hbh \in \reals^{N^2}} \Big\| \begin{bmatrix}\bbzero_{N^2} \\ \bby~\end{bmatrix} - \bbF \hbh \Big\|_2^2~\text{with}~\bbF:=\begin{bmatrix}\gamma \bbD \\ ~\bbE\end{bmatrix}
 \end{equation}
To prove that~\eqref{eq:step1_filterid_single_ls_reformulation} has a unique solution, it suffices to show that $\bbF$ is full column rank, i.e. $\nexists \; \bbn \in \reals^{N^2} $ such that $\bbF \bbn = \bbzero_{N+N^2}$. To show this, we first identify $\ccalN(\bbD)$, the null space of $\bbD$, and then show that $\bbE \bbn \neq \bbzero_N \; \forall \; \bbn \in \ccalN(\bbD)\setminus \{\bbzero_{N^2} \}$.

We start with the characterization of $\ccalN(\bbD)$. Given the Kronecker structure of $\bbD$, each of its $N^2$ eigenvalues has the form $\lambda_k - \lambda_{k'}$, with $(\bbV^{-1})^\top \otimes \bbV$ being the associated eigenvectors. Leveraging that $\lambda_i\neq \lambda_{i'}$ for $i\neq i'$, it follows that only when $i=i'$ the eigenvalue of $\bbD$ is zero. As a result, $\rank (\bbD) = N^2 - N$ and $\text{dim} (\ccalN(\bbD)) = N$. Equally important, the $N$ eigenvectors associated with the $N$ zero eigenvalues are given by $(\bbV^{-1})^\top \odot \bbV$, which, as a result, constitutes a basis spanning $\ccalN(\bbD)$.
More formally, we concluded that $\ccalN(\bbD) = \{ ((\bbV^{-1})^\top \odot \bbV) \bbtheta \;| \forallsymb \; \bbtheta \in \reals^{N}\}$.
Note that $\odot$ denotes the Khatri-Rao product.

Thus, to show that $\bbF$ in~\eqref{eq:step1_filterid_single_ls_reformulation} is full column rank we just need to prove that the only element $\bbn\in\ccalN(\bbD)$ that renders $\bbE \bbn = \bbzero_N$ is the all-zero vector $\bbzero_{N^2}$. To do so, we leverage the characterization of $\ccalN(\bbD)$ and write $\bbE \bbn$ as
\begin{align}
	\nonumber \bbE \bbn &= (\bbx^\top \otimes \bbI) ((\bbV^{-1})^\top \odot \bbV) \bbtheta % \neq \bbzero \in \reals^{N^2} \\
	= (\bbx^\top (\bbV^{-1})^\top \odot \bbV) \bbtheta \\% \neq \bbzero,
	\nonumber &= \bbV \text{diag}(\bbtheta) (\bbx^\top (\bbV^{-1})^\top)^\top %\neq \bbzero \in \reals^{N \times N} \\
	\nonumber = \bbV \text{diag}(\bbtheta)  \bbV^{-1} \bbx \\%\neq \bbzero \\
	&= \bbV \text{diag}(\bbtheta) \tbx%\neq \bbzero \\
	= \bbV (\bbtheta \circ \tbx), \label{eq:Dtimesn_cannotbezero}
\end{align}
where we used the property $(a \otimes b) (c \odot d) = a c \odot b d$, and $\circ$ denotes the entry-wise product.
Since $\bbV$ is invertible, the first and last terms in \eqref{eq:Dtimesn_cannotbezero} demonstrate that $\bbE \bbn=\bbzero_N$ requires $\bbtheta \circ \tbx=\bbzero_N$. However, condition 2) in \cref{thm3} states that $\tilde{x}_i \neq 0 ~\forall~ i$; hence, $\bbtheta \circ \tbx=\bbzero_N$ requires $\bbtheta=\bbzero_N$.  This implies that the only element in $\ccalN(\bbD)$ that renders $\bbE \bbn = \bbzero_N$ is $\bbn=(\bbV^{-1})^\top \odot \bbV) \bbzero_N =\bbzero_{N^2}$, concluding the proof.

The proof can be generalized for $M > 1$. In that case, the matrix $\bbE$ has size $MN\times N^2$ and the counterpart to \eqref{eq:Dtimesn_cannotbezero} establishes that having $\bbE \bbn=\bbzero$ requires $\vvec(\text{diag}(\bbtheta) \tbX)= \bbzero_{MN}$. Since \cref{thm3} assumes that each row of $\tbX$ has at least one nonzero entry, it follows that $\bbtheta=\bbzero_{MN}$, concluding the proof. 
%only the steps associated xxxx follows analogous steps, it is enough if every frequency $k$ is excited by at least one graph signal, as this ensures that the product $\text{diag}(\bbtheta) \tbX \neq \bbzero_{N \times M}$ for every $\bbtheta \neq \bbzero_N$, concluding the proof.
\end{proof}
\chapter{Robust network topology inference}\label{chap:nti_hidden}
The last perturbation considered in our robust framework is the presence of hidden nodes, a perturbation particularly relevant (and common) in the context of network topology inference.
Even though most of the existing works approach the problem of identifying the network topology based on the assumption that the whole vertex set is observed, oftentimes this scenario is not realistic.
In fact, in many relevant settings, the \emph{observed} data may correspond only to a subset of the nodes from the original graph while the rest of them remains unobserved or \emph{hidden}.
The existence of these hidden nodes entails a challenge for most of the existing methods, which require important adjustments to develop a robust alternative.
Although noticeable less than its ``complete network'' counterpart, topology inference of networks with hidden variables has attracted some attention, with examples including Gaussian graphical model selection \cite{chandrasekaran2012latent}, inference of linear Bayesian networks \cite{anandkumar2013learning}, nonlinear regression \cite{mei2018silvar}, or brain connectivity \cite{chang2019graphical}, to name a few.

Motivated by this, the primary goal of this chapter is to develop a joint network topology algorithm that is robust to the presence of hidden nodes and harnesses the similarity between the graphs being estimated to enhance the quality of the estimation.
In this chapter, the key assumption that enables learning the graph topology from a set of nodes is that the observed signals are observations from a random network process stationary on the unknown graphs.
Therefore, we first introduce how the presence of hidden nodes influences the graph stationarity assumption, and then, we analyze how to leverage the graph similarity between nodes that are not observed.
This work has been published in~\cite{buciulea2019network,rey2022joint}.

After a brief introduction presented in \cref{sec:introduction_nti_hidden}, the remaining chapter is organized as follows.
First, \cref{sec:hidden_variables_inference} provides a general overview of the structure resulting from the presence of hidden nodes.
Then, \cref{sec:stationary_inf} introduces the structure and problem formulation when a single graph is being learned, and \cref{sec:joint_stationary_inf} addresses the joint network topology inference form stationary observations.
%Then, Sections~\ref{sec:stationary_inf} and \ref{sec:joint_stationary_inf} respectively address the network topology inference with hidden nodes when the goal is estimating a single graph and several graphs.
Finally, \cref{sec:experiments_nti_hidden} evaluates the performance of the proposed algorithms, and \cref{sec:conclusion_nti_hidden} provides some concluding remarks.

\section{Introduction}\label{sec:introduction_nti_hidden}
When reviewing the literature addressing the problem of network topology inference, it holds that the standard approach entails learning the topology of a single graph when observations (measurements) from all the nodes in the network are available.
Nonetheless, in many relevant settings we only have access to observations from a subset of nodes, with the remaining nodes being unobserved or hidden.
The existence of these hidden nodes constitutes a relevant and challenging problem since closely related values from two observed nodes may be explained not only by an edge between the two nodes but by a third latent node connected to both of them.
Furthermore, many contemporary setups involve \emph{multiple related networks}, each of them with a subset of available signals.
This is the case, for example, in multi-hop communication networks in dynamic environments, in social networks where the same set of users may present different types of interactions, or in brain analytics where observations from different patients are available and the goal is to estimate their brain functional networks. 
When there exist several closely related networks, we can boost the performance of network topology inference by approaching the problem in a joint fashion that allows us to capture the relationship between the different graphs~\cite{murase2014multilayer,danaher2014joint,navarro2020joint,arroyo2021inference}.

% Summary of work
Based on the previous discussion, in this chapter we approach the problem of joint network topology inference with hidden nodes from stationary observations.
Although the assumption of graph stationarity has been successfully adopted in the context of the network-topology inference problem, a formulation robust to the presence of hidden variables is still missing.
To fill this gap, first we detail how hidden nodes impact the classical definition of graph stationarity and introduce the formulation of the single-network-recovery problem as a constrained optimization that accounts explicitly for the modified definitions.
Then, we propose a topology inference method that simultaneously performs joint estimation of \emph{multiple} graphs and accounts for the presence of hidden variables.
Then, to fully benefit from the joint inference formulation, a critical aspect is to capture the similarity among graphs not only accounting for the observed nodes but also for the hidden ones.
This is achieved by carefully exploiting the structure inherent to the presence of latent variables with a regularization inspired by group Lasso~\cite{simon2013sparse}.
The proposed method is evaluated using synthetic and real-world graphs and compared with other related approaches.

\vspace{3mm}
\noindent\textbf{Contributions.}
To summarize, the main contributions of this chapter are the following.
\begin{itemize}
    \item[(i)] We develop a joint network topology inference optimization-based framework that accounts for the presence of hidden variables and exploit the graph similarity on the whole graph, not only on the observed nodes.
    \item[(ii)] We quantify the performance of the proposed algorithms using numerical experiments. 
\end{itemize}

To introduce notation and facilitate exposition, we start by reviewing topology inference methods in the presence of hidden variables for partial correlation networks as well as their generalization for graph stationary signals. This review (including the generalization to stationary signals) was published in \cite{buciulea2022learning}. While being one of the authors of that paper, we emphasize that the contents of \cite{buciulea2022learning} are not included as a contribution of this Ph.D. thesis and, that, the contributions of this chapter are limited to those listed above as (i) and (ii).

\section{Topology inference model in the presence of hidden variables}\label{sec:hidden_variables_inference}
The first step towards formally stating the network topology inference from a robust perspective is to properly describe the structure resulting from the presence of hidden variables.
To that end, let us assume that the unknown graph $\ccalG$ is an undirected graph with $N$ nodes, consider that there is a random process associated with $\ccalG$, and denote by $\bbX := [\bbx_1,...,\bbx_M]\in\reals^{N \times M}$ a collection of $M$ independent realizations of such a process.
In this chapter, to model the presence of hidden nodes, we assume that we only have access to the subset of nodes $\ccalO \subseteq \ccalV$ with cardinality $O\leq N$.
Meanwhile, the remaining $H=N-O$ nodes collected in the subset $\ccalH=\ccalV\setminus\ccalO$ stay unobserved, and thus, only the entries of $\bbX$ associated with the subset $\ccalO$ are observed.
For the sake of simplicity and without loss of generality, we assume that the observed nodes correspond to the first $O$ nodes of the graph.
Hence, defining $\bbXo\in\reals^{O\times M}$ as the submatrix formed by the first $O$ rows of $\bbX$, it is clear that $\bbXo$ collects the values of the $M$ available signals at the $O$ observed nodes.
Lastly, if the $N\times N$ matrices $\bbS$ and $\bbC$ represent, respectively, the full GSO and covariance matrix of the random graph process, these matrices, as well as the signals $\bbX$, present the following block structure
\begin{equation}\label{eq:block_def_S_C}
    \bbX =
    \begin{bmatrix}
    \bbXo  \\
    \bbXh
    \end{bmatrix}, \quad
    \bbS \!=\! 
    \begin{bmatrix}
        \bbSo & \bbSoh  \\
        \bbSho & \bbSh
        \end{bmatrix}, \quad
        \bbC \!=\! 
        \begin{bmatrix}
        \bbCo & \bbCoh  \\
        \bbCho & \bbC_H
    \end{bmatrix}.
\end{equation}
Here, the submatrices $\bbSo$ and $\bbCo$, both of dimension $O \times O$, are associated with the observed variables.
The observed GSO $\bbSo$ describes the connections between the observed nodes while the remaining blocks collect connections involving hidden nodes.
Similarly, the observed covariance $\bbCo$ corresponds to the covariance of the random variables defined at the observed nodes.
Since the graph is undirected, both $\bbS$ and $\bbC$ are symmetric, and thus, $\bbSoh = \bbSho^\top$ and $\bbCoh = \bbCho^\top$.
Clearly, the sample covariance matrix $\hbC = \frac{1}{M}\bbX\bbX^\top$ present the same block structure as $\bbC$.

With the previous definitions in place, the problem of network topology inference in the presence of hidden nodes aims at estimating the submatrix $\bbSo$ from the observed graph signals $\bbXo$ while accounting for the presence of hidden nodes.
This is depicted in \cref{fig:nti_hidden}, where the detrimental effects of ignoring the hidden nodes is also illustrated.
However, despite having observations from $O$ nodes, there are still $H=N-O$ nodes that remain unseen and influence the observed signals $\bbXo$, rendering the inference problem challenging and severely ill-conditioned.
As a result, to ensure the tractability of the problem, we assume that the number of hidden nodes is substantially smaller than the number of observed nodes ($ O \lessapprox N$) and, more importantly, we consider that there exists a known property relating the full graph signals $\bbX$ to the full GSO $\bbS$.
The particular relationship is further developed in subsequent sections, where we assume that $\bbX$ is stationary (\cref{sec:stationary_inf} and \cref{sec:joint_stationary_inf}) on $\bbS$, giving rise to different network topology inference problem that are formally introduced later in the chapter.
Then, the key issue to address is how the relation between $\bbX$ and $\bbS$, which involves the full signals and GSO, translates to the submatrices $\bbXo$, $\bbSo$, and $\bbCo$ in \eqref{eq:block_def_S_C}.
Nonetheless, before discussing our specific solution, a relevant question that is addressed next is how classical topology-inference approaches handle the problem of latent nodal variables.

%%%%%%%%%%   FIGURE   %%%%%%%%%%
\begin{figure}[tb]
    \centering
    \includegraphics[width=.9\textwidth]{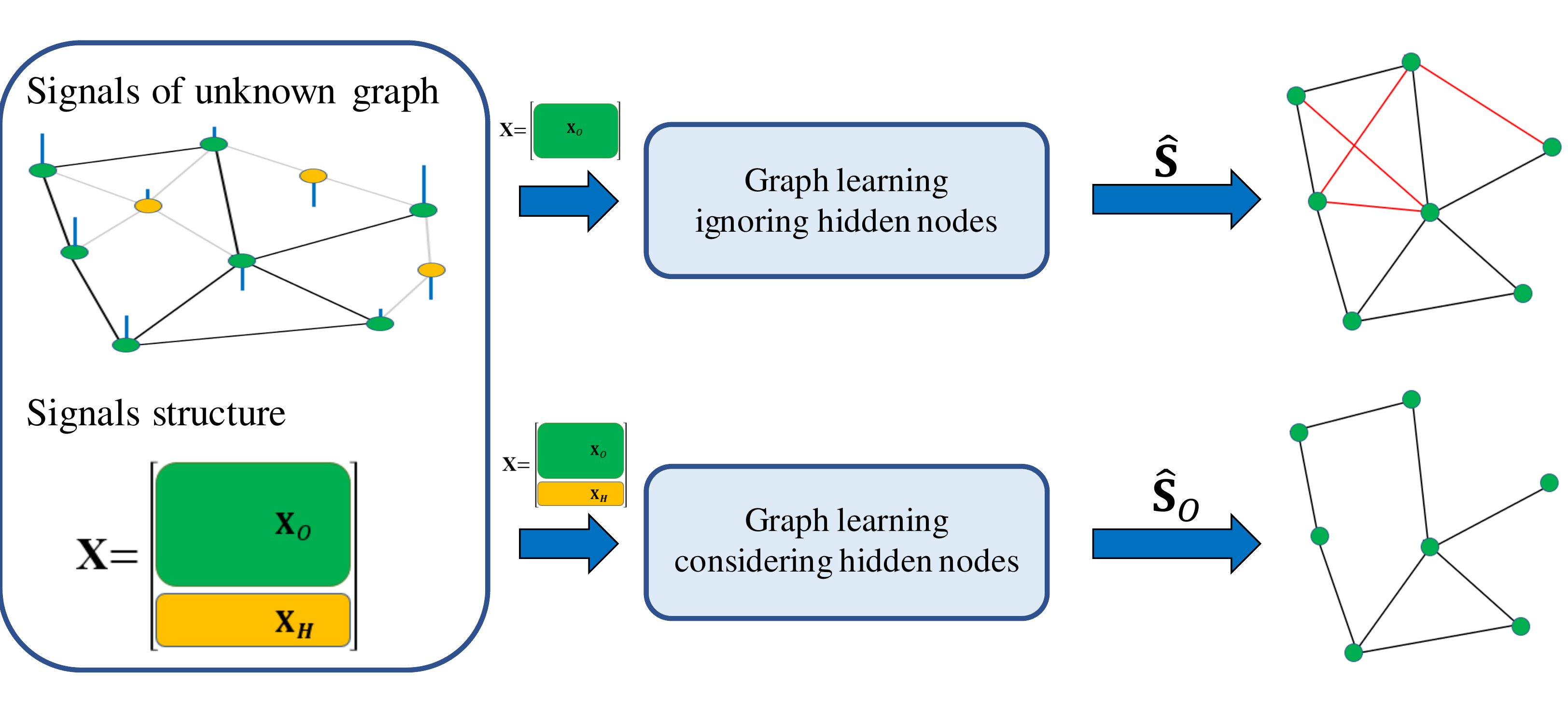}
    \caption{Depiction of the importance of modeling the influence of hidden variables. 
    The left box contains the true (unknown) graph with the nodes in $\ccalH$ represented in yellow, and the graph signals collected in $\bbX$, which is a block matrix as in \eqref{eq:block_def_S_C}.
    Then, the diagram on the top infers the graph assuming that observations from all nodes are available, i.e., assuming $\bbX = \bbXo$, and hence it wrongly estimates some edges (represented in red).
    On the other hand, the diagram below takes into account the whole matrix $\bbX$ even though only $\bbXo$ is observed, and thus, it takes into account the presence of hidden variables and provides an accurate estimation.}
    \label{fig:nti_hidden}
\end{figure}
%%%%%%%%%%%%%%%%%%%%%%%%%%%%%%%%

\subsection{Correlation and partial correlation networks with hidden variables}
A key question when addressing the network topology inference problem in the (more general) scenario of  hidden nodes is how to modify the existing formulations that deal with the full observable case (i.e. $\ccalO=\ccalV$).
The answer to that question is very different for the so-called direct methods (which consider that a link between $i$ and $j$ exists based only on correlation/similarity metrics between the signals observed at $i$ and $j$) and indirect methods (which consider that the link $(i,j)$ exists based on the global relations/dependencies among all nodes and, hence, depends on the full observation matrices).
Correlation networks, which basically assume that $\hbS$ corresponds to (a thresholded version) of $\bbC$, fall into the first category.
The generalization to setups with hidden nodes is indeed trivial and given by $\hbSo = \bbCo$, where only the direct influence between each pair of observed nodes is considered.
A relevant example within the second category are partial correlation methods (including, those for GMRF) which, in their simplest form, assume that $\bbS = \bbC^{-1}$.
Under this setting, the key to include hidden variables resides in noticing that, using the expression for recovering a block of the inverse of a matrix, we can write $\bbCo^{-1} = \bbSo - \bbQ$, with $\bbQ = \bbSoh\bbSh^{-1}\bbSoh$ being a low-rank matrix since $H \ll O$. Leveraging this, the authors in~\cite{chandrasekaran2012latent} modified the celebrated \acrfull{gl} algorithm to deal with hidden variables via an augmented maximum-likelihood estimator.
The resulting algorithm is known as \acrfull{lvgl} and is given by
\begin{alignat}{2}\label{eq:graphical_lasso}
    \!\!&\! \max_{\bbSo, \bbQ} && \;\;\;\; \log\det(\bbSo - \bbQ) -\tr(\hbCo (\bbSo - \bbQ)) - \lambda\|\bbSo\|_1  -\gamma\|\bbQ\|_*,     \\ 
    \!\!&\! \mathrm{\;\;s. \;t. } && \;\;\;\; \bbSo - \bbQ \succeq \bbzero,\;\; \bbQ \succeq \bbzero,  \;\; \nonumber
\end{alignat}
where $\hbCo$ represents the sample covariance estimate that can be obtained from the samples in $\bbXo$, the nuclear norm $\|\cdot\|_*$ promotes low-rank solutions, and $\lambda$ and $\gamma$ are (tunable) regularization constants. Clearly, if all the nodes are observed, we have that $\hbCo=\hbC$ and $\bbQ=\bb0$ and, if those are replaced in \eqref{eq:graphical_lasso}, the classical GL formulation is recovered.

Rather than assuming that the relation between $\bbX$ and $\bbS$ is given by either correlations or partial-correlations, this chapter looks at setups where the operating assumption is that the observed signals are stationary on the graph.
\cref{sec:stationary_inf} first introduces this setup, and then, \cref{sec:joint_stationary_inf} present a joint network topology inference approach to leverage the similarity of several related graphs. \cref{sec:experiments_nti_hidden} evaluates numerically the performance of the developed algorithms and compares it with that of classical correlation and LVGL schemes.

\section{Topology inference from stationary signals}\label{sec:stationary_inf}
Now, we describe a setup in which the signals $\bbX$ are assumed to be independent samples of a random network process stationary in $\bbS$. 
The resulting inference problem is formally stated next.
\begin{problem}\label{prob:stationary_problem}
    Given the matrix $\bbXo$ collecting the signal values at the observed nodes of the graph $\ccalG$, find the sparsest matrix $\bbSo$ encoding the structure of $\ccalG$ under the assumptions that: \\
    \textbf{(AS1)} The number of hidden nodes is substantially smaller than the number of observed nodes, i.e., $ H \ll O$. \\
    \textbf{(AS2)} The graph signals in $\bbX\in\reals^{N \times M}$ are independent realizations of a process that is stationary in $\bbS$.
\end{problem}

Intuitively, \textbf{(AS1)} promotes the feasibility of the inference problem by ensuring that a considerable amount of the nodes are observed.
More interestingly, assumption \textbf{(AS2)} implies that the GSO and the covariance of the graph process share the eigenvectors, which is tantamount to assuming that the mapping between $\bbS$ and $\bbC$ can be accurately represented by a matrix polynomial.
However, because the stationarity assumption involves the whole matrices $\bbS$ and $\bbC$ but only the observed submatrix $\hbCo$ is available, we need to generalize this mapping to include the effect of the hidden variables over the observed ones.
Key to that end is to note that the assumption of stationarity \textbf{(AS1)} implies that $\bbS$ and $\bbC$ are simultaneously diagonalizable. The formulation in~\cite{segarra2017network}, which is valid only if $\ccalO=\ccalV$, approached this by extracting the eigenvectors of the (sample) covariance and, then, formulated an optimization problem guaranteeing that the extracted eigenvectors and those of the GSO were the same.
Since obtaining the submatrix of the eigenvectors corresponding to the observed nodes is challenging, here we impose the graph stationarity by requiring the equality $\bbC\bbS=\bbS\bbC$ to hold \cite{segarra2017joint}.
Indeed, suppose that we focus on the upper left $O \times O$ block in both sides of the equality. Then, leveraging the expressions introduced in \eqref{eq:block_def_S_C}, we have that
\begin{equation}\label{eq:block_commutation}
\bbCo\bbSo+\bbCoh\bbSho=\bbSo\bbCo+\bbSoh\bbCho.
\end{equation}
Equation \eqref{eq:block_commutation} shows that, when hidden variables are present, we cannot simply ask $\bbSo$ and $\bbCo$ to commute, but we need also to account for the terms $\bbCoh\bbSho$ and $\bbSoh\bbCho$.
Since \textbf{(AS1)} states that $H \ll O$, a reasonable approach is to \emph{lift} the problem by defining the matrix $\bbP:=\bbCoh\bbSho \in\reals^{O\times O}$ and, then, exploit \textbf{(AS1)} to impose that $\rank(\bbP)\leq H \ll O$.
Moreover, since both $\bbS$ and $\bbC$ are symmetric matrices, it follows that $(\bbCoh\bbSho)^\top =\bbSoh\bbCho$ and, hence, $\bbP^\top=\bbSoh\bbCho$. Then, by making the general assumption that graphs are typically sparse, we can formulate \cref{prob:stationary_problem} as an optimization framework in which we attempt to find a sparse matrix $\bbSo$ and a low-rank matrix $\bbP$ that satisfy \textbf{(AS1)} and \textbf{(AS2)} by solving 
\begin{alignat}{2}\label{eq:eqn_zero_norm}
    \!\!&\!\min_{\bbSo, \bbP} \
    &&\|\bbSo\|_0                                    \\ 
    \!\!&\!\mathrm{\;\;s. \;t. } && \;\;\;\;\bbCo\bbSo+\bbP = \bbSo\bbCo+\bbP^\top, \;\; \nonumber
    \\
    \!\!&\! && 
    \;\;\;\; \rank(\bbP)\leq H, \;\; \nonumber
    \\
    \!\!&\! && 
    \;\;\;\; \bbSo \in \ccalS, \;\; \nonumber
\end{alignat}
where the equality constraint accounts for \textbf{(AS2)} and the second constraint for \textbf{(AS1)}, as explained before. Notice that the first constraint assumes perfect knowledge of the observed covariance matrix $\bbCo$, but this might not be attainable in practice, motivating the need for robust versions of this formulation as discussed in \cref{subsec:robust_inference}.
Finally, recall that the set $\ccalS$ specify additional properties that $\bbSo$ must satisfy.
In the remainder of this section we will focus on the case where the GSO is an adjacency matrix, so we consider the set $\ccalS = \ccalS_\bbA$. 

% Need of convex approach
Even though \eqref{eq:eqn_zero_norm_joint} indeed solves \cref{prob:stationary_problem}, it is a non-convex optimization problem, and hence, challenging to solve, motivating the development of convex approximations.

\subsection{Topology inference with stationary observations as a convex optimization}\label{subsec:convex_formulation}
It is important to notice that the presence of the rank constraint and the  $\ell_0$ norm in \eqref{eq:eqn_zero_norm} render the problem non-convex and computationally hard to solve, motivating the need for convex relaxations. 
To achieve this, instead of enforcing a rank constraint, we will augment the objective with the \emph{nuclear norm} as its convex surrogate. Similarly, the $\ell_0$ norm can be replaced with the $\ell_1$ norm, its closest convex approximation.
After applying these relaxations to \eqref{eq:eqn_zero_norm} we obtain the following \emph{convex} optimization problem
\begin{alignat}{2}\label{eq:eqn_one_norm} 
    \!\!&\!\min_{\bbSo, \bbP} \
    &&\|\bbSo\|_1 + \gamma\|\bbP\|_*   \\ 
    \!\!&\!\mathrm{\;\;s. \;t. } && \bbCo\bbSo+\bbP = \bbSo\bbCo+\bbP^\top,  \;\;
    \bbSo \in \ccalS_\bbA, \nonumber
\end{alignat} 
where $\gamma$ is a regularization constant encoding the relative importance of the low-rank vs the sparsity promoting term.

% norm1 rew
Even though replacing the original $\ell_0$ norm with the convex $\ell_1$ norm constitutes a common approach, it is well-known that non-convex surrogates can lead to sparser solutions.
Indeed, a more sophisticated alternative in the context of sparse recovery is to define $\delta$ as a small positive number and replace the $\ell_0$ norm with a (non-convex) logarithmic penalty~\cite{candes2008enhancing} as
\begin{equation}
\|\bbSo\|_0 \approx \sum_{i,j=0}^{O}\log(|S_{ij}|+\epsilon_0)
\end{equation}
with $\epsilon_0$ being a small positive constant. 
However, because the logarithm is a concave function it renders the optimization non-convex.
An efficient way to handle this issue consists on relying on a MM approach \cite{sun2016majorization}, which considers an iterative linear approximation to the concave objective and leads to an \emph{iterative} re-weighted $\ell_1$ minimization.
For the problem at hand, and with $t=1,...,T$ being the iteration index, the resulting optimization is 
\begin{alignat}{2}\label{eq:eqn_one_norm_rew}
    \!\!&\! \bbSo^{(t+1)} := \argmin_{\bbSo, \bbP} && \sum_{i,j=1}^O W^{(t)}_{ij}[\bbSo]_{ij} + \gamma\|\bbP\|_*                                    \\ 
    \!\!&\!\mathrm{\;\;s. \;t. } && \bbCo\bbSo+\bbP = \bbSo\bbCo+\bbP^\top, \nonumber \\
     \!\!&\!  && \bbSo \in \ccalS_\bbA, \nonumber
\end{alignat}
with $W_{ij}^{(t)}$ being defined as $W_{ij}^{(t)} = (S_{ij}^{(t-1)}+\epsilon_0)^{-1}$.
Since the iterative algorithm penalizes (assigns a larger weight to) entries of $\bbSo$ that are close to zero, the obtained solution is typically sparser at the expense of a higher computational cost. Finally, note that the absolute values can be removed whenever the constraint $[\bbSo ]_{ij}\geq 0$ is enforced.

\subsection{Robust network inference}\label{subsec:robust_inference}
In most scenarios the covariance matrix  $\bbCo$ is not known perfectly but, instead, only an estimate $\hbCo$ is available. 
This is indeed the case for the setting described in \cref{prob:stationary_problem}, where we only have access to a finite number $M$ of graph signals and we estimate the covariance as $\hbCo=M^{-1}(\bbXo\bbXo\top)$. While different choices to accommodate the discrepancies between the true covariance and its sampled version exist, here we simply chose to relax the commutative constraint in \eqref{eq:eqn_one_norm}, rewriting it using a Frobenius norm
\begin{alignat}{2}\label{eq:eqn_robust_norm} 
\!\!&\!\min_{\bbSo, \bbP} \
&&\|\bbSo\|_1 + \gamma\|\bbP\|_*                                    \\ 
\!\!&\!\mathrm{\;\;s. \;t. } && \|\hbCo\bbSo+\bbP-\bbSo\hbCo-\bbP^\top\|_F\leq\epsilon,    \;\;
\bbSo \in \ccalS_\bbA. \nonumber
\end{alignat} 
The value of the non-negative constant $\epsilon$ should be selected based on prior knowledge on the noise level present in the observations and, more importantly, the number of samples $M$ used to estimate the covariance.
It must be large enough to ensure the problem is feasible but not too large so the constraint is not active.
In the limit, when an infinite number of realizations is available, then $\bbC = \hbC$ so we can set $\epsilon=0$ and \eqref{eq:eqn_robust_norm} boils down to (\ref{eq:eqn_one_norm}).
The same update on the constraint can be applied to \eqref{eq:eqn_one_norm_rew} to account for imperfect knowledge of the covariance.

\section{Joint inference from stationary signals in the presence of hidden variables}\label{sec:joint_stationary_inf}
The last network topology inference task that we will be analyzing from a robust perspective deals with the estimation of multiple related graphs.
The driving idea is to exploit the existing relation between the different graphs in a joint optimization framework to enhance the quality of the estimated networks. 

% Introduce the notation
To formally introduce the problem of joint graph topology inference in the presence of hidden variables, let us assume that $K$ undirected graphs $\{\ccalG^{(k)}\}_{k=1}^K$ are defined over the same set of nodes $\ccalV$, and denote as $\bbX^{(k)}=[\bbx_1^{(k)},...,\bbx_{M_k}^{(k)}]\in\reals^{N \times M_k}$ the collection of (zero-mean) $M_k$ signals defined on top of each unknown graph $\ccalG^{(k)}$.
Furthermore, in analogy to previous sections, consider that for each graph only $O$ nodes contained in the subset $\ccalO$ are observed while the remaining $H$ node in the set $\ccalH$ remain hidden.
Without loss of generality, let the signals associated with the observed nodes be collected in the first $O$ rows of $\bbX^{(k)}$ and denote them as $\bbXo^{(k)} \in \reals^{O \times M_k}$.
Then, the distinction between observed and hidden nodes renders each matrix $\bbS^{(k)}$ and $\hbC^{(k)}$ with the characteristic block structure introduced in \eqref{eq:block_def_S_C}.
Also note that, since the $K$ graphs are undirected, the matrices $\bbS^{(k)}$ and $\hbC^{(k)}$ are symmetric.

% Problem formulation
With these considerations in place, the problem of joint topology inference in the presence of hidden variables is described next.
\begin{problem}\label{prob:joint_stationary_inf}
    Given the $O \times M_k$ matrices $\{\bbXo^{(k)}\}_{k=1}^K$ collecting the signal values at the observed nodes for each graph $\ccalG^{(k)}$, find the sparsest matrices $\{\bbSo^{(k)}\}_{k=1}^K$ encoding the structure of the $K$ graphs under the assumptions that: \\
    \textbf{(AS1)} The number of hidden nodes is much smaller than the number of observed nodes, i.e., $ H \ll O$. \\
    \textbf{(AS2)} The signals $\bbX^{(k)}$ are realizations of a random process that is stationary in $\bbS^{(k)}$. \\
    \textbf{(AS3)} The distance between the $K$ graphs is small according to a particular metric $d(\bbS^{(k)},\bbS^{(k')})$.
\end{problem}

% Explain the problem
Accounting for the hidden variables implies modeling their influence over the observed nodes without any additional observation, thus rendering the inference problem a challenge.
To ensure the tractability of the problem, \textbf{(AS1)} guarantees that most of the nodes are observed while \textbf{(AS2)} establishes a relation between the graph signals and the whole unknown graph via graph stationarity.
Then, \textbf{(AS3)} guarantees that the $K$ graphs are similar so we can benefit from inferring them in a joint setting.
The key question now is how to exploit the graph similarity from edges involving observed nodes but also from edges involving hidden nodes.  

% Introduce next section
In the following section, we exploit the aforementioned assumptions and the block structure resulting from the presence of hidden variables to approach \cref{prob:joint_stationary_inf} by solving a convex optimization problem.

\subsection{Modeling hidden variables in the joint inference problem}
The first step towards formulating an optimization problem that solves \cref{prob:joint_stationary_inf} is to account for the presence of the hidden nodes in the stationary assumption \textbf{(AS2)}.
This can be achieved by following the same reasoning as in \cref{sec:stationary_inf}.
Then, for each graph $k$, we define the associated $O \times O$ matrix $\bbP^{(k)}:=\bbCoh^{(k)}(\bbSoh^{(k)})^\top$, which combined with the commutativity of $\bbC^{(k)}$ and $\bbS^{(k)}$ that stems from the stationarity assumption results in
\begin{equation}\label{eq:block_commutativity}
    \bbCo^{(k)}\bbSo^{(k)}+\bbP^{(k)}=\bbSo^{(k)}\bbCo^{(k)}+(\bbP^{(k)})^\top.
\end{equation}
% low rank of P
Also note that the matrices $\bbP^{(k)}$ are the product of two matrices of sizes $O \times H$ and $H \times O$ so due to \textbf{(AS1)} it follows that the rank of $\bbP^{(k)}$ is upper bounded by $H$.

% Non-convex formulation
With the previous considerations in place, we approach the sparse joint topology inference problem in the presence of hidden nodes by means of the following non-convex optimization problem
\begin{alignat}{2}\label{eq:eqn_zero_norm_joint}
    \!\!&\!\min_{\{\bbSo^{(k)},\bbP^{(k)}\}_{k=1}^K} \
    &&\sum_k \alpha_k \|\bbSo^{(k)}\|_0 + \sum_{k<k'}\beta_{k,k'}d_S(\bbSo^{(k)},\bbSo^{(k')})                                   \\ 
    \!\!&\! && + \sum_{k<k'}\eta_{k,k'}d_P(\bbP^{(k)},\bbP^{(k')}) \nonumber \\
    \!\!&\!\mathrm{\;\;s. \;t. } && \;\;\;\;\rank(\bbP^{(k)})\leq H, \;\; \nonumber
    \\
    \!\!&\! && 
    \;\;\;\; \|\hbCo^{(k)}\bbSo^{(k)}\!+\!\bbP^{(k)} \!-\!\bbSo^{(k)}\hbCo^{(k)}\!-\!(\bbP^{(k)})^{T}\|_F^2 \leq  \epsilon, \;\; \nonumber
    \\
    \!\!&\! && 
    \;\;\;\; \bbSo^{(k)} \in \ccalS. \;\; \nonumber
\end{alignat}
The first and second constraints capture assumptions \textbf{(AS1)} and \textbf{(AS2)}, with $\epsilon$ being a small positive number capturing the fidelity of the sample covariance. The set $\ccalS$ encodes the properties of the desired GSOs.
In this section we will focus on the case where the GSO is given by the adjacency matrix of the underlying undirected graph with non-negative weights and no self-loops.
Thus, from now onwards we set the feasibility set $\ccalS = \ccalS_\bbA$, with $\ccalS_\bbA$ as defined in \eqref{eq:A_set}.
Other GSOs such as the normalized Laplacian can be accommodated via minor adaptations to $\ccalS$; see~\cite{segarra2017network}.

% Effect of the joint inference in the hidden variables
Similar to standard joint inference approaches~\cite{navarro2020joint}, the objective function of \eqref{eq:eqn_zero_norm_joint} captures the similarity of the $K$ graphs with the function $d_S(\cdot,\cdot)$.
Nevertheless, when accounting for the presence of hidden variables, assumption \textbf{(AS3)} is also reflected in the unobserved blocks of the GSOs.
This important observation, captured by the function $d_P(\cdot,\cdot)$, allows us to incorporate additional structure reducing the degrees of freedom and rendering the problem more manageable.
More specifically, note that the matrix $\bbP^{(k)}$ is given by the product of $\bbCoh^{(k)}$ and $(\bbSoh^{(k)})^\top$ with the latter being a submatrix of a sparse  GSO, so it can be seen that the matrices $\bbP^{(k)}$ present a column-sparse structure.
Furthermore, since the $K$ graphs are similar, the submatrices $\bbSoh^{(k)}$ are also similar, which implies that the matrices $\bbP^{(k)}$ present a similar column-sparsity pattern.
In other words, the columns with non-zero entries are likely to be placed in the same positions for the different matrices $\bbP^{(k)}$.
By designing a distance function $d_P(\cdot,\cdot)$ that exploits this additional structure we improve the estimation of the matrices $\bbP^{(k)}$, resulting in a better estimation of the matrices $\bbSo^{(k)}$. 

% Intro next section
The non-convexity of \eqref{eq:eqn_zero_norm_joint}, which arises from the presence of the rank constraint and the $\ell_0$ norm, renders the optimization problem computationally hard to solve, leading us to implement some convex relaxations that are detailed next.

\subsection{Convex relaxations for the joint topology inference}
% Convex relaxation --> group lasso and reweighted
The rank constraints are commonly avoided by augmenting the objective function with a nuclear norm penalty, which promotes low-rank solutions by seeking matrices with sparse singular values. 
However, this penalty does not preserve the characteristic column sparsity of the matrices $\bbP^{(k)}$.
To circumvent this issue, in contrast to~\cite{buciulea2022learning}, we employ the group Lasso regularization~\cite{simon2013sparse} and rely on the fact that, in this particular setting, we can promote low rankness by reducing the number of non-zero columns while still achieving a reliable estimate.
Then, we replace the $\ell_0$ norm by a reweighted $\ell_1$ minimization~\cite{candes2008enhancing}, an iterative algorithm rooted on a logarithmic penalty that: i)~converges to a stationary point~\cite{hong2015unified}; and ii)~usually outperforms the widely used $\ell_1$ norm.

% Our algorithm
%% - Comment on convergence of the algorithm?
By leveraging the aforementioned relaxations we address the joint topology inference problem in the presence of hidden variables by solving an iterative method.
Under this approach, for each iteration, we solve the following convex problem
\begin{alignat}{2}\label{eq:eqn_convex}
    \!\!&\!  \min_{\scriptscriptstyle\{\bbSo^{(k)},\bbP^{(k)}\}_{k=1}^K} &&
    \;\sum_k \!\alpha_k \mathrm{vec}(\bbW^{(k)})^\top\!\mathrm{vec}(\bbSo^{(k)})
    \\
    \!\!&\!  && \!\!\!+\! \sum_{k<k'}\!\beta_{k,k'}\|\bbSo^{(k)}-\bbSo^{(k')}\|_1 \! \nonumber
    \\ 
    \!\!&\!  && \!\!\!+\!  \sum_k\!\gamma_k\|\bbP^{(k)}\|_{2,1} +\! \sum_{k<k'}\eta_{k,k'} \norm{\begin{bmatrix}
        \bbP^{(k)} \\
        \bbP^{(k')}
    \end{bmatrix}}_{2,1} \; \nonumber
    \\
    \!\!&\!  && \!\!\!+\! \sum_k\!\mu_k \|\! \hbCo^{(k)}\bbSo^{(k)} \!+\! \bbP^{(k)} \!-\! \bbSo^{(k)}\hbCo^{(k)} \!-\! (\bbP^{(k)})^{\top} \!\|_F^2 \nonumber
    \\
    \!\!&\!\mathrm{\;\;s. \;t. } && \bbSo^{(k)} \in \ccalS. \;\; \nonumber
\end{alignat}
To compute the weight matrices $\bbW^{(k)}$, let $t=1...T$ denote the iteration index (omitted in the expression above to alleviate the notation), and compute the $k$-th weight matrix for the $t$-th iteration as $W^{(k,t)}_{ij}=(S_{\scriptscriptstyle \ccalO_{ij}}^{(k,t-1)}+\delta)^{-1}$ with $S_{\scriptscriptstyle \ccalO_{ij}}^{(k,t-1)}$ being the solution obtained during the $(t-1)$-th iteration and $\delta$ a small positive constant. 
Hence, for each iteration $t$ we first compute the weight matrices $\bbW^{(k,t)}$ and, then, employ those to estimate the matrices $\bbSo^{(k,t)}$ and $\bbP^{(k,t)}$.
Coming back to the formulation in \eqref{eq:eqn_convex}, note that the distance $d_S(\cdot,\cdot)$ is set to the $\ell_1$ norm to promote similar edges on the $K$ graphs.
The norm $\|\cdot\|_{2,1}$ represents the group Lasso penalty by first computing the $\ell_2$ norm of the columns of the input matrix and then the $\ell_1$ norm of the resulting vector.
To capture the similar column-sparsity pattern of the matrices $\bbP^{(k)}$ resulting from the similarity of the $K$ graphs, we design the function $d_P(\cdot,\cdot)$ relying on the group Lasso penalty.
More specifically, we concatenate each pair of matrices $\bbP^{(k)}$ and $\bbP^{(k')}$ to create a tall matrix and then promote column sparsity on the tall matrix with the $\ell_{2,1}$ norm.
Note that a column of all zeros in the tall matrix implies that the same column in $\bbP^{(k)}$ and $\bbP^{(k')}$ will only contain zeros, thus promoting the desired structure.
Finally, it is worth noting that we moved the commutativity constraint to the objective function.
Due to the iterative nature of the proposed method, the estimation of the observed GSO during the first iteration might be far from the true GSO, and hence, a more restrictive constraint as the one employed in \eqref{eq:eqn_zero_norm_joint} might be misleading.
Augmenting the objective function with the commutativity penalty is more amenable to an iterative approach.

\section{Numerical experiments}\label{sec:experiments_nti_hidden}
In this section we asses numerically the performance of the proposed algorithms, comparing them with different baselines, including the LVGL counterpart presented in \eqref{eq:graphical_lasso}.
To compare the different scenarios raised in this chapter, first we consider the setting where a single graph is learned from stationary observations (\cref{sec:stationary_inf}), and then, we evaluate the performance of the proposed joint network topology algorithm (\cref{sec:joint_stationary_inf}).

\subsection{Numerical experiments based on joint inference}
Lastly, we evaluate the performance of the joint network topology inference algorithms developed in \cref{sec:joint_stationary_inf}, which are evaluated over synthetic and real-world graphs.
The error of the recovered graphs is computed as
\begin{equation}\label{eq:err_joint_inf}
    \sum_{k=1}^K\|\bbSo^{(k)}-\hbSo^{(k)}\|_F^2/K\|\bbSo^{(k)}\|_F^2,    
\end{equation}
and when the graphs are randomly generated, they are sampled from an ER model with $N=20$ nodes and edge probability $p=0.2$.
The code for the following experiments is available on
GitHub\footnote{{\scriptsize \url{https://github.com/reysam93/hidden_joint_inference/tree/ICASSP2022}}} and the interested reader is referred there for specific implementation details.

\vspace{3mm}
\noindent\textbf{Test case 1.}
We evaluate the influence of the hidden variables and its detrimental effect on the topology inference task
when the true covariance matrix is known.
The results are depicted in \cref{fig:main_fig}a, where we report the error of the recovered graphs, computed according to \eqref{eq:err_joint_inf}, for several models as the number of hidden variables increases on the x-axis.
The error is averaged over 64 realizations with ER graphs.
The considered models are: (i)~``PGL'', which stands for the method introduced in \eqref{eq:eqn_convex}; (ii)~``PNN'', which denotes the reweighed algorithm proposed in~\cite{buciulea2019network} augmented with the joint penalty $d_S(\cdot,\cdot)$ to perform the joint inference; and (iii)~``No hidden'', which is a joint inference method unaware of the presence of hidden variables similar to the work in~\cite{navarro2020joint}.
In addition, for each model we let $K$ take the values in $\{3,6\}$.
Looking at the results, we can observe that ``PGL'' and ``PNN'', which take into account the presence of hidden variables, outperform the method ``No hidden'', showcasing the benefit of a robust formulation.
Also, the method proposed in \eqref{eq:eqn_convex} outperforms ``PNN'', the other alternative accounting for hidden variables.
This reflects the advantage of employing the group Lasso regularization and incorporating the graph similarity through the careful design of the function $d_P(\cdot,\cdot)$. 
Lastly, it is worth noting that  the performance improves for higher values of $K$, achieving better results when more related graphs are available.

%%%%%%%%%   FIGURES   %%%%%%%%%%%%%%%
\begin{figure*}[!t]
	\centering
	\begin{subfigure}{0.49\textwidth}
		\centering
		\includegraphics[width=1.1\textwidth]{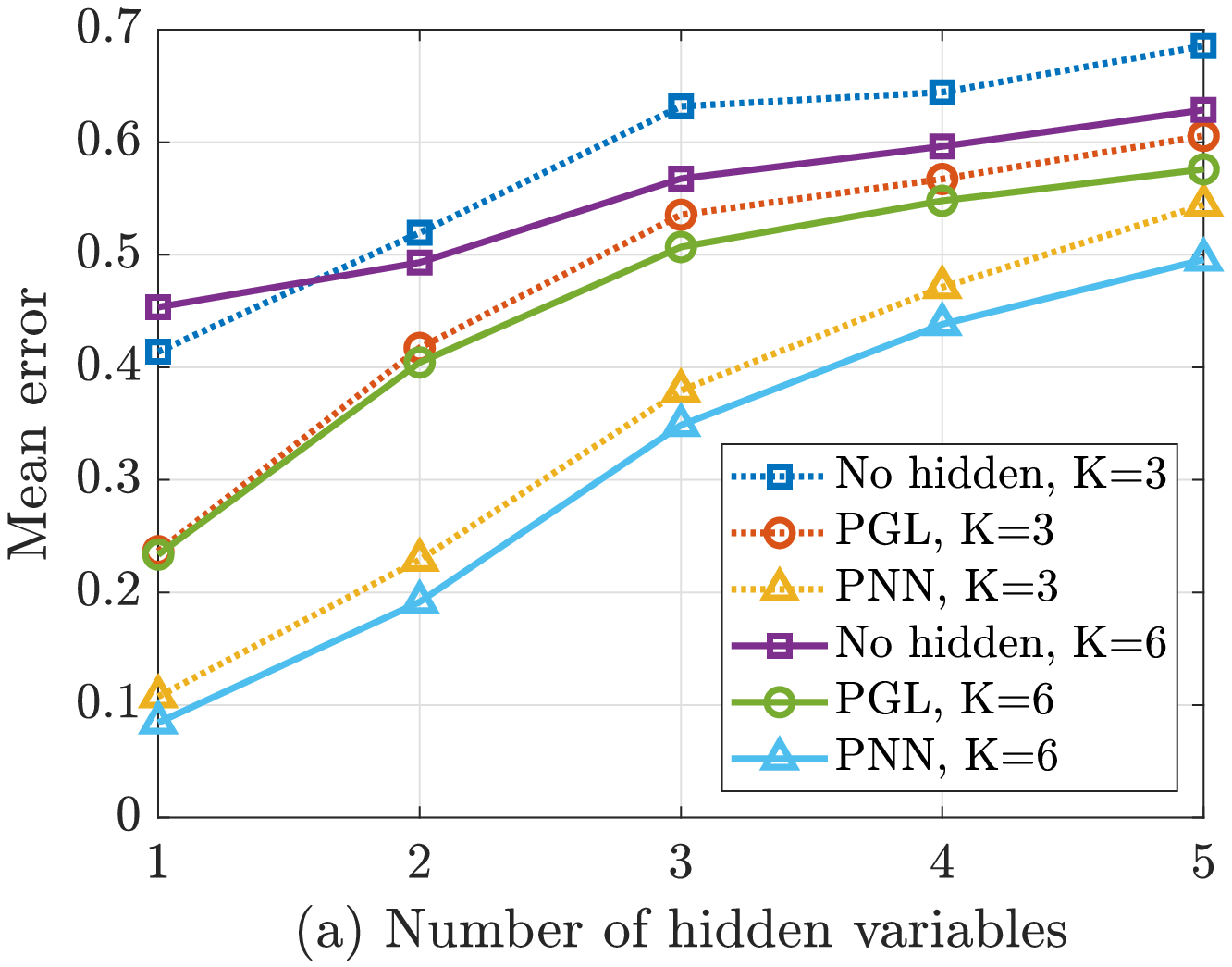}
		\label{fig:main_fig_a}
	\end{subfigure}
	\begin{subfigure}{0.49\textwidth}
		\centering
		\includegraphics[width=1.1\textwidth]{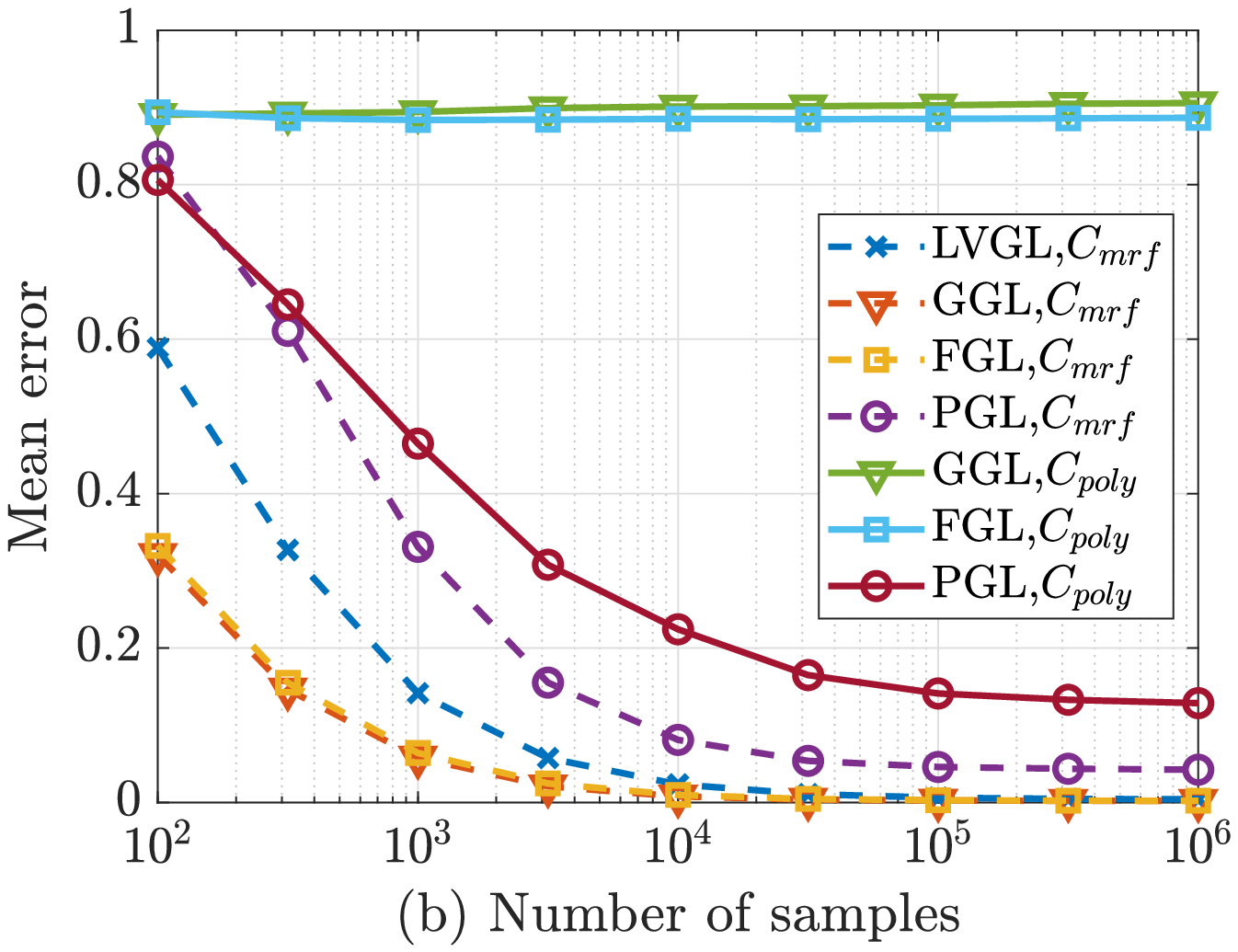}
	    \label{fig:main_fig_b}
	\end{subfigure}
	\begin{subfigure}{0.49\textwidth}
		\centering
		\includegraphics[width=1.1\textwidth]{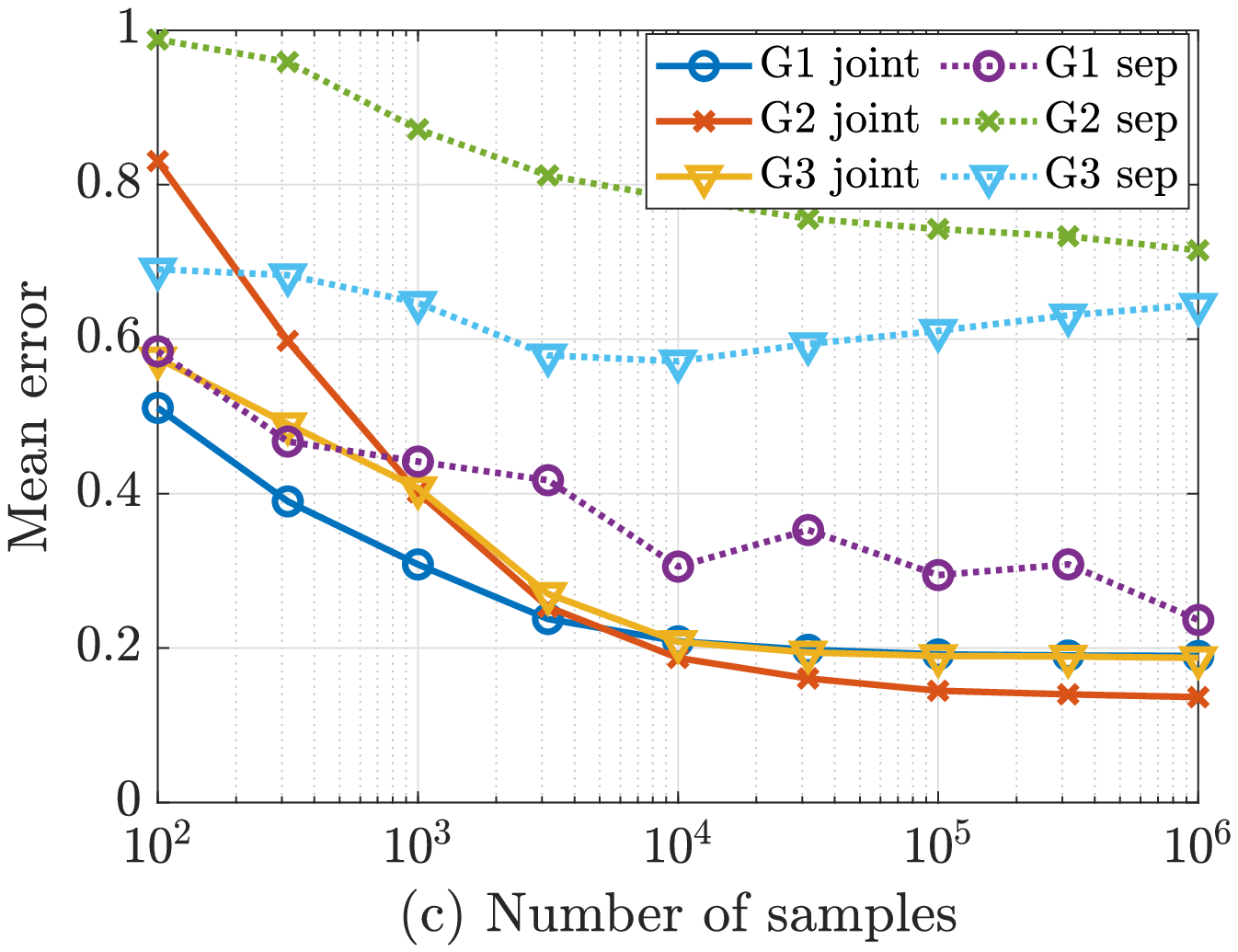}
	    \label{fig:main_fig_c}
	\end{subfigure}
	\caption{Numerical validation of the proposed algorithm. 
	a)~Mean error of 100 realizations as the number of hidden variables increases for different models and values of $K$.
	b)~Mean error of the recovered graphs for several algorithms as the number of samples increases. 
	c)~Mean error of the recovered graphs for joint and separate approaches as the number of samples increases. 
	The first two experiments use ER graphs with $N=20$ and $p=0.2$ and the third one employs real-world graphs.} \label{fig:main_fig}
\end{figure*}
%%%%%%%%%%%%%%%%%%%%%%%%%%%%%%%%%%%%%

\vspace{3mm}
\noindent\textbf{Test case 2.}
Next, we evaluate the influence of the number of observed signals and compare the performance of the proposed approach with other related alternatives.
In this experiment, only a single hidden node is considered.
To that end, in \cref{fig:main_fig}b we show the mean normalized error of the recovered graphs on the y-axis as the number of samples increases on the x-axis.
The error is computed as in the previous experiment and the mean is considered over 30 realizations of $K=3$ ER graphs with 10 realizations of random covariance matrices for each, resulting in a total of 300 realizations.
We compare the proposed model (``PGL'') with latent variable graphical Lasso (``LVGL'')\cite{chandrasekaran2012latent}, and with group and fusion graphical Lasso (``GGL'' and ``FGL''), both from \cite{danaher2014joint}.
For each model, signals are generated using two different types of covariance matrices: (i) $\bbC_{MRF}=(\sigma\bbI+\phi\bbS)^{-1}$ where $\phi$ is a positive random number and $\sigma$ is a positive number so that  $\bbC_{MRF}^{-1}$ is positive semi-definite; and $\bbC_{poly}=\bbH^2$ where the matrix $\bbH$ is a graph filter with random coefficients $\bbh$.
By looking at the \cref{fig:main_fig}b, it can be observed that, when $\bbC_{MRF}$ is employed, the graphical Lasso models slightly outperform the proposed approach. 
This is expected since they are tailored for this specific type of covariance matrices.
However, we can also see that the performance of the proposed algorithm is close to that of the alternatives, illustrating the benefits of considering both the joint optimization and the presence of hidden variables.
On the other hand, when we focus on the covariance matrices $\bbC_{poly}$, it is evident that the proposed method ``PGL'' clearly outperforms the alternatives, demonstrating that the proposed method is based on more lenient assumptions.
Note that the results for ``LVGL'' for the polynomial covariance are not included since the error was too high.

\vspace{3mm}
\noindent\textbf{Test case 3.}
Finally, we test the proposed algorithm and the impact of performing the topology inference in a joint fashion using real-world graphs.
We employ three graphs defined on a common set of 32 nodes.
Nodes represent students from the University of Ljubljana and the different networks encode different types of interactions among the students\footnote{\label{fn:students_dataset}The original data can be found at {\scriptsize \url{http://vladowiki.fmf.uni-lj.si/doku.php?id=pajek:data:pajek:students}}}.
The error is computed as detailed in  and one hidden variable is considered.
The results, illustrated in \cref{fig:main_fig}c, show the error of the recovered graphs as the number of samples increases.
The displayed error is the mean of 30 realizations of random stationary graph signals and only one hidden variable is considered.
Also, for each of the three graphs we include the performance of both the joint and the separate estimation.
It can be observed that the recovery of the three graphs improves when a joint approach is followed, showcasing the benefits of exploiting the existing relationship between the different networks.
Furthermore, this experiment confirms that the developed method is also suitable for real applications.

\section{Conclusion}\label{sec:conclusion_nti_hidden}
In this chapter, we faced the challenging problem of joint graph topology inference in the presence of hidden nodes.
First, we detailed the simpler case where the goal was to learn a single graph from stationary observations, and then, we presented a new method for approaching the joint graph topology inference in the presence of hidden nodes.
To tackle this ill-posed inference problem, we assume that (i) the number of hidden nodes $H$ is much smaller than the number of observed nodes $O$; (ii) the observed signals are realizations from a random process stationary in $\bbS^{(k)}$; and (iii) the $K$ graphs are closely related.
Furthermore, we exploit the inherent block structure of the matrices $\bbC^{(k)}$ and $\bbS^{(k)}$ to solve the joint topology inference problem by solving an optimization framework.
A reweighted $\ell_1$ norm to promote sparse solutions is employed, and the stationarity assumption is adapted to the presence of hidden nodes by defining the (unknown) low-rank lifting matrices $\bbP^{(k)}$.
Instead of relying on the nuclear norm, low-rank matrices $\bbP^{(k)}$ are achieved by promoting column-sparsity with the group Lasso penalty.
Moreover, the similarity of the $K$ graphs is leveraged in two ways.
First, for each pair of graphs, we look for matrices $\bbSo^{(k)}$ with a similar edge pattern by minimizing the $\ell_1$ penalty, and second, we look for matrices $\bbP^{(k)}$ with a similar column sparsity pattern.
The proposed method is evaluated using synthetic and real-world graphs, and a comparison with other baseline methods based on graph stationarity and on graphical Lasso is provided.

\chapter{Concluding remarks}\label{chap:conclusions}
The encompassing objective of this thesis was to contribute to building the foundations of a robust paradigm to address classical problems within GSP while modeling the detrimental influence of perturbations in the observed data.
To that end, we considered several types of perturbations that were classified into two broad classes: (i) perturbations in the graph signals; and (ii) perturbations in the topology of the graph.
First, we dealt with imperfections in the observed signals, which often lead to tractable problems and have been more studied in related literature.
Since incorporating the influence of imperfections in the signals in GSP tasks is well studied, when the magnitude of the perturbations is small with respect to the original signal (and especially when the perturbations capture the presence of noise), we addressed settings where the magnitude of the perturbations was large.
Furthermore, the focus in \cref{chap:denoising} and \cref{chap:interpolation} was on using processing schemes that had not received too much attention in the literature, trying to characterize their estimation performance (including those based on GNN architectures).
In the second part of the thesis,  we dealt with imperfections in the topology of the graph, which give rise to more challenging problems and have been significantly less studied in the literature.
In \cref{chap:robust_filter_id}, the focus was on estimating a graph filter but we also recovered an enhanced estimate of the unperturbed topology as a byproduct. Finally, \cref{chap:nti_hidden} tackled directly the problem of learning a graph, but for the challenging setup of network tomography, where hidden (unobservable) nodes are present. Finally, in each chapter we run an extensive battery of experiments to gain additional insights, test the robustness of the methods, compare them to existing state-of-the-art alternatives, and assess their potential applicability in real-world problems.  

The remainder of the chapter is devoted to describing the degree of fulfillment of the objectives listed in \cref{sec:objectives} (\cref{sec:conclusions_objectives}), and propose future lines of work (\cref{sec:future_work}).

\section{Revisiting the proposed goals}\label{sec:conclusions_objectives}
% O1: denoising
First, the objective \textbf{(O1)} considered that the observed graph signals were corrupted with noise and the goal was to develop non-linear algorithms to separate the noise from the signal.
This was addressed in \cref{chap:denoising}. Since linear denoising schemes had been investigated extensively, we designed two (untrained) graph-aware NNs that incorporated the information encoded in the GSO through different strategies.
The GCG relied on fixed (non-learnable) GFs to model convolutions in the vertex domain while the GDec employed a nested collection of graph upsampling operators to progressively increase the input size, limiting the degrees of freedom of the architecture, and providing more robustness to noise.
Furthermore, we provided theoretical guarantees on the denoising performance of both architectures when denoising $K$-bandlimited graph signals under some simplifying assumptions, and then, we numerically illustrated that the proposed architectures were also capable of denoising graph signals in more general settings.
Interestingly, the proposed GDec is not limited to the graph signal denoising task.
In fact, we successfully leveraged the GDec to learn a mapping from an input graph signal $\bbx$ defined on some graph $\ccalG_1$ to a graph signal $\bby$ defined on some graph $\ccalG_2$ by learning a latent space common to $\bbx$ and $\bby$, a problem intimately related with canonical correlation analysis. 

% O2: interpolation
Objective \textbf{(O2)}, approached in \cref{chap:interpolation}, considered signals with missing values and the goal was to interpolate the original signals under the assumption that they were DSGS.
Interpreting the observed (non-missing) values as samples gathered via an AGSS, first we studied the recovery of the original signal from the local observations when the seeding nodes were known.
Then, we considered more challenging sampling configurations where the seeding nodes, the diffusing filter, or both, were unknown.

% O3: filter ID
In the objective \textbf{(O3)}, the perturbations represented uncertainty in the edges of the observed graph and the aim was to develop a robust GF identification scheme from input-output observations.
To that end, \cref{chap:robust_filter_id} recast the true graph as an additional estimation variable and formulated an optimization problem that jointly estimated the GF and the true (unknown) GSO.
First, we focused on the case where only one GF needed to be estimated and, then, shifted to setups where multiple GFs have to be jointly identified.
The optimization problem was formulated completely in the vertex domain bypassing the error propagation associated with high-order matrix polynomials as well as the challenges of dealing with the influence of perturbations in the graph spectrum.
Since the optimization problem was non-convex, we blended techniques from alternating optimization and MM to obtain an iterative convex algorithm capable to find a stationary point in polynomial time.
This algorithm was later modified so that the scaling of the computational complexity with respect to the number of nodes in the graph is reduced.

% O4: NTI hidden
Finally, \cref{chap:nti_hidden} deals with objective \textbf{(O4)}, which focused on developing a joint network topology inference algorithm robust to the presence of hidden nodes from stationary observations.
To ensure the tractability of the problems, we assumed that the number of observed nodes was substantially larger than the number of hidden nodes, and formulated constrained optimization problems that accounted for the topological and signal constraints.
The cornerstone of the proposed algorithm was to exploit the block-matrix structure resulting from the presence of hidden variables. This structure allowed us to reformulate the classical definition of stationarity to account for the presence of hidden variables and, moreover, it revealed a column sparsity pattern on the matrices $\bbP^{(k)}$ (associated with hidden nodes) that we exploited to promote graph similarity between edges involved with hidden nodes via a group Lasso regularization.

\section{Future lines of research}\label{sec:future_work}
To conclude the document, we present several research directions to strengthen and grow the robust GSP methodology implemented in this thesis.
The suggested lines range from considering tractable generalizations of the schemes discussed in the previous chapters to more general and ambitious research directions. The lines in the first group are in general well-defined and likely to be successfully addressed in the short/medium term, while the ones in the second group can be understood as a research plan for the medium/long term.

% Incremental lines
\vspace{3mm}\noindent
\textbf{Generalizing the theoretical characterization of the denoising with GNNs to more general settings.}
Enhancing our current understanding of NNs (including GNNs) is a relevant ongoing research problem.
In this sense, generalizing the theoretical analysis from \cref{chap:denoising} to hold with more lenient assumptions would be beneficial.
To be precise, our analysis considered the graph signals being bandlimited, the noise to be white, and the graphs being drawn from an SBM. Considering more general noise distributions is a tractable task. Regarding the assumptions about the signals, dealing with other generative models (such as diffused or smooth signals) requires being able to establish a link between either the spectrum or the vertex distribution of the SBM and the signals at hand. Lastly, considering graphs beyond SBMs requires finding random graph models such that (i) the distance between $\bbA$ and $\ccalA$ goes to 0 as $N$ grows; and (ii) the sparsity pattern of $\ccalA$ is retained after the operation $\arccos(\ccalbH\ccalbH)$. This is highly non-trivial, but it may be doable in some cases (e.g., graphon models more general than SBMs, or non-graphon graphs with a strong clustered structure). Last but not least, application of our GNN architectures (along with the associated theoretical analysis) to problems other than graph denoising are also worth investigating. 
 
\vspace{3mm}\noindent
\textbf{Network topology inference with hidden nodes for multiple graphs.}
The growing adoption of graph-based approaches has revealed that in many datasets one needs more than one graph to describe the data. When the set of nodes does not change, multi-layer graphs are the preferred way to address this problem. However, the theoretical research in this area has been a bit slow and only recently, generalizations of graphical Lasso to the multi-layer graph case have been proposed. As a result, network topology inference algorithms for multi-layer graphs in the presence of perturbations (including hidden nodes) are mostly missing.
To develop such an algorithm in the graphical Lasso case, one should: (i) consider the effect of the hidden nodes in each of the graphs as well as (ii) incorporate a term that promotes similarity among the different graphs (according to the prior information and the application at hand). The challenge to achieve the latter is how to split the similarity metric among observed nodes and hidden nodes. Moreover, the goal would be not only designing an effective algorithm, but also characterizing analytically its performance. In this sense, the works in~\cite{ravikumar2011high,navarro2020joint} provide two promising starting points for the case where the observations are either Gaussian or stationary, respectively. Another equally interesting but more challenging line of research is to look at setups where not all the nodes participate in all relations (layers of the graph).

\vspace{3mm}\noindent
\textbf{Accounting for perturbations in the observed topology in other GSP problems.}
The presence of perturbations in the observed topology is critical in most GSP approaches, and hence, generalizing the ideas from \cref{chap:robust_filter_id} to other GSP problems like (blind) deconvolution, denoising, or sampling and reconstruction constitutes an interesting line of research.
In this sense, while the general approach in \cref{chap:robust_filter_id} can be preserved (i.e., defining the true GSO as an explicit estimation variable, formulate a problem that solves jointly over the variables of interest and the true GSO, and link the true GSO and the variables of interest via tractable constraints), the specific formulation and algorithmic approach will depend on the problem at hand.
On top of this, even for the GF identification task investigated in \cref{chap:robust_filter_id} (as well as for the generalizations suggested above), one can enhance the problem formulation by incorporating additional information of the GSO and/or its perturbations. A first avenue to pursue is considering alternative types of link perturbations as well as the presence of hidden nodes. A second avenue is to consider additional information about the true GSO (either in the form of a statistical prior or, e.g., having access to other graphs that we know to belong to the same class/family than the true one). The previous discussion illustrates that, by strengthening the estimation/denoising of GSO in the formulation, the GSP paradigm shifts from a two-step approach where first one observes/estimates a graph and then uses the graph to solve SP tasks to an encompassing methodology where the estimation/denoising of graph and the GSP task of interest are solved jointly. 

\vspace{3mm}\noindent
\textbf{Exploiting prior information about the graph topology.}
When dealing with a perturbed GSO $\barbS$, we have assumed that the magnitude of the perturbations was small, so that we could augment our formulation with a term that promotes similarity between $\bbS$ and $\barbS$.
However, there may be practical settings where leveraging only this information may not be enough. Examples include setups where the magnitude of the perturbations in $\barbS$ is large, or when $\barbS$ corresponds to a (possibly perturbed) subgraph of the whole $\ccalG$, so there is no information about the topology of the remaining graph.
In these cases, having access to prior information about the underlying graph is key to obtaining an accurate estimate of $\bbS$ that can be exploited. As briefly pointed out in this chapter, two potential avenues to achieve this are: (i) viewing the graph as a realization of a random graph model and incorporating some (tractable) statistical prior to the formulation; and (ii) assuming that we have access to other graphs that are similar / related / belong to the same family as the graph at hand. Regarding (i), postulation of meaningful and tractable probabilistic graphs is an entire area of research, so that the efforts there should be on identifying models that are particularly suited for the (vertex based, spectral based, polynomial based...) approach exploited by the GSP task at hand. Regarding (ii), our work in~\cite{rey2022enhanced} (which finds a graph with a density of motifs similar to that of another given graph and is able to deal with graphs with different numbers of nodes) provides an early example of how to achieve that. 
Intuitively, prior information of this sort could be included in our robust GSP approaches to enhance the resilience of the algorithms to perturbations and, moreover, the proposed similarity metric can be used in joint (multi-layer) network topology inference problems even when the sought graphs have different sizes.
Although interesting, employing this prior information about the density of motifs is a challenging task, because the similarity metric from~\cite{rey2022enhanced} is formulated in the spectral domain. Since most of the algorithms presented in this thesis were formulated in the node domain, they will need to be carefully redesigned to obtain convex solutions exploiting the similarity of motifs.

The lines above represent a few examples of open problems that can be addressed using the results of this thesis as starting point. There are a number of emerging areas in GSP (tensor models, time-varying graph signals, link-based signals and GSOs, categorical graph signals, Montecarlo graph and graph-signal based estimation schemes...) that are at their infancy and, as a result, are currently ignoring the effect of perturbations. Carefully updating and coming up with new formulations to render those novel schemes robust to perturbations in the data and the supporting graph will certainly be a problem of interest that will receive substantial attention in future years.

%%%%%%%%%%%%%%% Bibliografía %%%%%%%%%%%%%%%
\backmatter
% \nocite{*}
\bibliographystyle{IEEEtran}
\bibliography{biblio}
\clearpage

%%%%%%%%%%%%%%% Acrónimos %%%%%%%%%%%%%%%%%%
%\addcontentsline{toc}{chapter}{Acronyms}
%\printglossary[type=\acronymtype]

\end{document}